\pdfoutput=1
\documentclass[a4paper,BCOR=10mm,DIV=8]{scrbook}
\addtokomafont{disposition}{\normalfont}
\setkomafont{sectioning}{\rmfamily\bfseries\boldmath}

\KOMAoptions{fontsize=10pt}
\KOMAoptions{titlepage=true}
\KOMAoptions{twoside=true}
\KOMAoptions{parskip=half}
\KOMAoptions{cleardoublepage=empty}
\KOMAoptions{open=right}
\KOMAoptions{numbers=noendperiod}
\KOMAoptions{chapterprefix=true}
\KOMAoptions{appendixprefix=true}
\KOMAoptions{headinclude}

\usepackage[english,ngerman]{babel}
\usepackage[T1]{fontenc}
\usepackage[utf8]{inputenc}
\usepackage{xcolor}
\usepackage{paralist}
\usepackage{amssymb}
\usepackage{amsmath}
\usepackage{amsthm}
\usepackage[automark,headsepline]{scrpage2}
\usepackage[avantgarde]{quotchap}
\usepackage{etoolbox}

\usepackage{bookmark,hyperref}
\usepackage[firstpage]{draftwatermark}
\usepackage{cite}
\usepackage[square,sort,comma,numbers]{natbib}
\usepackage{bibentry}
\newcommand{\bibverse}[1]{\begin{verse} \bibentry{#1}. \end{verse}}

\usepackage{epigraph}
\usepackage{epipart}
\usepackage[colorinlistoftodos]{todonotes}

\providebool{draftmode}
\providebool{draftwatermark}

\setbool{draftmode}{false}
\setbool{draftwatermark}{false}
\newcommand{\thetitle}{On the Cost of Concurrency in \\Transactional Memory}
\newcommand{\theauthor}{Srivatsan Ravi (M.E.)}
\newcommand{\degree}{Doktor der Ingenieurwissenschaften (Dr.-Ing.)}
\newcommand{\defensedate}{18 June, 2015}
\newcommand{\thesisyear}{2015}
\newcommand{\phdboard}{
  Vorsitzender:    & Prof. Uwe Nestmann, Ph.\,D., TU Berlin \\
  Gutachterin:     & Prof.\ Anja Feldmann, Ph.\,D., TU Berlin \\
  Gutachter:     & Prof.\ Petr Kuznetsov, Ph.\,D., T\'el\'ecom ParisTech \\
  Gutachterin:       & Prof. Hagit Attiya, Ph.\,D., The Technion\\
  Gutachter:       & Prof. Rachid Guerraoui, Ph.\,D., EPFL\\
  Gutachter:       & Prof. Michel Raynal, Ph.\,D., INRIA, Rennes
}

\usepackage{enumitem}
\usepackage{multicol}
\usepackage{float}
\floatstyle{boxed} 
\restylefloat{figure}

\usepackage{graphicx}
\usepackage{subfig}

\usepackage{tikz}
\usepackage{tabularx}
\usepackage[top=1.25in, bottom=1.25in, left=1in, right=1in]{geometry}
\newif\ifcode
\codefalse
\codetrue
\usepackage{macros}
\usepackage{algcompatible}
\ifcode
\usepackage[ruled,chapter]{algorithm}
\usepackage[nottoc,numbib]{tocbibind}
\usepackage{algpseudocode}
\usepackage{ioa_code}
\newtheorem{theorem}{Theorem}
\numberwithin{theorem}{chapter}

\newtheorem{claim}[theorem]{Claim}

\newtheorem{proposition}{Proposition}
\numberwithin{proposition}{chapter}

\newtheorem{corollary}[theorem]{Corollary}
\newtheorem{definition}{Definition}
\numberwithin{definition}{chapter}
\newtheorem{lemma}[theorem]{Lemma}
\newtheorem{observation}[theorem]{Observation}

\newcounter{linenumber}

\def\S{\ensuremath{\mathcal{S}}}

\def\N{\ensuremath{\mathcal{N}}}

\def\M{\ensuremath{\mathcal{M}}}

\def\T{\ensuremath{\mathcal{T}}}
\def\X{\ensuremath{\mathcal{X}}}

\def\Nat{\ensuremath{\mathbb{N}}}

\def\TS{\mathit{TS}}
\def\shared{\mathit{shared}}
\def\exclusive{\mathit{exclusive}}

\newcommand{\id}[1]{\mbox{\textit{#1}}}
\newcommand{\LS}{LS}

\newcommand{\true}{\mathit{true}}
\newcommand{\false}{\mathit{false}}

\newcommand{\remove}[1]{}

\newcommand{\Wset}{\textit{Wset}}
\newcommand{\Rset}{\textit{Rset}}
\newcommand{\Dset}{\textit{Dset}}

\newcommand{\txns}{\textit{txns}}

\newcommand{\Read}{\textit{read}}
\newcommand{\Write}{\textit{write}}
\newcommand{\TryC}{\textit{tryC}}
\newcommand{\TryA}{\textit{tryA}}
\newcommand{\ok}{\textit{ok}}
\newcommand{\LL}{\ms{LL}}

\newcommand{\RNum}[1]{\uppercase\expandafter{\romannumeral #1\relax}}

\newcommand{\VCSmakeinfo}{Use 'make' to get VCS information.}
\newcommand{\VCSversion}{\VCSinfo}
\newcommand{\VCSdate}{\VCSmakeinfo}

\newcommand{\VCSauthor}{\VCSmakeinfo}
\newcommand{\VCSlog}{\VCSmakeinfo}
\newcommand{\VCSbranch}{\VCSmakeinfo}
\InputIfFileExists{vcs}{\input{vcs}}{}

\ifbool{draftmode}{
  \newcommand{\VCSline}{\footnotesize \sffamily \color{gray!90} %
    {\bf Preprint version compiled on \today. Not intended for publication!} \\%
    Revision:~\VCSversion~Branch:~\VCSbranch~Commit:~\VCSauthor~(\VCSdate)\\
    Log: \VCSlog
  }
  \clearscrheadfoot
  \cfoot[\VCSline]{\VCSline}
  \ofoot[\thepage]{\thepage}
  \ifbool{draftwatermark}{
	\SetWatermarkLightness{0.92}
	\SetWatermarkScale{5}
	\SetWatermarkAngle{45}
	\SetWatermarkText{Draft}
  }{}
}{
  \SetWatermarkText{}
}

\begin{document}
\bibliographystyle{abbrv}
\selectlanguage{english}
\ohead{\headmark}
\title{
  \thetitle
  \ifbool{draftmode}{
    \\\Large \bf \sl Entwurf / Draft from \VCSdate\\
    Version \VCSversion~(compiled \today)\\
  }{}
}

\titlehead{
    \centering
    \includegraphics[width=3cm]{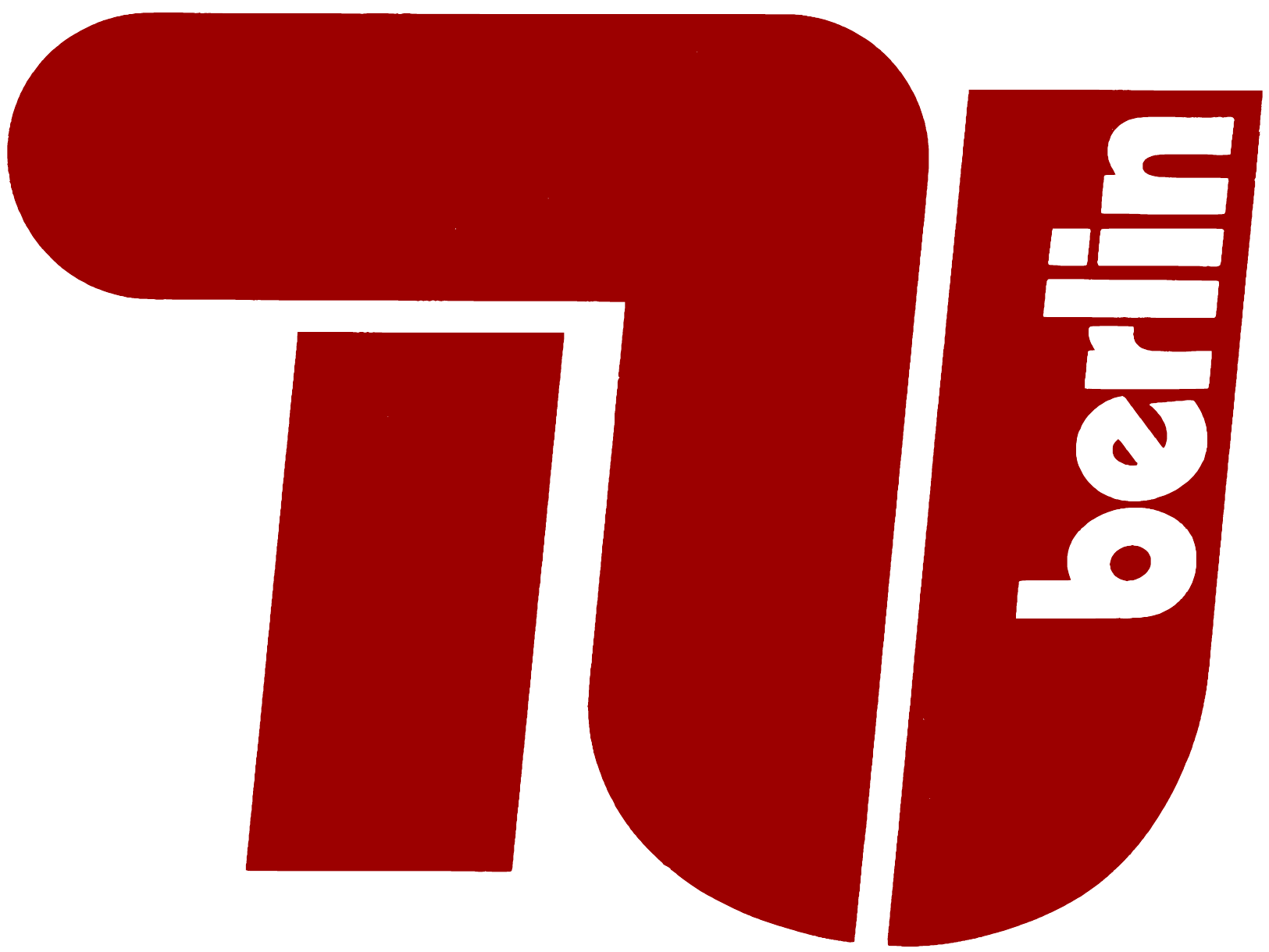}\\
    \vspace{0.8em}
    \textsc{
        \LARGE Technische Universit\"at Berlin\\
        \large Fakult\"at f\"ur Elektrotechnik und Informatik\\
               Lehrstuhl f\"ur Intelligente Netze\\
               und Management Verteilter Systeme
    }
    \vspace{2\baselineskip}
}

\author{
    \large vorgelegt von\\
    \Large \theauthor\\
    \vspace{2\baselineskip}\\
    \large von der Fakult\"at IV -- Elektrotechnik und Informatik\\
           der Technischen Universit\"at Berlin\\
           zur Erlangung des akademischen Grades\\
           \textsc{\degree}\\
           genehmigte Dissertation
}

\publishers{
    \vspace{1\baselineskip}
    \large
    \textbf{Promotionsausschuss:}\\
    \vspace{0.5\baselineskip}
    \begin{tabular}{ll}
        \phdboard
     \end{tabular}\\
     \vspace{1\baselineskip}
     Tag der wissenschaftlichen Aussprache: \defensedate\\
    \ifbool{draftmode}{
        \vspace{2\baselineskip}
     }{
        \vspace{3\baselineskip}
     }

     Berlin \thesisyear\\
     D 83
}
\date{}
\pagestyle{empty}
\maketitle
\pagestyle{scrheadings}

\chapter*{Eidesstattliche Erklärung}

\vspace*{\stretch{1}}

Ich versichere an Eides statt, dass ich diese Dissertation selbständig
verfasst und nur die angegebenen Quellen und Hilfsmittel verwendet habe.

\vspace{1cm}

\begin{flushright}
  \begin{tabular}{p{2cm}p{6cm}}
    \hline
    Datum & \theauthor \\
  \end{tabular}
\end{flushright}
\cleardoublepage
\section*{Abstract}
Current general-purpose CPUs are \emph{multicores}, offering multiple computing units within
a single chip.
The performance of programs on these architectures, however, does not necessarily increase proportionally
with the number of cores.
Designing concurrent programs to exploit these multicores emphasizes the need for
achieving efficient \emph{synchronization} among \emph{threads} of computation. 
When there are several threads that \emph{conflict} on the same data,
the threads will need to coordinate their actions for ensuring correct program behaviour.

Traditional techniques for synchronization are based on \emph{locking} that provides threads with exclusive access to shared data.
\emph{Coarse-grained} locking typically forces threads to access large amounts of data sequentially and, thus, 
does not fully exploit hardware concurrency.
Program-specific \emph{fine-grained} locking or \emph{non-blocking} (\emph{i.e.}, not using locks) 
synchronization, on the other hand, is 
a dark art to most programmers and trusted to the wisdom of a few computing experts.
Thus, it is appealing to seek a middle ground between these two extremes: a synchronization mechanism that
relieves the programmer of the overhead
of reasoning about data conflicts that may arise from concurrent operations without severely 
limiting the program's performance.
The \emph{Transactional Memory (TM)} abstraction is proposed as such a mechanism: 
it intends to combine an easy-to-use programming interface
with an efficient utilization of the concurrent-computing abilities provided by multicore architectures.
TM allows the programmer to \emph{speculatively} execute sequences of shared-memory
operations as \emph{atomic transactions} with \emph{all-or-nothing} semantics: 
the transaction can either \emph{commit}, in which case it appears as executed sequentially,
or \emph{abort}, in which case its update operations do not take effect.
Thus, the programmer can design software having only sequential
semantics in mind and let TM take care, at run-time, of resolving the
conflicts in concurrent executions.

Intuitively, we want TMs to allow for as much \emph{concurrency} as possible: in the absence of severe data conflicts,
transactions should be able to progress in parallel.
But what are the inherent costs associated with providing high degrees of concurrency in TMs? 
This is the central question of the thesis.

To address this question, we first focus on the \emph{consistency} criteria that must be satisfied by a TM implementation.
We precisely characterize what it means for a TM implementation to be
\emph{safe}, \emph{i.e.}, to ensure that the view of \emph{every} transaction could have been observed 
in some sequential execution.
We then present several lower and upper bounds on the complexity of three classes of safe TMs:
\emph{blocking} TMs that allow transactions to delay or abort due to \emph{overlapping} transactions, 
\emph{non-blocking} TMs which adapt to \emph{step contention} by ensuring that a transaction not encountering steps
of overlapping transactions must commit, 
and \emph{partially non-blocking} TMs that provide strong non-blocking guarantees (\emph{wait-freedom})
to only a subset of transactions.
We then propose a model for \emph{hybrid} TM implementations that complement hardware transactions with software transactions.
We prove that there is an inherent trade-off on the degree of concurrency
allowed between hardware and software transactions and the costs introduced on the hardware.
Finally, we show that \emph{optimistic} synchronization techniques based on speculative executions
are, in a precise sense, better equipped to exploit concurrency
than inherently \emph{pessimistic} techniques based on locking.

\section*{Zusammenfassung}
\selectlanguage{ngerman}
Aktuelle Allzweck-CPUs haben \emph{mehrere} Rechenkerne innerhalb
eines einzelnen Chipsatzes.
Allerdings erhöht sich die Leistung der Programme auf diesen Architekturen nicht notwendigerweise proportional in der 
Anzahl der Kerne.
Das Entwerfen nebenläufiger Programme um diese Multicores zu nutzen,
erfordert die Überwindung einiger nicht-trivialer Herausforderungen;
die wichtigste ist, eine effiziente \emph{Synchronisierung} der \emph{Threads} der Berechnung herzustellen.
Greifen mehrere Threads gleichzeitig auf dieselben Daten zu, müssen diese ihre Aktionen koordinieren, um 
ein korrektes Programmverhalten zu gewährleisten.

Die traditionelle Methode zur Synchronisierung ist "Locking", welches jeweils nur einem einzelnen Thread Zugriff auf gemeinsam genutzten Daten gewährt.
Bei \emph{grobkörnigem} "Locking" erfolgt der Zugang zu einer großen Menge von Daten meist seriell, 
sodass die Hardware-Parallelität nich in vollem Unfang ausgenutzt wird.
Auf der anderen Seite stellt programmspezifisches \emph{feinkörniges} Locking, 
oder auch nicht-blockierende (d.h. keine Locks benutzende) Synchronisierung, 
eine dunkle Kunst für die meisten Programmierer dar, welche auf die Weisheit weniger Computerexperten vertraut.
So ist es angebracht, einen Mittelweg zwischen diesen beiden Extremen zu suchen: einen Synchronisierungsmechanismus, 
der den Programmierer bezüglich der \emph{Datenkonflikte}, die aus gleichzeitigen Operationen entstehen, entlastet, 
ohne jedoch die Leistung des Programms zu stark zu beeinträchtigen. 
Die \emph{Transactional Memory (TM)} Abstraktion wird als solcher Mechanismus vorgeschlagen:
ihr Ziel ist es, eine einfach zu bedienende Programmierschnittstelle mit einer effizienten Nutzung der gleichzeitigen 
Computing-Fähigkeiten von Multicore-Architekturen zu kombinieren. 
TM erlaubt es dem Programmierer, Sequenzen von Operationen auf dem gemeinsamen Speicher als \emph{atomare Transaktionen} 
mit \emph{Alles-oder-Nichts} Semantik zu erklären: Die Transaktion wird entweder \emph{übergeben}, 
und somit sequentiell ausgeführt, oder \emph{abgebrochen}, 
sodass ihre Operationen nicht durchgeführt werden.
Dies ermöglicht dem Programmierer, Software mit nur sequentieller Semantik zu konzipieren, 
und die aus gleichzeitger Ausführung entstehenden Konflikte TM zu überlassen.

Intuitiv sollen die TMs so viel Nebenläufigkeit wie möglich berücksichtigen: Falls keine Datenkonflikte vorhanden sind, 
sollen alle Transaktionen parallel ausgeführt werden. 
Gibt es in TMs Kosten, die durch diesen hohen Grad an \emph{Nebenläufigkeit} entstehen? 
Das ist die zentrale Frage dieser Arbeit.

Um diese Frage zu beantworten, konzentrieren wir uns zunächst auf das Kriterium der \emph{Konsistenz}, 
welche von der TM-Implementierung erfüllt werden muss.
Wir charakterisieren auf präzise Art, was es für eine TM-Implementierung heißt, 
\emph{sicher} zu sein, d.h. zu gewährleisten, dass die Sicht \emph{einer jeden} 
Transaktion auch von einer sequentiellen Transaktion hätte beobachtet 
werden können.
Danach präsentieren wir mehrere untere und obere Schranken für die Komplexität dreier Klassen von sicheren TMs:
\emph{blockierende} TMs, die Blockierungen oder Abbrüche der Transaktionen erlauben, sollten diese sich überlappen,
\emph{nicht-blockierende} TMs die einen schrittweisen Zugriffskonflikt berücksichtigen, d.h. 
Transaktionen, die keinen Zugriff überlappender anderer Transaktionen beobachten, müssen übergeben, und
\emph{partiell nicht-blockierende} TMs, die nur für eine Teilmenge von Transaktionen nicht-blockierend sind.
Wir schlagen daraufhin ein Modell für \emph{hybride} TM-Implementierungen vor, welches die 
Hardware Transaktionen mit Software Transaktionen ergänzt. Wir beweisen, dass es eine inherente Trade-Off
zwischen Grad der erlaubten Nebenläufigkeit 
zwischen Hard- und Software Transaktionen und den Kosten der Hardware gibt.
Schlussendlich beweisen wir, dass \emph{optimistische}, auf \emph{spekulativen} Ausführungen 
basierende, Synchronisierungstechniken, in einem präzisen Sinne,  besser geeignet sind um 
Nebenläufigkeit auszunutzen als \emph{pessimistische} Techniken, die auf "Locking" 
basieren.
\selectlanguage{english}
\clearpage
\tableofcontents
\section*{Acknowledgements}
In his wonderfully sarcastic critique of the scientific community in \emph{His Master's Voice}, the great Polish writer
Stanis\l aw Lem refers to a \emph{specialist} as a \emph{barbarian whose ignorance is not well-rounded}.
Writing a Ph.\,D. thesis is essentially an attempt at becoming a specialist on some topic; whether this thesis on
Transactional Memory makes me one is a questionable claim, but I am culturedly not totally ignorant, I think.
The thesis itself was a long time in the making and would not have been possible without
the wonderful support and gratitudes I have received these past four years.

My advisors Anja Feldmann and Petr Kuznetsov guided me throughout my Ph.\,D. term.

A fair amount of whatever good I have learnt these past few years, both scientifically and meta-scientifically, I owe it
to Petr. He taught me, by example, what it takes to achieve nontrivial scientific results. He spent several hours schooling me
when I had misunderstood some topic and as such suffered the worst of my writing, especially in the first couple of years.
He was hard on me when I did badly, but always happy for me when I did well. Apart from being a deep thinker and a brilliant
researcher, his scientific integrity and mental discipline have indelibly made me a better human being and student of science.
At a personal level, I wish to thank him and his family for undeserved kindness shown to me over the years.

I would like to thank Anja for the extraordinary amount of freedom she gave me to pursue my own research and the trust
she placed in me, as she does in all her students.

I am especially grateful to Robbert Van Renesse and Bryan Ford, who gave me a taste for independent research and
in many ways, helped shape the course of my graduate career.

I am of course extremely grateful to all my co-authors who allowed me to include content, written in conjunction with them, 
in the thesis. So special thanks to Dan Alistarh, Hagit Attiya, Vincent Gramoli, Sandeep Hans, Justin Kopinsky, Petr Kuznetsov
and Nir Shavit. 

The results in Chapter~\ref{ch:pc1} were initiated during a memorable visit to the Technion, Haifa
in the Spring of '12. I am thankful to have been hosted by Hagit Attiya and to have had the chance to work with her 
and Sandeep Hans. I am also very grateful to David Sainz for taking time off and introducing me to some beautiful parts of Israel.

Chapter~\ref{ch:p4c4} essentially stemmed from a visit to MIT in the summer of '13. I am very grateful to
Dan Alistarh and Nir Shavit for hosting me. Special thanks to Justin Kopinsky, who helped keep our discussions alive
during our lengthy dry spells when we were seemingly spending all our time thinking about the problem, but without
producing any tangible results.

Chapter~\ref{ch:p2c1} represents the most excruciatingly painful part of the thesis purely in terms of the number of iterations
the paper based on this chapter went through. Yet, it was a procedure from which I learnt a lot and I am very thankful to Vincent
Gramoli, who initiated the topic during my visit to EPFL in Spring '11.

In general, I have benefitted immensely from just talking to researchers in distributed computing during
conferences, workshops and research visits. 
These include Yehuda Afek, Panagiota Fatourou, Rachid Guerraoui, Maurice Herlihy, Victor Luchangco, Adam Morrison and 
Michel Raynal. Also, great thanks to the anonymous reviewers of my paper submissions whose critiques and comments helped improve the contents
of the thesis.

Back here in Berlin, so many of my INET colleagues have shaped my thought processes and enriched my experience in grad school.
Thanks to Stefan Schmid, whose ability to execute several tasks concurrently with minimal synchronization
overhead, never ceases to amaze anyone in this group. I am also very grateful to Anja, Petr and Stefan for
allowing me to be a Teaching Assistant in their respective courses.
Apart from being great friends, Felix Poloczek and Matthias Rost have been wonderful office mates and indulged my random discussions about
life and research. Arne Ludwig and Carlo Fürst have been great friends; Arne, thanks for all the football discussions
and Carlo, for exposing me to some social life. Thanks to Dan Levin, who has been a great friend and 
always been there to motivate and give me a fillip whenever I needed it.
Ingmar Poese has been a wonderful friend as well as a constant companion to the movies.
Franziska Lichtblau, Enric Pujol, Philipp Richter and others have been willing companions in ordering several
late night dinners at the lab. Thomas Krenc and Philipp Schmidt were great co-TA's.
Great thanks to all the other current and past members of INET: our system admin Rainer May, Marco Canini, Damien Foucard, 
Juhoon Kim, Gregor Schaffrath, Julius Schulz-Zander, Georgios Smaragdakis, Florian Streibelt, Lalith Suresh, Steve Uhlig and all the other members I have missed.
Special thanks to our group secretaries, Birgit Hohmeier-Touré and Nadine Pissors, without whom there would be absolute chaos.

I would like to thank my flat mates of the last two years: Ingmar, Jennifer, Jose, Lily and Mathilda.
Lastly, thanks to my family and friends outside of the academic sphere for tolerating me all these years.

\chapter{Introduction}
\label{ch:intro}
While the performance of programs would increase proportionally with the performance of a \emph{singlecore}
CPU, the performance of programs on \emph{multicore} CPU architectures, however, does not necessarily increase proportionally
with the number of cores. In order to exploit these multicores, the amount of concurrency provided by programs will need to
increase as well.
Designing concurrent programs that exploit the hardware concurrency provided by modern multicore CPU architectures requires
achieving efficient \emph{synchronization} among \emph{threads} of computation. 
However, due to the \emph{asynchrony} resulting from the CPU's context switching and
scheduling policies, it is hard to specify reasonable bounds
on relative thread speeds. This makes the design of efficient and correct concurrent programs a difficult task. 
The \emph{Transactional Memory (TM)} abstraction~\cite{HM93,ST95} is a synchronization mechanism  
proposed as a solution to this problem: 
it combines an easy-to-use programming interface
with an efficient utilization of the concurrent-computing abilities provided by multicore architectures.
This chapter introduces the
TM abstraction and presents an overview of the thesis.
\section{Concurrency and synchronization}
\label{sec:introconcur}
In this section, we introduce the challenges of concurrent computing and overview the drawbacks associated with
traditional synchronization techniques.
\subsection{Concurrent computing overview}
\vspace{1mm}\noindent\textbf{Shared memory model.}
A \emph{process} represents a thread of computation that is provided with its own private memory
which cannot be accessed
by other processes. However, these independent processes will have to synchronize their actions
in an asynchronous environment in order to 
\emph{implement} a user application, which they do by communicating via the CPU's \emph{shared memory}.

In the shared memory model of computation, 
processes communicate by \emph{reading} and \emph{writing} to a fragment of the shared memory,
referred to as a \emph{base object}, in a single \emph{atomic} (\emph{i.e.}, indivisible) instruction. 
Modern CPU architectures additionally allow processes
to invoke certain powerful atomic \emph{read-modify-write (rmw)} instructions~\cite{Her91}, which allow processes
to write to a base object subject to the check of an invariant.
For example, the \emph{compare-and-swap} instruction is a rmw instruction that 
is supported by most modern architectures: it takes as input $\langle \ms{old},\ms{new} \rangle$ and
atomically updates the value of a base object to $\ms{new}$ and returns $\true$ \emph{iff} 
its value prior to applying the instruction 
is equal to $\ms{old}$; otherwise it returns $\false$.

\vspace{1mm}\noindent\textbf{Concurrent implementations.}
A \emph{concurrent implementation} provides
each process with an algorithm to apply CPU instructions on the
shared base objects for the \emph{operations} of the user application.
For example, consider the problem of implementing a concurrent \emph{list-based set}~\cite{HS08-book}. the \emph{set}
abstraction implemented as a \emph{sorted linked list} supporting operations $\lit{insert}(v)$, $\lit{remove}(v)$ and 
$\lit{contains}(v)$; $v\in \mathbb{Z}$.
The set abstraction stores a set of integer values,
initially empty. 
The update operations, $\lit{insert}(v)$ and $\lit{remove}(v)$, return
a boolean response, $\true$ if and only if $v$ is absent (for
$\lit{insert}(v)$) or present (for $\lit{remove}(v)$) in the list.  
After $\lit{insert}(v)$ is complete, $v$ is present in the list, and 
after $\lit{remove}(v)$ is complete, $v$ is absent in the list.
The $\lit{contains}(v)$ returns a boolean, $\true$ if and
only if $v$ is present in the list.
A concurrent implementation of the list-based set is simply an \emph{emulation} of the set abstraction that is realized
by processes applying the available CPU instructions on the underlying base objects.

\vspace{1mm}\noindent\textbf{Safety and liveness.}
What does it mean for a concurrent implementation to be correct?
Firstly, the implementation must satisfy a \emph{safety property}: 
there are no bad reachable states in any \emph{execution} of the implementation.
Intuitively, we characterize safety for a concurrent implementation of a data abstraction by verifying if the responses returned
in the concurrent execution may have been observed
in a \emph{sequential execution} of the same.
For example, the safety property for a concurrent list-based set implementation stipulates that
the response of the set operations in a concurrent execution is consistent with some sequential execution
of the list-based set. 
However, a concurrent set implementation that does not return any response trivially
ensures safety; thus, the implementation must satisfy some \emph{liveness property} specifying the conditions
under which the operations must return a response.
For example, one liveness property we may wish to impose on the concurrent list-based set is \emph{wait-freedom}: 
every process completes the operations it invokes irrespective of the behaviour of other processes.

As another example, consider 
the \emph{mutual exclusion} problem~\cite{Dijkstra1}
which involves sharing some critical data resource among processes.
The safety property for mutual exclusion stipulates that at most 
one process has access to the resource in any execution, in which case, we say that the process is
inside the \emph{critical section}.
However, one may notice that an implementation which ensures that no process ever enters the critical section
is trivially safe, but not very useful. 
Thus, the mutual exclusion implementation must satisfy some liveness property specifying the conditions under which
the processes must eventually enter the critical section. 
For example, we expect that the implementation is \emph{deadlock-free}: if every process
is given the chance to execute its algorithm, some process will enter the critical section.
In contrast to safety, a liveness property can be violated
only in an infinite execution, \emph{e.g.}, by no process ever entering the critical section. 

In shared memory computing, we are concerned with deriving
concurrent implementations with strong safety and liveness properties, thus emphasizing
the need for efficient synchronization among processes.
\subsection{Synchronization using locks}
\label{sec:locks}
A \emph{lock} is a concurrency abstraction that implements mutual exclusion and 
is the traditional solution for achieving synchronization among processes.
Processes \emph{acquire} a lock prior to executing code inside the critical section and \emph{release}
the lock afterwards, thereby allowing other processes to modify the data accessed by the code within the critical section.
In essence, after acquiring the lock, the code within the critical section can be executed atomically.
However, \emph{lock-based} implementations suffer from some fundamental drawbacks.

\vspace{1mm}\noindent\textbf{Ease of designing lock-based programs.}
Ideally, to reduce the programmer's burden, we would like to take any sequential implementation and transform
it to an efficient concurrent one with minimal effort.
Consider a simple locking protocol that works for most applications: \emph{coarse-grained} locking which typically
serializes access to a large amount of data. Although trivial for the programmer to implement,
it does not exploit hardware concurrency. In contrast, \emph{fine-grained} locking may exploit concurrency better,
but requires the programmer to have a good understanding of the data-flow relationships in the application and
precisely specify which locks provide exclusive access to which data.

For example, consider the problem of implementing a concurrent list-based set.
The \emph{sequential} implementation of the list-based set uses a sorted linked list data structure 
in which each data item (except the \emph{tail} of the list) maintains a next field to provide a pointer to the
successor.
Every operation ($\lit{insert}$, $\lit{remove}$ and $\lit{contains}$) invoked with a parameter $v\in \mathbb{Z}$ traverses the list
starting from the \emph{head} up to the data item storing value $v' \geq v$.
If $v'=v$, then $\lit{contains}$ returns $\true$, $\lit{remove}(v)$ unlinks the 
corresponding element and returns $\lit{true}$ and $\lit{insert}(v)$ returns $\false$. Otherwise, 
$\lit{contains}(v)$ and $\lit{remove}(v)$ return $\lit{false}$, while $\lit{insert}(v)$ adds a new data item
with value $v$ to the list and returns $\true$.

Given such a sequential implementation, we may derive a coarse-grained implementation of the list-based set by having 
processes acquire a lock on the head
of the list, thus, forcing one operation to complete before the next starts. Alternatively, a fine-grained protocol 
may involve acquiring locks \emph{hand-over-hand}~\cite{BS88}: a process holds 
the lock on at most two adjacent data items of the list.
Yet, while such a protocol produces a correct set implementation~\cite{ARR10}, it is not a universal strategy that applies
to other popular data abstractions like \emph{queues} and \emph{stacks}.

\vspace{1mm}\noindent\textbf{Composing lock-based programs.}
It is hard to compose smaller atomic operations based on locks to produce a larger atomic operation without affecting 
safety~\cite{ST95,HS08-book}.
Consider the \emph{fifo queue} abstraction supporting the $\lit{enqueue}(v)$; $v\in \mathbb{Z}$
and $\lit{dequeue}$ operations.
Suppose that we wish to solve the problem of atomically dequeuing from a queue $Q_1$ and 
enqueuing the item returned, say $v$, to a queue $Q_2$. 
While the individual actions of dequeuing from $Q_1$ and enqueuing $v$ to $Q_2$ may be atomic, we wish
to ensure that the combined action is atomic: no process must observe the absence of $v$ or that it is present in both
$Q_1$ and $Q_2$. 
A possible solution to this specific problem is to force a process attempting atomic
modification of $Q_1$ and $Q_2$ to acquire a lock.
Firstly, this requires prior knowledge of the identities of the two queue instances.
Secondly, this solution does not exploit hardware concurrency since the lock itself becomes a concurrency bottleneck.
Moreover, imagine that processes $p_1$ and $p_2$ need to acquire two locks $L_1$ and $L_2$
in order to atomically modify a set of queue instances.
Without imposing a pre-agreed upon order on such lock acquisitions, there is the
possibility of introducing \emph{deadlocks} where processes wait infinitely long without completing their operations. 
For example, imagine the following concurrency scenario: 
process $p_1$ (and resp., $p_2$) holds the lock $L_1$ (and resp., $L_2$) and 
attempts to acquire the lock $L_2$ (and resp., $L_1$). Thus, process $p_1$ (and resp., $p_2$)
waits infinitely long for $p_2$ (and resp., $p_1$) to complete its operation.
%
\subsection{Non-blocking synchronization}
It is impossible to derive lock-based implementations that provide \emph{non-blocking} liveness 
guarantees, \emph{i.e.}, some process completes
its operations irrespective of the behaviour of other processes.
In fact, even the weak non-blocking liveness property of \emph{obstruction-freedom}~\cite{AGHK09}
cannot be satisfied by lock-based implementations: a process must complete its operation if it eventually runs \emph{solo}
without interleaving events of other processes.

Concurrent implementations providing non-blocking liveness properties are appealing in practice since they overcome
problems like deadlocks and \emph{priority inversions}~\cite{dobbs-pinv} inherent to lock-based implementations.
Thus, non-blocking (without using locks) solutions using \emph{conditional} rmw instructions like compare-and-swap
have been proposed as an alternative to lock-based implementations. 
However, as with fine-grained locking, implementing
correct non-blocking algorithms can be hard and requires hand-crafted problem-specific strategies.
For example, the state-of-the-art list-based set implementation by Harris-Michael~\cite{HS08-book,harris-set,michael-set}
is non-blocking: the $\lit{insert}$ and $\lit{remove}$ operations, as they traverse the list, 
\emph{help} concurrent operations to physically remove data items (using compare-and-swap) that are
logically deleted, \emph{i.e.}, ``marked for deletion''. 
But one cannot employ an identical algorithmic technique for implementing a non-blocking queue~\cite{MS96}, whose
semantics is orthogonal to that of the set abstraction.
Moreover, addressing the compositionality issue, as with lock-based solutions, requires ad-hoc strategies that are
not easy to realize~\cite{HS08-book}.
\begin{figure*}
\input{files/l1-example}
\caption{Transforming a sequential implementation of the list-based set to a TM-based concurrent one
\label{fig:list-tm}} 
\end{figure*}
\section{Transactional Memory (TM)}
\label{sec:introtm}
\emph{Transactional Memory} (TM)~\cite{HM93,ST95}
addresses the challenge of resolving \emph{conflicts} (concurrent reading and writing to the same data) 
in an efficient and safe manner by offering a simple interface in which sequences
of shared memory operations on \emph{data items} can be declared as \emph{optimistic transactions}.
The underlying idea of TM, inspired by databases~\cite{gray1992}, is to treat each
transaction as \emph{atomic}: a transaction may either \emph{commit}, in
which case it appears as executed sequentially, or \emph{abort}, in
which case none of its update operations \emph{take effect}.
Thus, it enables the programmer to design software applications having only sequential
semantics in mind and let TM take care of \emph{dynamically} handling the
conflicts resulting from concurrent executions at run-time.

A TM \emph{implementation} provides processes with algorithms for implementing \emph{transactional operations} such as
\emph{read}, \emph{write}, \emph{tryCommit} and \emph{tryAbort} on \emph{data items} using base objects.
TM implementations typically ensure that 
all committed transactions appear to execute sequentially in some
total order respecting the timing of non-overlapping transactions.
Moreover, unlike database transactions, intermediate states witnessed by the read operations of an incomplete transaction
may affect the user application.
Thus, to ensure that the TM implementation does not export any pathological executions,
it is additionally expected that every transaction (including aborted and incomplete ones) must return responses that is 
consistent with some \emph{sequential} execution of the TM implementation.

In general, given a sequential implementation of a data abstraction, 
a corresponding \emph{TM-based} concurrent one encapsulates the sequential (high-level) operations within a
transaction.
Then, the TM-based concurrent implementation of the data abstraction replaces each read and write of a data item with the
transactional read and write implementations, respectively.
If the transaction commits, then the result of the high-level operation is
returned to the application. Otherwise, one of the transactional operations may be aborted,
in which case, the write operations performed by the transaction do not take effect and the high-level 
operation is typically re-started with a new transaction.

To illustrate this, we refer to the sequential implementation of the $\lit{remove}$ operation of list-based set depicted in 
Algorithm~\ref{alg:seqremove} of Figure~\ref{fig:list-tm}.
In a TM-based concurrent implementation of the list-based set (Algorithm~\ref{alg:tm-remove}), 
each read (and resp. write) operation performed by $\lit{remove}(v)$ on a 
data item $X$ of the list is replaced with 
$\lit{tx-read}(X)$ (and resp., $\lit{tx-write}(X,\ms{arg})$).
$\lit{tx-read}(X)$ returns the value of the data item $X$ or aborts the transaction while $\lit{tx-write}(X,\ms{arg})$
writes the value $\ms{arg}$ to $X$ or aborts the transaction.
Finally, the process attempts to commit the transaction by invoking the tryCommit operation. If the tryCommit is successful, 
the response of $\lit{remove}(v)$ is returned; otherwise a failed response (denoted $\bot$) is returned, in which case, 
the write operations performed by the transaction are ``rolled back''.

Intuitively, it is easy to understand how TM simplifies concurrent programming.
Deriving a TM-based concurrent implementation of the list-based set simply requires encapsulating
the operations to be executed atomically within a transaction using an \emph{atomic} delimiter~\footnote{Different compilers
may use different names for the delimiter; in \emph{GCC}, it is transaction\_atomic~\cite{gcc-tm}}.
The underlying TM implementation endeavours to dynamically execute the transactions by resolving the conflicts
that might arise from processes reading and writing to the same data item at run-time.
Intuitively, since the TM implementation enforces a strong safety
property, the resulting list-based implementation is also safe: the responses of its operations are consistent with
some sequential execution of the list-based set.

One may view TM as a \emph{universal construction}~\cite{Her91,EFKMT12,Fatourou11, Barnes93,Anderson99-multi}
that accepts as input the operations of a sequential implementation
and strives to execute them concurrently.
Specifically, TM is designed to work in a dynamic environment
where neither the sequence of operations nor the data items accessed by a transaction are known a priori.
Thus, the response of a read operation performed by a transaction is returned immediately to the application,
and the application determines the next data item that must accessed by the transaction.

TM-based implementations overcome the drawbacks of traditional synchronization techniques based
on locks and compare-and-swap.
Firstly, the TM interface places minimal overhead on the programmer: using TM only requires encapsulating
the sequential operations within transactions and handling an exception should the transaction be aborted.
Secondly, the ability to execute multiple operations atomically allows TM-based implementations to seamlessly compose
smaller atomic operations to produce larger ones.
For example, suppose that we wish to atomically $\lit{dequeue}$ from a queue $Q_1$ and $\lit{enqueue}(v)$ in queue
$Q_2$, where $v$ is the value returned by $Q_1.\lit{dequeue}$. 
Solving this problem using TM simply requires encapsulating the sequential implementation of 
$Q_1.\lit{dequeue}$, followed by $Q_2.\lit{enqueue}(v)$ within a transaction.

Note that a TM implementation may internally employ locks or conditional rmw instructions like compare-and-swap. 
However, TM raises
the level of abstraction by exposing an easy-to-use compositional transactional interface for the user application that is
\emph{oblivious} to the specifics of the implementation and the semantics of the user application.
\section{Summary of the thesis}
\label{sec:introoutline}
TM allows the programmer to \emph{speculatively} execute sequences of shared-memory
operations as atomic transactions: 
if the transaction commits, the operations appear as executed sequentially,
or if the transaction aborts, the update operations do not take effect.
The combination of speculation and the simple programming interface provided by TM
seemingly overcomes the problems associated with traditional synchronization techniques based on locks
and compare-and-swap. But are there some fundamental drawbacks associated with the TM abstraction?
Does providing high degrees of concurrency in TMs
come with inherent costs? This is the central question of the thesis.
In the rest of this introductory chapter, 
we provide a summary of the results in the thesis that give some answers to this question.
\subsection{Safety for transactional memory}
\label{sec:o1}
We first need to define the consistency criteria that must be satisfied by a TM implementation.
We formalize the semantics of a \emph{safe} TM: 
every transaction, including aborted and incomplete ones, must observe a view 
that is \emph{consistent} with some sequential execution.
This is important, since if the intermediate view is not consistent with \emph{any} sequential execution, the application
may experience a fatal irrevocable error or enter an infinite loop.
Additionally, the response of a transaction's read should not depend on an ongoing transaction that has not started committing yet.
This restriction, referred to as \emph{deferred-update semantics} appears desirable, 
since the ongoing transaction may still abort, thus rendering the read inconsistent.
We define the notion of deferred-update semantics formally and apply it to several TM consistency criteria 
proposed in literature. We then verify if the resulting TM consistency criterion is a \emph{safety} property~\cite{OL82,AS85,Lyn96}
in the formal sense, \emph{i.e.}, the set of \emph{histories} (interleavings of invocations and 
responses of transactional operations) 
is \emph{prefix-closed}
and \emph{limit-closed}.

We first consider the popular consistency criterion of
\emph{opacity}~\cite{tm-book}.
Opacity requires the states observed by all transactions, included
uncommitted ones, to be consistent with a global \emph{serialization},
\emph{i.e.}, a serial execution constituted by committed transactions.
Moreover, the serialization should respect the \emph{real-time
  order}: a transaction that completed before (in real time) another transaction
started should appear first in the serialization.

By definition, opacity reduces correctness of a TM history to correctness of all
its prefixes, and thus is prefix-closed and limit-closed.
Thus, to verify that a history is opaque, one needs to verify that each of its prefixes is
consistent with some global serialization.
To simplify verification and explicitly introduce deferred-update
semantics into a TM correctness criterion,
we specify a general criterion of \emph{du-opacity}~\cite{icdcs-opacity},
which requires the global serial execution to respect
the deferred-update property.
Informally, a du-opaque history must be
indistinguishable from a totally-ordered history, with respect to which
no transaction reads from a transaction that has not started committing.
Assuming that in an infinite history, every
transaction completes each of the operations it invoked,
we prove that du-opacity is a safety property.

One may notice that the intended safety semantics does not
require, as opacity does, that all transactions observe the same serial
execution.
As long as committed transactions constitute a serial
execution and every transaction witnesses a consistent state, the
execution can be considered ``safe'': no run-time error that cannot
occur in a serial execution can happen.
Two definitions in literature have adopted this approach~\cite{damien-vw-jv,DGLM13}.
We introduce ``deferred-update'' versions of these properties and discuss
how the resulting properties relate to du-opacity.
\subsection{Complexity of transactional memory}
\label{sec:ob3}
One may observe that a TM implementation that aborts or never commits any transaction is trivially safe, but not very useful.
Thus, the TM implementation must satisfy some nontrivial
\emph{liveness} property specifying the conditions under which the transactional operations must
return some response and a
\emph{progress} property specifying the conditions under which the transaction is allowed to abort.

Two properties considered important for TM performance are 
\emph{read invisibility}~\cite{attiyaH13} and 
\emph{disjoint-access parallelism}~\cite{israeli-disjoint}.
Read invisibility may boost the concurrency of a TM
implementation by ensuring that no reading transaction can cause any other
transaction to abort.  
The idea of disjoint-access parallelism is to allow transactions that do not access the
same data item to proceed independently of each other without memory contention. 

We investigate the inherent complexities in terms of time and memory resources 
associated with implementing safe TMs that provide strong liveness and progress properties, possibly combined with
attractive requirements like read invisibility and disjoint-access parallelism.
Which classes of TM implementations are (im)possible to solve?

\vspace{1mm}\noindent\textbf{Blocking TMs.}
We begin by studying TM implementations that are \emph{blocking}, in the sense that, 
a transaction may be delayed or aborted due to concurrent transactions.
\begin{itemize}
\item
We prove that, even inherently \emph{sequential} TMs, that allow a transaction to be aborted due to a concurrent transaction, incur 
significant complexity costs when combined with read invisibility and disjoint-access parallelism.
\item
We prove that, \emph{progressive} TMs, that 
allow a transaction to be aborted only if it encounters a read-write or write-write conflict with a 
concurrent transaction~\cite{GK09-progressiveness}, may need to exclusively control a linear number of data items
at some point in the execution.
\item
We then turn our focus to \emph{strongly progressive}
TMs~\cite{tm-book} that, in addition to progressiveness, 
ensures that \emph{not all} concurrent transactions conflicting over a single
data item abort.   
We prove that in any strongly progressive TM
implementation that accesses the shared memory with \emph{read}, \emph{write} and \emph{conditional}
primitives, such as compare-and-swap, 
the total number of \emph{remote memory references}~\cite{anderson-90-tpds,rmr-mutex} (RMRs) 
that take place in an execution in which $n$
concurrent processes perform transactions on a single data item might
reach $\Omega(n \log n)$ in the worst-case. 
\item
We show that, with respect to the amount of \emph{expensive synchronization} patterns like
compare-and-swap instructions and \emph{memory barriers}~\cite{AGK11-popl,McKenney10},
progressive implementations are asymptotically optimal. 
We use this result to establish a linear (in the transaction's data set size)
separation between the worst-case transaction expensive synchronization complexity of progressive TMs and 
\emph{permissive} TMs that allow a transaction to abort only if committing it would violate opacity.
\end{itemize}
\vspace{1mm}\noindent\textbf{Non-blocking TMs.}
Next, we focus on TMs that avoid using locks and rely on non-blocking
synchronization: a prematurely halted transaction cannot not prevent
other transactions from committing.
Possibly the weakest non-blocking
progress condition is obstruction-freedom~\cite{HLM03,HS11-progress}
stipulating that  every transaction running in the absence
of \emph{step contention}, \emph{i.e.}, not encountering steps of concurrent
transactions, must commit.
In fact, several early TM implementations~\cite{HLM+03, astm, ST95,nztm,fraser} satisfied obstruction-freedom.
However, circa. 2005, several papers presented the case for a shift from TMs that provide obstruction-free TM-progress
to lock-based progressive TMs~\cite{DStransaction06,DSS06, Ennals05}. 
They argued that lock-based TMs tend to outperform
obstruction-free ones by allowing for simpler algorithms with lower complexity overheads.
We prove the following lower bounds for obstruction-free TMs.
\begin{itemize}
\item 
Combining invisible reads with even \emph{weak} forms of
disjoint-access parallelism~\cite{AHM09} in obstruction-free TMs is impossible,
\item 
A read operation in a $n$-process obstruction-free TM
implementation incurs $\Omega(n)$ \emph{memory stalls}~\cite{G05,AGHK09}.
\item 
A \emph{read-only} transaction may need to
perform a linear (in $n$) number of expensive synchronization patterns.
\end{itemize}
We then present a progressive TM implementation
that beats all of these lower bounds, thus suggesting that
the course correction from non-blocking (obstruction-free) TMs to blocking (progressive) TMs was indeed justified.

\vspace{1mm}\noindent\textbf{Partially non-blocking TMs.}
Lastly, we explore the costs of providing non-blocking progress to only a \emph{subset} of transactions.
Specifically, we require \emph{read-only} transactions to commit \emph{wait-free}, \emph{i.e.}, every transaction commits
within a finite number of its steps, but \emph{updating}
transactions are guaranteed to commit only if they run in the absence of concurrency.
We show that combining this kind of \emph{partial} wait-freedom with
read invisibility \emph{or} disjoint-access parallelism comes with inherent costs.
Specifically, we establish the following lower bounds for TMs that provide this kind of partial wait-freedom. 
\begin{itemize}
\item 
This kind of partial wait-freedom equipped with invisible reads results in maintaining unbounded sets of \emph{versions}
for every data item.
\item 
It is impossible to implement a \emph{strict} form of disjoint-access parallelism~\cite{OFTM}.
\item
Combining with the weak form of disjoint-access parallelism means that a read-only transaction (with
an arbitrarily large read set) must sometimes perform at least one
expensive synchronization pattern per read operation in some executions.
\end{itemize}
%
%
\subsection{Hybrid transactional memory}
\label{sec:ob4}
We turn our focus on \emph{Hybrid transactional memory (HyTM)}~\cite{hybridnorec,damronhytm, kumarhytm,phasedtm}.
The TM abstraction, in its original manifestation, augmented the processor's \emph{cache-coherence 
protocol} and extended the CPU's instruction set with
instructions to indicate which memory accesses must be transactional~\cite{HM93}.
Most popular TM designs, subsequent to the original proposal in \cite{HM93} 
have implemented all the functionality in software~\cite{norec, ST95,HLM+03, astm, fraser}.
More recently, CPUs have included hardware extensions to 
support small transactions~\cite{Rei12, asf, bluegene}.
Hardware transactions may be spuriously aborted due to several reasons: cache capacity overflow, interrupts \emph{etc}.
This has led to proposals for \emph{best-effort} HyTMs in which the fast, but potentially unreliable hardware
transactions are complemented with slower, but more reliable software transactions.
However, the fundamental limitations of building a HyTM with nontrivial concurrency between hardware and software transactions
are not well understood.
Typically, hardware transactions usually employ \emph{code instrumentation} techniques to
detect concurrency scenarios and abort in the
case of contention.
But are there inherent instrumentation costs of implementing a HyTM, and what are the trade-offs between these costs
ands provided concurrency, \emph{i.e.}, the ability of the HyTM to execute hardware and software transactions in parallel?

The thesis makes the following contributions which help determine the cost of concurrency in HyTMs.
\begin{itemize}
\item
We propose a general model for HyTM implementations, 
which captures the notion of \emph{cached} accesses as performed by hardware transactions,
and precisely defines instrumentation costs in a quantifiable way.
\item
We derive lower and upper bounds in this model, which capture for the first time, an inherent trade-off 
on the degree of concurrency
allowed between hardware and software transactions and the \emph{instrumentation} overhead introduced on the hardware.
\end{itemize}
\subsection{Optimism for boosting concurrency}
\label{sec:o2}
Lock-based implementations are conventionally \emph{pessimistic} in nature: the operations invoked by 
processes are not ``abortable'' and return only after they are successfully completed.
The TM abstraction is a realization of \emph{optimistic} concurrency control: speculatively execute transactions,
abort and roll back on dynamically detected conflicts.
But are optimistic implementations fundamentally better equipped to exploit concurrency than pessimistic ones?

We compare the \emph{amount of concurrency} one can obtain by converting a sequential implementation
of a data abstraction into a concurrent one using optimistic
or pessimistic synchronization techniques.
To establish fair comparison of such implementations, we introduce a new correctness criterion for concurrent implementations, 
called \emph{locally serializable linearizability},
defined independently of the synchronization techniques they use.

We treat an implementation's concurrency as its ability to accept \emph{schedules} of sequential operations 
from different processes.
More specifically, we assume an external scheduler that defines which
processes execute which steps of the corresponding sequential implementation 
in a dynamic and unpredictable fashion. 
This allows us to define concurrency provided by an implementation as the set of 
interleavings of steps of sequential operations (or schedules)
it \emph{accepts}, \emph{i.e.}, is able to effectively process.
Then, the more schedules the implementation would accept without hampering correctness, the more concurrent it would be.

The thesis makes the following contributions.
\begin{itemize}
\item 
We provide a framework to analytically capture the
inherent concurrency provided by two broad classes of synchronization techniques:
pessimistic implementations that implement some form of mutual exclusion
and optimistic implementations based on speculative executions.
\item
We show that, implementations based on pessimistic synchronization
and ``semantics-oblivious'' TMs are suboptimal, in the sense that,
there exist there exist simple schedules of the list-based set
which cannot be accepted by \emph{any} pessimistic or TM-based implementation.
Specifically, we prove that TM-based implementations accept schedules of the list-based set that
cannot be accepted by \emph{any} pessimistic implementation.
However, we also show pessimistic implementations of the list-based set which accept schedules
that cannot be accepted by \emph{any} TM-based implementation.
\item
We show that, there exists an 
optimistic implementation
of the list-based set that is \emph{concurrency optimal}, \emph{i.e.}, it accepts \emph{all} correct schedules.
\end{itemize}
Our results suggest that ``semantics-aware'' optimistic implementations may be better suited to exploiting concurrency
than their pessimistic counterparts.
%
\section{Roadmap of the thesis}
\label{sec:roadmap}
We first define the TM model, the TM properties proposed in literature and the complexity metrics considered
in Chapter~\ref{ch:tm-model}. Chapter~\ref{ch:pc1} is on safety for TMs.
Chapter~\ref{ch:p3c2} is on the complexity of blocking TMs,
non-blocking TMs that satisfy obstruction-freedom are covered in Chapter~\ref{ch:p3c3}
and we present lower bounds for partially non-blocking TMs in Chapter~\ref{ch:p3c4}.
Chapter~\ref{ch:p4c4} is devoted to the study of hybrid TMs.
In Chapter~\ref{ch:p2c1}, we compare the relative abilities of optimistic and pessimistic
synchronization techniques in exploiting concurrency.
Chapter~\ref{ch:conc} presents closing comments and future directions.
Viewed collectively, the results hopefully shine light on the foundations of the TM abstraction that
is widely expected to be the \emph{Zeitgeist} of the concurrent computational model. 
%
\chapter{Transactional Memory model}
\label{ch:tm-model}
\epigraph{All models are wrong, but some models are useful.}
{\textit{George Edward Pelham Box}}
In this chapter, we formalize the TM model and discuss some important TM properties proposed in literature.
In Section~\ref{sec:c21}, we formalize the specification of TMs.
In Section~\ref{sec:tm-correctness}, we introduce the basic TM-correctness property of \emph{strict serializability}
that we consider in the thesis.
Sections~\ref{sec:tm-progress} and \ref{sec:tm-liveness} overview progress and liveness properties for TMs
respectively and identifies the relations between them.
Section~\ref{sec:inv} defines the notion of \emph{invisible reads} while
Section~\ref{sec:dap} is on \emph{disjoint-access parallelism}.
Finally, in Section~\ref{sec:complexity}, we introduce some of the complexity metrics considered in the thesis.
\section{TM interface and TM implementations}
\label{sec:c21}
In this section, we first describe the \emph{shared memory} model of computation and then introduce the TM abstraction.

\vspace{1mm}\noindent\textbf{The shared memory model.}
The thesis considers the standard \emph{asynchronous shared memory} model of computation in which
a set of $n\in \mathbb{N}$ processes (that may \emph{fail by crashing}), 
communicate by applying \emph{operations} on shared \emph{objects}~\cite{AGHK09}.
An object is an instance of an \emph{abstract data type}.
An abstract data type $\tau$ is a \emph{mealy machine} that is specified as a tuple
$(\Phi,\Gamma, Q, q_0, \delta)$ where
$\Phi$ is a set of operations,
$\Gamma$ is a set of responses, $Q$ is a set of states, $q_0\in Q$ is an
initial state and $\delta \subseteq Q\times \Phi \times Q\times \Gamma$ 
is a transition relation that determines, for each state
and each operation, the set of possible
resulting states and produced responses~\cite{AFHHT07}. 
Here, $(q,\pi,q',r) \in \delta$ implies that when
an operation $\pi \in \Phi$ is applied on an object of type $\tau$
in state $q$, the object moves to state $q'$ and returns response $r$.

An \emph{implementation} of an object type $\tau$ provides a specific data-representation of $\tau$ that is realized by
processes applying \emph{primitives} on shared \emph{base objects}, each of which is assigned an initial value. 
In order to implement an object, processes are provided with an algorithm, which is a set of deterministic
state-machines, one for each process.
In the thesis, we use the term primitive to refer to operations on base objects and reserve the term operation
for the object that is implemented from the base objects.

A primitive is a generic atomic \emph{read-modify-write} (\emph{rmw}) procedure applied to a base object~\cite{G05,Her91}.
It is characterized by a pair of functions $\langle g,h \rangle$:
given the current state of the base object, $g$ is an \emph{update function} that
computes its state after the primitive is applied, while $h$ 
is a \emph{response function} that specifies the outcome of the primitive returned to the process.
A rmw primitive is \emph{trivial} if it never changes the value of the base object to which it is applied.
Otherwise, it is \emph{nontrivial}.
An rmw primitive $\langle g,h \rangle$ is \emph{conditional} if there exists $v$, $w$ such that
$g(v,w)=v$ and there exists $v$, $w$ such that
$g(v,w)\neq v$~\cite{cond-04}.

\emph{Read} is an example of a trivial rmw primitive that takes no input arguments: when applied
to a base object with value $v$, the update function leaves the state of the base object unchanged and 
the response function returns the value $v$.
\emph{Write} is an example of a nontrivial rmw primitive that takes an input argument $v'$: when applied to a base object
with value $v$, its update function changes the value of the base object to $v'$ and its response function
returns $ok$.
\emph{Compare-and-swap} is an example of a nontrivial conditional rmw primitive: its update function receives an input argument 
$\langle \ms{old},\ms{new} \rangle$ and changes the value $v$ of the base object to which it is applied \emph{iff} $v=\ms{old}$.
\emph{Load-linked/store-conditional} is another example of a nontrivial conditional rmw primitive:
the \emph{load-linked} primitive executed by some process $p_i$ returns the value of the base object to which
it is applied and the \emph{store-conditional} primitive's update function receives an input $\ms{new}$ and 
atomically changes the value of the base object to $\ms{new}$
\emph{iff} the base object has not been updated by any other process since the load-linked event by $p_i$.
\emph{Fetch-and-add} is an example of a nontrivial rmw primitive that is not conditional: its update
function applied to base object with an integer value $v$ takes an integer $w$ as input
and changes the value of the base object to $v+w$.

\vspace{1mm}\noindent\textbf{Transactional memory (TM).}
\emph{Transactional memory}
allows a set of data items (called \emph{t-objects}) to be accessed 
via \emph{transactions}.
Every transaction $T_k$ has a unique identifier $k$. 
We make no assumptions on the \emph{size} of a t-object, \emph{i.e.}, 
the cardinality on the set $V$ of possible values a t-object can store.
A transaction $T_k$ may contain the following \emph{t-operations},
each being a matching pair of an \emph{invocation} and a \emph{response}:
$\Read_k(X)$ returns a value in $V$, denoted $\Read_k(X) \rightarrow v$, 
or a special value $A_k\notin V$ (\emph{abort});
$\Write_k(X,v)$, for a value $v \in V$,
returns \textit{ok} or $A_k$;
$\TryC_k$ returns $C_k\notin V$ (\emph{commit}) or $A_k$.
As we show in the subsequent Section~\ref{sec:tm-correctness}, we can specify
TM as an abstract data type.

Note that a TM interface may additionally provide a $\textit{start}_k$ t-operation that returns $ok$ or $A_k$, 
which is the first t-operation
transaction $T_k$ must invoke, or a $\textit{tryA}_k$ t-operation that returns $A_k$.
However, the actions performed inside the $\textit{start}_k$ may be performed as part of the first t-operation performed
by the transaction. The $\textit{tryA}_k$ t-operation allows the user application
to explicitly abort a transaction and can be useful, but since each of the individual t-read or t-write are allowed to abort,
the $\textit{tryA}_k$ t-operation provides no additional expressive power to the TM interface.
Thus, for simplicity, we do not incorporate these t-operations in our TM specification.

\vspace{1mm}\noindent\textbf{TM implementations.}
A TM \emph{implementation} provides processes with algorithms
for implementing $\Read_k$, $\Write_k$ and $\TryC_k()$
of a transaction $T_k$ by applying primitives from a set of shared base objects, each of which is 
assigned an initial value.
We assume that a process starts a new transaction
only after its previous transaction has committed or aborted.

In the rest of this section, we define the terms specifically in the context of TM implementations, but
they may be used analogously in the context of any concurrent implementation of an abstract data type.

\vspace{1mm}\noindent\textbf{Executions and configurations.}
An \emph{event} of a process $p_i$ (sometimes we say \emph{step} of $p_i$)
is an invocation or response of an operation performed by $p_i$ or a 
rmw primitive $\langle g,h \rangle$ applied by $p_i$ to a base object $b$
along with its response $r$ (we call it a \emph{rmw event} and write $(b, \langle g,h\rangle, r,i)$).

A \emph{configuration} (of an implementation) specifies the value of each base object and 
the state of each process.
The \emph{initial configuration} is the configuration in which all 
base objects have their initial values and all processes are in their initial states.

An \emph{execution fragment} is a (finite or infinite) sequence of events.
An \emph{execution} of an implementation $M$ is an execution
fragment where, starting from the initial configuration, each event is
issued according to $M$ and each response of a rmw event $(b, \langle
g,h\rangle, r,i)$ matches the state of $b$ resulting from all
preceding events.
An execution $E\cdot E'$ denotes the concatenation of $E$ and execution fragment $E'$,
and we say that $E'$ is an \emph{extension} of $E$ or $E'$ \emph{extends} $E$.

Let $E$ be an execution fragment.
For every transaction (resp., process) identifier $k$,
$E|k$ denotes the subsequence of $E$ restricted to events of
transaction $T_k$ (resp., process $p_k$).
If $E|k$ is non-empty,
we say that $T_k$ (resp., $p_k$) \emph{participates} in $E$, else we say $E$ is \emph{$T_k$-free} (resp., \emph{$p_k$-free}).
Two executions $E$ and $E'$ are \emph{indistinguishable} to a set $\mathcal{T}$ of transactions, if
for each transaction $T_k \in \mathcal{T}$, $E|k=E'|k$.
A TM \emph{history} is the subsequence of an execution consisting of the invocation and 
response events of t-operations.
Two histories $H$ and $H'$ are \emph{equivalent} if $\txns(H) = \txns(H')$
and for every transaction $T_k \in \txns(H)$, $H|k=H'|k$.

\vspace{1mm}\noindent\textbf{Data sets of transactions.}
The \emph{read set} (resp., the \emph{write set}) of a transaction $T_k$ in an execution $E$,
denoted $\Rset_E(T_k)$ (resp., $\Wset_E(T_k)$), is the set of t-objects that $T_k$ reads (resp., writes to) in $E$.
More specifically, if $E$ contains an invocation of $\Read_k(X)$ (resp., $\Write_k(X,v)$), we say
that $X\in \Rset_E(T_k)$ (resp., $\Wset_E(T_k)$) (for brevity, we sometimes omit the subscript $E$ from the notation).
The \emph{data set} of $T_k$ is $\Dset(T_k)=\Rset(T_k)\cup\Wset(T_k)$.
A transaction is called \emph{read-only} if $\Rset(T_k)\neq \emptyset \wedge \Wset(T_k)=\emptyset$; 
\emph{write-only} if $\Wset(T_k)\neq \emptyset \wedge \Rset(T_k)=\emptyset$ and
\emph{updating} if $\Wset(T_k)\neq\emptyset$.
Note that, in our TM model, the data set of a transaction is not known apriori, \emph{i.e.}, at the start of the transaction
and it is identifiable only by the set of data items the transaction has invoked a read or write on in the given execution.

\vspace{1mm}\noindent\textbf{Transaction orders.}
Let $\txns(E)$ denote the set of transactions that participate in $E$.
In an infinite history $H$, 
we assume that each $T_k\in \txns(H)$, $H|k$ is finite;
\emph{i.e.}, transactions do not issue an infinite number of t-operations.
An execution $E$ is \emph{sequential} if every invocation of
a t-operation is either the last event in the history $H$ exported by $E$ or
is immediately followed by a matching response.
We assume that executions are \emph{well-formed}, \emph{i.e.}, for all $T_k$, $E|k$ begins with the invocation of a t-operation, is
sequential and has no events after $A_k$ or $C_k$.
A transaction $T_k\in \txns(E)$ is \emph{complete in $E$} if
$E|k$ ends with a response event.
The execution $E$ is \emph{complete} if all transactions in $\txns(E)$
are complete in $E$.
A transaction $T_k\in \txns(E)$ is \emph{t-complete} if $E|k$
ends with $A_k$ or $C_k$; otherwise, $T_k$ is \emph{t-incomplete}.
$T_k$ is \emph{committed} (resp., \emph{aborted}) in $E$
if the last event of $T_k$ is $C_k$ (resp., $A_k$).
The execution $E$ is \emph{t-complete} if all transactions in
$\txns(E)$ are t-complete.

For transactions $\{T_k,T_m\} \in \txns(E)$, we say that $T_k$ \emph{precedes}
$T_m$ in the \emph{real-time order} of $E$, denoted $T_k\prec_E^{RT} T_m$,
if $T_k$ is t-complete in $E$ and
the last event of $T_k$ precedes the first event of $T_m$ in $E$.
If neither $T_k\prec_E^{RT} T_m$ nor $T_m\prec_E^{RT} T_k$,
then $T_k$ and $T_m$ are \emph{concurrent} in $E$.
An execution $E$ is \emph{t-sequential} if there are no concurrent
transactions in $E$.

\vspace{1mm}\noindent\textbf{Latest written value and legality.}
Let $H$ be a 
t-sequential history.
For every operation $\Read_k(X)$ in $H$,
we define the \emph{latest written value} of $X$ as follows:
if $T_k$ contains a $\Write_k(X,v)$ preceding $\Read_k(X)$,
then the latest written value of $X$ is the value of the latest such write to $X$.
Otherwise,
the latest written value of $X$ is the value
of the argument of the latest $\Write_m(X,v)$ that precedes
$\Read_k(X)$ and belongs to a committed transaction in $H$.
(This write is well-defined since $H$ starts with $T_0$ writing to
all t-objects.)

We say that $\Read_k(X)$ is \emph{legal} in a t-sequential history $H$ 
if it returns the latest written value of $X$, and $H$ is \emph{legal}
if every $\Read_k(X)$ in $H$ that does not return $A_k$ is legal in $H$.

We also assume, for simplicity, that the user application invokes a $\Read_k(X)$
at most once within a transaction $T_k$.
This assumption incurs no loss of generality,
since a repeated read can be assigned to return a previously returned
value without affecting the history's legality.

\vspace{1mm}\noindent\textbf{Contention.}
We say that a configuration $C$ after an execution $E$ is \emph{quiescent} (resp., \emph{t-quiescent})
if every transaction $T_k \in \ms{txns}(E)$ is complete (resp., t-complete) in $C$.
If a transaction $T$ is incomplete in an execution $E$, it has exactly one \emph{enabled} event, 
which is the next event the transaction will perform according to the TM implementation.
Events $e$ and $e'$ of an execution $E$  \emph{contend} on a base
object $b$ if they are both events on $b$ in $E$ and at least 
one of them is nontrivial (the event is trivial (resp., nontrivial) 
if it is the application of a trivial (resp., nontrivial) primitive).

We say that $T$ is \emph{poised to apply an event $e$ after $E$} if $e$ is the next enabled event for $T$ in $E$.
We say that transactions $T$ and $T'$ \emph{concurrently contend on $b$ in $E$} 
if they are poised to apply contending events on $b$ after $E$.

We say that an execution fragment $E$ is \emph{step contention-free for t-operation $op_k$} if the events of $E|op_k$ 
are contiguous in $E$.
We say that an execution fragment $E$ is \emph{step contention-free for $T_k$} if the events of $E|k$ are contiguous in $E$.
We say that $E$ is \emph{step contention-free} if $E$ is step contention-free for all transactions that participate in $E$.
%
\section{TM-correctness}
\label{sec:tm-correctness}
Correctness for TMs is specified as a safety property on TM histories~\cite{OL82,AS85,Lyn96}.
In this section, we introduce the popular TM-correctness condition \emph{strict serializability}~\cite{Pap79-serial}:
all committed transactions appear to execute sequentially in some total order respecting the real-time transaction orders.
We then explain how strict serializability is related to \emph{linearizability}~\cite{HW90}.

In the thesis, we only consider TM-correctness conditions like strict serializability and its restrictions.
We formally define strict serializability below, but other TM-correctness conditions studied in the
thesis can be found in Chapter~\ref{ch:pc1}.

First, we define how to derive a t-complete history from a t-incomplete one.
\begin{definition}[Completions]
\label{def:comp}
Let $H$ be a history.
A \emph{completion of $H$}, denoted ${\overline{H}}$,
is a history derived from $H$ as follows:
\begin{itemize}
\item[--]
First, for every transaction $T_k \in \txns(H)$
with an incomplete t-operation $op_k$ in $H$,
if $op_k=\Read_k \vee \Write_k$,
insert $A_k$ somewhere after the invocation of $op_k$;
otherwise, if $op_k=\TryC_k()$,
insert $C_k$ or $A_k$ somewhere after the last event of $T_k$.
\item[--]
After all transactions are complete,
for every transaction $T_k$ that is not t-complete,
insert $\TryC_k\cdot A_k$ after the last event of transaction $T_k$.
\end{itemize}
\end{definition}
\begin{definition}[Strict serializability]
\label{def:oser}
A finite history $H$ is \emph{serializable} if there
is a legal t-complete t-sequential history $S$,
such that
there is a completion $\overline{H}$ of $H$,
such that $S$ is equivalent to $\ms{cseq}(\overline{H})$,
where $\ms{cseq}(\overline{H})$ is the subsequence of $\overline{H}$
reduced to committed transactions in $\overline{H}$.

We refer to $S$ as a \emph{serialization} of $H$.

We say that $H$ is \emph{strictly serializable} if there exists a serialization $S$ of $H$ such that
for any two transactions $T_k,T_m \in \txns(H)$,
if $T_k \prec_H^{RT} T_m$, then $T_k$ precedes $T_m$ in $S$.
\end{definition}
In general, given a TM-correctness condition $C$, we say that a TM implementation $M$ satisfies $C$ if every
execution of $M$ satisfies $C$.

\vspace{1mm}\noindent\textbf{Strict serializability as linearizability.}
We now show we can specify TM as an abstract data type.

The sequential specification of a TM is specified as follows:
\begin{enumerate}
\item
$\Phi$ is the set of all transactions $\{T_i\}_{i\in \mathbb{N}}$
\item 
$\Gamma$ is the set of incommensurate vectors $\{[r_1,\ldots, r_i ]\};i\in \mathbb{N}$; where each 
$r_j$; $1\leq j \leq i-1$ $\in \{v \in V, A, ok\}$ and $r_i \in \{A,C\}$
\item 
The state of TM is a vector of the state of each t-object $X_m$.
The state of a t-object $X_m$ is a value $v_m\in V$ of $X_m$.
Thus, $Q \subseteq \{ [v^i_1,\ldots, v^i_m, \ldots ]\}$; where each $v^i_m \in V$
\item 
$q_0 \in Q=[ov_1,\ldots, ov_m,\ldots ]$, where each $ov_m\in V$
\item 
$\delta$ is defined as follows:
Let $T_k$ be a transaction
applied to the TM in state $q=[v_1,\ldots, v_m, \ldots]$.
\begin{itemize}
\item For every $X\in \Rset(T_k)$, the response of $\Read_k(X)$ is defined as follows:
If $T_k$ contains a $\Write_k(X,v)$ prior to $\Read_k(X)$, then the response is $v$; else the response is the current state of $X$.
\item 
For every $X\in \Wset(T_k)$, the response of $\Write_k(X,v)$ is $ok$.
\item
Transaction $T_k$ returns the response $C$ in which case the TM moves to state $q'$ defined as follows:
every $X_j\in \Wset(T_k)$ to which $T_k$ writes values $nv_j$, $q'[j]=nv_j$; else if $X_j\not\in \Wset(T_k)$, $q'[j]$ is unchanged.
Otherwise, $T_k$ returns the response $A$ in which case $q'=q$.
\end{itemize}
\end{enumerate}
In general, the correctness of an implementation of a data type is commonly captured by the criterion 
of linearizability.
In the TM context, a t-complete history $H$ is \emph{linearizable} with 
respect to the TM type if there exists
a t-sequential history $S$ equivalent to $H$ such that
(1) $S$ respects the real-time ordering of transactions in $H$ and
(2) \emph{$S$ is consistent with the sequential specification of TM}.

The following lemma, which illustrates the similarity between strict serializability
and linearizability with respect to the TM type, is now immediate.
\begin{lemma}
\label{lm:serlin}
Let $H$ be any t-complete history. Then, $H$ is strictly serializable \emph{iff} $H$ is linearizable with respect to the
TM type.
\end{lemma}
%
%
\section{TM-progress}
\label{sec:tm-progress}
One may notice that a TM implementation that forces, in every execution to abort every transaction is trivially strictly
serializable, but not very useful.
A TM-progress condition specifies the conditions under which a transaction is allowed to abort.
Technically, a TM-progress condition specified this way is a \emph{safety property} since it can be violated
in a finite execution (cf. Chapter~\ref{ch:pc1} for details on safety properties).

Ideally, a TM-progress condition must provide \emph{non-blocking} progress, in the sense that a prematurely halted
transaction cannot prevent all other transactions from committing.
Such a TM-progress condition is also said to be \emph{lock-free}
since it cannot be achieved by use of locks and mutual-exclusion.
A non-blocking TM-progress condition is considered useful in asynchronous systems with process failures
since it prevents the TM implementation from deadlocking (processes wait infinitely long without
committing their transactions).

\vspace{1mm}\noindent\textbf{Obstruction-freedom.}
Perhaps, the weakest non-blocking TM-progress condition is obstruction-freedom, which stipulates
that a transaction may be aborted only if it encounters steps of a concurrent transaction~\cite{HS11-progress}.
\begin{definition}[Obstruction-free (OF) TM-progress]
We say that a TM implementation $M$ provides \emph{obstruction-free (OF) TM-progress} if for every execution $E$ of $M$, 
if any transaction $T_k \in \ms{txns}(E)$ returns $A_k$ in $E$, then $E$ is not step contention-free for $T_k$. 
\end{definition}
We now survey the popular \emph{blocking} TM-progress properties proposed in literature. Intuitively,
unlike non-blocking TM-progress conditions that adapt to \emph{step contention}, 
a blocking TM-progress condition allows a transaction to be aborted due to \emph{overlap contention}.

\vspace{1mm}\noindent\textbf{Minimal progressiveness.}
Intuitively, the most basic TM-progress condition is one which provide only \emph{sequential} TM-progress, \emph{i.e.},
a transaction may be aborted due to a concurrent transaction.
In literature, this is referred to as \emph{minimal progressiveness}~\cite{tm-book}.
\begin{definition}[Minimal progressiveness]
We say that a TM implementation $M$ provides \emph{minimal progressive} TM-progress (or \emph{minimal progressiveness}) 
if for every execution $E$ of $M$
and every transaction $T_k\in \ms{txns}(E)$ that returns $A_k$ in $E$, there exists a transaction $T_m\in \ms{txns}(E)$
that is concurrent to $T_k$ in $E$~\cite{tm-book}.
\end{definition}
Given TM conditions $C_1$ and $C_2$, if every TM implementation that satisfies 
$C_1$ also satisfies $C_2$, 
but the converse is not true, we say that $C_2$~$\ll$~$C_1$.
\begin{observation}
Minimal progressiveness~$\ll$ Obstruction-free.
\end{observation}
\begin{proof}
Clearly, every TM implementation that satisfies obstruction-freedom also satisfies minimal progressiveness, but the converse 
is not true.
Consider any execution of a TM implementation $M$ in which a transaction $T$ run step contention-free.
If $M$ is minimally progressive, then $T$ may be aborted in such an execution since $T$ may be concurrent with another transaction.
However, if $M$ satisfies obstruction-freedom, $T$ cannot be aborted in such an execution.
\end{proof}
\vspace{1mm}\noindent\textbf{Progressiveness.}
In contrast to the ``single-lock'' minimal progressive TM-progress
condition (also referred to as \emph{sequential} TM-progress in the thesis), 
state-of-the-art TM implementations
allow a transaction to abort only if it encounters a \emph{conflict} on a 
t-object with a concurrent transaction.
\begin{definition}[Conflicts]
We say that transactions $T_i,T_j$ \emph{conflict} in an execution $E$ on a t-object $X$ if
$T_i$ and $T_j$ are concurrent in $E$ and $X\in\Dset(T_i)\cap\Dset(T_j)$,  and $X\in\Wset(T_i)\cup\Wset(T_j)$.
\end{definition}
\begin{definition}[Progressiveness]
\label{def:progdef}
A TM implementation $M$ provides \emph{progressive} TM-progress (or \emph{progressiveness}) 
if for every execution $E$ of $M$ and every transaction $T_i \in \ms{txns}(E)$ that returns $A_i$ in $E$,
there exists a transaction $T_k \in \ms{txns}(E)$ such that $T_k$ and $T_i$
conflict in $E$~\cite{tm-book}. 
\end{definition}
Note that progressiveness is incomparable to obstruction-freedom.
\begin{observation}
Progressiveness~$\not\ll$~Obstruction-free and Obstruction-free~$\not\ll$~Progressiveness.
\end{observation}
\begin{proof}
We can show that there exists an execution exported by an obstruction-free TM, but not by any
progressive TM and vice-versa.

Consider a t-read $X$ by a transaction $T$ that runs step contention-free
from a configuration that contains an incomplete write to $X$. Weak progressiveness
does not preclude $T$ from being aborted in such an execution.
Obstruction-free TMs however, must ensure that $T$ must complete its read of $X$ without blocking or aborting in such executions.
On the other hand, weak progressiveness requires two non-conflicting transactions to not be aborted even
in executions that are not step contention-free; but this is not guaranteed by obstruction-freedom.
\end{proof}
In general, progressive TMs (including the ones described in the thesis) satisfy the following stronger definition: 
for every transaction $T_i \in \ms{txns}(E)$ that returns $A_i$ in an execution $E$, 
there exists prefix $E'$ of $E$ and a transaction $T_k \in \ms{txns}(E')$ such that $T_k$ and $T_i$
conflict in $E$. However, for the lower bound results stated in the thesis, we stick to Definition~\ref{def:progdef}.

\vspace{1mm}\noindent\textbf{Strong progressiveness.}
One may observe that the definition of progressiveness does not preclude two conflicting 
transactions (over a single t-object) from each being aborted.
Thus, we study a stronger notion of progressiveness called \emph{strong progressiveness}~\cite{tm-book}.

Let $CObj_E(T_i)$ denote the set of t-objects over which transaction $T_i \in \ms{txns}(H)$ conflicts with any other 
transaction in an execution $E$,
\emph{i.e.}, $X \in CObj_E(T_i)$, \emph{iff} there exist transactions $T_i$ and $T_k$
that conflict on $X$ in $E$.
Let $ Q\subseteq \ms{txns}(E)$ and $CObj_E(Q)=\bigcup\limits_{ T_i \in Q} CObj_E(T_i)$.

Let $\ms{CTrans}(E)$ denote the set of non-empty subsets of $\ms{txns}(E)$ such that a set $Q$ is in $\ms{CTrans}(E)$ 
if no transaction in $Q$ conflicts with a transaction not in $Q$.
\begin{definition}[Strong progressiveness]
\label{def:sprog}
A TM implementation $M$ is \emph{strongly progressive} if $M$ is weakly progressive 
and for every execution $E$ of $M$ and for every set $Q \in \ms{CTrans}(E)$ such that $|CObj_{E}(Q)| \leq 1$, 
some transaction in $Q$ is not aborted in $E$~\cite{tm-book}.
\end{definition}
The above definitions imply:
\begin{corollary}
Minimal progressiveness (sequential TM-progress)~$\ll$~Progressiveness~$\ll$~Strong progressiveness.
\end{corollary}
\vspace{1mm}\noindent\textbf{Mv-permissiveness.}
Perelman \emph{et al.} introduced the notion of \emph{mv-permissiveness}, a TM-progress property
designed to prevent read-only transactions from being aborted.
\begin{definition}[Mv-permissiveness]
\label{def:sprog}
A TM implementation $M$ is \emph{mv-permissive} if for every execution $E$ of $M$ and
for every transaction $T_k\in \ms{txns}(E)$ that returns $A_k$ in $E$, we have that $\Wset(T_k)\neq \emptyset$
and there exists an updating transaction $T_m\in \ms{txns}(E)$ such that $T_k$ and $T_m$ conflict in $E$.
\end{definition}
We observe that mv-permissiveness is strictly stronger than progressiveness, but incomparable to strong progressiveness.
\begin{observation}
Progressiveness~$\ll$~Mv-permissiveness.
\end{observation}
\begin{proof}
Since mv-permissive TMs allow a transaction to be aborted only on read-write conflicts, they also
satisfy progressiveness. But the converse is not true.
Consider an execution in which a read-only transaction $T_i$ that runs concurrently with a conflicting
updating transaction $T_j$.
By the definition of progressiveness, both $T_i$ and $T_j$ may be aborted in such an execution.
However, a mv-permissive TM would not allow $T_i$ to be aborted since it is read-only.
\end{proof}
\begin{observation}
Strong progressiveness~$\not\ll$~Mv-permissiveness and Mv-permissiveness~$\not\ll$~Strong progressiveness.
\end{observation}
\begin{proof}
Consider an execution in which a read-only transaction $T_i$ that runs concurrently with an
updating transaction $T_j$ such that $T_i$ and $T_j$ conflict on at least two t-objects.
By the definition of strong progressiveness, both $T_i$ and $T_j$ may be aborted in such an execution.
However, a mv-permissive TM would not allow $T_i$ to be aborted since it is read-only.

On the other hand, consider an execution in which two updating transactions $T_i$ and $T_j$ that conflict on a single t-object.
A mv-permissive TM allows both $T_i$ and $T_j$ to be aborted, but strong progressiveness ensures that at least one of
$T_i$ or $T_j$ is not aborted in such an execution.
\end{proof}
%
\section{TM-liveness}
\label{sec:tm-liveness}
Observe that a TM-progress condition only specifies the conditions under which a transaction is aborted, but does not
specify the conditions under which it must commit.
For instance, the OF TM-progress condition specifies that a transaction $T$ may be aborted only in executions
that are not step contention-free for $T$, but does not guarantee that $T$ is committed in a step contention-free execution.
Thus, in addition to a progress condition, we must stipulate a \emph{liveness}~\cite{AS85,Lyn96} condition.

We now define the TM-liveness conditions considered in the thesis.
\begin{definition}[Sequential TM-liveness]
A TM implementation $M$ provides \emph{sequential TM-liveness} if for every finite execution $E$ of $M$, 
and every transaction $T_k$ that runs t-sequentially and applies the invocation of a t-operation $op_k$ immediately after $E$, 
the finite step contention-free extension for $op_k$ contains a response.
\end{definition}
\begin{definition}[Interval contention-free (ICF) TM-liveness]
A TM implementation $M$ provides \emph{interval contention-free (ICF)} TM-liveness
if for every finite execution $E$ of $M$ such that the configuration after $E$ is quiescent, 
and every transaction $T_k$ that applies the invocation of a t-operation $op_k$ immediately after $E$, 
the finite step contention-free extension for $op_k$ contains a response.
\end{definition}
\begin{definition}[Starvation-free TM-liveness]
A TM implementation $M$ provides \emph{starvation-free TM-liveness} if in every execution of $M$, each t-operation 
eventually returns a matching response, assuming that no concurrent t-operation stops indefinitely before returning.  
\end{definition}
\begin{definition}[Obstruction-free (OF) TM-liveness]
A TM implementation $M$ provides \emph{obstruction-free (OF) TM-liveness} if for every finite execution $E$ of $M$, 
and every transaction $T_k$ that applies the invocation of a t-operation $op_k$ immediately after $E$, 
the finite step contention-free extension for $op_k$ contains a matching response.
\end{definition}
\begin{definition}[Wait-free (WF) TM-liveness]
A TM implementation $M$ provides \emph{wait-free (WF) TM-liveness} if
in every execution of $M$, every t-operation returns a response in a finite number of its steps.
\end{definition}
%
The following observations are immediate from the definitions:
\begin{observation}
Sequential TM-liveness~$\ll$~ICF TM-liveness~$\ll$~OF TM-liveness~$\ll$~WF TM-liveness
and
Starvation-free TM-liveness~$\ll$~WF TM-liveness.
\end{observation}
Since ICF TM-liveness guarantees that a t-operation returns a response if there is no other concurrent t-operation,
we have:
\begin{observation}
ICF TM-liveness~$\ll$~Starvation-free TM-liveness.
\end{observation}
However, we observe that OF TM-liveness and starvation-free TM-liveness are incomparable.
\begin{observation}
Starvation-free TM-liveness~$\not\ll$~OF TM-liveness and OF TM-liveness~$\not\ll$~Starvation-free TM-liveness.
\end{observation}
\begin{proof}
Consider the step contention-free execution of t-operation $op_k$ concurrent with t-operation $op_m$: $op_k$
must return a matching response within a finite number of its steps, but this is not necessarily ensured by starvation-free
TM-liveness ($op_m$ may be delayed indefinitely).
On the other hand, in executions where two concurrent t-operations $op_k$ and $op_k$ encounter step contention, but
neither stalls indefinitely, both must return matching responses. But this is not guaranteed by OF TM-liveness.
\end{proof}
%
\section{Invisible reads}
\label{sec:inv}
In this section, we introduce the notion of \emph{invisible reads} that intuitively ensures that
a reading transaction does not cause a concurrent transaction to abort.
Since most TM worklods are believed to be read-dominated, this is considered to be an important TM property for
performance~\cite{stmbench7, Attiya09-tmread}.

\vspace{1mm}\noindent\textbf{Invisible reads.}
Informally, in a TM using invisible reads, 
a transaction cannot reveal any information about its read set to
other transactions. Thus, given an execution $E$ and some transaction
$T_k$ with a non-empty read set, transactions other than $T_k$ cannot
distinguish $E$ from an execution in which $T_k$'s read set is empty.
This prevents TMs from applying nontrivial primitives
during t-read operations and from announcing read sets of transactions during tryCommit.
Most popular TM implementations like \emph{TL2}~\cite{DSS06} and
\emph{NOrec}~\cite{norec} satisfy this property.
\begin{definition}[Invisible reads~\cite{attiyaH13}]
We say that a TM implementation $M$ uses \emph{invisible reads} if
for every execution $E$ of $M$:
\begin{itemize}
\item
for every read-only transaction $T_k \in \ms{txns}(E)$, 
no event of $E|k$ is nontrivial in $E$, 
\item
for every updating transaction $T_k \in \ms{txns}(E)$; $\Rset_E(T_k)\neq \emptyset$, 
there exists an execution $E'$ of $M$ such that
\begin{itemize}
\item
$\Rset_{E'}(T_k)=\emptyset$,
\item
$\ms{txns}(E)=\ms{txns}(E')$ and $\forall T_m \in \ms{txns}(E) \setminus \{T_k\}$: $E|m=E'|m$
\item
for any two transactions $T_i, T_j \in \ms{txns}(E)$, 
if the last event of $T_i$ precedes the first event of $T_j$ in $E$, 
then the last event of $T_i$ precedes the first event of $T_j$ in $E'$.
\end{itemize}
\end{itemize}
\end{definition}
\vspace{1mm}\noindent\textbf{Weak invisible reads.}
We introduce the notion of \emph{weak} invisible reads that prevents t-read operations
from applying nontrivial primitives only in the absence of concurrent transactions.
Specifically, weak read invisibility allows t-read operations of a transaction
$T$ to be ``visible'', \emph{i.e.}, write to base objects, only if $T$ is concurrent with
another transaction.
\begin{definition}[Weak invisible reads]
For any execution $E$ and any t-operation $\pi_k$ invoked by some transaction $T_k\in \ms{txns}(E)$,
let $E|\pi_k$ denote the subsequence of $E$ restricted to events of $\pi_k$ in $E$.

We say that a TM implementation $M$ satisfies \emph{weak invisible reads}
if for any execution $E$ of $M$ and every transaction $T_k\in \ms{txns}(E)$; $\Rset(T_k)\neq \emptyset$ that is
not concurrent with any transaction $T_m\in \ms{txns}(E)$, $E|\pi_k$ does not contain any nontrivial events, where $\pi_k$ is
any t-read operation invoked by $T_k$ in $E$.
\end{definition}
For example, the popular TM implementation \emph{DSTM}~\cite{HLM+03} satisfies weak invisible reads, but not invisible reads.
Algorithm~\ref{alg:oftm} in Chapter~\ref{ch:p3c2} 
depicts a TM implementation that is based on DSTM satisfying weak invisible reads, but not
the stronger definition of invisible reads.
\section{Disjoint-access parallelism (DAP)}
\label{sec:dap}
The notion of \emph{disjoint-access parallelism (DAP)}~\cite{israeli-disjoint}
is considered important in the TM context since it allows two transactions accessing unrelated t-objects to execute
without memory contention.
In this section, we preview the DAP definitions proposed in literature and identify the relations between them.

\vspace{1mm}\noindent\textbf{Strict data-partitioning.}
Let $E|X$ denote the subsequence of the execution $E$ derived by removing all events associated with t-object $X$.
A TM implementation $M$ is \emph{strict data-partitioned}~\cite{tm-book}, if for every t-object $X$, 
there exists a set of base objects $\ms{Base}_M(X)$ such that
\begin{itemize}
\item
for any two t-objects $X_1, X_2$; $\ms{Base}_M(X_1) \cap \ms{Base}_M(X_2)=\emptyset$,
\item 
for every execution $E$ of $M$ and every transaction $T \in \ms{txns}(E)$,
every base object accessed by $T$ in $E$ is contained in $\ms{Base}_M(X)$ for some $X\in \Dset(T)$
\item
for all executions $E$ and $E'$ of $M$, if $E|X=E|X$ for some t-object $X$, then the configurations after $E$ and
$E'$ only differ in the states of the base objects in $\ms{Base}_M(X)$.
\end{itemize}
\vspace{1mm}\noindent\textbf{Strict disjoint-access parallelism.}
A TM implementation $M$ is \emph{strictly disjoint-access parallel
  (strict DAP)} if, for
all executions $E$ of $M$, and for all transactions $T_i$ and $T_j$ that participate in $E$, 
$T_i$ and $T_j$ contend on a base object in $E$ only if 
$\Dset(T_i)\cap \Dset(T_j)\neq \emptyset$~\cite{tm-book}.
\begin{proposition}
Strict DAP $\ll$ Strict data-partitioning.
\end{proposition}
\begin{proof}
Let $M$ be any strict data-partitioned TM implementation. Then, $M$ is also strict DAP. Indeed,
since any two transactions accessing mutually disjoint data sets in a strict data-partitioned implementation
cannot access a common base object in any execution $E$ of $M$, $E$ also ensures that
any two transactions
contend on the same base object in $E$ only if   
they access a common t-object.

Consider the following execution $E$ of a strict DAP TM implementaton $M$ 
that begins with two transactions $T_1$ and $T_2$ that access disjoint data sets in $E$. 
A strict data-partitioned TM implementation
would preclude transactions $T_1$ and $T_2$ from accessing the same base object, but a strict DAP
TM implementation does not preclude this possibility.
\end{proof}
We now describe two relaxations of strict DAP. For the formal definitions, we introduce the notion of a
\emph{conflict graph} which captures the dependency relation among t-objects accessed by transactions.

\vspace{1mm}\noindent\textbf{Read-write (RW) disjoint-access parallelism.}
Informally, read-write (RW) DAP means that two transactions
can \emph{contend}  
on a common base object only if their data 
sets are connected in the \emph{conflict graph}, capturing 
write-set overlaps among all concurrent transactions.

We denote by $\tau_{E}(T_i,T_j)$, the set of transactions ($T_i$ and $T_j$ included)
that are concurrent to at least one of $T_i$ and $T_j$ in an execution $E$.

Let ${\tilde G}(T_i,T_j,E)$ be an undirected graph whose vertex set is $\bigcup\limits_{T \in \tau_{E}(T_i,T_j)} \Dset(T)$
and there is an edge
between t-objects $X$ and $Y$ \emph{iff} there exists $T \in \tau_{E}(T_i,T_j)$ such that 
$\{X,Y\} \in \Wset(T)$.
We say that $T_i$ and $T_j$ are \emph{read-write disjoint-access} in $E$
if there is no path between a t-object in $\Dset(T_i)$ and a t-object in $\Dset(T_j)$ in ${\tilde G}(T_i,T_j,E)$.
A TM implementation $M$ is \emph{read-write disjoint-access parallel (RW DAP)} if, for
all executions $E$ of $M$, 
transactions $T_i$ and $T_j$ 
contend on the same base object in $E$ only if   
$T_i$ and $T_j$ are not read-write disjoint-access in $E$ or there exists a t-object $X \in \Dset(T_i) \cap \Dset(T_j)$.
\begin{proposition}
RW DAP $\ll$ Strict DAP.
\end{proposition}
\begin{proof}
From the definitions, it is immediate that every strict DAP TM implementation satisfies RW DAP.

But the converse is not true (Algorithm~\ref{alg:oftm} describes a TM implementation that satisfies RW and weak DAP, but not strict DAP).
Consider the following execution $E$ of a weak DAP or RW DAP TM implementaton $M$ 
that begins with the t-incomplete execution of a transaction $T_0$ that 
accesses t-objects $X$ and $Y$, followed by the step contention-free executions of two transactions $T_1$ and $T_2$ 
which access $X$ and $Y$ respectively. Transactions $T_1$ and $T_2$ may contend on a base object since 
there is a path between $X$ and $Y$ in $G(T_1,T_2,E)$. However, a strict DAP TM implementation
would preclude transactions $T_1$ and $T_2$ from contending on the same base object since $\Dset(T_1) \cap \Dset(T_2)=\emptyset$
in $E$.
\end{proof}
\vspace{1mm}\noindent\textbf{Weak disjoint-access parallelism.}
Informally, weak DAP means that two transactions
can \emph{concurrently contend}  
on a common base object only if their data 
sets are connected in the \emph{conflict graph}, capturing 
data-set overlaps among all concurrent transactions.

Let $G(T_i,T_j,E)$ be an undirected graph whose vertex set is $\bigcup\limits_{T \in \tau_{E}(T_i,T_j)} \Dset(T)$
and there is an edge
between t-objects $X$ and $Y$ \emph{iff} there exists $T \in \tau_{E}(T_i,T_j)$ such that 
$\{X,Y\} \in \Dset(T)$.
We say that $T_i$ and $T_j$ are \emph{disjoint-access} in $E$
if there is no path between a t-object in $\Dset(T_i)$ and a t-object in $\Dset(T_j)$ in $G(T_i,T_j,E)$.
A TM implementation $M$ is \emph{weak disjoint-access parallel (weak DAP)} if, for
all executions $E$ of $M$,
transactions $T_i$ and $T_j$ 
concurrently contend on the same base object in $E$ only if   
$T_i$ and $T_j$ are not disjoint-access in $E$ or there exists a t-object $X \in \Dset(T_i) \cap \Dset(T_j)$~\cite{AHM09,PFK10}.

We now prove an auxiliary lemma, inspired by \cite{AHM09}, concerning weak DAP TM implementations 
that will be useful in subsequent proofs. Intuitively, the lemma states that, two transactions that are disjoint-access
and running one after the other in an execution of a weak DAP TM cannot contend on the same base object.
\begin{lemma}
\label{lm:dap}
Let $M $ be any weak DAP TM implementation.
Let $\alpha\cdot \rho_1 \cdot \rho_2$ be any execution of $M$ where
$\rho_1$ (resp., $\rho_2$) is the step contention-free
execution fragment of transaction $T_1 \not\in \ms{txns}(\alpha)$ (resp., $T_2 \not\in \ms{txns}(\alpha)$) 
and transactions $T_1$, $T_2$ are disjoint-access in $\alpha\cdot \rho_1 \cdot \rho_2$. 
Then, $T_1$ and $T_2$ do not contend on any base object in $\alpha\cdot \rho_1 \cdot \rho_2$.
\end{lemma}
\begin{proof}
Suppose, by contradiction that $T_1$ and $T_2$ contend on the same base object in $\alpha\cdot \rho_1\cdot \rho_2$.

If in $\rho_1$, $T_1$ performs a nontrivial event on a base object on which they contend, let $e_1$ be the last
event in $\rho_1$ in which $T_1$ performs such an event to some base object $b$ and $e_2$, the first event
in $\rho_2$ that accesses $b$.
Otherwise, $T_1$
only performs trivial events in $\rho_1$ to base objects on which it contends with $T_2$ in $\alpha\cdot \rho_1\cdot \rho_2$:
let $e_2$ be the first event in $\rho_2$ in which $\rho_2$ performs a nontrivial event to some base object $b$
on which they contend and $e_1$, the last event of $\rho_1$ in $T_1$ that accesses $b$.

Let $\rho_1'$ (and resp. $\rho_2'$) be the longest prefix of $\rho_1$ (and resp. $\rho_2$) that does not include
$e_1$ (and resp. $e_2$).
Since before accessing $b$, the execution is step contention-free for $T_1$, $\alpha \cdot
\rho_1'\cdot \rho_2'$ is an execution of $M$.
By construction, $T_1$ and $T_2$ are disjoint-access in $\alpha \cdot \rho_1'\cdot \rho_2'$
and $\alpha\cdot \rho_1 \cdot \rho_2'$ is indistinguishable to $T_2$ from $\alpha\cdot \rho_1' \cdot \rho_2'$.
Hence, $T_1$ and
$T_2$ are poised to apply contending events $e_1$ and $e_2$ on $b$ in the configuration after 
$\alpha\cdot \rho_1' \cdot \rho_2'$---a contradiction since $T_1$ and $T_2$ cannot concurrently contend on the same base object.   
\end{proof}
We now show that weak DAP is a weaker property than RW DAP.
\begin{proposition}
Weak DAP $\ll$ RW DAP.
\end{proposition}
\begin{proof}
Clearly, every implementation that satisfies RW DAP also satisfies weak DAP since the conflict graph
${\tilde G}(T_i,T_j,E)$ (for RW DAP) is a subgraph of ${G}(T_i,T_j,E)$ (for weak DAP).

However, the converse is not true (Algorithm~\ref{alg:oftm2} describes a TM implementation that 
satisfies weak DAP, but not RW DAP).
Consider the following execution $E$ of a weak DAP TM implementaton $M$ 
that begins with the t-incomplete execution of a transaction $T_0$ that 
reads $X$ and writes to $Y$, followed by the step contention-free executions of two transactions $T_1$ and $T_2$ 
which write to $X$ and read $Y$ respectively. Transactions $T_1$ and $T_2$ may contend on a base object since 
there is a path between $X$ and $Y$ in $G(T_1,T_2,E)$. However, a RW DAP TM implementation
would preclude transactions $T_1$ and $T_2$ from contending on the same base object: there is no edge
between t-objects $X$ and $Y$ in the corresponding conflict graph ${\tilde G}(T_1,T_2,E)$ because
$X$ and $Y$ are not contained in the write set of $T_0$.
\end{proof}
%
Thus, the above propositions imply:
\begin{corollary}
Weak DAP~$\ll$~ RW DAP~$\ll$~Strict DAP~$\ll$~Strict data-partitioning.
\end{corollary}
%
\section{TM complexity metrics}
\label{sec:complexity}
We now present an overview of some of the TM complexity metrics we consider in the thesis.

\vspace{1mm}\noindent\textbf{Step complexity.}
The step complexity metric, is the total number of events that a process performs on the shared memory,
in the worst case, in order to complete its operation on the implementation.

\vspace{1mm}\noindent\textbf{RAW/AWAR patterns.}
Attiya \emph{et al.} identified two common expensive synchronization patterns that frequently arise in
the design of concurrent algorithms: \emph{read-after-write (RAW) or atomic write-after-read (AWAR)}~\cite{AGK11-popl,McKenney10}
and showed that it is 
impossible to derive RAW/AWAR-free implementations of
a wide class of data types that include \emph{sets}, \emph{queues} and \emph{deadlock-free mutual exclusion}.

Note the shared memory model in the thesis makes the assumption that CPU \emph{events} are performed atomically:
every ``read'' of a base object returns the value of ``latest write'' to the base object. In practice however,
the CPU architecture's \emph{memory model}~\cite{AdveG96} that specifies the outcome of CPU instructions
is \emph{relaxed} without enforcing a strict order among the shared memory instructions.
Intuitively, RAW (read-after-write) or AWAR (atomic-write-after-read)
patterns~\cite{AGK11-popl} capture the amount of ``expensive
synchronization'', \emph{i.e.}, the number of costly memory barriers or conditional primitives~\cite{AdveG96} incurred by the
implementation in relaxed CPU architectures.
The metric appears to be more practically relevant than simply counting the number of steps performed by a process, 
as it accounts for expensive cache-coherence operations or instructions like compare-and-swap.
Detailed coverage on memory fences and the RAW/AWAR metric can be found in \cite{McKenney10}.
\begin{definition}[Read-after-write metric]
A \emph{RAW} (read-after-write) pattern  performed by a transaction
$T_k$ in an execution $\pi$ 
is a pair of its events $e$ and $e'$, such that: (1) $e$ is a write to a
base object $b$ by $T_k$, 
(2) $e'$ is a subsequent read of a base object $b'\neq b$ by $T_k$, and 
(3) no event on $b$ by $T_k$ takes place between $e$ and $e'$. 
\end{definition}
In the thesis, we are concerned only with \emph{non-overlapping} RAWs,
\emph{i.e.}, the read performed by one RAW precedes the write 
performed by the other RAW.
\begin{definition}[Atomic write-after-read metric]
An \emph{AWAR} (atomic-write-after-read) pattern $e$ in an execution
$\pi\cdot e$ is a nontrivial rmw event on an object $b$ which
atomically returns the value of $b$ (resulting after $\pi$) and updates $b$ with a
new value.  
\end{definition}
For example, consider the execution $\pi\cdot e$ where $e$ is the application of a \emph{compare-and-swap} rmw primitive that
returns $\true$.

\vspace{1mm}\noindent\textbf{Stall complexity.}
Intuitively, the stall metric captures the fact that the time a process might have to spend before it applies a 
primitive on a base object can be proportional to the number of processes that try to update the object concurrently.  

Let $M$ be any TM implementation.
Let $e$ be an event applied by process $p$ to a base object $b$ as it performs a transaction $T$ during an execution $E$ of $M$.
Let $E=\alpha\cdot e_1\cdots e_m \cdot e \cdot \beta$ be an execution of $M$, where $\alpha$ and $\beta$ are execution 
fragments and $e_1\cdots e_m$
is a maximal sequence of $m\geq 1$ consecutive nontrivial events by distinct distinct processes other than $p$ that access $b$.
Then, we say that $T$ incurs $m$ \emph{memory stalls in $E$ on account of $e$}.
The \emph{number of memory stalls incurred by $T$ in $E$} is the sum of memory stalls incurred by all events of $T$ in $E$~\cite{G05,AGHK09}.

In the thesis, we adopt the following definition of a \emph{k-stall execution} from \cite{AGHK09,G05}.
\begin{definition}
\label{def:stalls}
An execution $\alpha\cdot \sigma_1 \cdots \sigma_i$ is a $k$-stall execution for t-operation $op$ executed by process $p$ if
\begin{itemize}
\item 
$\alpha$ is $p$-free,
\item
there are distinct base objects $b_1,\ldots , b_i$ and disjoint sets of processes $S_1,\ldots , S_i$
whose union does not include $p$
and has cardinality $k$ such that, for $j=1,\ldots i $,
\begin{itemize}
\item
each process in $S_j$ has an enabled nontrivial event about to access base object $b_j$ after $\alpha$, and
\item
in $\sigma_j$, $p$ applies events by itself until it is the first about to apply an event to $b_j$,
then each of the processes in $S_j$ applies an event that accesses $b_j$, and finally, $p$ applies an event that accesses $b_j$,
\end{itemize}
\item
$p$ invokes exactly one t-operation $op$ in the execution fragment $\sigma_1\cdots \sigma_i$
\item
$\sigma_1\cdots \sigma_i$ contains no events of processes not in $(\{p\}\cup S_1\cup \cdots \cup S_i)$
\item
in every $(\{p\}\cup S_1\cup \cdots \cup S_i)$-free execution fragment that extends $\alpha$, 
no process applies a nontrivial event to any base object accessed in $\sigma_1 \cdots \sigma_i$.
\end{itemize}
\end{definition}
Observe that in a $k$-stall execution $E$ for t-operation $op$, the number of memory stalls incurred by $op$
in $E$ is $k$.

The following lemma will be of use in our proofs.
\begin{lemma}
\label{lm:stalls}
Let $\alpha\cdot \sigma_1 \cdots \sigma_i$ be a $k$-stall execution for t-operation $op$ executed by process $p$.
Then, $\alpha\cdot \sigma_1 \cdots \sigma_i$ is indistinguishable to $p$ from a step contention-free execution~\cite{AGHK09}.
\end{lemma}
\vspace{1mm}\noindent\textbf{Remote memory references(RMR)~\cite{rmr-mutex}.}
Modern shared memory CPU architectures employ a \emph{memory hierarchy}~\cite{hennessy-patterson}:
a hierarchy of memory devices with different capacities and costs.
Some of the memory is \emph{local} to a given process while the rest of the memory is \emph{remote}.
Accesses to memory locations (\emph{i.e.} base objects) that are \emph{remote} to a given process
are often orders of magnitude slower than a \emph{local} access of the base object.
Thus, the performance of concurrent implementations in the shared memory model may depend on the number
of \emph{remote memory references} made to base objects~\cite{anderson-90-tpds}.

In the \emph{cache-coherent (CC) shared memory}, each process maintains \emph{local}
copies of shared base objects inside its cache, whose consistency is ensured by a \emph{coherence protocol}.
Informally, we say that an access to a base object $b$ is \emph{remote} to a process $p$ and 
causes a \emph{remote memory reference (RMR)} if $p$'s cache contains a 
cached copy of the object that is out of date or \emph{invalidated}; otherwise the access is \emph{local}.

In the \emph{write-through (CC) protocol}, to read a base object $b$, process $p$ must have a cached copy of $b$ that
has not been invalidated since its previous read. Otherwise, $p$ incurs a RMR. 
To write to $b$, $p$ causes a RMR that invalidates all cached copies
of $b$ and writes to the main memory.

In the \emph{write-back (CC) protocol}, $p$ reads a base object $b$ without causing a RMR if it holds a cached copy of $b$
in shared or exclusive mode; otherwise the access of $b$ causes a RMR that (1) 
invalidates all copies of $b$ held in exclusive mode, and writing $b$ back to the main memory,
(2) creates a cached copy of $b$ in shared mode.
Process $p$ can write to $b$ without causing a RMR if it holds a copy of $b$ in exclusive mode; otherwise
$p$ causes a RMR that invalidates all cached copies of $b$ and creates a cached copy of $b$ in exclusive mode.

In the \emph{distributed shared memory (DSM)}, each base object is forever assigned to a single process and it is
\emph{remote} to the others. Any access of a remote register causes a RMR.
%
\chapter{Safety for transactional memory}
\label{ch:pc1}
\epigraph{\textbf{Arthur}: If I asked you where the hell we were, would I regret it?\\
\textbf{Ford}: We're safe.\\
\textbf{Arthur}: Oh good.\\
\textbf{Ford}: We're in a small galley cabin in one of the spaceships of the Vogon Constructor Fleet.\\
\textbf{Arthur}: Ah, this is obviously some strange use of the word safe that I wasn't previously aware of.}
{\textit{Douglas Adams}-The Hitchhiker's Guide to the Galaxy}
\section{Overview}
\label{sec:pc1}
In the context of Transactional memory,
intermediate states witnessed by the read operations of an incomplete transaction
may affect the user application through the outcome of its read operations.
If the intermediate state is not consistent with any sequential
execution, the application may experience a fatal irrecoverable
error or enter an infinite loop.
Thus, it is important that \emph{each} transaction, including \emph{aborted} ones observes a
\emph{consistent} state so that the implementation does not export any pathological executions.

A state should be considered consistent if it could
result from a serial application of transactions 
observed in the current execution.
In this sense, every transaction should witness  a state
that \emph{could have been} observed in \emph{some} execution
of the sequential code put by the programmer within the transactions.
Additionally, a consistent state should not depend on
a transaction that has not started committing yet
(referred to as \emph{deferred-update} semantics).
This restriction appears desirable, since the ongoing transaction may still
abort (explicitly by the user or because of consistency reasons) and,
thus, render the read inconsistent.
Further, the set of histories specified by the consistency criterion must constitute a \emph{safety
property}, as defined by Owicki and Lamport~\cite{OL82},
Alpern and Schneider~\cite{AS85} and refined by Lynch~\cite{Lyn96}:
it must be non-empty, \emph{prefix-closed} and \emph{limit-closed}.

In this chapter,
we define the notion of deferred-update semantics formally, which we then apply to
a spectrum of TM consistency criteria. Additionally, we verify if the resulting TM consistency criterion is a safety property,
as defined by Lynch~\cite{Lyn96}.

We begin by considering the popular criterion of
\emph{opacity}~\cite{tm-book}, which was
the first TM
consistency criterion that was proposed to grasp this semantics formally.
Opacity requires the states observed by all transactions, included
uncommitted ones, to be consistent with a global \emph{serialization},
\emph{i.e.}, a serial execution constituted by committed transactions.
Moreover, the serialization should respect the \emph{real-time order}: a transaction that 
completed before (in real time) another transaction
started should appear first in the serialization.

By definition, opacity reduces correctness of a history to correctness of all
its prefixes, and thus is prefix-closed and limit-closed by definition.
Thus, to verify that a history is opaque, one needs to verify that each of its prefixes is
consistent with some global serialization.
To simplify verification and explicitly introduce deferred-update
semantics into a TM correctness criterion,
we specify a general criterion of \emph{du-opacity}~\cite{icdcs-opacity},
which requires the global serial execution to respect
the deferred-update property.
Informally, a du-opaque history must be
indistinguishable from a totally-ordered history, with respect to which
no transaction reads from a transaction that has not started committing.

We show that du-opacity is \emph{prefix-closed}, that is,
every prefix of a du-opaque history is also du-opaque.
We then show that extending opacity (and du-opacity) to infinite histories in a
non-trivial way (\emph{i.e.}, requiring that even infinite histories should have proper
serializations), does not result in a limit-closed property.
However, under certain restrictions,
we show that du-opacity is \emph{limit-closed}.
In particular, assuming that in an infinite history, every
transaction completes each of the operations it invoked,
the limit of any sequence of
ever extending du-opaque histories is also du-opaque.
Therefore,
under this assumption,
du-opacity is a \emph{safety property}~\cite{OL82,AS85,Lyn96}, and
to prove that a TM implementation
that complies with the assumption
is du-opaque, it suffices to prove that all its \emph{finite} histories are du-opaque.

One may notice that the intended safety semantics does not
require that all transactions observe the same serial
execution.
Intuitively, to avoid pathological executions, we only need that every transaction witnesses \emph{some} consistent state,
while the views of different aborted and incomplete transactions do not have to be
consistent with \emph{the same} serial execution.
As long as committed transactions constitute a serial
execution and every transaction witnesses a consistent state, the
execution can be considered ``safe'': no run-time error that cannot
occur in a serial execution can happen.
Several definitions like \emph{virtual-world consistency (VWC)}~\cite{damien-vw-jv}
and \emph{Transactional Memory Specification~1 (TMS1)}\cite{DGLM13} have adopted this approach.
We introduce ``deferred-update'' versions of these properties and discuss
how the resulting properties relate to du-opacity.

Finally, we also study the consistency criterion \emph{Transactional Memory Specification~2 (TMS2)}~\cite{DGLM13,TMS-WTTM}, 
which was proposed as a restriction of opacity
and verify if it is a safety property.

\vspace{1mm}\noindent\textbf{Roadmap of Chapter~\ref{ch:pc1}.}
In Section~\ref{sec:safety} of this chapter, we formally define safety properties.
In Section~\ref{sec:pc2}, we introduce the notion of deferred-update semantics and apply it to the correctness criterions of
opacity and strict serializability in Sections~\ref{sec:gko} and \ref{sec:strictser} respectively. 
Section~\ref{sec:pc3} studies two relaxations of opacity: VWC and TMS1 and a restriction of opacity, TMS2.
Section~\ref{sec:pc4} summarizes the relations between the TM correctness properties proposed in the thesis and 
presents our concluding remarks.
%
\section{Safety properties}
\label{sec:safety}
A \emph{property} $\mathcal{P}$ is a set of (transactional) histories.
Intuitively, a \emph{safety property} says that ``no bad thing ever happens''.
\begin{definition}[Lynch~\cite{Lyn96}]
\label{def:pc}
A property $\mathcal{P}$ is a \emph{safety} property if it satisfies
the following two conditions:
\begin{description}
\item[\emph{Prefix-closure:}]
For every history $H \in \mathcal{P}$,
every prefix $H'$ of $H$ (\emph{i.e.}, every
prefix of the sequence of the events in $H$)
is also in $\mathcal{P}$.
\item[\emph{Limit-closure:}]
For every infinite sequence of finite histories
$H^0,H^1,\ldots $ such that for every $i$, $H^i \in \mathcal{P}$ and
$H^i$ is a prefix of $H^{i+1}$,
the limit of the sequence is also in $\mathcal{P}$.
\end{description}
\end{definition}
Notice that the set of histories produced by a TM implementation $M$
is, by construction,  prefix-closed.
Therefore, every infinite history of $M$
is the limit of an infinite sequence of ever-extending finite
histories of $M$. Thus, to prove that $M$ satisfies a safety property
$P$, it is enough to show that all finite histories of $M$ are in
$P$. Indeed, limit-closure of $P$ then implies that every infinite
history of $M$ is also in $P$.
%
\section{Opacity and deferred-update(DU) semantics}
\label{sec:pc2}
In this section, we formalize the notion of deferred-update semantics and apply to the 
TM correctness condition of \emph{opacity}~\cite{tm-book}.
%
%
\begin{definition}[Guerraoui and Kapalka~\cite{tm-book}]
\label{def:opacity GK}
A finite history $H$ is \emph{final-state opaque} if there
is a legal t-complete t-sequential history $S$,
such that
\begin{enumerate}
\item
for any two transactions $T_k,T_m \in \txns(H)$,
if $T_k \prec_H^{RT} T_m$, then $T_k <_S T_m$, and
\item
$S$ is equivalent to a completion of $H$.
\end{enumerate}
We say that $S$ is a \emph{final-state serialization} of $H$.
\end{definition}
Final-state opacity is not prefix-closed.
Figure~\ref{fig:one} depicts
a t-complete sequential history $H$ that is final-state opaque, 
with $T_1\cdot T_2$ being a legal t-complete t-sequential history 
equivalent to $H$.
Let $H'=\Write_1(X,1), \Read_2(X)$ be a prefix of $H$ in which
$T_1$ and $T_2$ are t-incomplete.
Transaction $T_i$ ($i=1,2$) is completed by inserting $\textit{tryC}_i\cdot
A_i$ immediately after the last event of $T_i$ in $H$.
Observe that neither $T_1\cdot T_2$ nor $T_2\cdot T_1$
allow us to derive a serialization of $H'$ (we assume that the initial value of $X$ is $0$).

A restriction of final-state opacity, which we refer to as
\emph{opacity}~\cite{tm-book} explicitly filters out
histories that are not prefix-closed.
\begin{figure}[h]
\begin{center}
\scalebox{.7}[0.7]{\begin{tikzpicture}
\node (w2) at (1.5,0) [] {};
\node (r1) at (2.5,-1) [] {};
\node (c2) at (5,0) [] {};
\node (c1) at (5.5,-1) [] {};
\draw (c1) node [above] {\small {$\TryC_2 $}}
(r1) node [above] {\small {$R_2(X)\rightarrow 1$}};
   
\draw (c2) node [above] {\small {$\TryC_1$}}
   (w2) node [above] {\small {$W_1(X,1)$}};



\begin{scope}   
\draw [|-|,thick] (.5,0) node[left] {} to (2.5,0);
\draw [|-|,thick] (4,0) node[left] {} to (6,0);
\draw [-,dotted] (-.5,0) node[left] {$T_1$} to (6,0) node[right] {$C_1$};
\end{scope}
\begin{scope}
\draw [|-|,thick] (1.5,-1) node[left] {} to (3.5,-1);
\draw [|-|,thick] (4.5,-1) node[left] {} to (6.5,-1);
\draw [-,dotted] (1.5,-1) node[left] {$T_2$} to (6.5,-1) node[right] {$C_2$} ; 
\end{scope}  
\begin{scope}
\draw [-,dashed] (3.9,1) node[left] {$H'$} to (3.9,-1.5) ; 
\end{scope}  
\begin{scope}
\draw [-,dashed] (7.2,1) node[left] {$H$} to (7.2,-1.5) ; 
\end{scope}  
\end{tikzpicture}}
\end{center}
\caption{History $H$ is final-state opaque,
while its prefix $H'$ is not final-state opaque.}
\label{fig:one}
\end{figure}
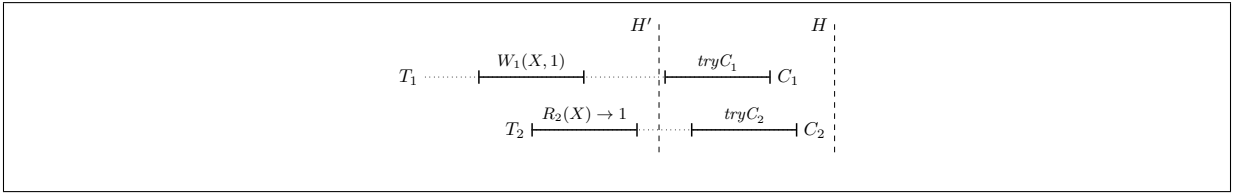
\begin{definition}[Guerraoui and Kapalka~\cite{tm-book}]
\label{def:opaque}
A history $H$ is \emph{opaque} if and only if every finite prefix $H'$
of $H$ (including $H$ itself if it is finite) is final-state opaque.
\end{definition}
%
It can be easily seen that opacity is prefix- and limit-closed,
and, thus, it is a safety property.

We now give a formal definition of opacity
with deferred-update semantics.
Then we show that the property is prefix-closed and,
under certain \emph{liveness} restrictions,  limit-closed.

Let $H$ be any history and let $S$ be a legal t-complete
t-sequential history that is equivalent to some completion of $H$.
Let $<_S$ be the total order on 
transactions in $S$.
\begin{definition}[Local serialization]
For any $\Read_k(X)$ that does not return $A_k$,
let $S^{k,X}$ be the prefix of $S$ up to the response of $\Read_k(X)$
and $H^{k,X}$ be the prefix of $H$ up to the response of $\Read_k(X)$.\\
$S_{H}^{k,X}$, the \emph{local serialization}
of $\Read_k(X)$ with respect to $H$ and $S$,
is the subsequence of $S^{k,X}$ derived by
removing from $S^{k,X}$ the events of all transactions
$T_m \in \txns(H) \setminus \{T_k\}$ such that
$H^{k,X}$ does not contain an invocation of $\TryC_m()$.
\end{definition}
%
We are now ready to present our correctness condition, \emph{du-opacity}.
\begin{definition}[Du-opacity]
\label{def:opacityKR}
A history $H$ is \emph{du-opaque} if there is
a legal t-complete t-sequential history $S$ such that
\begin{enumerate}
\item there is a completion of $H$ that is equivalent to $S$, and
\item for every pair of transactions $T_k,T_m \in \txns(H)$, if $T_k \prec_H^{RT} T_m$,
then $T_k <_S T_m$, i.e., $S$ respects the real-time ordering of
transactions in $H$, and
\item each $\Read_k(X)$ in $S$ that does not return $A_k$ is
  legal in $S_{H}^{k,X}$.
\end{enumerate}
We then say that $S$ is a (du-opaque) \emph{serialization} of $H$.
\end{definition}
Informally, a history $H$ is du-opaque if
there is a legal t-sequential history $S$  that is equivalent to $H$, respects the
real-time ordering of transactions in $H$ and every t-read is legal in
its local serialization with respect to $H$ and $S$.
The third condition reflects the implementation's deferred-update semantics, i.e.,
the legality of a t-read in a serialization does not depend on transactions that start committing after the response of the t-read.

For any du-opaque serialization $S$,
$\ms{seq}(S)$ denotes the \emph{sequence of transactions} in $S$
and $\ms{seq}(S)[k]$ denotes the $k^{th}$ transaction in this sequence.
\section{On the safety of du-opacity}
\label{sec:gko}
In this section, we examine the safety properties of du-opacity, \emph{i.e.}, whether it is prefix-closed and limit-closed.
\subsection{Du-opacity is prefix-closed}
\label{sec:dupc}
\begin{lemma}
\label{lm:dusep}
Let $H$ be a du-opaque history and let $S$ be a serialization of $H$.
For any $i\in \Nat$, there is a serialization $S^i$ of $H^i$ 
(the prefix of $H$ consisting of the first $i$ events),
such that $\ms{seq}(S^i)$ is a subsequence of $\ms{seq}(S)$.
\end{lemma}
\begin{proof}
Given $H$, $S$ and $H^i$,
we construct a t-complete t-sequential history $S^i$ as follows:
\begin{itemize}
\item[--]
for every transaction $T_k$ that is t-complete in $H^i$,
$S^i|k=S|k$.
\item[--]
for every transaction $T_k$ that is complete but not t-complete in $H^i$,
    $S^i|k$ consists of the sequence of events in $H^i|k$,
    immediately followed by $\TryC_k()\cdot A_k$.
\item[--]
for every transaction $T_k$ with an incomplete t-operation,
$op_k=\Read_k \vee \Write_k \vee \TryA_k()$ in $H^i$, $S^i|k$ is the
sequence of events in $S|k$ up to the invocation of $op_k$, immediately
followed by $A_k$.
\item[--]
for every transaction $T_k \in \txns(H^i)$ with an incomplete
t-operation, $op_k=\TryC_k()$, $S^i|k=S|k$.
\end{itemize}

By the above construction, $S^i$ is indeed a t-complete history
and every transaction that appears in $S^i$ also appears in $S$.
We order transactions in $S^i$ so that $\ms{seq}(S^i)$ is a subsequence of $\ms{seq}(S)$.

Note that $S^i$ is derived from events contained in some completion $\overline{H}$ of
$H$ that is equivalent to $S$ and some other events to derive a completion of $S^i$.
Since $S^i$ contains events from every complete t-operation in $H^i$ and
other events included satisfy Definition~\ref{def:comp},
there is a completion of $H^i$ that is equivalent to $S^i$.

We now argue that $S^i$ is a serialization of $H^i$.
First we observe that $S^i$ respects the real-time order of $H^i$.
Indeed, if $T_j \prec_{H^i}^{RT} T_k$, then $T_j \prec_{H}^{RT} T_k$
and $T_j <_{S} T_k$. Since $\ms{seq}(S^i)$ is a subsequence of
$\ms{seq}(S)$, we have  $T_j <_{S^i} T_k$.

To show that $S^i$ is legal,
suppose, by way of contradiction,
that there is some $\Read_k(X)$ that returns $v\neq A_k$ in $H^i$
such that $v$ is not the latest written value of $X$ in $S^i$.
If $T_k$ contains a $\Write_k(X,v')$ preceding $\Read_k(X)$ such that
$v\neq v'$ and $v$ is not the latest written value for $\Read_k(X)$
in $S^i$, it is also not the latest written value  for $\Read_k(X)$
in $S$, which is a contradiction.
Thus, the only case to consider is when $\Read_k(X)$ should return a value
written by another transaction.

Since $S$ is a serialization of $H$,
there is a committed transaction
$T_m$ that performs the last $\Write_m(X,v)$ that precedes
$\Read_k(X)$ in $T_k$ in $S$.
Moreover, since $\Read_k(X)$ is legal in the local serialization of
$\Read_k(X)$ in $H$ with respect to $S$,
the prefix of $H$ up to the response of $\Read_k(X)$ must contain an invocation of $\TryC_m()$.
Thus, $\Read_k(X) \not\prec_H^{RT} \TryC_m()$  and $T_m\in \txns(H^i)$.
By construction of $S^i$, $T_m \in \txns(S^i)$ and $T_m$ is committed in $S^i$.

We have assumed, towards a contradiction, that
$v$ is not the latest written value for $\Read_k(X)$  in $S^i$.
Hence, there is a committed transaction $T_j$ that performs
$\Write_j(X,v');v' \neq v$ in $S^i$ such that $T_m <_{S^i} T_j <_{S^i} T_k$.
But this is not possible since $\ms{seq}(S^i)$ is a subsequence of
$\ms{seq}(S)$.

Thus, $S^i$ is a legal t-complete t-sequential history equivalent to
some completion of $H^i$.
Now, by the construction of $S^i$, for every $\Read_k(X)$ that does not
return $A_k$ in $S^i$, we have ${S^i}_{H^i}^{k,X}={S}_{H}^{k,X}$.
Indeed, the transactions that appear before $T_k$ in
${S^i}_{H^i}^{k,X}$ are those with a $\TryC$ event before the
response of  $\Read_k(X)$ in $H$ and are committed in $S$.
Since $\ms{seq}(S^i)$ is a subsequence of $\ms{seq}(S)$, we have
${S^i}_{H^i}^{k,X}={S}_{H}^{k,X}$.
Thus, $\Read_k(X)$ is legal in  ${S^i}_{H^i}^{k,X}$.
\end{proof}
%
%
Lemma~\ref{lm:dusep} implies that every prefix of a
du-opaque history has a du-opaque serialization and thus:
\begin{corollary}
\label{cr:pc}
Du-opacity is a prefix-closed property.
\end{corollary}
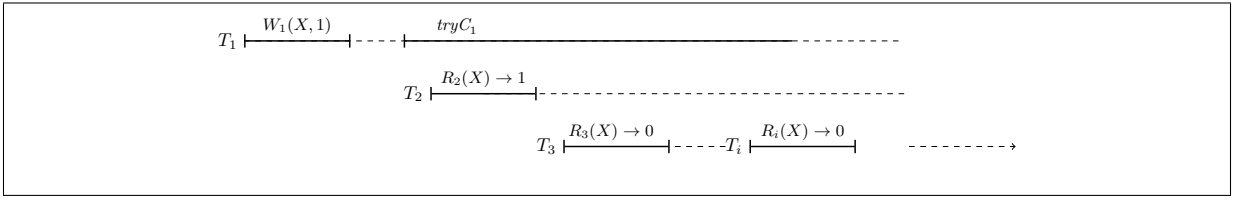
\begin{figure*}[t]
\begin{center}
\scalebox{.7}[0.7]{\begin{tikzpicture}
\node (w1) at (.5,0) [] {};
\node (c1) at (3.5,0) [] {};
\node (r2) at (4,-1) [] {};
\node (r4) at (10,-2) [] {};
\node (r5) at (6.4,-2) [] {};

\draw (w1) node [above] {\small {$W_1(X,1)$}};
\draw (c1) node [above] {\small {$\TryC_1$}};
\draw (r2) node [above] {\small {$R_2(X) \rightarrow 1$}};
\draw (r4) node [above] {\small {$R_i(X) \rightarrow 0$}};
\draw (r5) node [above] {\small {$R_3(X) \rightarrow 0$}};

\begin{scope}   
\draw [|-|,thick] (-0.5,0) node[left] {} to (1.5,0);
\draw [|-,thick] (2.5,0) node[left] {} to (9.8,0);
\draw [-,dashed] (-0.5,0) node[left] {$T_1$} to (9.8,0);
\draw [-,dashed] (9.8,0) node[left] {} to (11.8,0);
\end{scope}

\begin{scope}
\draw [|-|,thick] (3,-1) node[left] {$T_2$} to (5,-1);
\draw [-,dashed] (4,-1) node[left] {} to (12,-1);
\end{scope}  
\begin{scope}
\draw [|-|,thick] (5.5,-2) node[left] {$T_3$} to (7.5,-2);
\draw [-,dashed] (7.6,-2) node[left] {} to (8.6,-2);
\draw [|-|,thick] (9,-2) node[left] {$T_i$} to (11,-2);
\draw [->,dashed] (12,-2) to (14,-2);
\end{scope}  
\end{tikzpicture}}
\end{center}
\caption{
An infinite history in which $\TryC_1$ is incomplete
and any two transactions are concurrent.
Each finite prefix of the history is du-opaque,
but the infinite limit of the ever-extending sequence is not du-opaque.}
\label{fig:op-example}
\end{figure*}
\subsection{The limit of du-opaque histories}
We observe, however,  that du-opacity is, in general, not limit-closed.
We present an infinite history that is not du-opaque,
but each of its prefixes is.

\begin{proposition}\label{lm:lmn}
Du-opacity is not a limit-closed property.
\end{proposition}

\begin{proof}
Let $H^j$ denote a finite prefix of $H$ of length $j$.
Consider an infinite history $H$ that is the limit of the
histories $H^j$ defined as follows (see Figure~\ref{fig:op-example}):
\begin{itemize}
\item[--]
Transaction $T_1$ performs a $\Write_1(X,1)$ and then
invokes $\TryC_1()$ that is incomplete in $H$.
\item[--]
Transaction $T_2$ performs a $\Read_2(X)$ that overlaps with $\TryC_1()$
and returns $1$.
\item[--]
There are infinitely many transactions $T_i$, $i\geq 3$,
each of which performing a single $\Read_i(X)$ that returns $0$
such that each $T_i$ overlaps with both $T_1$ and $T_2$.
\end{itemize}

We now prove that, for all $j\in \mathbb{N}$, $H^j$ is a du-opaque history.
Clearly, $H^0$ and $H^1$ are du-opaque histories.
For all $j>1$, we first derive a completion of $H^j$ as follows:
\begin{enumerate}
\item
$\TryC_1()$ (if it is contained in $H^j$) is completed by
inserting $C_1$ immediately after its invocation,
\item
for all $i \geq 2$, any incomplete $\Read_i(X)$ that is contained
in $H^j$ is completed by inserting $A_i$ and $\TryC_i\cdot A_i$
immediately after its invocation,
and
\item
for all $i\geq 2$ and every complete $\Read_j(X)$ that is contained in $H^j$,
we include $\TryC_i \cdot A_i$ immediately after the response of this $\Read_j(X)$.
\end{enumerate}

We can now derive a t-complete t-sequential history $S^j$ equivalent
to the above derived completion of $H^j$ from the sequence of transactions
$T_3, \ldots, T_i,T_1,T_2$
(depending on which of these transactions participate in $H^j$),
where $i\geq 3$.
It is easy to observe that $S^j$ so derived is indeed a serialization of $H^j$.

However, there is no serialization of $H$.
Suppose that such a serialization $S$ exists.
Since every transaction that participates in $H$ must participate in $S$,
there exists $n \in \Nat$ such that $\ms{seq}(S)[n]=T_1$.
Consider the transaction at index $n+1$, say $T_i$ in $\ms{seq}(S)$.
But for any $i\geq 3$, $T_i$ must precede $T_1$ in any serialization (by legality),
which is a contradiction.
\end{proof}
Notice that all finite prefixes of the infinite history depicted in
Figure~\ref{fig:op-example} are also opaque.
Thus, if we extend the definition of opacity to cover infinite
histories in a non-trivial way, i. e., by explicitly defining opaque
serializations for infinite histories, we can reformulate
Proposition~\ref{lm:lmn} for opacity.
%
\subsection{Du-opacity is limit-closed for complete histories}
We show now that du-opacity is limit-closed if the only infinite
histories
we consider are those
in which every transaction eventually completes (but not
necessarily t-completes).

We first prove an auxiliary lemma on du-opaque serializations.
For a transaction $T\in \txns(H)$, the \emph{live set of $T$ in $H$},
denoted $\ms{Lset}_H(T)$ ($T$ included), is defined as follows:
every transaction $T' \in \txns(H)$ such that neither the last event of $T'$ precedes the first event of $T$ in
$H$ nor the last event of $T$ precedes the first event of $T'$ in $H$ is contained in $\ms{Lset}_H(T)$.
We say that transaction $T'\in \txns(H)$ \emph{succeeds the live set of $T$} and we write $T\prec_H^{LS} T'$ if in $H$, for all $T''\in
\ms{Lset}_H(T)$, $T''$ is complete and the last event of $T''$
precedes the first event of $T'$.
\begin{lemma}
\label{lm:full}
Let $H$ be a finite du-opaque history and assume
$T_k \in \txns(H)$ is a complete transaction in $H$,
such that every transaction in $\ms{Lset}_H(T_k)$ is complete in $H$.
Then there is a serialization $S$ of $H$,
such that for all $T_k, T_m \in\txns(H)$,
if $T_k\prec_H^{LS} T_m$, then $T_k<_S T_m$.
\end{lemma}
\begin{proof}
Since $H$ is du-opaque, there is a serialization ${\tilde S}$ of $H$.

Let ${ S}$ be a t-complete t-sequential history such that $\txns(\tilde S)=\txns(S)$, and $\forall~T_i \in \txns(\tilde S): S|i={\tilde S}|i$.
We now perform the following procedure iteratively to derive $\ms{seq}(S)$ from $\ms{seq}(\tilde S)$.
Initially $\ms{seq}(S)=\ms{seq}(\tilde S)$.
For each $T_k \in \txns(H)$, let $T_{\ell}\in \txns(H)$ denote the earliest transaction in ${\tilde S}$ such that $T_k \prec_H^{LS} T_{\ell}$.
If $T_{\ell} <_{\tilde S} T_k$ (implying $T_k$ is not t-complete), then move $T_k$ to immediately precede $T_{\ell}$ in $\ms{seq}(S)$.
%

By construction, $S$ is equivalent to ${\tilde S}$ and for all $T_k, T_m \in\txns(H)$; $T_k\prec_H^{LS} T_m$, $T_k<_{S} T_m$
We claim that $S$ is a serialization of $H$.
Observe that any two transactions that are complete in $H$, but not t-complete are not related by real-time order in $H$.
By construction of $S$, for any transaction $T_k\in \txns(H)$, the set of transactions that precede $T_k$ in ${\tilde S}$, but succeed $T_k$ in $S$ are not related to $T_k$ by real-time order.
Since $\tilde S$ respects the real-time order in $H$, this holds also
for $S$.

We now show that $S$ is legal.
Consider any $\Read_k(X)$ performed by some transaction $T_k$ that returns $v\in V$
in $S$ and let $T_{\ell}\in \txns(H)$ be the earliest transaction in ${\tilde S}$ such that $T_k \prec_H^{LS} T_{\ell}$.
Suppose, by contradiction, that $\Read_k(X)$ is not legal in $S$.
Thus, there is a committed transaction $T_m$ that performs
$\Write_m(X,v)$ in ${\tilde S}$ such that $T_m=T_{\ell}$ or $T_{\ell} <_{\tilde S} T_m <_{\tilde S} T_k$.
Note that, by our assumption, $\Read_k(X) \prec_H^{RT} \TryC_{\ell}()$.
Since $\Read_k(X)$ must be legal in its local serialization
with respect to $H$ and $\tilde S$,
$\Read_k(X) \not\prec_H^{RT} \TryC_m()$.
Thus, $T_m \in \ms{Lset}_H(T_k)$. Therefore $T_m\neq T_{\ell}$.
Moreover, $T_m$ is complete, and since it commits in $\tilde S$, it
is also t-complete in $H$ and the last event of $T_m$
precedes the first event of $T_{\ell}$ in $H$, i.e., $T_m\prec_H^{RT}
T_{\ell}$. Hence, $T_{\ell}$ cannot precede $T_m$ in ${\tilde S}$---a contradiction.

Observe also that since $T_k$ is complete in $H$ but not t-complete,
$H$ does not contain an invocation of $\TryC_k()$. Thus,
the legality of any other transaction is unaffected by moving $T_k$ to
precede $T_{\ell}$ in $S$.
Thus, $S$ is a legal t-complete t-sequential history equivalent to some completion of $H$.
The above arguments also prove that every t-read in $S$ is legal in
its local serialization with respect to $H$ and $S$ and, thus, $S$ is
a serialization of $H$.
\end{proof}
%


The proof uses K\"{o}nig's Path Lemma~\cite{konig}
formulated as follows.
Let $G$ on a rooted directed graph and let $v_0$ be the root of $G$.
We say that $v_k$, a vertex of $G$, is \emph{reachable} from $v_0$,
if there is a sequence of vertices $v_0 \ldots, v_k$ such that for each $i$,
there is an edge from $v_{i}$ to $v_{i +1}$.
$G$ is \emph{connected} if every vertex in $G$ is reachable from $v_0$.
$G$ is \emph{finitely branching} if every vertex in $G$ has a finite out-degree.
$G$ is \emph{infinite} if it has infinitely many vertices.
\begin{lemma}[K\"{o}nig's Path Lemma~\cite{konig}]
\label{lm:konig}
If $G$ is an infinite connected finitely branching rooted directed graph,
then $G$ contains an infinite sequence of distinct vertices $v_0,v_1, \ldots$,
such that $v_0$ is the \emph{root},
and for every $i \geq 0$, there is an edge from $v_i$ to $v_{i+1}$.
\end{lemma}

\begin{theorem}
\label{th:lc}
Under the restriction that in any infinite history $H$,
every transaction $T_k \in \txns(H)$ is complete,
du-opacity is a limit-closed property.
\end{theorem}
\begin{proof}
We want to show that the limit $H$ of an infinite sequence of finite ever-extending du-opaque
histories is du-opaque. 
By Corollary~\ref{cr:pc}, we can assume the sequence of du-opaque histories
to be $H^0,H^1, \ldots H^i,H^{i+1},\ldots$ such that for all $i\in \mathbb{N}$, $H^{i+1}$ is the one-event extension of $H^i$.

We construct a rooted directed graph $G_H$ as follows:
%
\begin{enumerate}
\item
The root vertex of $G_H$ is $(H^0, S^0)$ where $S^0$ and $H^0$ contain the initial transaction $T_0$.
\item
Each non-root vertex of $G_H$ is a tuple $({H}^i, S^i)$,
where $S^i$ is a du-opaque serialization of ${H}^i$ that satisfies
the condition specified in Lemma~\ref{lm:full}: for all $T_k, T_m \in\txns(H)$; $T_k\prec_{H^i}^{LS} T_m$ implies $T_k<_{S^i} T_m$.
Note that there exist several possible serializations for any $H^i$.
For succinctness, in the rest of this proof, when we refer to a specific $S^i$,
it is understood to be associated with the prefix $H^i$ of $H$.
\item
Let $\ms{cseq}_i(S^j)$, $j\geq i$, denote the subsequence of
$\ms{seq}(S^j)$ restricted to transactions whose last event in $H$
is a response event and it is contained in $H^i$.
For every pair of vertices $v=({H}^i,S^i)$ and $v'=({H}^{i+1}, S^{i+1})$ in $G_H$,
there is an edge from $v$ to $v'$ if $\ms{cseq}_i(S^i)=\ms{cseq}_i(S^{i+1})$.
\end{enumerate}

The out-degree of a vertex $v=(H^i,S^i)$ in $G_H$ is defined by
the number of possible serializations of ${H}^{i+1}$,
bounded by the number of possible permutations
of the set $\txns(S^{i+1})$, implying that $G_H$ is \emph{finitely branching}.

By Lemma~\ref{lm:dusep},
given any serialization $S^{i+1}$ of $H^{i+1}$,
there is a serialization $S^i$ of $H^i$
such that $\ms{seq}(S^i)$ is a subsequence of $\ms{seq}(S^{i+1})$.
Indeed, the serialization $S^i$ of $H^i$ also respects the restriction specified in Lemma~\ref{lm:full}.
Since $\ms{seq}(S^{i+1})$ contains every complete transaction that takes its last step in $H$ in $H^i$,
$\ms{cseq}_i(S^i)=\ms{cseq}_i(S^{i+1})$.
Therefore, for every vertex $({H}^{i+1},S^{i+1})$,
there is a vertex $({H}^{i},S^i)$ such that $\ms{cseq}_i(S^i)={cseq}_i(S^{i+1})$.
Thus, we can iteratively construct a path from
$(H^0,S^0)$ to every vertex $(H^i,S^i)$ in $G_H$,
implying that $G_H$ is \emph{connected}.

We now apply K\"{o}nig's Path Lemma (Lemma~\ref{lm:konig}) to $G_H$.
Since $G_H$ is an infinite connected finitely branching rooted directed graph,
we can derive an infinite sequence of distinct vertices
\[
\mathcal{L}=(H^0,S^0), (H^1, S^1), \ldots, (H^i,S^i), \ldots
\]
such that $\ms{cseq}_i(S^i)=\ms{cseq}_i(S^{i+1})$.

The rest of the proof explains how to use $\mathcal{L}$ to construct a serialization of $H$.
We begin with the following claim concerning $\mathcal{L}$.
\begin{claim}
\label{cl:lclaim}
For any $j>i$, $\ms{cseq}_i(S^i)=\ms{cseq}_i(S^{j})$.
\end{claim}
\begin{proof}
Recall that $\ms{cseq}_i(S^i)$ is a prefix
of $\ms{cseq}_i(S^{i+1})$, and
$\ms{cseq}_{i+1}(S^{i+1})$ is a prefix of $\ms{cseq}_{i+1}(S^{i+2})$.
Also, $\ms{cseq}_i(S^{i+1})$ is a subsequence of $\ms{cseq}_{i+1}(S^{i+1})$.
Hence, $\ms{cseq}_i(S^{i})$ is a subsequence of
$\ms{cseq}_{i+1}(S^{i+2})$.
But, $\ms{cseq}_{i+1}(S^{i+2})$ is a subsequence of
$\ms{cseq}_{i+2}(S^{i+2})$.
Thus, $\ms{cseq}_i(S^i)$ is a subsequence of
$\ms{cseq}_{i+2}(S^{i+2})$.
Inductively, for any $j>i$, $\ms{cseq}_i(S^i)$ is a subsequence of
$\ms{cseq}_{j}(S^{j})$.
But $\ms{cseq}_{i}(S^{j})$ is the
subsequence  of $\ms{cseq}_{j}(S^{j})$ restricted to 
complete transactions in $H$ whose last step is in $H^i$.
Thus, $\ms{cseq}_i(S^i)$ is indeed equal to $\ms{cseq}_{i}(S^{j})$.
\end{proof}

Let $f:\Nat \rightarrow \txns(H)$ be defined as follows:
$f(1) = T_0$. For every integer $k > 1$, let
\[ i_k = \min \{ \ell \in \Nat | \forall j>\ell:
\ms{cseq}_{\ell}(S^{\ell})[k]=\ms{cseq}_j(S^j)[k]\} \]
Then, $f(k)=\ms{cseq}_{i_k}(S^{i_k})[k]$.
%
%
\begin{claim}
\label{cl:bij}
The function $f$ is \emph{total} and \emph{bijective}.
\end{claim}
\begin{proof}
\textit{(Totality and surjectivity)}

Since each transaction $T \in \txns(H)$ is complete in some prefix $H^i$ of $H$,
for each $k\in \Nat$, 
there exists $i \in \Nat$ such that $\ms{cseq}_i(S^i)[k]=T$.
By Claim~\ref{cl:lclaim},
for any $j>i$, $\ms{cseq}_i(S^i)=\ms{cseq}_i(S^{j})$.
Since a transaction that is complete in $H^i$ w.r.t $H$ is also
complete in $H^j$ w.r.t $H$,
it follows that for every $j >i$, $\ms{cseq}_j(S^j)[k']=T$, with $k' \geq k$.
By construction of $G_H$ and the assumption that each transaction is complete in $H$,
there exists $i\in \Nat$
such that each $T \in \ms{Lset}_{H^i}(T)$ is complete in $H$ 
and its last step is in $H^i$, 
and $T$ precedes in $S^i$ every transaction whose first event succeeds 
the last event of each $T'\in \ms{Lset}_{H^i}(T)$ in $H^i$.
Indeed, this implies that for each $k\in \Nat$, 
there exists $i \in \Nat$ such that $\ms{cseq}_i(S^i)[k]=T$;
$\forall j>i: \ms{cseq}_j(S^j)[k]=T$.
%

This shows that for every $T \in \txns(H)$,
there are $i,k\in \Nat$; $\ms{cseq}_i(S^i)[k]=T$,
such that for every $j > i$, $\ms{cseq}_j(S^j)[k]=T$.
Thus, for every $T \in \txns(H)$, there is $k$ such that $f(k)=T$.

\textit{(Injectivity)}

If $f(k)$ and $f(m)$ are transactions at indices $k$, $m$
of the same $\ms{cseq}_i(S^i)$, then clearly $f(k) = f(m)$ implies $k=m$.
Suppose $f(k)$ is the transaction at index $k$ in some $\ms{cseq}_i(S^i)$
and $f(m)$ is the transaction at index $m$ in some $\ms{cseq}_{\ell}(S^{\ell})$.
For every $\ell > i$ and $k<m$,
if $\ms{cseq}_i(S^i)[k]=T$, then $\ms{cseq}_{\ell}(S^{\ell})[m] \neq T$
since $\ms{cseq}_i(S^i)=\ms{cseq}_i(S^{\ell})$.
If $\ell > i$ and $k>m$, it follows from the definition that $f(k)\neq f(m)$.
Similar arguments for the case when $\ell < i$ prove that if $f(k)= f(m)$, then $k=m$.
\end{proof}
By Claim~\ref{cl:bij}, $\mathcal{F}=f(1),f(2),\ldots, f(i) ,\ldots$
is an infinite sequence of transactions.
Let $S$ be a t-complete t-sequential history such that
$\ms{seq}(S)=\mathcal{F}$
and for each t-complete transaction $T_k$ in $H$, $S|k=H|k$; and for
transaction that is complete, but not t-complete in $H$, $S|k$ consists of the sequence of
events in $H|k$, immediately followed by $\TryA_k()\cdot A_k$.
Clearly, there is a completion of $H$ that is equivalent to $S$.

Let $\mathcal{F}^i$ be the prefix of $\mathcal{F}$ of length $i$,
and ${\widehat S^i}$ be the prefix of $S$ such that
$\ms{seq}({\widehat S^i})=\mathcal{F}^i$.
\begin{claim}
\label{cl:final}
Let ${\widehat H^j}_i$ be a subsequence of $H^j$ reduced
to transactions $T_k \in \txns({\widehat S^i})$ such that the last
event of $T_k$ in $H$ is a response event and it is contained in $H^j$.
Then, for every $i$, there is $j$ such that
${\widehat S^i}$ is a serialization of ${\widehat H^j}_i$.
\end{claim}
\begin{proof}
%
Let $H^j$ be the shortest prefix of $H$ (from $\mathcal{L}$)
such that for each $T\in \txns({\widehat S^i})$,
if $\ms{seq}(S^j)[k]=T$, then for every $j'>j$, $\ms{seq}(S^{j'})[k]=T$.
From the construction of $\mathcal{F}$, such $j$ and $k$ exist.
Also, we observe that $\txns({\widehat S^i}) \subseteq \txns(S^j)$
and $\mathcal{F}^i$ is a subsequence of $\ms{seq}(S^j)$.
Using arguments similar to the proof of Lemma~\ref{lm:dusep},
it follows that ${\widehat S^i}$ is indeed a serialization of ${\widehat H^j}_i$.
\end{proof}
Since $H$ is complete, there is exactly one completion of $H$, where
each transaction $T_k$ that is not t-complete in $H$ is completed with
$\textit{tryC}_k\cdot A_k$ after its last event.
By Claim~\ref{cl:final}, the limit t-sequential t-complete history is
equivalent to this completion, is legal,
respects the real-time order of $H$, and ensures that
every read is legal in the corresponding local serialization.
Thus, $S$ is a serialization of $H$.
\end{proof}
Theorem~\ref{th:lc} implies the following:
\begin{corollary}
\label{cr:safetytm}
Let $M$ be a TM implementation that ensures that
in every infinite history $H$ of $M$,
every transaction $T\in \txns(H)$ is complete in $H$.
Then, $M$ is du-opaque if and only if every finite history of $M$ is du-opaque.
\end{corollary}
%
%
\subsection{Du-opacity vs. opacity}
\label{sec:opacitygk}
We now compare our deferred-update requirement with the conventional
TM correctness property of opacity~\cite{tm-book}.
\begin{figure}[t]
\begin{center}
\scalebox{.7}[0.7]{\begin{tikzpicture}
\node (w1) at (-1,0) [] {};
\node (c1) at (3.5,0) [] {};
\node (r2) at (3.2,-1) [] {};
\node (w3) at (1,-2) [] {};
\node (c3) at (5.5,-2) [] {};

\draw (w1) node [above] {\small {$W_1(X,1)$}};
\draw (c1) node [above] {\small {$\TryC_1 $}};
\draw (r2) node [above] {\small {$R_2(X) \rightarrow 1$}};
\draw (w3) node [above] {\small {$W_3(X,1)$}};
\draw (c3) node [above] {\small {$\TryC_3 $}};

\begin{scope}   
\draw [|-|,thick] (-2.2,0) node[left] {} to (-0.2,0);
\draw [|-|,thick] (2.5,0) node[left] {} to (5.5,0);
\draw [-,dotted] (-2.2,0) node[left] {$T_1$} to (5.5,0) node[right] {$A_1$};
\end{scope}

\begin{scope}
\draw [|-|,thick] (2.5,-1) node[left] {$T_2$} to (3.8,-1);
\end{scope}  
\begin{scope}
\draw [|-|,thick] (0,-2) node[left] {} to (2,-2);
\draw [|-|,thick] (4.5,-2) node[left] {} to (6.5,-2);
\draw [-,dotted] (0,-2) node[left] {$T_3$} to (6.5,-2) node[right] {$C_3$}; 
\end{scope}  
\end{tikzpicture}}
\end{center}
\caption{A history that is opaque, but not du-opaque.}
\label{fig:lin-example}
\end{figure}
\begin{theorem}
\label{th:gkkr}
Du-opacity $\subsetneqq$ Opacity.
\end{theorem}
\begin{proof}
We first claim that every finite du-opaque history is opaque.
Let $H$ be a finite du-opaque history.
By definition, there is a final-state serialization $S$ of $H$.
Since du-opacity is a prefix-closed property,
every prefix of $H$ is final-state opaque.
Thus, $H$ is opaque.

Again, since
every prefix of a du-opaque history is also du-opaque,
by Definition~\ref{def:opaque},
every infinite du-opaque history is also opaque.

To show that the inclusion is strict, we present an an opaque history
that is not du-opaque.
Consider the finite history $H$ depicted in Figure~\ref{fig:lin-example}:
transaction $T_2$ performs a $\Read_2(X)$ that returns the value $1$.
Observe that $\Read_2(X) \rightarrow 1$ is concurrent to $\TryC_1$, but precedes $\TryC_3$ in real-time order.
Although $\TryC_1$ returns $A_1$ in $H$, the response of $\Read_2(X)$ can be justified since $T_3$ concurrently writes $1$ to $X$
and commits. Thus, $\Read_2(X)\rightarrow 1$ \emph{reads-from} transaction $T_2$ in any serialization of $H$, but
since $\Read_2(X) \prec_H^{RT} \TryC_3$, $H$ is not du-opaque even though each of its prefixes is final-state opaque.

We now formally prove that $H$ is opaque. We proceed by examining every prefix of $H$.
\begin{enumerate}
\item
Each prefix up to the invocation of $\Read_2(X)$ is trivially final-state opaque.
\item
Consider the prefix, $H^i$ of $H$ where the $i^{th}$ event is the response of $\Read_2(X)$. Let $S^i$ be a t-complete t-sequential history derived from the sequence $T_1,T_2$ by inserting $C_1$ immediately after the invocation of $\TryC_1()$. It is easy to see that $S^i$ is a final-state serialization of $H^i$.
\item
Consider the t-complete t-sequential history $S$ derived from the sequence $T_1,T_3,T_2$ in which each transaction is t-complete in $H$. Clearly, $S$ is a final-state serialization of $H$.
\end{enumerate}
Since $H$ and every (proper) prefix of it are final-state opaque,
$H$ is opaque.

Clearly, the required final-state serialization $S$ of $H$ is specified by $\ms{seq}(S)=T_1,T_3,T_2$
in which $T_1$ is aborted while $T_3$ is committed in $S$ (the position of $T_1$ in the serialization does not affect legality).
Consider $\Read_2(X)$ in $S$;
since $H^{2,X}$, the prefix of $H$ up to the response of $\Read_2(X)$
does not contain an invocation of $\TryC_3()$,
the local serialization of $\Read_2(X)$ with respect to $H$ and $S$,
$S_{H}^{2,X}$ is $T_1\cdot \Read_2(X)$.
But $\Read_2(X)$ is not legal in $S_{H}^{2,X}$,
which is a contradiction.
Thus, $H$ is not du-opaque.
\end{proof}
\vspace{1mm}\noindent\textbf{The unique-write case}
We now show that du-opacity is equivalent to opacity
assuming that no two transactions write identical values
to the same t-object (``unique-write'' assumption).

Let Opacity$_{uw}$ $\subseteq$ Opacity, be a property defined as follows:
\begin{enumerate}
\item
an infinite opaque history $H \in$ Opacity$_{uw}$ if and only if
every transaction $T\in \txns(H)$ is complete in $H$, and
\item
an opaque history $H \in$ Opacity$_{uw}$ if and only if
for every pair of write operations $\Write_k(X,v)$ and $\Write_m(X,v')$,
$v\neq v'$.
\end{enumerate}

\begin{theorem}
\label{th:opeq}
Opacity$_{uw}= $du-opacity.
\end{theorem}

\begin{proof}
We show first that every finite history $H \in $Opacity$_{uw}$
is also du-opaque.
Let $H$ be any finite opaque history such that for every pair of write
operations $\Write_k(X,v)$ and $\Write_m(X,v)$,
performed by transactions $T_k,T_m \in \txns(H)$, respectively,
$v\neq v'$.

Since $H$ is opaque, there is a final-state serialization $S$ of $H$.
Suppose by contradiction that $H$ is not du-opaque.
Thus, there is a $\Read_k(X)$ that returns a value $v\in V$ in $S$
that is not legal in ${S}_{H}^{k,X}$,
the local serialization of $\Read_k(X)$ with respect to $H$ and $S$.
Let ${H}^{k,X}$ and ${S}^{k,X}$ denote the prefixes of $H$ and $S$,
respectively, up to the response of $\Read_k(X)$ in $H$ and $S$.
Recall that ${S}_{H}^{k,X}$,
the local serialization of $\Read_k(X)$ with respect to $H$ and $S$,
is the subsequence of $S^{k,X}$ that does not contain events of
any transaction $T_i \in \txns(H)$ so that the invocation of $\TryC_i()$
is not in $H^{k,X}$.
Since $\Read_k(X)$ is legal in $S$,
there is a committed transaction $T_m \in \txns(H)$ that performs
$\Write_m(X,v)$ that is the latest such write in $S$ that precedes $T_k$.
Thus, if $\Read_k(X)$ is not legal in ${S}_{H}^{k,X}$,
the only possibility is that $\Read_k(X) \prec_H^{RT} \TryC_m()$.
Under the assumption of unique writes,
there does not exist any other transaction $T_j \in \txns(H)$
that performs $\Write_j(X,v)$.
Consequently, there does not exist any ${\overline{H}^{k,X}}$
(some completion of $H^{k,X}$) and a t-complete t-sequential history $S'$,
such that $S'$ is equivalent to ${\overline{H}^{k,X}}$
and $S'$ contains any committed transaction that writes $v$ to $X$.
This is, $H^{k,X}$ is not final-state opaque.
However, since $H$ is opaque, every prefix of $H$ must be final-state opaque,
which is a contradiction.

By Definition~\ref{def:opaque}, an infinite history $H$ is opaque if
every finite prefix of $H$ is final-state opaque.
Theorem~\ref{th:lc} now implies that
Opacity$_{uw}$ $\subseteq$ du-Opacity.

Definition~\ref{def:opaque} and Corollary~\ref{cr:pc} imply that
du-Opacity $\subseteq$ Opacity$_{uw}$.
\end{proof}
\vspace{1mm}\noindent\textbf{The sequential-history case}
The deferred-update semantics was mentioned by Guerraoui et
al.~\cite{GHS08-permissiveness}
and later adopted by Kuznetsov and Ravi~\cite{KR11}.
In both papers, opacity is only defined for sequential histories,
where every invocation of a t-operation is immediately
followed by a matching response.
In particular, these definitions require the final-state serialization
to respect the \emph{read-commit order}:
in these definitions, a history $H$ is opaque
if there is a final-state serialization $S$ of $H$
such that if a t-read of a t-object $X$
by a transaction $T_k$ precedes the tryC of a transaction $T_m$
that commits on $X$ in $H$, then $T_k$ precedes $T_m$ in $S$.
As we observed in Figure~\ref{fig:du-example}, this definition is not equivalent to opacity
even for sequential histories.

The property considered in~\cite{GHS08-permissiveness,KR11} is strictly
stronger than du-opacity: the sequential history $H$ in Figure~\ref{fig:du-example}
is du-opaque (and consequently opaque by Theorem~\ref{th:gkkr}):
a du-opaque serialization (in fact the only possible one) for this history is $T_1,T_3,T_2$.
However, in the restriction of opacity defined above, $T_2$
must precede $T_3$ in any serialization, since the
response of $\Read_2(X)$ precedes the invocation of $\TryC_3()$.

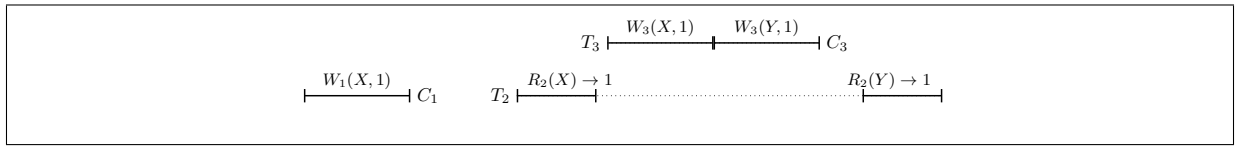
\begin{figure*}[t]
\begin{center}
\scalebox{.7}[0.7]{\begin{tikzpicture}
\node (w1) at (0.5,0) [] {};
\node (r1) at (4.5,0) [] {};
\node (r2) at (10.5,0) [] {};

\node (w2) at (6.2,1) [] {};
\node (w22) at (8.2,1) [] {};

\draw (w1) node [above] {\small {$W_1(X,1)$}};
\draw (r1) node [above] {\small {$R_2(X) \rightarrow 1$}};
\draw (r2) node [above] {\small {$R_2(Y) \rightarrow 1$}};

\draw (w2) node [above] {\small {$W_3(X,1)$}};
\draw (w22) node [above] {\small {$W_3(Y,1)$}};
%
\begin{scope}   
\draw [|-|,thick] (-0.5,0) node[left] {} to (1.5,0) node[right] {$C_1$};
\draw [|-|,thick] (3.5,0) node[left] {} to (5,0);
\draw [|-|,thick] (10,0) node[left] {} to (11.5,0);
\draw [-,dotted] (3.5,0) node[left] {$T_2$} to (11.5,0);
\end{scope}
\begin{scope}   
\draw [|-|,thick] (4.2+1,1) node[left] {} to (7.2,1);
\draw [|-|,thick] (7.2,1) node[left] {} to (9.2,1);
\draw [-,dotted] (4.2+1,1) node[left] {$T_3$} to (9.2,1) node[right] {$C_3$};
\end{scope}
\end{tikzpicture}}
\end{center}
\caption{
A sequential du-opaque history,
which is not opaque by the definition of~\cite{GHS08-permissiveness}.}
\label{fig:du-example}
\end{figure*}
%
\section{Strict serializability with DU semantics}
\label{sec:strictser}
In this section, we discuss the deferred-update restriction of strict serializability from Definition~\ref{def:oser}.
First, we remark that, just as final-state opacity, strict serializability is not prefix-closed (cf. Figure~\ref{fig:one}).
However, we show that the restriction of deferred-update semantics applied to strict serializability induces a safety property.
\begin{definition}[Strict serializability with du semantics]
\label{def:ser}
A finite history $H$ is \emph{strictly serializable} if there
is a legal t-complete t-sequential history $S$,
such that
\begin{enumerate}
\item there is a completion $\overline{H}$ of $H$,
such that $S$ is equivalent to $\ms{cseq}(\overline{H})$,
where $\ms{cseq}(\overline{H})$ is the subsequence of $\overline{H}$
reduced to committed transactions in $\overline{H}$,
\item for any two transactions $T_k,T_m \in \txns(H)$,
if $T_k \prec_H^{RT} T_m$, then $T_k$ precedes $T_m$ in $S$, and
\item each $\Read_k(X)$ in $S$ that does not return $A_k$ is
legal in $S_{H}^{k,X}$.
\end{enumerate}
\end{definition}
%
Notice that every du-opaque history is strictly serializable,
but not vice-versa.
\begin{theorem}
\label{th:sersafety}
Strict serializability is a safety property.
\end{theorem}
\begin{proof}
Observe that any
strictly serializable
serialization of a finite history $H$ does not include events of
any transaction that has not invoked \textit{tryC} in $H$.

To show prefix-closure, 
a proof almost identical to that of Lemma~\ref{lm:dusep} implies that, 
given a strictly serializable history $H$ and a serialization $S$,
there is a serialization $S'$ of $H'$ ($H'$ is some prefix of $H$)
such that $\ms{seq}(S')$ is a prefix of $\ms{seq}(S)$.

Consider an infinite sequence of finite histories
\[ H^0,\ldots , H^i, H^{i+1},\ldots , \]
where $H^{i+1}$ is a one-event extension of $H^i$,
we prove that the infinite limit $H$ of this ever-extending
sequence is strictly serializable.
As in Theorem~\ref{th:lc},
we construct an infinite rooted directed graph $G_H$:
a vertex is a tuple $(H^i,S^i)$ (note that for each $i\in \mathbb{N}$,
there are several such vertices of this form), where $S^i$ is a 
serialization of $H^i$ and there is an edge from
$(H^i,S^i)$ to $(H^{i+1},S^{i+1})$ if $\ms{seq}(S^i)$
is a prefix of $\ms{seq}(S^{i+1})$.
The resulting graph is finitely branching since the out-degree of
a vertex is bounded by the number of possible serializations of a history.
Observe that for every vertex $(H^{i+1},S^{i+1})$, there is a vertex $H^i,S^i)$ such that
$\ms{seq}(S^i)$ is a prefix of $\ms{seq}(S^{i+1})$.
Thus, $G_H$ is connected since we can iteratively construct a path from the root $(H^0,S^0)$
to every vertex $(H^i,S^i)$ in $G_H$.
Applying K\"{o}nig's Path Lemma to $G_H$,
we obtain an infinite sequence of distinct vertices,
$(H^0,S^0), (H^1, S^1), \ldots, (H^i,S^i), \ldots $.
Then, \(S=\lim\limits_{i \to \infty } S_i \)
gives the desired 
serialization of $H$.
\end{proof}
%
%
\section{Du-opacity vs. other deferred-update criteria}
\label{sec:pc3}
In this section, we first study two relaxations of opacity: \emph{Virtual-world consistency}~\cite{damien-vw-jv}
and \emph{Transactional Memory Specification~1}~\cite{DGLM13}. We then study \emph{Transactional Memory Specification~2}
which is a restriction of opacity.
\subsection{Virtual-world consistency}
\label{sec:vwc}
Virtual World Consistency (VWC)~\cite{damien-vw-jv} was proposed as a
relaxation of opacity (in our case, du-opacity), where each aborted transaction should be
consistent with its \emph{causal past} (but not necessarily with a
serialization formed by committed transactions). Intuitively, a transaction
$T_1$ causally precedes $T_2$ if $T_2$ reads a value written and
committed by $T_1$.  The original definition~\cite{damien-vw-jv} required
that no two write operations are ever invoked with the same argument
(the \emph{unique-writes} assumption).
Therefore, the causal precedence is unambiguously identified for each
transactional read. Below we give a more general definition.

Given a t-sequential legal history $S$ and  transactions
$T_i,T_j\in\txns(S)$, we say that $T_i$ \emph{reads $X$ from $T_j$} if
(1) $T_i$ reads $v$ in $X$ and (2) $T_j$ is the last committed  transaction that writes $v$ to
$X$ and precedes $T_i$ in $S$.

Now consider a (not necessarily t-sequential)  history $H$.
We say that $T_i$ \emph{could have read $X$ from} $T_j$ in $H$
if $T_j$ writes a value $v$ to a t-object $X$,
$T_i$ reads $v$ in $X$, and $\Read_i(X) \not\prec_H^{RT} \TryC_j()$.

Given $\T\subseteq\txns(H)$,
let $H^{\T}$ denote the subsequence of $H$ restricted to events of
transactions in $\T$.

\begin{definition}[du-VWC]
\label{def:vwc}
A finite history $H$ is \emph{du-virtual-world consistent} if
it is strictly serializable (with du-semantics), and
for every aborted or t-incomplete transaction $T_i \in \ms{txns}(H)$,
there is $\T\subseteq \txns(H)$ including $T_i$ and a t-sequential
t-complete legal history $S$ such that:
\begin{enumerate}
\item
$S$ is equivalent to a completion of $H^{\T}$,
\item
For all $T_j,T_k\in\txns(S)$, if $T_j$ reads $X$ from $T_k$ in $S$, then $T_j$ could have read $X$ from
$T_k$ in $H$,
\item $S$ respects the per-process order of $H$: if $T_j$ and $T_k$
  are executed by the same process and $T_j\prec_H^{RT} T_k$, then $T_j\prec_S T_k$.
\end{enumerate}
We refer to $S$ as a \emph{du-VWC serialization for $T_i$ in $H$}.
\end{definition}
Intuitively, with every t-read on $X$
performed by $T_i$ in $H$, the du-VWC serialization $S$ associates
some transaction $T_j$ from which $T_i$ could have read the value of $X$.
Recursively, with every read performed by $T_j$, $S$ associates some $T_m$
from which $T_j$ could have read, etc.
Altogether, we get a ``plausible'' causal past of $T_i$ that
constitutes a serial history.
Notice that to ensure deferred-update semantics,
we only allow a transaction $T_j$ to read from a
transaction $T_k$  that
invoked $\TryC_k$ by the time of the read operation of $T_j$.

We now prove that du-VWC is a strictly weaker property than du-opacity.
Since du-TMS2 is strictly weaker than du-opacity, it follows that
Du-TMS2 $\subsetneqq$ du-VWC.
\begin{theorem}
\label{th:vwc}
Du-opacity $\subsetneqq$ du-VWC.
\end{theorem}
\begin{proof}
If a history $H$ is du-opaque,
then there is a du-opaque serialization $S$ equivalent
to $\overline{H}$, where $\overline{H}$ is some completion of $H$.
By construction, $S$ is a total-order on the set of all transactions
that participate in $S$.
Trivially, by taking $\T=\txns(H)$, we derive that $S$ is a du-VWC serialization for
every aborted or t-incomplete transaction $T_i \in \txns(H)$.
Indeed, $S$ respects the real-time order and, thus, the
per-process order of $H$.
Since $S$ respects the deferred-update order in $H$, every t-read in $S$
``could have happened'' in $H$.

\begin{figure*}[t]
\begin{center}
\scalebox{.7}[0.7]{\begin{tikzpicture}

\node (r1) at (2.5,0) [] {};
\node (r2) at (4.5,0) [] {};
\node (w1) at (0,-1) [] {};

\node (r3) at (-3,-2) [] {};
\node (w3) at (9,-2) [] {};

\draw (r1) node [above] {\small {$R_1(X)\rightarrow 1$}}
   
   (r2) node [above] {\small {$R_1(Y)\rightarrow 0$}};
   
\draw (w1) node [above] {\small {$W_2(X,1)$}};

\draw (w3) node [above] {\small {$W_3(Y,1)$}};

\draw (r3) node [above] {\small {$R_3(X)\rightarrow 0$}};

\begin{scope}   
\draw [|-|,thick] (1.5,0) node[left] {{\small $T_1$}} to (3.5,0) ;
\draw [|-|,thick] (3.5,0) node[left] {} to (5.5,0) ;
\draw [|-|,dotted] (1.5,0) node[right] {} to (5.5,0) node[right] {{\small $A_1$}};
\end{scope}
\begin{scope}

\draw [|-|,thick] (-1,-1) node[left] {{\small $T_2$}} to (1,-1) node[right] {{\small $C_2$}} ; 
\end{scope}  
\begin{scope}
\draw [|-|,thick] (-4,-2) node[left] {} to (-2,-2) ;
\draw [|-|,thick] (8,-2) node[left] {} to (10,-2) ;
\draw [|-|,dotted] (-4,-2) node[left] {{\small $T_3$}} to (10,-2) node[right] {{\small $C_3$}}; 
\end{scope} 
\end{tikzpicture}}
\end{center}
\caption{A history that is du-VWC, but not du-opaque.}
\label{fig:vwc}
\end{figure*}
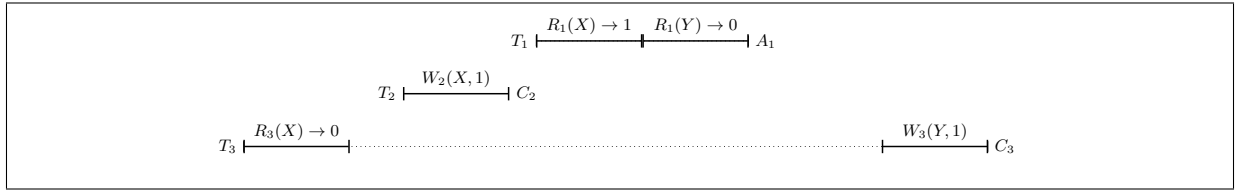

To show that the inclusion is strict,
Figure~\ref{fig:vwc} depicts a history $H$ that is du-VWC, but not du-opaque.
Clearly, $H$ is strictly serializable.
Here $T_2$, $T_1$ is the required du-VWC serialization for aborted
transaction $T_1$.
However, $H$ has no du-opaque serialization.
\qed
\end{proof}
\begin{theorem}
\label{th:lcvwc}
Du-VWC is a safety property.
\end{theorem}
\begin{proof}
By Definition~\ref{def:vwc}, a history $H$ is du-VWC if and only if
$H$ is strictly serializable and there is a du-VWC serialization
for every transaction $T_i \in \ms{txns}(H)$
that is aborted or t-incomplete in $H$.

To prove prefix-closure, recall that
strict serializability is a prefix-closed property
(Theorem~\ref{th:sersafety}).
Therefore, any du-VWC serialization $S$ for a transaction $T_i$
in history $H$ is also a du-VWC serialization $S$ for a transaction
$T_i$ in any prefix of $H$ that contains events of $T_i$.

To prove limit-closure,
%
consider an infinite sequence of du-VWC histories
$H^0$, $H^1$, $\ldots$,  $H^i$, $H^{i+1}$ , $\ldots$, where each $H^{i+1}$ is the one-event extension of
$H^i$ and prove that the infinite limit, $H$ of this sequence is also a du-VWC history.
Theorem~\ref{th:sersafety} establishes that there is a
strictly serializable serialization for $H$.

Since, for all $i\in \mathbb{N}$, $H^i$ is du-VWC,
for every transaction $T_i$ that is t-incomplete or aborted in $H^i$,
there is a VWC serialization for $T_i$.
Consequently, there is a du-VWC serialization for
every aborted or incomplete transaction $T_i$ in $H$. 
\end{proof}
  %
\subsection{Transactional memory specification (TMS)}
\label{sec:tms}
\emph{Transactional Memory Specification} (TMS) 1 and 2 were
formulated in I/O automata~\cite{DGLM13}.
Following~\cite{AttiyaGHR2014},
we adapt these definitions to our framework and
explicitly introduce the deferred-update requirement.

\vspace{1mm}\noindent\textbf{TMS1.}
Given a history $H$,
TMS1 requires us to justify the behavior of all committed transactions in $H$
by a legal t-complete t-sequential history that
preserves the real-time order in $H$ (strict serializability),
and to justify the response of each complete t-operation performed in $H$
by a legal t-complete t-sequential history $S$.
The t-sequential history $S$ used to justify a complete t-operation $op_{i,k}$
(the $i^{th}$ t-operation performed by transaction $T_k$)
includes $T_k$ and
a subset of transactions from $H$ whose operations justify $op_{i,k}$.
(Our description follows~\cite{AttiyaGHR2014}.)

Let ${H}^{k,i}$ denote the prefix of a history $H$ up to (and including)
the response of $i^{th}$ t-operation $op_{k,i}$ of transaction $T_k$.
We say that a history $H''$ is a \emph{possible past} of $H^{k,i}$
if $H''$ is a subsequence of $H^{k,i}$ and
consists of all events of transaction $T_k$ and
all events from some subset of committed transactions
and transactions that have invoked $\TryC$ in $H^{k,i}$
such that if a transaction $T\in H''$,
then for a transaction $T'\prec_{H^{k,i}}^{RT} T$, $T'\in H''$
if and only if $T'$ is committed in $H^{k,i}$.
Let $\ms{cTMSpast}(H,op_{k,i})$ denote the set of possible pasts of $H^{k,i}$.

For any history $H'' \in \ms{cTMSpast}(H,op_{k,i})$,
let $\ms{ccomp}(H'')$ denote the history generated from $H''$ by the following procedure:
for all $m\neq k$, replace every event $A_m$ by $C_m$
and complete every incomplete $\TryC_m$ with including $C_m$ at the end of $H''$;
include $\TryC_k\cdot A_k$ at the end of $H''$.
%
%
\begin{definition}[du-TMS1]
\label{def:tms1}
A history $H$ satisfies \emph{du-TMS1} if
\begin{enumerate}
\item $H$ is strictly serializable (with du-semantics), and
\item for each complete t-read $op_{i,k}$ that returns a non-$A_k$ response in
$H$, there exist a legal t-complete t-sequential history $S$ and a history $H'$ such that:
\begin{itemize}
\item[--] $H'=\ms{ccomp}(H'')$, where $H''\in \ms{cTMSpast}(H,op_{k,i})$
\item[--] $H'$ is equivalent to $S$
\item[--] for any two transactions $T_k$ and $T_m$ in $H'$, if $T_k \prec_{H'}^{RT} T_m$ then $T_k <_S T_m$
\end{itemize}
\end{enumerate}

We refer to $S$ as the du-TMS1 serialization for $op_{i,k}$.
\end{definition}
\begin{theorem}
\label{th:tms1safety}
Du-TMS1 is a safety property.
\end{theorem}
\begin{proof}
A history $H$ is du-TMS1 if and only if $H$ is strictly serializable
and there is a du-TMS1 serialization for every t-operation $op_{k,i}$
that does not return $A_k$ in $H$.

To see that du-TMS1 is prefix closed,
recall that strict serializability is a prefix-closed property.
Let $H$ be any du-TMS1 history and $H^i$, any prefix of $H$.
We now need to prove that, for every t-operation $op_{k,i} \neq \TryC_k$ that returns a non-$A_k$ response in $H^i$,
there is a du-TMS1 serialization for $op_{k,i}$.
But this is immediate since the du-TMS1 serialization for $op_{i,k}$ in $H$ is also the required du-TMS1 serialization
for $op_{k,i}$ in $H^i$.

To see that du-TMS1 is limit closed,
consider an infinite sequence
\[ H^0,H^1, \ldots H^i,H^{i+1},\ldots \]
of finite du-TMS1 histories,
such that $H^{i+1}$ is a one-event extension of $H^i$.
Let let $H$ be the corresponding infinite limit history.
We want to show that $H$ is also du-TMS1.

Since strict serializability is a limit-closed property
(Theorem~\ref{th:sersafety}), $H$ is strictly serializable.
By assumption, for all $i\in \mathbb{N}$, $H^i$ is du-TMS1.
Thus, for every transaction $T_i$ that participates in $H^i$,
there is a du-TMS1 serialization $S^{i,k}$
for each t-operation $op_{k,i}$.
But $S^{i,k}$ is also the required
du-TMS1 serialization for $op_{k,1}$ in $H$.
The claim follows. 
\end{proof}
It has been shown~\cite{TMS-WTTM} that
Opacity is a strictly stronger property than du-TMS1,
that is, Opacity $\subsetneqq$ du-TMS1.
Since Du-Opacity $\subsetneqq$ Opacity (Theorem~\ref{th:gkkr})
it follows that Du-Opacity $\subsetneqq$ du-TMS1.
On the other hand, du-TMS1 is incomparable to du-VWC,
as demonstrated by the following examples.
\begin{figure*}[t]
\begin{center}
\scalebox{.7}[0.7]{ \begin{tikzpicture}[font=\small]
\node (w1) at (7.3,-1) [] {};
\node (c1) at (9,-1) [] {};

\node (w2) at (2.7,-2) [] {};
\node (c2) at (5.2,-2) [] {};

\node (r3) at (12,-3) [] {};
\node (a3) at (13.8,-3) [] {};

\draw (w1) node [above] { \small {$W_1(X,1)$}};
\draw (c1) node [above] {\small {$\TryC_1$}};

\draw (w2) node [above] {\small {$W_2(X,0)$}};
\draw (c2) node [above] {\small {$\TryC_2$}};

\draw (r3) node [above] {\small {$R_3(X) \rightarrow 0$}};
\draw (a3) node [above] {\small {$\TryC_3$}};

\begin{scope}
\draw [|-|,dotted] (6.6,-1) node[left] {$T_1$} to (9.5,-1) node[right] {$C_1$};
\draw [|-|,thick] (6.6,-1) node[above] {} to (8,-1);
\draw [|-|,thick] (8,-1) node[above] {} to (9.5,-1);
\end{scope}
\begin{scope}
\draw [|-|,dotted] (2,-2) node[left] {$T_2$} to (6,-2) node[right] {$C_2$};
\draw [|-|,thick] (2,-2) node[above] {} to (4,-2);
\draw [|-|,thick] (4,-2) node[above] {} to (6,-2);
\end{scope}
 \begin{scope}
\draw [|-|,dotted] (11.4,-3) node[left] {$T_3$} to (14.4,-3) node[right] {$A_3$};
\draw [|-|,thick] (11.4,-3) node[above] {} to (13,-3);
\draw [|-|,thick] (13,-3) node[above] {} to (14.4,-3);
\end{scope}
\end{tikzpicture}}
\end{center}
\caption{A history which is du-VWC but not du-TMS1.}
\label{fig:vwctms1}
\end{figure*}
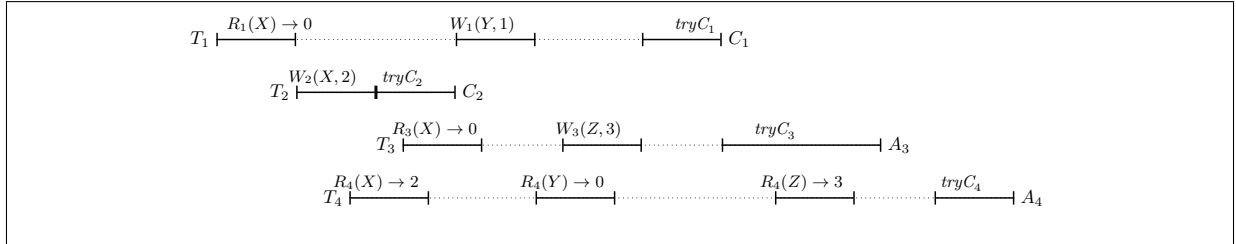
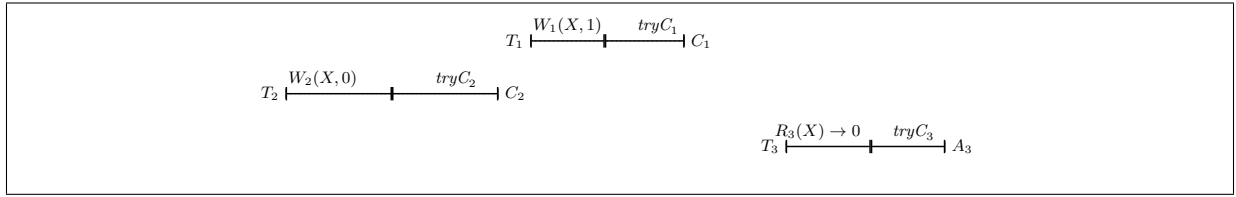
\begin{figure*}[t]
\begin{center}
\scalebox{.7}[0.7]{\begin{tikzpicture}
\node (r1) at (1.5,0) [] {};
\node (w1) at (5.5,0) [] {};
\node (w2) at (2.5,-1) [] {};
\node (r3) at (4.6,-2) [] {};
\node (w3) at (7.5,-2) [] {};
\node (r41) at (3.5,-3) [] {};
\node (r42) at (7,-3) [] {};
\node (r43) at (11.5,-3) [] {};
\node (c1) at (9.5,0) [] {};
\node (c2) at (4,-1) [] {};
\node (a3) at (11,-2) [] {};
\node (a4) at (14.5,-3) [] {};
\draw (c1) node [above] {\small {$\TryC_1$}};
\draw (c2) node [above] {\small {$\TryC_2 $}};
\draw (a3) node [above] {\small {$\TryC_3 $}};
\draw (a4) node [above] {\small {$\TryC_4 $}};

\draw (r1) node [above] {\small {$R_1(X)\rightarrow 0$}};
\draw (w1) node [above] {\small {$W_1(Y,1)$}};
\draw (w2) node [above] {\small {$W_2(X,2)$}};
\draw (r3) node [above] {\small {$R_3(X)\rightarrow 0$}};
\draw (w3) node [above] {\small {$W_3(Z,3)$}};
\draw (r41) node [above] {\small {$R_4(X)\rightarrow 2$}};
\draw (r42) node [above] {\small {$R_4(Y)\rightarrow 0$}};
\draw (r43) node [above] {\small {$R_4(Z)\rightarrow 3$}};

\begin{scope}
\draw [-,dotted] (.5,0) node[left] {$T_1$} to (10,0) node[right] {$C_1$};
\draw [|-|,thick] (0.5,0) node[above] {} to (2,0);
\draw [|-|,thick] (5,0) node[above] {} to (6.5,0);
\draw [|-|,thick] (8.5,0) node[right] {} to (10,0);
\end{scope}
%
\begin{scope}
\draw [-,dotted] (2,-1) node[left] {$T_2$} to (5,-1) node[right] {$C_2$};
\draw [|-|,thick] (2,-1) node[left] {} to (3.5,-1);
\draw [|-|,thick] (3.5,-1) node[right] {} to (5,-1);
\end{scope}
%
\begin{scope}
\draw [-,dotted] (4,-2) node[left] {$T_3$} to (13,-2) node[right] {$A_3$};
\draw [|-|,thick] (4,-2) node[left] {} to (5.5,-2);
\draw [|-|,thick] (7,-2) node[above] {} to (8.5,-2);
\draw [|-|,thick] (10,-2) node[right] {} to (13,-2);
\end{scope}
%
\begin{scope}
\draw [-,dotted] (3,-3) node[left] {$T_4$} to (15.5,-3) node[right] {$A_4$};
\draw [|-|,thick] (3,-3) node[left] {} to (4.5,-3);
\draw [|-|,thick] (6.5,-3) node[above] {} to (8,-3);
\draw [|-|,thick] (11,-3) node[right] {} to (12.5,-3);
\draw [|-|,thick] (14,-3) node[above] {} to (15.5,-3);

\end{scope}
\end{tikzpicture}}
\end{center}
\caption{A history which is du-TMS1 but not du-VWC.}
\label{fig:tms1vwc}
\end{figure*}
\begin{proposition}
There is a history that is du-TMS1, but not du-VWC.
\end{proposition}
\begin{proof}
Figure~\ref{fig:tms1vwc} depicts a history $H$ that is du-TMS1, but not du-VWC.
Observe that $H$ is strictly serializable.
To prove that $H$ is du-TMS1, we need to prove that there is a TMS1 serialization for each
t-read that returns a non-abort response in $H$.
Clearly, the serialization in which only $T_3$ participates is the required
TMS1 serialization for $\Read_3(X)\rightarrow 0$.
Now consider the aborted transaction $T_4$.
The TMS1 serialization for $\Read_4(X) \rightarrow 2$ is $T_2,T_4$, while the TMS1 serialization that justifies the response of
$\Read_4(Y)->0$ includes just $T_4$ itself.
The only nontrivial t-read whose response needs to be justified is $\Read_4(Z)\rightarrow 3$.
Indeed, $\TryC_3$ overlaps with $\Read_4(Z)$ and thus, the response of $\Read_4(Z)$ can be justified by
choosing transactions in $\ms{cTMSpart}(H,\Read_4(Z))$ to be $\{T_3,T_2,T_4\}$ and then
deriving a TMS1 serialization $S=T_3,T_2,T_4$ for $\Read_4(Z)\rightarrow 3$ in which
$\TryC_3$ may be completed by including the commit response.

However, $H$ is not du-VWC. Consider transaction $T_3$ which returns $A_3$ in $H$: $T_3$ must be aborted in any serialization
equivalent to some direct causal past of $T_4$. But $\Read_4(Z)$ returns the value $3$ that is written by $T_3$.
Thus, $\Read_4(Z)$
cannot be legal in any du-VWC serialization for $T_4$.
\end{proof}

\begin{proposition}
There is a history that is du-VWC, but not du-TMS1.
\end{proposition}

\begin{proof}
Figure~\ref{fig:vwctms1} depicts a history $H$ that is du-VWC, but not du-TMS1.
Clearly, $H$ is strictly serializable.
Observe that $T_3$ could have read only from $T_1$ in $H$ ($T_1$ writes the value $0$ to $X$ that is returned by $\Read_3(X)$).
Therefore, $T_1, T_3$ is the required du-VWC serialization for aborted transaction $T_3$.

However, $H$ is not du-TMS1: since both transactions $T_1$ and $T_2$ are committed and
precede $T_3$ in real-time order, they must be included
in any du-TMS1 serialization for $\Read_3(X) \rightarrow 0$.
But there is no such du-TMS1 serialization that would ensure the legality of $\Read_3(X)$.
\end{proof}
\vspace{1mm}\noindent\textbf{TMS2.}
We now study the TMS2 definition which imposes an extra restriction on the opaque serialization.
\begin{definition}[du-TMS2]
\label{def:tms2}
A history $H$ is \emph{du-TMS2} if there is
a legal t-complete t-sequential history $S$
equivalent to some completion, $\overline{H}$ of $H$ such that
\begin{enumerate}
\item
for any two transactions $T_k,T_m \in \txns(H)$, such that $T_m$ is a committed updating transaction,
if $C_k \prec_H^{RT} \TryC_m$ or $A_k \prec_H^{RT} \TryC_m$, then $T_k \prec_S T_m$, and
\item
for any two transactions $T_k,T_m \in \txns(H)$,
if $T_k \prec_H^{RT} T_m$, then $T_k <_S T_m$, and
\item
each $\Read_k(X)$ in $S$ that does not return $A_k$ is
legal in $S_{H}^{k,X}$.
\end{enumerate}
We refer to $S$ as the du-TMS2 serialization of $H$.
\end{definition}

It has been shown~\cite{TMS-WTTM} that
TMS2 is a strictly stronger property than Opacity,
\emph{i.e.}, TMS2 $\subsetneqq$ Opacity.
We now show that du-TMS2 is strictly stronger than du-opacity.
Indeed, from Definition~\ref{def:tms2},
we observe that every history that is du-TMS2 is also du-opaque.
The following proposition completes the proof.
\begin{proposition}
\label{lm:tms}
There is a history that is du-opaque, but not du-TMS2.
\end{proposition}
\begin{proof}
Figure~\ref{fig:tms2} depicts a history $H$ that is du-opaque, but not du-TMS2.
Indeed, there is a du-opaque serialization $S$ of $H$ such that
$\ms{seq}(S)=T_2,T_1$.
On the other hand, since $T_1$ commits before $T_2$, $T_1$ must precede $T_2$ in any du-TMS2 serialization,
there does not exist any such serialization that ensures every t-read is legal. Thus, $H$ is not du-TMS2.
\end{proof}
\begin{figure*}[t]
\begin{center}
\scalebox{.7}[0.7]{\begin{tikzpicture}
\node (r1) at (0.7,0) [] {};
\node (w1) at (3.5,0) [] {};
\node (c1) at (7.5,0) [] {};

\node (r2) at (2.2,-1) [] {};
\node (w2) at (5.5,-1) [] {};
\node (c2) at (9,-1) [] {};

\draw (r1) node [above] {\small {$R_1(X)\rightarrow 0$}};
\draw (w1) node [above] {\small {$W_1(X,1)$}};
\draw (c1) node [above] {\small {$\TryC_1$}};

\draw (r2) node [above] {\small {$R_2(X)\rightarrow 0$}};
\draw (w2) node [above] {\small {$W_2(Y,1)$}};
\draw (c2) node [above] {\small {$\TryC_2$}};
\begin{scope}   
\draw [|-|,thick] (-0.5,0) node[left] {} to (1,0);
\draw [|-|,thick] (3,0) node[left] {} to (4.5,0);
\draw [|-|,thick] (7,0) node[left] {} to (8.5,0);
\draw [-,dotted] (-0.5,0) node[left] {$T_1$} to (8.5,0) node[right] {$C_1$};
\end{scope}
\begin{scope}   
\draw [|-|,thick] (1.5,-1) node[left] {} to (3,-1);
\draw [|-|,thick] (5,-1) node[left] {} to (6.5,-1);
\draw [|-|,thick] (8.5,-1) node[left] {} to (10,-1);
\draw [-,dotted] (1.5,-1) node[left] {$T_2$} to (10,-1) node[right] {$C_2$};
\end{scope}
\end{tikzpicture}}
\end{center}
\caption{A history that is du-opaque, but not TMS2~\cite{DGLM13}.}
\label{fig:tms2}
\end{figure*}
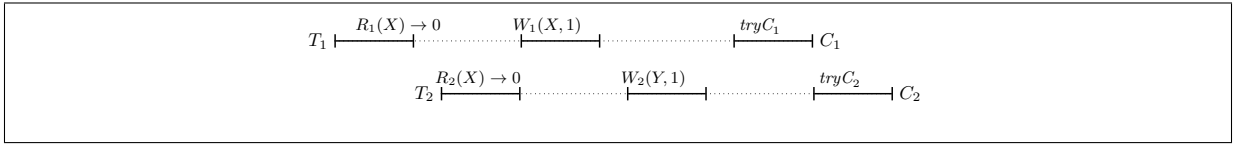
\begin{theorem}
\label{th:tms2pc}
Du-TMS2 is prefix-closed.
\end{theorem}
\begin{proof}
Let $H$ be any du-TMS2 history. Then, $H$ is also du-opaque.
By Corollary~\ref{cr:pc}, for every $i\in \mathbb{N}$,
there is a du-opaque serialization $S^i$ for $H^i$.
We now need to prove that,
for any two transactions $T_k,T_m \in \txns(H^i)$, such that $T_m$ is a committed updating transaction,
if $C_k \prec_{H^i}^{RT} \TryC_m$ or $A_k \prec_{H^i}^{RT} \TryC_m$,
there is a du-opaque serialization $S^i$ with the restriction that
$T_k \prec_{S^i} T_m$.

Suppose by contradiction that there exist transactions $T_k,T_m \in \txns(H^i)$,
such that $T_m$ is a committed updating transaction and $C_k \prec_{H^i}^{RT} \TryC_m$ or $A_k \prec_{H^i}^{RT} \TryC_m$,
but $T_m$ must precede $T_k$ in any du-opaque serialization $S^i$.
Since $T_m \not\prec_{H^i}^{RT} T_k$, the only possibility is that $T_m$ performs $\Write_m(X,v)$
and there is $\Read_k(X) \rightarrow v$.
However, by our assumption, $\Write_k(X,v) \prec_{H^i}^{RT} \TryC_m$:
thus, $\Read_k(X)$ is not legal in its local serialization
with respect to $H^i$ and $S^i$---contradicting the assumption that
$S^i$ is a du-opaque serialization of $H^i$.
Thus, there is a du-TMS2 serialization for $H^i$,
proving that du-TMS2 is a prefix-closed property.
\end{proof}

\begin{proposition}
Du-TMS2 is not limit-closed.
\end{proposition}

\begin{proof}
The counter-example to establish that du-opacity is not limit-closed
(Figure~\ref{fig:op-example}) also shows that du-TMS2 is not limit-closed:
all histories discussed in the counter-example are in du-TMS2.
\end{proof}
%
%
\section{Related work and Discussion}
\label{sec:pc4}
The properties discussed in this chapter explicitly preclude reading
from a transaction that has not yet invoked \textit{tryCommit},
which makes them prefix-closed and facilitates their verification.
We believe that this constructive definition is useful to TM
practitioners, since it streamlines possible implementations of t-read
and tryCommit operations.

We showed that du-opacity is limit-closed under the restriction
that every operation eventually terminates,
while du-VWC and du-TMS1 are (unconditionally) limit-closed,
which makes them safety properties~\cite{Lyn96}.

Table~\ref{table:relations} 
summarizes the containment relations between the
properties discussed in this chapter:
opacity, du-opacity, du-VWC, du-TMS1 and du-TMS2.
For example, ``du-opacity $\subsetneq$ opacity'' means that the set of
du-opaque histories is a proper subset of the set of opaque histories,
\emph{i.e.}, du-opacity is a strictly stronger property than opacity.
Incomparable (not related by containment) properties, such as du-TMS1 and
du-VWC are marked with $\times$.

%
%


\emph{Linearizability}~\cite{HW90,AW04},
when applied to objects with \emph{finite nondeterminism}
(\emph{i.e.}, an operation applied to a given state
may produce only finitely many outcomes) sequential
specifications is a safety property~\cite{Lyn96,GR14-lin}.
Recently, it has been shown~\cite{GR14-lin} that 
linearizability is not limit-closed if the implemented 
object may expose infinite non-determinism~\cite{GR14-lin},
that is, an operation applied to a given state may produce
infinitely many different outcomes.
The limit-closure proof (cf. Theorem~\ref{th:lc}), 
using K\"onig's lemma, 
cannot be applied with infinite non-determinism,
because the out-degree of the graph $G_H$, constructed 
for the limit infinite history $H$, is not finite.
%

In contrast, the TM abstraction is \emph{deterministic},
since reads and writes behave deterministically
in serial executions, yet du-opacity is not limit-closed.
It turns out that the graph $G_H$ for 
the counter-example history $H$ in Figure~\ref{fig:op-example}
is not connected.
For example, one of the finite prefixes of $H$ can be
serialized as $T_3,\,T_1,\,T_2$, but no prefix has a serialization
$T_3,\,T_1$ and, thus, the root is not connected to the corresponding
vertex of $G_H$.
Thus, the precondition of K\"onig's lemma does not hold for $G_H$:
the graph is in fact an infinite set of isolated vertices.
This is because du-opacity requires even incomplete
reading transactions, such as $T_2$, to appear in the serialization,
which is not the case for linearizability, where incomplete operations
may be removed from the linearization.
\begin{table}
      \begin{center}
      \scalebox{0.85}[0.85]{
      
     \begin{tabular}{l|l|l|l|l|l}
	~~~~ & du-opacity & du-VWC & du-TMS1 & du-TMS2\\ \hline
	du-opacity   &  & $\subsetneqq$ & $\subsetneqq$  & $\supsetneqq$\\ \hline
	du-VWC   & $\supsetneqq$ &  &$\times$  &$\supsetneqq$   \\ \hline
	du-TMS1  & $\supsetneqq$ & $\times$ &  &$\supsetneqq$  \\ \hline
	du-TMS2  & $\subsetneqq$ & $\subsetneqq$ & $\subsetneqq$ &

     \end{tabular}
     }
     \end{center}
     \caption{Relations between TM consistency definitions.\label{table:relations}}
\end{table}


\chapter{Complexity bounds for blocking TMs}
\label{ch:p3c2}
\epigraph{"I can't believe that!" said Alice\\
"Can't you?" the Queen said in a pitying tone. "Try again: draw a long breath, and shut your eyes."\\
Alice laughed. "There's no use trying," she said: "one can't believe impossible things."\\
"I daresay you haven't had much practice," said the Queen. "When I was your age, I always did it for half-an-hour a day. 
Why, sometimes I've believed as many as six impossible things before breakfast."}
{\textit{Lewis Carroll}-Through the Looking-Glass}
%
\section{Overview}
\label{sec:p3c2s0}
In this chapter, we present complexity bounds for TM implementations that provide no non-blocking progress guarantees for 
transactions and typically allow a transaction to \emph{block} (delay) or abort in concurrent executions.
We refer to Section~\ref{sec:complexity} in Chapter~\ref{ch:tm-model} for an overview of the complexity metrics
considered in the thesis.

\vspace{1mm}\noindent\textbf{Sequential TMs.}
We start by presenting complexity bounds for \emph{single-lock} TMs that satisfy sequential TM-progress.
We show that a read-only transaction in an opaque TM featured with 
weak DAP, 
weak invisible reads, ICF TM-liveness and sequential TM-progress must \emph{incrementally} validate
every next read operation. This results in a quadratic (in the size of the transaction's read
set) step-complexity lower bound.
Secondly, we prove that if the TM-correctness property is weakened to strict serializability, there exist executions
in which the tryCommit of some transaction 
must access a linear (in the size of the transaction's read set) number of distinct base objects.
We then show that expensive synchronization in TMs cannot be eliminated: even single-lock TMs must perform a 
RAW (read-after-write) or AWAR (atomic-write-after-read)
pattern~\cite{AGK11-popl}. 

\vspace{1mm}\noindent\textbf{Progressive TMs.}
We turn our focus to \emph{progressive} TM implementations which allow a transaction to be aborted only
due to read-write conflicts with concurrent transactions.
We introduce a new metric called \emph{protected data size}
that, intuitively, captures the amount of data that a transaction 
must exclusively control at some point of its execution. 
All progressive TM implementations we are aware of 
(see, \emph{e.g.}, an overview in~\cite{GK09-progressiveness})
use locks or timing assumptions to give an updating transaction exclusive access  
to all objects in its write set at some point of its execution.
For example, lock-based progressive implementations like \emph{TL}~\cite{DStransaction06} and \emph{TL2}~\cite{DSS06} 
require that a
transaction grabs all locks on its write set
before updating  the corresponding base objects.    
Our result shows that this is an inherent price to pay for
providing progressive concurrency: every committed transaction in a progressive and 
strict DAP TM implementation providing starvation-free TM-liveness
must, at some point of its execution, protect every t-object 
in its write set.

We also present a very cheap progressive opaque strict DAP TM implementation from read-write base objects
with constant expensive synchronization and constant memory stall complexities.

\vspace{1mm}\noindent\textbf{Strongly progressive TMs.}
We then prove that in any \emph{strongly progressive} strictly serializable TM
implementation that accesses the shared memory with read, write and conditional
primitives, such as compare-and-swap and
load-linked/store-conditional, 
the total number of \emph{remote memory references} (RMRs) 
that take place in an execution of a progressive TM in which $n$
concurrent processes perform transactions on a single t-object might
reach $\Omega(n \log n)$. 
The result is obtained via a reduction to an analogous lower bound for
mutual exclusion~\cite{rmr-mutex}.
In the reduction, we show that any TM with the above properties can be
used to implement a \emph{deadlock-free} mutual exclusion, employing
transactional operations on only one t-object and incurring a constant RMR overhead. 
The lower bound applies to RMRs in both the \emph{cache-coherent (CC)} and
\emph{distributed shared memory (DSM)} models, and it appears to be the first RMR complexity
lower bound for transactional memory.

We also present a constant expensive synchronization strongly progressive TM implementation from read-write base objects.
Our implementation provides \emph{starvation-free} TM-liveness, thus showing one means of circumventing the lower bound
of Rachid \emph{et al.}~\cite{tm-book} who proved the impossibility of implementing strongly progressive
strictly serializable TMs providing \emph{wait-free} TM-liveness from read-write base objects.

\vspace{1mm}\noindent\textbf{Permissive TMs.}
We conclude our study of blocking TMs by establishing a linear (in the transaction's data set size)
separation between the worst-case transaction expensive synchronization complexity of strongly progressive TMs and 
\emph{permissive} TMs that allow a transaction to abort only if committing it would violate opacity.
Specifically, we show that an execution of a transaction 
in a \emph{permissive opaque} TM implementation that provides starvation-free TM-liveness may require to perform
at least one RAW/AWAR pattern \emph{per} t-read.

\vspace{1mm}\noindent\textbf{Roadmap of Chapter~\ref{ch:p3c2}.}
Section~\ref{sec:p3c2s1}
studies ``single-lock'' TMs that provide minimal progressiveness or sequential TM-progress,
Section~\ref{sec:p3c2s2} is devoted to progressive TMs while Section~\ref{sec:p3c2s3} is on
strongly progressive TMs.
In Section~\ref{sec:perm}, we study the cost of permissive TMs that allow a transaction to abort
only if committing it would violate opacity. Finally, we present related work and open questions in Section~\ref{sec:p3c2disc}.
\section{Sequential TMs}
\label{sec:p3c2s1}
We begin with ``single-lock'', \emph{i.e.}, sequential TMs.
Our first result proves that a read-only transaction in a sequential TM featured with weak DAP
and weak invisible reads must incur the cost of validating its read set. 
This results in a quadratic (and resp., linear)  (in the size of the transaction's read
set) step-complexity lower bound if we assume opacity (and resp., strict serializability).
Secondly, we show that expensive synchronization cannot be avoided even in such sequential TMs, \emph{i.e.},
a serializable TM must perform a RAW/AWAR even when transactions are guaranteed to commit only in the absence of any concurrency.

\begin{figure*}[t]
\begin{center}
	\subfloat[$R_{\phi}(X_{i})$ must return $nv$ by strict serializability \label{sfig:readinv-1}]{\scalebox{0.6}[0.6]{\begin{tikzpicture}
\node (r1) at (3,0) [] {};
\node (r2) at (7.7,0) [] {};

\node (w1) at (-2,0) [] {};

\draw (r1) node [below] {\normalsize {$R_{\phi}(X_1) \cdots R_{\phi}(X_{i-1})$}};
\draw (r1) node [above] {\normalsize {$i-1$ t-reads}};

\draw (r2) node [above] {\normalsize {$R_{\phi}(X_i)\rightarrow nv$}};

\draw (w1) node [above] {\normalsize {$W_i(X_i,nv)$}}; 
\draw (w1) node [below] {\normalsize {$T_i$ commits}};

\begin{scope}   
\draw [|-|,thick] (0,0) node[left] {$T_{\phi}$} to (6,0);
\draw [|-|,thick] (6.5,0) node[left] {} to (9,0);
\draw [-,dotted] (0,0) node[left] {} to (9,0);
\end{scope}
\begin{scope}   
\draw [|-|,thick] (-3,0) node[left] {$T_i$} to (-1,0);
\end{scope}
\end{tikzpicture}}}
        \\
        \vspace{2mm}
	\subfloat[$T_i$ does not observe any conflict with $T_{\phi}$ \label{sfig:readinv-2}]{\scalebox{0.6}[0.6]{\begin{tikzpicture}
\node (r1) at (3,0) [] {};
\node (r3) at (12.2,0) [] {};


\node (w2) at (7.5,-2) [] {};

\draw (r1) node [below] {\small {$R_{\phi}(X_1) \cdots R_{\phi}(X_{i-1})$}};
\draw (r1) node [above] {\small {$i-1$ t-reads}};

\draw (w2) node [above] {\small {$W_{i}(X_{i},nv)$}}; 
\draw (w2) node [below] {\small {$T_{i}$ commits}};

\draw (r3) node [above] {\small {$R_{\phi}(X_{i})\rightarrow nv$}};
\draw (r3) node [below] {\small {new value}};

\begin{scope}   
\draw [|-|,thick] (0,0) node[left] {$T_{\phi}$} to (6,0);
\draw [|-|,dotted] (0,0) node[left] {$T_{\phi}$} to (13.5,0);
\draw [|-|,thick] (11,0) node[left] {} to (13.5,0);
\end{scope}
\begin{scope}   
\draw [|-|,thick] (6.5,-2) node[left] {$T_i$} to (9,-2);
\end{scope}
\end{tikzpicture}}}
	\caption{Executions in the proof of Lemma~\ref{lm:readdap}; By weak DAP, $T_{\phi}$ cannot distinguish this from the execution in Figure~\ref{sfig:readinv-1}
        \label{fig:indis}} 
\end{center}
\end{figure*}
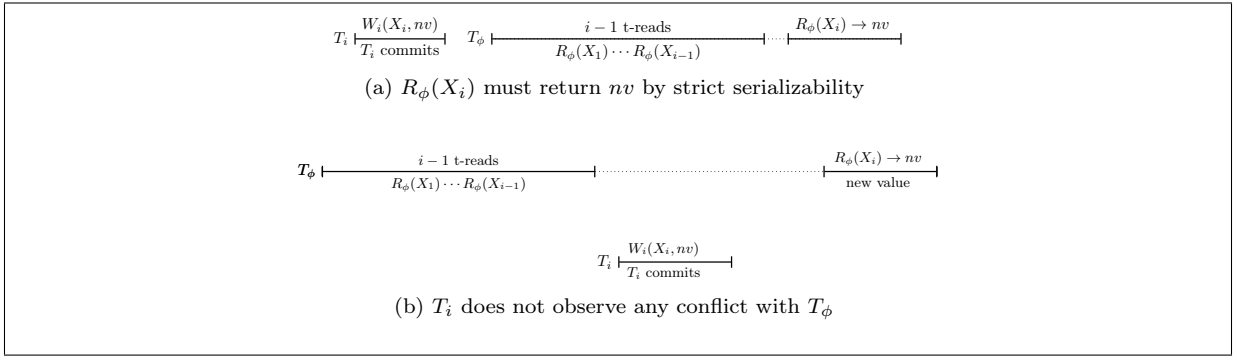
We first prove the following auxiliary lemma that will be of use in subsequent proofs.
\begin{lemma}
\label{lm:readdap}
Let $M$ be any strictly serializable, weak DAP TM implementation that provides sequential TM-progress and sequential TM-liveness.
Then, for all $i\in \mathbb{N}$, $M$ has an execution of the form $\pi^{i-1}\cdot \rho^i\cdot \alpha^i$ where,
\begin{itemize}
\item
$\pi^{i-1}$ is the complete step contention-free execution of read-only transaction $T_{\phi}$ that performs
$(i-1)$ t-reads: $\Read_{\phi}(X_1)\cdots \Read_{\phi}(X_{i-1})$,
\item
$\rho^i$ is the t-complete step contention-free execution of a transaction $T_{i}$
that writes $nv\neq v$ to $X_i$ and commits ($v$ is the initial value of $X-i$),
\item
$\alpha_i$ is the complete step contention-free execution fragment of $T_{\phi}$ that performs its $i^{th}$ t-read:
$\Read_{\phi}(X_i) \rightarrow nv_i$.
\end{itemize}
\end{lemma}
\begin{proof}
By sequential TM-progress and sequential TM-liveness, $M$ has an execution of the form $\rho^i\cdot \pi^{i-1}$.
Since $\Dset(T_k) \cap \Dset(T_{i}) =\emptyset$ in $\rho^i\cdot \pi^{i-1}$,
by Lemma~\ref{lm:dap}, transactions $T_{\phi}$ and $T_i$ do not contend
on any base object in execution $\rho^i\cdot \pi^{i-1}$.
Thus, $\rho^i\cdot \pi^{i-1}$ is also an execution of $M$.

By assumption of strict serializability, $\rho^i\cdot \pi^{i-1} \cdot \alpha_i$ is an execution
of $M$ in which the t-read of $X_i$ performed by $T_{\phi}$ must return $nv$.
But $\rho^i \cdot \pi^{i-1} \cdot \alpha_i$ is indistinguishable to $T_{\phi}$ from
$\pi^{i-1}\cdot \rho^i \cdot \alpha_i$.
Thus, $M$ has an execution of the form $\pi^{i-1}\cdot \rho^i \cdot \alpha_i$.
\end{proof}
\subsection{A quadratic lower bound on step complexity}
\label{sec:p3c2ss1}
In this section, we present our step complexity lower bound for sequential TMs.
\begin{theorem}
\label{th:iclb}
For every weak DAP TM implementation $M$ that provides ICF TM-liveness, sequential TM-progress and uses weak invisible reads,
\begin{enumerate}
\item[(1)]
If $M$ is opaque,
for every $m\in \mathbb{N}$,
there exists an execution $E$ of $M$ 
such that some transaction $T\in \ms{txns}(E)$ performs $\Omega(m^2)$ steps, where $m=|\Rset(T_k)|$.
\item[(2)]
if $M$ is strictly serializable,
for every $m\in \mathbb{N}$,
there exists an execution $E$ of $M$
such that some transaction $T_k\in \ms{txns}(E)$
accesses at least $m-1$ distinct base objects during the executions of the $m^{th}$ t-read operation
and $\TryC_k()$, where $m=|\Rset(T_k)|$.
\end{enumerate}
\end{theorem}
\begin{proof}
For all $i\in \{1,\ldots , m\}$, let $v$ be the initial value of t-object $X_i$.

($1$) Suppose that $M$ is opaque.
Let $\pi^{m}$ denote the complete step contention-free execution of a transaction
$T_{\phi}$ that performs ${m}$ t-reads: $\Read_{\phi}(X_1)\cdots \Read_{\phi}(X_{m})$
such that for all $i\in \{1,\ldots , m \}$, $\Read_{\phi}(X_i) \rightarrow v$.

By Lemma~\ref{lm:readdap}, for all $i\in \{2,\ldots, m\}$, $M$ has an execution of the form 
$E^{i}=\pi^{i-1}\cdot \rho^i \cdot \alpha_i$.

For each $i\in \{2,\ldots, m\}$, $j\in \{1,2\}$ and $\ell \leq (i-1)$, 
we now define an execution of the form  $\mathbb{E}_{j\ell}^{i}=\pi^{i-1}\cdot \beta^{\ell}\cdot \rho^i \cdot \alpha_j^i$
as follows:
\begin{itemize}
\item
$\beta^{\ell}$ is the t-complete step contention-free execution fragment of a transaction $T_{\ell}$
that writes $nv_{\ell}\neq v$ to $X_{\ell}$ and commits
\item
$\alpha_1^i$ (and resp. $\alpha_2^i$) is the complete step contention-free execution fragment of 
$\Read_{\phi}(X_i) \rightarrow v$ (and resp. $\Read_{\phi}(X_i) \rightarrow A_{\phi}$).
\end{itemize}
\begin{claim}
\label{cl:ic2}
For all $i\in \{2,\ldots, m\}$ and $\ell \leq (i-1)$, $M$ has an execution of the form $\mathbb{E}_{1\ell}^{i}$ or 
$\mathbb{E}_{2\ell}^{i}$.
\end{claim}
\begin{proof}
For all $i \in \{2,\ldots, m\}$, $\pi^{i-1}$
is an execution of $M$.
By assumption of weak invisible reads and sequential TM-progress, $T_{{\ell}}$ must be committed in $\pi^{i-1}\cdot \rho^{\ell}$
and $M$ has an execution of the form $\pi^{i-1}\cdot \beta^{\ell}$.
By the same reasoning, since $T_i$ and $T_{\ell}$ have disjoint data sets,
$M$ has an execution of the form $\pi^{i-1}\cdot\beta^{\ell}\cdot \rho^i$.

Since the configuration after $\pi^{i-1}\cdot\beta^{\ell}\cdot \rho^i$ is quiescent,
by ICF TM-liveness, $\pi^{i-1}\cdot\beta^{\ell}\cdot \rho^i$ extended with $\Read_{\phi}(X_i)$
must return a matching response.
If $\Read_{\phi}(X_i) \rightarrow v_i$, then clearly $\mathbb{E}_{1}^{i}$
is an execution of $M$ with $T_{\phi}, T_{i-1}, T_i$ being a valid serialization
of transactions.
If $\Read_{\phi}(X_i) \rightarrow A_{\phi}$, the same serialization
justifies an opaque execution.

Suppose by contradiction that there exists an execution of $M$ such that
$\pi^{i-1}\cdot\beta^{\ell}\cdot \rho^i$ is extended with the complete execution
of $\Read_{\phi}(X_i) \rightarrow r$; $r \not\in \{A_{\phi},v\}$. 
The only plausible case to analyse is when $r=nv$.
Since $\Read_{\phi}(X_i)$ returns the value of $X_i$ updated by $T_i$, 
the only possible serialization for transactions is $T_{\ell}$, $T_i$, $T_{\phi}$; but $\Read_{\phi}(X_{\ell})$
performed by $T_k$ that returns the initial value $v$
is not legal in this serialization---contradiction.
\end{proof}
We now prove that, for all $i\in \{2,\ldots, m\}$, $j\in \{1,2\}$ and $\ell \leq (i-1)$, transaction $T_{\phi}$ must access
$(i-1)$ different base objects during the execution of $\Read_{\phi}(X_i)$ in the execution
$\pi^{i-1}\cdot \beta^{\ell}\cdot \rho^i \cdot \alpha_j^i$.

By the assumption of weak invisible reads,
the execution $\pi^{i-1}\cdot \beta^{\ell}\cdot \rho^i \cdot \alpha_j^i$
is indistinguishable to
transactions $T_{\ell}$ and $T_{i}$
from the execution ${\tilde \pi}^{i-1}\cdot \beta^{\ell}\cdot \rho^i \cdot \alpha_j^i$, where $\Rset(T_{\phi})=\emptyset$
in ${\tilde \pi}^{i-1}$.
But transactions $T_{\ell}$ and $T_{i}$ are disjoint-access in ${\tilde \pi}^{i-1}\cdot \beta^{\ell}\cdot \rho^i$ and by Lemma~\ref{lm:dap},
they cannot contend on the same base object in this execution.

Consider the $(i-1)$ different executions: 
$\pi^{i-1}\cdot\beta^{1}\cdot \rho^i$, $\ldots$, $\pi^{i-1}\cdot\beta^{i-1}\cdot \rho^i$.
For all $\ell, \ell' \leq (i-1)$;$\ell' \neq \ell$, 
$M$ has an execution of the form $\pi^{i-1}\cdot \beta^{\ell}\cdot \rho^i \cdot \beta^{\ell'}$
in which transactions $T_{\ell}$ and $T_{\ell'}$ access mutually disjoint data sets.
By weak invisible reads and Lemma~\ref{lm:dap}, the pairs of transactions $T_{\ell'}$, $T_{i}$ and $T_{\ell'}$, $T_{\ell}$
do not contend on any base object in this execution.
This implies that $\pi^{i-1}\cdot \beta^{\ell} \cdot \beta^{\ell'} \cdot \rho^i$ is an execution of $M$ in which
transactions $T_{\ell}$ and $T_{\ell'}$ each apply nontrivial primitives
to mutually disjoint sets of base objects in the execution fragments $\beta^{\ell}$ and $\beta^{\ell'}$ respectively 
(by Lemma~\ref{lm:dap}).

This implies that for any $j\in \{1,2\}$, $\ell \leq (i-1)$, the configuration $C^i$ after $E^i$ differs from the configurations
after $\mathbb{E}_{j\ell}^{i}$ only in the states of the base objects that are accessed in the fragment $\beta^{\ell}$.
Consequently, transaction $T_{\phi}$ must access at least $i-1$ different base objects
in the execution fragment $\pi_j^i$
to distinguish configuration $C^i$ from the configurations
that result after the $(i-1)$ different executions 
$\pi^{i-1}\cdot\beta^{1}\cdot \rho^i$, $\ldots$, $\pi^{i-1}\cdot\beta^{i-1}\cdot \rho^i$ respectively.

Thus, for all $i \in \{2,\ldots, m\}$, transaction $T_{\phi}$ must perform at least $i-1$ steps 
while executing the $i^{th}$ t-read in $\pi_{j}^i$ and $T_{\phi}$ itself must perform 
$\sum\limits_{i=1}^{m-1} i=\frac{m(m-1)}{2}$ steps.

($2$) Suppose that $M$ is strictly serializable, but not opaque.
Since $M$ is strictly serializable, by Lemma~\ref{lm:readdap}, it has an execution of the form 
$E=\pi^{m-1}\cdot \rho^{m} \cdot \alpha_m$.

For each $\ell \leq (i-1)$, we prove that $M$ has an execution of the form 
$E_{\ell}= \pi^{m-1}\cdot \beta^{\ell}\cdot \rho^m \cdot {\bar \alpha}^m$
where ${\bar \alpha}^m$ is the complete step contention-free execution fragment of $\Read_{\phi}(X_m)$ followed
by the complete execution of $\TryC_{\phi}$.
Indeed, by weak invisible reads, $\pi^{m-1}$ does not contain any nontrivial events
and the execution $\pi^{m-1}\cdot \beta^{\ell}\cdot \rho^m$ is indistinguishable to transactions
$T_{\ell}$ and $T_m$ from the executions ${\tilde \pi}^{m-1}\cdot \beta^{\ell}$ and ${\tilde \pi}^{m-1}\cdot\beta^{\ell} \cdot \rho^m$ 
respectively, where $\Rset(T_{\phi})=\emptyset$ in ${\tilde \pi}^{m-1}$.
Thus, applying Lemma~\ref{lm:dap}, transactions $\beta^{\ell} \cdot \rho^m$ do not contend
on any base object in the execution $\pi^{m-1}\cdot \beta^{\ell}\cdot \rho^m$. 
By ICF TM-liveness, $\Read_{\phi}(X_m)$ and $\TryC_{\phi}$ must return matching responses in the execution
fragment ${\bar \alpha}^m$ that extends $\pi^{m-1}\cdot \beta^{\ell}\cdot \rho^m$.
Consequently, for each $\ell \leq (i-1)$, $M$ has an execution of the form 
$E_{\ell}= \pi^{m-1}\cdot \beta^{\ell}\cdot \rho^m \cdot {\bar \alpha}^m$
such that transactions $T_{\ell}$ and $T_m$ do not contend on any base object.

Strict serializability of $M$ means that if $\Read_{\phi}(X_m)\rightarrow nv$ in the execution fragment ${\bar \alpha}^m$, 
then $\TryC_{\phi}$ must return $A_{\phi}$.
Otherwise if $\Read_{\phi}(X_m)\rightarrow v$ (i.e. the initial value of $X_m$), then
$\TryC_{\phi}$ may return $A_{\phi}$ or $C_{\phi}$.

Thus, as with ($1$), in the worst case, $T_{\phi}$ must access at least $m-1$ distinct base objects 
during the executions of
$\Read_{\phi}(X_m)$ and $\TryC_{\phi}$ to distinguish the configuration $C^i$ from the configurations
after the $m-1$ different executions
$\pi^{m-1}\cdot\beta^{1}\cdot \rho^m$, $\ldots$, $\pi^{m-1}\cdot\beta^{m-1}\cdot \rho^m$ respectively.
\end{proof}
\subsection{Expensive synchronization in Transactional memory cannot be eliminated}
\label{sec:p3c2ss11}
In this section, we show that serializable TMs must perform a RAW/AWAR even if they are guaranteed to commit
only when they run in the absence of any concurrency.
\begin{theorem}
\label{th:sl}
Let $M$ be a serializable TM implementation providing sequential TM-progress and sequential TM-liveness.
Then, every execution of $M$ in which a transaction running t-sequentially performs at least one 
t-read and at least one t-write contains a RAW/AWAR pattern.     
\end{theorem}
\begin{proof}
Consider an execution $\pi$ of $M$ in which a transaction $T_1$ running t-sequentially
performs (among other events) $\Read_1(X)$, $\Write_1(Y,v)$ and $\TryC_1()$.
Since $M$ satisfies sequential TM-progress and sequential TM-liveness, $T_1$ must commit in $\pi$. 
Clearly $\pi$ must contain a write to a base object. 
Otherwise a subsequent transaction reading $Y$ would return the
initial value of $Y$ instead of the value written by $T_1$.
 
Let $\pi_w$ be the first write to a base object in $\pi$.
Thus, $\pi$ can be represented as $\pi_s\cdot\pi_w\cdot\pi_f$.

Now suppose by contradiction that $\pi$ contains neither RAW nor
AWAR patterns. 

Since $\pi_s$ contains no writes, the states of base objects in the
initial configuration and in the configuration after $\pi_s$ is performed are the same. 
Consider an execution $\pi_s\cdot\rho$ where in $\rho$, 
a transaction $T_2$ performs $\Read_2(Y)$, $\Write_2(X,1)$,
$\TryC_2()$ and commits.
Such an execution exists, since $\rho$ is indistinguishable to $T_2$
from an execution in which $T_2$ runs t-sequentially and thus $T_2$ cannot
be aborted in $\pi_s\cdot\rho$. 

Since $\pi_w$ contains no AWAR, $\pi_s\cdot\rho \cdot \pi_w$ is an execution of $M$.

Since $\pi_w\cdot\pi_f$ contains no RAWs, every read performed in
$\pi_w\cdot\pi_f$ is applied to base objects which were previously
written in  $\pi_w\cdot\pi_f$. Thus, there exists an execution
$\pi_s\cdot\rho\cdot\pi_w\cdot\pi_f$,
such that $T_1$ cannot distinguish  $\pi_s\cdot\pi_w\cdot\pi_f$ 
and $\pi_s\cdot\rho\cdot\pi_w\cdot\pi_f$.
Hence, $T_1$ commits in  $\pi_s\cdot\rho\cdot\pi_w\cdot\pi_f$.

But $T_1$ reads the initial value of $X$ and $T_2$ reads the initial value of $Y$ in
$\pi_s\cdot\rho\cdot\pi_w\cdot\pi_f$, and thus $T_1$ and $T_2$ cannot be both committed
(at least one of the committed transactions must read the value
written by the other)---a contradiction.
\end{proof}
%
%
\section{Progressive TMs}
\label{sec:p3c2s2}
We move on to the stronger (than sequential TMs) class of progressive TMs.
We introduce a new metric called \emph{protected data size}
that, intuitively, captures the number of t-objects that a transaction 
must exclusively control at some prefix of its execution. 
We first prove that any strict DAP progressive opaque TM must protect its entire write set at some point in its execution.

Secondly, we describe a constant stall, constant RAW/AWAR strict DAP opaque progressive TM that provides invisible reads
and is implemented from read-write base objects.
\subsection{A linear lower bound on the amount of protected data}
\label{sec:p3c2ss2}
Let $M$ be a progressive TM implementation providing starvation-free TM-liveness.
Intuitively, a t-object $X_j$ is protected at the end of some finite execution $\pi$ of $M$ 
if some transaction $T_0$ is about to atomically change the value of $X_j$ in its next step 
(\emph{e.g.}, by performing a compare-and-swap) or does not
allow any concurrent transaction to read $X_j$ (\emph{e.g.}, by holding a
``lock'' on $X_j$).

Formally, let $\alpha\cdot\pi$ be an execution of $M$   
such that $\pi$ is a t-sequential t-complete execution of 
a transaction $T_0$, where $\Wset(T_0)=\{X_1,\ldots,X_m\}$.
Let $u_j$ ($j=1,\ldots,m$) denote the value written by $T_0$ to t-object $X_j$ in $\pi$. 
In this section, let $\pi^t$ denote the $t$-th shortest prefix of $\pi$. 
Let $\pi^0$ denote the empty prefix.

For any $X_j\in\Wset(T_0)$, 
let $T_{j}$ denote a transaction that tries to read $X_j$ 
and commit. 
Let $E_{j}^{t}=\alpha\cdot\pi^t\cdot\rho_j^t$ denote the extension of $\alpha\cdot\pi^t$ in which    
$T_j$ runs solo until it completes.
Note that, since we only require the implementation to be starvation-free,  $\rho_j^t$ can be infinite. 

We say that $\alpha\cdot\pi^t$ is $(1,j)$-valent if the read
operation performed by $T_j$ in $\alpha\cdot\pi^t\cdot\rho_j^t$
returns $u_j$ (the value written by $T_0$ to $X_j$).
We say that $\alpha\cdot\pi^t$ is $(0,j)$-valent if the read operation performed by $T_j$ in $\alpha\cdot\pi^t\cdot\rho_j^t$
does not abort and returns an "old" value $u\neq u_j$.  
Otherwise, if the read operation of $T_j$  aborts or never returns in
$\alpha\cdot\pi^t\cdot\rho_j^t$, we say that  
$\alpha\cdot\pi^t$ is $(\bot,j)$-valent.
\begin{definition}
We say that $T_0$ \emph{protects} an object $X_j$ in
$\alpha\cdot\pi^t$, where $\pi^t$ is the $t$-th shortest prefix
of $\pi$ ($t>0$) if one of the following conditions holds:
(1) $\alpha\cdot\pi^{t}$ is $(0,j)$-valent and $\alpha\cdot\pi^{t+1}$
is $(1,j)$-valent, or 
(2) $\alpha\cdot\pi^{t}$ or  $\alpha\cdot\pi^{t+1}$ is $(\bot,j)$-valent.
\end{definition}
%
%
%
%
For \emph{strict disjoint-access parallel} progressive TM, we show that 
every transaction running t-sequentially must protect every t-object in its write set 
at some point of its execution.

We observe that the no prefix of $\pi$ can be $0$ and $1$-valent at the same time.
\begin{lemma}\label{lem:valence}
There does not exist $\pi^t$, a prefix of $\pi$, and
$i,j\in\{1,\ldots,m\}$ such that 
$\alpha\cdot\pi^t$ is both $(0,i)$-valent and  $(1,j)$-valent. 
\end{lemma}
\begin{proof}
By contradiction, suppose that there exist $i,j$ and $\alpha\cdot\pi^t$ 
that is both $(0,i)$-valent and  $(1,j)$-valent. 
Since the implementation is strict DAP,
there exists an execution of $M$,
$E_{ij}^{t}=\alpha\cdot\pi^t\cdot\rho_j^t\cdot\rho_i^t$
that is indistinguishable to $T_i$ from
$\alpha\cdot\pi^t\cdot\rho_i^t$.
In $E_{ij}^{t}$, the only possible serialization is $T_0$, $T_j$, $T_i$.
But $T_i$ returns the ``old'' value of $X_i$ and, thus, the
serialization is not legal---a contradiction.
\end{proof}
If $\alpha\cdot\pi^t$ is $(0,i)$-valent (resp., $(1,i)$-valent) for
some $i$, we say that it is $0$-valent (resp., $1$-valent).
By Lemma~\ref{lem:valence}, the notions of $0$-valence and $1$-valence
are well-defined.
\begin{theorem}
\label{th:protected}
Let $M$ be a progressive, opaque and strict disjoint-access-parallel TM implementation that provides starvation-free TM-liveness.
Let $\alpha\cdot\pi$ be an execution of $M$, where $\pi$ is a
t-sequential t-complete execution of a transaction $T_0$.
Then, there exists $\pi^t$, a prefix of $\pi$, such that $T_0$
protects $|\Wset(T_0)|$ t-objects in $\alpha\cdot\pi^t$.  
\end{theorem}
\begin{proof}
Let $\Wset_{T_0}=\{X_1,\ldots,X_m\}$. 
Consider two cases:

\begin{enumerate}
\item[(1)]
Suppose that 
$\pi$ has a prefix $\pi^{t}$ such that $\alpha\cdot\pi^t$ is
$0$-valent and $\alpha\cdot\pi^{t+1}$ is
$1$-valent.
By Lemma~\ref{lem:valence}, there does not exists $i$, such that $\alpha\cdot\pi^t$ is $(1,i)$-valent and $\alpha\cdot\pi^{t+1}$ is $(0,i)$-valent.
Thus, one of the following are true
\begin{itemize}
\item For every $i \in \{1,\ldots ,m\}$, $\alpha\cdot\pi^t$ is $(0,i)$-valent and $\alpha\cdot\pi^{t+1}$ is
$(1,i)$-valent
\item At least one of $\alpha\cdot\pi^t$ and $\alpha\cdot\pi^{t+1}$ is $(\bot , i)$-valent, \emph{i.e.}, the t-operation of $T_i$ aborts or never returns
\end{itemize}
In either case, $T_0$ protects $m$ t-objects in $\alpha\cdot\pi^t$.  

\item[(2)] Now suppose that such $\pi^t$ does not exists, i.e., there
  is no $i\in\{1,\ldots,m\}$ and $t\in\{0,|\pi|-1\}$ such that $E_i^t$
  exists and returns an old value, and $E_i^{t+1}$ exists and returns
  a new value. 

Suppose there exists $s,t$, $0< s+1<t$,
$S\subseteq\{1,\ldots,m\}$,
such that:
\begin{itemize}
\item  $\alpha\cdot\pi^s$ is $0$-valent,  
\item  $\alpha\cdot\pi^t$ is $1$-valent,
\item  for all $r$, $s<r<t$, and for all $i\in S$, $\alpha\cdot\pi^r$ is $(\bot,i)$-valent.
\end{itemize}
We say that $s+1,\ldots,t-1$ is a \emph{protecting fragment} for
t-objects $\{X_j | j\in S\}$.  

Since $M$ is opaque and progressive, $\alpha\cdot\pi^0=\alpha$ is
$0$-valent and $\alpha\cdot\pi$ is $1$-valent.
Thus, the assumption of Case (2) implies that for each $X_i$, there exists a
protecting fragment for $\{X_i\}$.
In particular, there exists a protecting fragment for $\{X_1\}$.

Now we proceed by induction.
Let $\pi_{s+1},\ldots,\pi_{t-1}$ be a protecting fragment for
$\{X_1,\ldots,X_{u-1}\}$ such that $u\leq m$.

Now we claim that there must be a subfragment of $s+1,\ldots,t-1$ 
that protects $\{X_1,\ldots,X_u\}$. 

Suppose not. Thus, there exists $r$, $s<r<t$, such that 
$\alpha\cdot\pi^r$ is $(0,u)$-valent or $(1,u)$-valent. 
Suppose first that $\alpha\cdot\pi^r$ is $(1,u)$-valent.
Since $\alpha\cdot\pi^s$ is $(0,i)$-valent for some $i \neq u$, by Lemma~\ref{lem:valence}
and the assumption of Case (2),
there must exist $s',t'$, $s< s'+1 < t'\le r$ such that 
\begin{itemize}
\item  $\alpha\cdot\pi^{s'}$ is $0$-valent,  
\item  $\alpha\cdot\pi^{t'}$ is $1$-valent,
\item  for all $r'$, $s'<r'<t'$, $\alpha\cdot\pi^{r'}$ is $(\bot,u)$-valent.
\end{itemize}
As a result, $s'+1,\ldots,t'-1$ is a protecting fragment for $\{X_1,\ldots,X_u\}$.  
The case when $\alpha\cdot\pi^r$ is $(0,u)$-valent is symmetric, except that now we should 
consider fragment $r,\ldots,t$ instead of $s,\ldots,r$.

Thus, there exists a subfragment of ${s+1},\ldots,{t-1}$ that protects $\{X_1,\ldots,X_u\}$. 
By induction, we obtain a protecting fragment $s''+1,\ldots,t''-1$ for $\{X_1,\dots,X_m\}$.
Thus, any prefix $\alpha\cdot\pi^r$, where $s''<r<t''$ protects exactly
$m$ t-objects.  
\end{enumerate}

In both cases, there is a prefix of $\alpha\cdot\pi$ that protects
exactly $m$ t-objects.
\end{proof}
The lower bound of Theorem~\ref{th:protected} is tight: it is matched by all progressive implementations 
we are aware of, including Algorithm~\ref{alg:ic} described in the next section. 
\subsection{A constant stall and constant expensive synchronization strict DAP opaque TM}
\label{sec:p3c2s2s1}
\begin{algorithm}[!h]
\caption{Strict DAP progressive opaque TM implementation $LP$; code for $T_k$ executed by process $p_i$}
\label{alg:ic}
\begin{algorithmic}[1]
  	\begin{multicols}{2}
  	{\footnotesize
	\Part{Shared base objects}{
		\State $v_j$, for each t-object $X_j$, allows reads and writes
		\State $r_{ij}$, for each process $p_i$ and t-object $X_j$
		\State ~~~~~single-writer bit
		\State ~~~~~allows reads and writes
		\State $L_j$, for each t-object $X_j$ 
		\State ~~~~~allows reads and writes
	}\EndPart
	\Part{Local variables}{
		\State $\ms{Rset}_k,\ms{Wset}_k$ for every transaction $T_k$;
		\State ~~~~dictionaries storing $\{X_m, v_m\}$
	}\EndPart	

	\Statex
	\Part{\Read$_k(X_j)$}{
		\If{$X_j \not\in \Rset(T_k)$}
		
		\State $[\textit{ov}_j,k_j ] := \Read(v_j)$ \label{line:read2}
		\State $\Rset(T_k) := \Rset(T_k)\cup\{X_j,[\textit{ov}_j,k_j]\}$ \label{line:rset}
		\If{$\Read(L_j)\neq 0$} \label{line:abort0}
			\Return $A_k$ \EndReturn
		\EndIf
		\If{$\lit{validate}()$} \label{line:read-validate}
			\Return $A_k$ \EndReturn
		\EndIf
		\Return $\textit{ov}_j$ \EndReturn
		
		\Else
		    
		\State $[\textit{ov}_j, \bot] :=\Rset(T_k).\lit{locate}(X_j)$
		\Return $\textit{ov}_j$ \EndReturn
		
		\EndIf
   	}\EndPart
	\Statex
	\Part{\Write$_k(X_j,v)$}{
		\State $\textit{nv}_j := v$
		\State $\Wset(T_k) := \Wset(T_k)\cup\{X_j\}$
		\Return $\ok$ \EndReturn
		
   	}\EndPart
	\Statex
%
	
	\Part{\TryC$_k$()}{
		\If{$|\Wset(T_k)|= \emptyset$}
			\Return $C_k$ \EndReturn \label{line:return}
		\EndIf
				
		\State locked := $\lit{acquire}(\Wset(T_k))$\label{line:acq} 
		\If{$\neg$ locked} \label{line:abort2} 
	 		\Return $A_k$ \EndReturn
	 	\EndIf
	 	
		\If{$\lit{isAbortable}()$} \label{line:abort3}
			\State $\lit{release}(\Wset(T_k))$ 
			\Return $A_k$ \EndReturn
		\EndIf
		\Statex
		\Comment{Exclusive write access to each $v_j$}
		\ForAll{$X_j \in \Wset(T_k)$}
	 		 \State $\Write(v_j,[\textit{nv}_j,k])$ \label{line:write}
	 	\EndFor		
		\State $\lit{release}(\Wset(T_k))$   	\label{line:rellock}	
   		\Return $C_k$ \EndReturn
   	 }\EndPart		
	
	\newpage
	\Part{Function: $\lit{release}(Q)$}{
		\ForAll{$X_j \in Q$}	
 			\State \Write$(L_{j},0)$ \label{line:wlockrelease}
		\EndFor
  		\ForAll{$X_j \in Q$}	
 			\State \Write$(r_{ij},0)$ \label{line:rel1}
		\EndFor
		\Return $ok$ \EndReturn
	}\EndPart

 	\Statex
 	\Part{Function: $\lit{acquire}(Q$)}{
   		\ForAll{$X_j \in Q$}	
			\State \Write$(r_{ij},1)$ \label{line:acq1}
		\EndFor
		\If{$\exists X_j \in Q;t\neq k : \Read(r_{tj})=1$} \label{line:lock}
			\ForAll{$X_j \in Q$}	
				\State \Write$(r_{ij},0)$
			\EndFor
			\Return $\false$ \EndReturn
		\EndIf
		\Statex
		\Comment{Exclusive write access to each $L_j$}
		\ForAll{$X_j \in Q$}
		  \State $\Write(L_j,1)$ \label{line:wlockwrite}
		\EndFor
		\Return $\true$ \EndReturn
	}\EndPart		
	 \Statex
	 
	\Part{Function: $\lit{isAbortable()}$ }{
		\If{$\exists X_j \in \Rset(T_k): X_j\not\in \Wset(T_k)\wedge \Read(L_j)\neq 0$} \label{line:valid0} \label{line:isl}
			\Return $\true$ \EndReturn
		\EndIf
		
		\If{$\lit{validate}()$} 
			\Return $\true$ \EndReturn
		\EndIf
		\Return $\false$ \EndReturn
	}\EndPart
	\Statex
	\Part{Function: $\lit{validate()}$ }{
		
		\Comment{Read validation}
		\If{$\exists X_j \in Rset(T_k)$:$[\textit{ov}_j,k_j]\neq \Read(v_j)$} \label{line:valid} 
			\Return $\true$ \EndReturn
		\EndIf
		\Return $\false$ \EndReturn
	}\EndPart
	}
	\end{multicols}
  \end{algorithmic}
\end{algorithm}
In this section, we describe a cheap progressive, opaque TM implementation $LP$ (Algorithm~\ref{alg:ic}).
Our TM $LP$, every transaction performs at most a single RAW, every t-read operation
incurs $O(1)$ memory stalls and maintains exactly one version of every
t-object at any prefix of an execution. Moreover, the implementation
is strict DAP and uses only read-write base objects.  

\vspace{1mm}\noindent\textbf{Base objects.}
For every t-object $X_j$, $LP$ maintains a base object $v_j$ that stores the \emph{value} of $X_j$.
Additionally, for each $X_j$, we maintain a bit $L_j$, which if set, indicates the presence of an updating transaction
writing to $X_j$.
Also, for every process $p_i$ and t-object $X_j$,
$LP$ maintains a \emph{single-writer bit} $r_{ij}$  to which  only $p_i$ is allowed to write.
Each of these base objects may be accessed only via read and write primitives.

\vspace{1mm}\noindent\textbf{Read operations.}
The implementation first reads the value of t-object $X_j$ from base object $v_j$ and then
reads the bit $L_j$ to detect contention with an updating transaction.
If $L_j$ is set, the transaction is aborted; if not, read validation is performed on the entire read set. If the validation fails,
the transaction is aborted. Otherwise, the implementation returns the value of $X_j$.
For a read-only transaction $T_k$, $\TryC_k$ simply returns the commit response.

\vspace{1mm}\noindent\textbf{Updating transactions.}
The $\Write_k(X,v)$ implementation by process $p_i$ simply stores the value $v$ locally, deferring the actual updates
to $\TryC_k$.
During $\TryC_k$, process $p_i$ attempts to obtain exclusive write access to every
$X_j\in \Wset(T_k)$. This is realized through the single-writer bits, which ensure that no other transaction
may write to base objects $v_j$ and $L_j$ until $T_k$ relinquishes its exclusive write access to $\Wset(T_k)$.
Specifically, process $p_i$ writes $1$ to each $r_{ij}$, then checks
that no other process $p_t$ has written $1$ to any $r_{tj}$ by executing a series of reads (incurring a single RAW).
If there exists such a process that concurrently contends on write set of $T_k$, 
for each $X_j\in \Wset(T_k)$, $p_i$ writes $0$ to $r_{ij}$ and aborts $T_k$. 
If successful in obtaining exclusive write access to $\Wset(T_k)$, $p_i$ sets
the bit $L_j$ for each $X_j$ in its write set.
Implementation of $\TryC_k$ now checks if any t-object in its read set is concurrently contended by another transaction
and then validates its read set. 
If there is contention on the read set or validation fails (indicating the presence of a conflicting
transaction), the transaction is aborted. If not, $p_i$ writes the values of the t-objects to shared memory and 
relinquishes exclusive write access to each $X_j \in \Wset(T_k)$ by writing $0$ to each of the base objects $L_j$ 
and $r_{ij}$.

\vspace{1mm}\noindent\textbf{Complexity.}
Read-only transactions do not apply any nontrivial primitives.
Any updating transaction performs at most a single RAW in the course of acquiring exclusive write access to
the transaction's write set. 
Thus, every transaction performs $O(1)$ non-overlapping RAWs in any execution. 

Recall that a transaction may write to base objects $v_j$ and $L_j$ only after obtaining exclusive write access to t-object $X_j$, 
which in turn is realized via single-writer base objects.
Thus, no transaction performs a write to any base object $b$ 
immediately after a write to $b$ by 
another transaction, \emph{i.e.}, every transaction incurs only $O(1)$ memory stalls on account of any event it performs.
The $\Read_k(X_j)$ implementation reads base objects $v_j$ and
$L_j$, followed by the validation phase in which it reads $v_k$ for
each $X_k$ in its current read set. 
Note that if the first read in the validation phase incurs a stall,
then  $\Read_k(X_j)$ aborts. It follows that each t-read incurs $O(1)$ stalls in every execution.

\vspace{1mm}\noindent\textbf{Proof of opacity.}
We now prove that $LP$ implements an opaque TM.

We introduce the following technical definition:
process $p_i$ \emph{holds a lock on $X_j$ after an execution $\pi$ of Algorithm~\ref{alg:ic}} if
$\pi$ contains the invocation of \textit{acquire($Q$)}, $X_j\in Q$ by
$p_i$ that returned \emph{true}, but does not contain a subsequent
invocation of \textit{release($Q'$)}, $X_j\in Q'$, by $p_i$ in $\pi$. 
\begin{lemma}
\label{lm:mutex}
For any object $X_j$, and any execution
$\pi$ of Algorithm~\ref{alg:ic}, there exists at most one process that \emph{holds} a lock on
$X_j$ after $\pi$. 
\end{lemma}
\begin{proof}
Assume, by contradiction, that there exists an execution $\pi$ after which processes $p_i$ and $p_k$
\emph{hold} a lock on the same object, say $X_j$. 
In order to hold the lock on $X_j$, process $p_i$ writes $1$ to register
$r_{ij}$ and then checks if any other process $p_k$ has written $1$ to $r_{kj}$.  
Since the corresponding operation {\it acquire(Q)}, $X_j \in Q$
invoked by $p_i$ returns {\it true}, $p_i$ read $0$ in $r_{kj}$ in Line~\ref{line:lock}. 
But then $p_k$ also writes $1$ to $r_{kj}$ and later reads that
$r_{ij}$ is 1. 
This is because $p_k$ can write $1$ to $r_{kj}$ only after the read of
$r_{kj}$ returned $0$ to $p_i$ which is preceded by the write of $1$ to
$r_{ij}$. 
Hence, there exists an object $X_j$ such that $r_{ij}=1;i\neq k$, 
but the conditional in Line~\ref{line:lock} returns {\it true} to process $p_k$--- a contradiction. 
\end{proof}
\begin{observation}
\label{ob:write}
Let $\pi$ be any execution of Algorithm~\ref{alg:ic}. Then, any updating transaction $T_k \in \ms{txns}(\pi)$
executed by process $p_i$
writes to base object $v_j$ (in Line~\ref{line:write}) 
for some $X_j \in \Wset(T_k)$ immediately after $\pi$ \emph{iff} $p_i$ holds the lock on $X_j$ after $\pi$.
\end{observation}
\begin{lemma}
\label{lm:icopaque}
Algorithm~\ref{alg:ic} implements an opaque TM.
\end{lemma}
\begin{proof}
Let $E$ by any finite execution of Algorithm~\ref{alg:ic}. 
Let $<_E$ denote a total-order on events in $E$.

Let $H$ denote a subsequence of $E$ constructed by selecting
\emph{linearization points} of t-operations performed in $E$.
The linearization point of a t-operation $op$, denoted as $\ell_{op}$ is associated with  
a base object event or an event performed between the invocation and response 
of $op$ using the following procedure. 

\vspace{1mm}\noindent\textbf{Completions.}
First, we obtain a completion of $E$ by removing some pending
invocations and adding responses to the remaining pending invocations
involving a transaction $T_k$ as follows:
every incomplete $\Read_k$, $\Write_k$ operation is removed from $E$;
an incomplete $\TryC_k$ is removed from $E$ if $T_k$ has not performed any write to a base object during the \emph{release}
function in Line~\ref{line:rellock}, otherwise it is completed by including $C_k$ after
$E$.

\vspace{1mm}\noindent\textbf{Linearization points.}
Now a linearization $H$ of $E$ is obtained by associating linearization points to
t-operations in the obtained completion of $E$ as follows:
\begin{itemize}
\item For every t-read $op_k$ that returns a non-A$_k$ value, $\ell_{op_k}$ is chosen as the event in Line~\ref{line:read2}
of Algorithm~\ref{alg:ic}, else, $\ell_{op_k}$ is chosen as invocation event of $op_k$
\item For every $op_k=\Write_k $ that returns, $\ell_{op_k}$ is chosen as the invocation event of $op_k$
\item For every $op_k=\TryC_k$ that returns $C_k$ such that $\Wset(T_k)
  \neq \emptyset$, $\ell_{op_k}$ is associated with the response
  of \emph{acquire} in Line~\ref{line:acq}, 
  else if $op_k$ returns $A_k$, $\ell_{op_k}$ is associated with the invocation event of $op_k$
\item For every $op_k=\TryC_k$ that returns $C_k$ such that $\Wset(T_k) = \emptyset$, 
$\ell_{op_k}$ is associated with Line~\ref{line:return}
\end{itemize}
$<_H$ denotes a total-order on t-operations in the complete sequential history $H$.

\vspace{1mm}\noindent\textbf{Serialization points.}
The serialization of a transaction $T_j$, denoted as $\delta_{T_j}$ is
associated with the linearization point of a t-operation 
performed within the execution of $T_j$.

We obtain a t-complete history ${\bar H}$ from $H$ as follows: 
for every transaction $T_k$ in $H$ that is complete, but not t-complete, 
we insert $\textit{tryC}_k\cdot A_k$ after $H$. 
 
A t-complete t-sequential history $S$ is obtained by associating serialization points to transactions in ${\bar H}$ as follows:
\begin{itemize}
\item If $T_k$ is an update transaction that commits, then $\delta_{T_k}$ is $\ell_{\TryC_k}$
\item If $T_k$ is a read-only or aborted transaction in $\bar H$,
$\delta_{T_k}$ is assigned to the linearization point of the last t-read that returned a non-A$_k$ value in $T_k$
\end{itemize}
$<_S$ denotes a total-order on transactions in the t-sequential history $S$.
\begin{claim}
\label{cl:seq}
If $T_i \prec_{H}T_j$, then $T_i <_S T_j$
\end{claim}
\begin{proof}
This follows from the fact that for a given transaction, its
serialization point is chosen between the first and last event of the transaction
implying if $T_i \prec_{H} T_j$, then $\delta_{T_i} <_{E} \delta_{T_j}$ implies $T_i <_S T_j$.
\end{proof}
\begin{claim}
\label{cl:ic1}
Let $T_k$ be any updating transaction that returns \emph{false} from the invocation of \emph{isAbortable}
in Line~\ref{line:abort3}. Then, $T_k$ returns $C_k$ within a finite number of its own steps in any extension of $E$.
\end{claim}
\begin{proof}
Observer that $T_k$ performs the write to base objects $v_j$ for every $X_j \in \Wset(T_k)$ and then invokes
\emph{release} in Lines~\ref{line:write} and \ref{line:rellock} respectively.
Since neither of these involve aborting the transaction or contain unbounded loops or waiting statements, it follows
that $T_k$ will return $C_k$ within a finite number of its steps.
\end{proof}
\begin{claim}
\label{cl:readfrom}
$S$ is legal.
\end{claim}
\begin{proof}
Observe that for every $\Read_j(X_m) \rightarrow v$, there exists some transaction $T_i$
that performs $\Write_i(X_m,v)$ and completes the event in Line~\ref{line:write} such that
$\Read_j(X_m) \not\prec_H^{RT} \Write_i(X_m,v)$.
More specifically, $\Read_j(X_m)$ returns as a non-abort response, the value of the base object $v_m$
and $v_m$ can be updated only by a transaction $T_i$ such that $X_m \in \Wset(T_i)$.
Since $\Read_j(X_m)$ returns the response $v$, the event in Line~\ref{line:read2}
succeeds the event in Line~\ref{line:write} performed by $\TryC_i$.
Consequently, by Claim~\ref{cl:ic1} and the assignment of linearization points,
$\ell_{\TryC_i} <_E \ell_{\Read_j(X_m)}$.
Since, for any updating
committing transaction $T_i$, $\delta_{T_i}=\ell_{\TryC_i}$, by the assignment of serialization points, it follows that
$\delta_{T_{i}} <_E \delta_{T_{j}}$.

Thus, to prove that $S$ is legal, it suffices to show that  
there does not exist a
transaction $T_k$ that returns $C_k$ in $S$ and performs $\Write_k(X_m,v')$; $v'\neq v$ such that $T_i <_S T_k <_S T_j$. 
Suppose that there exists a committed transaction $T_k$, $X_m \in \Wset(T_k)$ such that $T_i <_S T_k <_S T_j$.

$T_i$ and $T_k$ are both updating transactions that commit. Thus, 
\begin{center}
($T_i <_S T_k$) $\Longleftrightarrow$ ($\delta_{T_i} <_{E} \delta_{T_k}$) \\
($\delta_{T_i} <_{E} \delta_{T_k}$) $\Longleftrightarrow$ ($\ell_{\TryC_i} <_{E} \ell_{\TryC_k}$) 
\end{center}
Since, $T_j$ reads the value of $X$ written by $T_i$, one of the following is true:
$\ell_{\TryC_i} <_{E} \ell_{\TryC_k} <_{E} \ell_{\Read_j(X_m)}$ or
$\ell_{\TryC_i} <_{E} \ell_{\Read_j(X_m)} <_{E} \ell_{\TryC_k}$.
Let $T_i$ and $T_k$ be executed by processes $p_i$ and $p_k$ respectively.

Consider the case that $\ell_{\TryC_i} <_{E} \ell_{\TryC_k} <_{E} \ell_{\Read_j(X_m)}$.

By the assignment of linearization points, $T_k$ returns a response from the event in Line~\ref{line:acq} 
before the read of $v_m$ by $T_j$ in Line~\ref{line:read2}. 
Since $T_i$ and $T_k$ are both committed in $E$, $p_k$ returns \emph{true} from the event in
Line~\ref{line:acq} only after $T_i$ writes $0$ to $r_{im}$ in Line~\ref{line:rel1} (Lemma~\ref{lm:mutex}).

Recall that $\Read_j(X_m)$ checks if $X_m$ is locked by a concurrent transaction (i.e $L_j\neq 0$), 
then performs read-validation (Line~\ref{line:abort0}) before returning a matching response. 
Consider the following possible sequence of events: 
$T_k$ returns \emph{true} from the \emph{acquire} function invocation, 
sets $L_j$ to $1$ for every $X_j \in \Wset(T_k)$ (Line~\ref{line:wlockwrite}) and
updates the value of $X_m$ to shared-memory (Line~\ref{line:write}).
The implementation of $\Read_j(X_m)$ then reads the base object $v_m$ associated with $X_m$ after which
$T_k$ releases $X_m$ by writing $0$ to $r_{km}$ and finally $T_j$ performs the check in Line~\ref{line:abort0}. 
However, $\Read_j(X_m)$ is forced to return $A_j$ because $X_m \in \Rset(T_j)$ (Line~\ref{line:rset}) 
and has been invalidated since last reading its value. 
Otherwise suppose that $T_k$ acquires exclusive access to $X_m$ by writing $1$ to $r_{km}$ and returns \emph{true}
from the invocation of \emph{acquire}, updates $v_m$ in Line~\ref{line:write}), 
$T_j$ reads $v_m$, $T_j$ performs the check in Line~\ref{line:abort0} and finally $T_k$ 
releases $X_m$ by writing $0$ to $r_{km}$. 
Again, $\Read_j(X_m)$ returns $A_j$ since $T_j$ reads that $r_{km}$ is $1$---contradiction.

Thus, $\ell_{\TryC_i} <_E \ell_{\Read_j(X)} <_{E} \ell_{\TryC_k}$.

We now need to prove that $\delta_{T_{j}}$ indeed precedes $\ell_{\TryC_k}$ in $E$.

Consider the two possible cases:
\begin{itemize}
\item
Suppose that $T_j$ is a read-only or aborted transaction in $\bar H$. 
Then, $\delta_{T_j}$ is assigned to the last t-read performed by $T_j$ that returns a non-A$_j$ value. 
If $\Read_j(X_m)$ is not the last t-read performed by $T_j$ that returned a non-A$_j$ value, 
then there exists a $\Read_j(X_z)$ performed by $T_j$ such that 
$\ell_{\Read_j(X_m)} <_{E} \ell_{\TryC_k} <_E \ell_{\Read_j(X_z)}$.
Now assume that $\ell_{\TryC_k}$ must precede $\ell_{\Read_j(X_z)}$ to obtain a legal $S$.
Since $T_k$ and $T_j$ are concurrent in $E$, we are restricted to the case that
$T_k$ performs a $\Write_k(X_z,v)$ and $\Read_j(X_z)$ returns $v$.
However, we claim that this t-read of $X_z$ must abort by performing the checks in Line~\ref{line:abort0}.
Observe that $T_k$ writes $1$ to $L_m$, $L_z$ each (Line~\ref{line:wlockwrite}) and 
then writes new values to base objects $v_m$, $v_z$ (Line~\ref{line:write}).
Since $\Read_j(X_z)$ returns a non-$A_j$ response, $T_k$ writes $0$ to $L_z$ before the read
of $L_z$ by $\Read_j(X_z)$ in Line~\ref{line:abort0}.
Thus, the t-read of $X_z$ would return $A_j$ (in Line~\ref{line:read-validate} after validation of the read set since $X_m$
has been updated---
contradiction to the assumption that it the last t-read by $T_j$ to return a non-$A_j$ response.
\item
Suppose that $T_j$ is an updating transaction that commits, then $\delta_{T_j}=\ell_{\TryC_j}$ which implies that
$\ell_{\Read_j(X_m)} <_{E} \ell_{\TryC_k} <_E \ell_{\TryC_j}$. Then, $T_j$ must necessarily perform the checks
in Line~\ref{line:abort3} and read that $L_m$ is $1$. 
Thus, $T_j$ must return $A_j$---contradiction to the assumption that $T_j$ is a committed transaction.
\end{itemize}
\end{proof}
The conjunction of Claims~\ref{cl:seq} and \ref{cl:readfrom} establish that Algorithm~\ref{alg:ic} is opaque.
\end{proof}
We can now prove the following theorem:
\begin{theorem}
\label{th:ic}
Algorithm~\ref{alg:ic} describes a progressive, opaque and strict DAP TM implementation $LP$ that provides
wait-free TM-liveness, uses invisible reads, uses only read-write base
objects, and for every execution $E$ and transaction $T_k\in\txns(E)$:
\begin{itemize}
\item 
$T_k$ performs at most a single RAW, and 
\item every t-read operation invoked by $T_k$
incurs $O(1)$ memory stalls in $E$, and
\item
every complete t-read operation invoked by $T_k$ performs $O(|\Rset(T_k)|)$ steps in $E$.
\end{itemize}
\end{theorem}
\begin{proof}
\textit{(TM-liveness and TM-progress)}
Since none of the implementations of the t-operations in Algorithm~\ref{alg:ic}
contain unbounded loops or waiting statements, every t-operation $op_k$ returns a matching response
after taking a finite number of steps in every execution. Thus, Algorithm~\ref{alg:ic}
provides wait-free TM-liveness.

To prove progressiveness, we proceed by enumerating the cases under which a transaction $T_k$ may be aborted.
\begin{itemize}
\item
Suppose that there exists a $\Read_k(X_j)$ performed by $T_k$ that returns $A_k$
from Line~\ref{line:abort0}.
Thus, there exists a process $p_t$ executing a transaction
that has written $1$ to $r_{tj}$ in Line~\ref{line:acq1}, but has not yet written
$0$ to $r_{tj}$ in Line~\ref{line:rel1} or
some t-object in $\Rset(T_k)$ has been updated since its t-read by $T_k$.
In both cases, there exists a concurrent transaction performing a 
t-write to some t-object in $\Rset(T_k)$.
\item
Suppose that $\TryC_k$ performed by $T_k$ that returns $A_k$
from Line~\ref{line:abort2}.
Thus, there exists a process $p_t$ executing a transaction
that has written $1$ to $r_{tj}$ in Line~\ref{line:acq1}, but has not yet written
$0$ to $r_{tj}$ in Line~\ref{line:rel1}. Thus, $T_k$ encounters step-contention with another
transaction that concurrently attempts to update a t-object in $\Wset(T_k)$.
\item
Suppose that $\TryC_k$ performed by $T_k$ that returns $A_k$
from Line~\ref{line:abort3}.
Since $T_k$ returns $A_k$ from Line~\ref{line:abort3} for the same reason it
returns $A_k$ after Line~\ref{line:abort0}, the proof follows.
\end{itemize}
\textit{(Strict disjoint-access parallelism)}
Consider any execution $E$ of Algorithm~\ref{alg:ic} and let $T_i$
and $T_j$ be any two transactions that participate in $E$ and access the same
base object $b$ in $E$.
\begin{itemize}
\item
Suppose that $T_i$ and $T_j$ contend on base object $v_j$ or $L_j$.
Since for every t-object $X_j$, there exists distinct base objects $v_j$ and $L_j$,
$T_j$ and $T_j$ contend on $v_j$ only if $X_j \in \Dset(T_i) \cap \Dset(T_j)$.
\item
Suppose that $T_i$ and $T_j$ contend on base object $r_{ij}$.
Without loss of generality, let $p_i$ be the process executing 
transaction $T_i$; $X_j \in \Wset(T_i)$ that writes $1$ to $r_{ij}$ in Line~\ref{line:acq1}.
Indeed, no other process executing a transaction that writes to $X_j$ can write to $r_{ij}$.
Transaction $T_j$ reads $r_{ij}$ only if $X_j \in \Dset(T_j)$ as evident from the accesses performed
in Lines~\ref{line:acq1}, \ref{line:lock}, \ref{line:rel1}, \ref{line:isl}.
\end{itemize}
Thus, $T_i$ and $T_j$ access the same base object only if they access a common t-object.

\textit{(Opacity)}
Follows from Lemma~\ref{lm:icopaque}.

\textit{(Invisible reads)}
Observe that read-only transactions do not perform any nontrivial events.
Secondly, in any execution $E$ of Algorithm~\ref{alg:ic}, and any transaction $T_k\in \ms{txns}(E)$,
if $X_j\in \Rset(T_k)$, $T_k$ does not write to any of the base objects associated with $X_j$ nor
write any information that reveals its read set to other transactions.

\textit{(Complexity)} 
Consider any execution $E$ of Algorithm~\ref{alg:ic}.
\begin{itemize}
\item
For any $T_k \in \ms{txns}(E)$, each $\Read_k$ only applies trivial primitives in $E$ while $\TryC_k$ simply
returns $C_k$ if $\Wset(T_k)=\emptyset$. Thus, Algorithm~\ref{alg:ic} uses invisible reads.
\item
Any read-only transaction $T_k \in \ms{txns}(E)$ not perform any RAW or AWAR.
An updating transaction $T_k$ executed by process $p_i$ performs a sequence of writes (Line~\ref{line:acq1}
to base objects $\{r_{ij}\}:X_j \in \Wset(T_k)$, followed by a sequence of reads to base objects 
$\{r_{tj}\}:t\in \{1,\ldots , n\}, X_j \in \Wset(T_k)$
(Line~\ref{line:lock}) thus incurring a single multi-RAW.
\item
Let $e$ be a write event performed by some transaction $T_k$ executed by process $p_i$ in $E$ on 
base objects $v_j$ and $L_j$ (Lines~\ref{line:write} and \ref{line:wlockwrite}).
Any transaction $T_k$ performs a write to $v_j$ or $L_j$ only after $T_k$ writes $0$ to $r_{ij}$, for every $X_j\in \Wset(T_k)$.
Thus, by Lemmata~\ref{lm:mutex} and \ref{lm:icopaque}, it follows that
events that involve an access to either of these base objects incurs $O(1)$ stalls.

Let $e$ be a write event on base object $r_{ij}$ (Line~\ref{line:acq1}) while writing to t-object $X_j$.
By Algorithm~\ref{alg:ic}, no other process can write to $r_{ij}$.
It follows that any transaction $T_k \in \ms{txns}(E)$ incurs $O(1)$ memory stalls 
on account of any event it performs in $E$.
%
Observe that any t-read $\Read_k(X_j)$ only accesses base objects $v_j$, $L_j$ and other value base objects in $\Rset(T_k)$.
But as already established above, these are $O(1)$ stall events. Hence, every t-read operation
incurs $O(1)$-stalls in $E$.
\end{itemize}
\end{proof}
The following corollary follows from Theorems~\ref{th:ic} and \ref{th:iclb}.
\begin{corollary}
\label{cr:pg1}
Let $M$ be any weak DAP progressive opaque TM implementation providing ICF TM-liveness and weak invisible reads.
Then, for every execution $E$ and each read-only transaction $T_k\in \ms{txns}(E)$, $T_k$
performs $\Theta (m^2)$ steps in $E$, where $m=|\Rset_E(T_k)|$.
\end{corollary}
\begin{algorithm}[t]
\caption{Strict DAP progressive strictly serializable TM implementation; code for $T_k$ executed by process $p_i$}
\label{alg:icss}
\begin{algorithmic}[1]
  	{\footnotesize
	\Part{\Read$_k(X_j)$}{
		\If{$X_j \not\in \Rset(T_k)$}
		
		\State $[\textit{ov}_j,k_j ] := \Read(v_j)$
		\State $\Rset(T_k) := \Rset(T_k)\cup\{X_j,[\textit{ov}_j,k_j]\}$ 
		\If{$\Read(L_j)\neq 0$} 
			\Return $A_k$ \EndReturn
		\EndIf
		\Return $\textit{ov}_j$ \EndReturn
		
		\Else
		    
		\State $[\textit{ov}_j, \bot] :=\Rset(T_k).\lit{locate}(X_j)$
		\Return $\textit{ov}_j$ \EndReturn
		
		\EndIf
   	}\EndPart
   	}
  \end{algorithmic}
\end{algorithm}
Similarly, we can prove an almost matching upper bound for Theorem~\ref{th:iclb} for strictly serializable progressive TMs.

Consider Algorithm~\ref{alg:icss} that is a simplification of the opaque progressive TM in Algorithm~\ref{alg:ic}:
we remove the validation performed in the implementation of a t-read, \emph{i.e.}, Line~\ref{line:read-validate}
in Algorithm~\ref{alg:ic}; otherwise, the two algorithms are identical.
It is easy to see this results in a strictly serializable (but not opaque) TM implementation. Thus,
\begin{theorem}
Algorithm~\ref{alg:icss} describes a progressive, strictly serializable and strict DAP TM implementation that provides
wait-free TM-liveness, uses invisible reads, uses only read-write base
objects, and for every execution $E$ and transaction $T_k\in\txns(E)$:
every t-read operation invoked by $T_k$ performs $O(1)$ steps and $\TryC_k$ performs $O(|\Rset(T_k)|)$ steps in $E$.
\end{theorem}
\begin{corollary}
\label{cr:pg2}
Let $M$ be any weak DAP progressive strictly serializable TM implementation providing ICF TM-liveness and weak invisible reads.
Then, for every execution $E$ and each read-only transaction $T_k\in \ms{txns}(E)$, each $\Read_k$ performs
$O(1)$ steps and $\TryC_k$ performs $\Theta (m)$ steps in $E$, where $m=|\Rset_E(T_k)|$.
\end{corollary}
%
%
%
%
\section{Strongly progressive TMs}
\label{sec:p3c2s3}
In this section, we prove that every strongly progressive strictly serializable TM that uses only 
read, write and \emph{conditional}
primitives has an execution in which
in which $n$
concurrent processes perform transactions on a single data item and incur
$\Omega (\log n)$ \emph{remote memory references}~\cite{anderson-90-tpds}.

We then describe a constant RAW/AWAR strongly progressive TM providing starvation-free TM-liveness from read-write base objects.
\subsection{A $\Omega (n\log n)$ lower bound on remote memory references}
\label{sec:p3c2s3s2}
Our lower bound on RMR complexity of strongly progressive TMs is derived by reduction to \emph{mutual exclusion}.

\vspace{1mm}\noindent\textbf{Mutual exclusion.}
The \emph{mutex object} supports two operations: \emph{Entry} and \emph{Exit}, both of which return the response $ok$.
We say that a process $p_i$ \emph{is in the critical section after an execution $\pi$} if
$\pi$ contains the invocation of $\lit{Entry}$ by
$p_i$ that returns $ok$, but does not contain a subsequent
invocation of $\lit{Exit}$ by $p_i$ in $\pi$. 
     
A mutual exclusion implementation satisfies the following properties:

\begin{itemize}
\item (\emph{Mutual-exclusion}) After any execution
  $\pi$, there exists at most one process that is in the critical section. 

\item (\emph{Deadlock-freedom})  Let $\pi$ be any execution that contains the invocation of
  $\lit{Enter}$ by process $p_i$. Then, in every extension of $\pi$ in which every process takes infinitely many steps, some
  process is in the critical section.

\item (\emph{Finite-exit)} Every process completes the $\lit{Exit}$ operation within a finite number of steps.
\end{itemize}     
%
%
We describe an implementation of a mutex object $L(M)$ from a strictly serializable, strongly
progressive TM implementation $M$ providing wait-free TM-liveness (Algorithm~\ref{alg:mutex-dsm}).
The algorithm is based on the mutex implementation in \cite{lee-thesis}.

Given a sequential implementation, we use a TM to execute the sequential code in a concurrent environment
by encapsulating each sequential operation within an \emph{atomic}
transaction that replaces each read and write of a t-object with the
transactional read and write implementations, respectively. 
If the transaction commits, then the result of the operation is
returned; otherwise if one of the transactional operations aborts.
For instance, in Algorithm~\ref{alg:mutex-dsm}, we wish to atomically read a t-object $X$, write a new value to it
and return the old value of $X$ prior to this write. To achieve this, we
employ a strictly serializable TM implementation $M$. 
Since we assume that $M$ is strongly progressive, in every execution,
at least one transaction successfully commits and the value of $X$ is returned.

\vspace{1mm}\noindent\textbf{Shared objects.}
We associate each process $p_i$ with two alternating identities $[p_i,\ms{face}_i]$; $\ms{face}_i \in \{0,1\}$.
The strongly progressive TM implementation $M$ is used to enqueue processes that attempt to enter the critical section within
a single t-object $X$ (initially $\bot$).
For each $[p_i,\ms{face}_i]$, $L(M)$ uses a register bit $\ms{Done}[p_i,\ms{face}_i]$ that indicates if this face of the process
has left the critical section or is executing the $\lit{Entry}$ operation. Additionally, we use
a register $\ms{Succ}[p_i,\ms{face}_i]$ that stores the process expected to succeed $p_i$ in the critical section.
If $\ms{Succ}[p_i,\ms{face}_i]=p_j$, we say that $p_j$ is the \emph{successor of $p_i$} (and $p_i$ is the \emph{predecessor} of $p_j$). 
Intuitively, this means
that $p_j$ is expected to enter the critical section immediately after $p_i$.
Finally, $L(M)$ uses a $2$-dimensional bit array $\ms{Lock}$: for each process $p_i$, there are $n-1$ registers associated with
the other processes. For all $j\in \{0,\ldots , n-1\}\setminus \{i\}$, the registers $\ms{Lock}[p_i][p_j]$
are local to $p_i$ and registers $\ms{Lock}[p_j][p_i]$ are remote to $p_i$. Process $p_i$ can only access registers
in the $\ms{Lock}$ array that are local or remote to it.

\vspace{1mm}\noindent\textbf{Entry operation.}
A process $p_i$ adopts a new identity $\ms{face}_i$ and writes $\false$ to $\ms{Done}(p_i,\ms{face}_i)$ to indicate
that $p_i$ has started the $\lit{Entry}$ operation. Process $p_i$ now initializes the successor of $[p_i,\ms{face}_i]$
by writing $\bot$ to $\ms{Succ}[p_i,\ms{face}_i]$. Now, $p_i$ uses a strongly progressive TM implementation $M$
to atomically store its \emph{pid} and identity i.e., $\ms{face}_i$ to t-object $X$ and returns the \emph{pid}
and identity of its \emph{predecessor}, say $[p_j,\ms{face}_j]$. Intuitively, this suggests that
$[p_i,\ms{face}_i]$ is scheduled to enter the critical section immediately after $]p_j,\ms{face}_j]$ exits the critical
section.
Note that if $p_i$ reads the initial value of t-object $X$, then it immediately enters the critical section.
Otherwise it writes \emph{locked} to the register $\ms{Lock}[p_i,p_j]$ and sets itself to be the successor of $[p_j,\ms{face}_j]$
by writing $p_i$ to $\ms{Succ}[p_j,\ms{face}_j]$.
Process $p_i$ now checks if $p_j$ has started the $\lit{Exit}$ operation by checking if $\ms{Done}[p_j,\ms{face}_j]$
is set. If it is, $p_i$ enters the critical section; otherwise $p_i$ spins on the register $\ms{Lock}[p_i][p_j]$
until it is \emph{unlocked}.

\vspace{1mm}\noindent\textbf{Exit operation.}
Process $p_i$ first indicates that it has exited the critical section by setting $\ms{Done}[p_i,\ms{face}_i]$, following
which it \emph{unlocks} the register $\ms{Lock}[\ms{Succ}[p_i,\ms{face}_i]][p_i]$ to allow $p_i$'s successor to
enter the critical section.
\begin{algorithm}[t]
\caption{Mutual-exclusion object $L$ from a strongly progressive, strict serializable TM $M$; code for process $p_i$; $1\leq i \leq n$}
\label{alg:mutex-dsm}
\begin{algorithmic}[1]
  	\begin{multicols}{2}
  	{
  	\footnotesize
	\Part{Local variables}{
		\State bit $\ms{face}_i$, for each process $p_i$
	}\EndPart
	\Statex
	\Part{Shared objects}{
		\State strongly progressive, strictly 
		\State ~~serializable TM $M$
		\State ~~t-object $X$, initially $\bot$  
		\State ~~storing value $v \in \{[p_i, \ms{face}_i] \} \cup \{\bot\}$
				
		\State for each tuple $[p_i,\ms{face}_i]$
		\State ~~$\ms{Done}[p_i,\ms{face}_i] \in \{\true,\false\}$
		\State ~~$\ms{Succ}[p_i,\ms{face}_i] \in \{p_1,\ldots , p_n \} \cup \{\bot\}$
		\State for each $p_i$ and $j\in \{1,\ldots , n\}\setminus \{i\}$
		\State ~~$\ms{Lock}[p_i][p_j] \in \{\ms{locked},\ms{unlocked}\}$
	}\EndPart	
		
	\Statex	
	\Statex
	\Part{Function: \lit{func}()}{
		\State \textbf{atomic using $M$}
		
		 \State ~~~~$\ms{value}:= \lit{tx-read}(X)$
		 \State ~~~~$\lit{tx-write}(X,[p_i,\ms{face}_i])$
		\State \textbf{on abort Return} $\false$
		 \Return $\ms{value}$ \EndReturn
		 
	}\EndPart	
 	
 	\newpage
	\Part{Entry}{
		 \State $\ms{face}_i :=1- \ms{face}_i$
		 \State $\ms{Done}[p_i,\ms{face}_i].\lit{write}(\false)$ \label{line:entrydone1}
		 \State $\ms{Succ}[p_i,\ms{face}_i].\lit{write}(\bot)$ \label{line:entrysucc1}
		 \While{$(\ms{prev} \gets \lit{func})=\false$} \label{line:tm}
		      \State \textbf{no op}
		 \EndWhile \label{line:while}
		 \If{$\ms{prev} \neq \bot$}
		    \State $\ms{Lock}[p_i][\ms{prev}.pid].\lit{write}(\ms{locked})$ \label{line:entrylock}
		    \State $\ms{Succ}[\ms{prev}].\lit{write}(p_i)$ \label{line:entrysucc2}
		    \If{$\ms{Done}[\ms{prev}]=\false$} \label{line:entrydone2}
		      \While{$\ms{Lock}[p_i][\ms{prev}.pid]=\ms{unlocked}$} \label{line:entryunlock}
			\State \textbf{no op}
		      \EndWhile
		    \EndIf
		 \EndIf
		 \Return $ok$ \EndReturn
		 \State \Comment{Critical section}
		 
   	 }\EndPart
	\Statex
	
	\Part{Exit}{
		  \State $\ms{Done}[p_i,\ms{face}_i].\lit{write}(\true)$ \label{line:exitdone}
		  \State $\ms{Lock}[\ms{Succ}[p_i, \ms{face}_i]][p_i].\lit{write}(\ms{unlocked})$ \label{line:exitunlock}
		  \Return $ok$ \EndReturn
					
   	}\EndPart
		
	}
	\end{multicols}
  \end{algorithmic}
\end{algorithm}
\begin{lemma}
\label{lm:mutex}
The implementation $L(M)$ (Algorithm~\ref{alg:mutex-dsm}) satisfies mutual exclusion.
\end{lemma}
\begin{proof}
Let $E$ be any execution of $L(M)$.
We say that $[p_i,\ms{face}_i]$ is the \emph{successor} of $[p_j,\ms{face}_j]$ if $p_i$ reads the value of $\ms{prev}$
in Line~\ref{line:while} to be 
$[p_j,\ms{face}_j]$ (and $[p_j,\ms{face}_j]$ is the \emph{predecessor} of $[p_i,\ms{face}_i]$); 
otherwise if $p_i$ reads the value to be $\bot$, we say that $p_i$ has no predecessor.

Suppose by contradiction that there exist processes $p_i$ and $p_j$ that are both inside the critical section after $E$.
Since $p_i$ is inside the critical section, either (1) 
$p_i$ read $\ms{prev}=\bot$ in Line~\ref{line:tm}, or
(2) $p_i$ read that $\ms{Done}[\ms{prev}]$ is $\true$ (Line~\ref{line:entrydone2}) or $p_i$ reads that
$\ms{Done}[\ms{prev}]$ is $\false$ and $\ms{Lock}[p_i][\ms{prev.pid}]$
is \emph{unlocked} (Line~\ref{line:entryunlock}).

(Case $1$) Suppose that $p_i$ read $\ms{prev}=\bot$ and entered the critical section. Since in this case, $p_i$ does
not have any predecessor, some other process that returns successfully from the \emph{while} loop in Line~\ref{line:while}
must be successor of $p_i$ in $E$. 
Since there exists $[p_j,\ms{face}_j]$ also inside the critical section after $E$, $p_j$ reads that either
$[p_i,\ms{face}_i]$ or some other process to be its predecessor. Observe that there must exist some such 
process $[p_k,\ms{face}_k]$ whose predecessor is $[p_i,\ms{face}_i]$. Hence, without loss of generality, we can assume that
$[p_j,\ms{face}_j]$ is the successor of $[p_i,\ms{face}_i]$.
By our assumption, $[p_j,\ms{face}_j]$ is also inside the critical section. Thus, $p_j$
\emph{locked} the register $\ms{Lock}[p_j,p_i]$ in Line~\ref{line:entrylock} and set itself to be $p_i$'s successor
in Line~\ref{line:entrysucc2}. Then, $p_j$ read that $\ms{Done}[p_i,\ms{face}_i]$ is $\true$
or read that $\ms{Done}[p_i,\ms{face}_i]$ is $\false$ and waited until $\ms{Lock}[p_j,p_i]$ is \emph{unlocked}
and then entered the critical section. But this is possible only if $p_i$ has left the critical section and 
updated the registers $\ms{Done}[p_i,\ms{face}_i]$ and $\ms{Lock}[p_j,p_i]$ in Lines~\ref{line:exitdone}
and \ref{line:exitunlock} respectively---contradiction to the assumption that $[p_i,\ms{face}_i]$ is also
inside the critical section after $E$.

(Case $2$) Suppose that $p_i$ did not read $\ms{prev}=\bot$ and entered the critical section.
Thus, $p_i$ read that $\ms{Done}[\ms{prev}]$ is $\false$ in Line~\ref{line:entrydone2} and $\ms{Lock}[p_i][\ms{prev.pid}]$
is \emph{unlocked} in Line~\ref{line:entryunlock}, where $\ms{prev}$ is the predecessor of $[p_i,\ms{face}_i]$.
As with case $1$, without loss of generality, we can assume that
$[p_j,\ms{face}_j]$ is the successor of $[p_i,\ms{face}_i]$ or 
$[p_j,\ms{face}_j]$ is the predecessor of $[p_i,\ms{face}_i]$.

Suppose that $[p_j,\ms{face}_j]$ is the predecessor of $[p_i,\ms{face}_i]$, \emph{i.e.}, $p_i$ writes the value
$[p_i,\ms{face}_i]$ to the register $\ms{Succ}[p_j,\ms{face}_j]$ in Line~\ref{line:entrysucc2}.
Since $[p_j,\ms{face}_j]$ is also inside the critical section after $E$,
process $p_i$ must read that $\ms{Done}[p_j,\ms{face}_j]$ is $\true$ in Line~\ref{line:entrydone2}
and $\ms{Lock}[p_i,p_j]$ is \emph{locked} in Line~\ref{line:entryunlock}.
But then $p_i$ could not have entered the critical section after $E$---contradiction.

Suppose that $[p_j,\ms{face}_j]$ is the successor of $[p_i,\ms{face}_i]$, \emph{i.e.}, $p_j$ writes the value
$[p_j,\ms{face}_j]$ to the register $\ms{Succ}[p_i,\ms{face}_i]$. Since both $p_i$ and $p_j$ are inside the critical section
after $E$, process $p_j$ must read that $\ms{Done}[p_i,\ms{face}_i]$ is $\true$ in Line~\ref{line:entrydone2}
and $\ms{Lock}[p_j,p_i]$ is \emph{locked} in Line~\ref{line:entryunlock}.
Thus, $p_j$ must spin on the register $\ms{Lock}[p_j,p_i]$, waiting for it to be \emph{unlocked} by $p_i$
before entering the critical section---contradiction to the assumption that both $p_i$ and $p_j$ are inside the critical section.

Thus, $L(M)$ satisfies mutual-exclusion.
\end{proof}
\begin{lemma}
\label{lm:dead}
The implementation $L(M)$ (Algorithm~\ref{alg:mutex-dsm}) provides deadlock-freedom.
\end{lemma}
\begin{proof}
Let $E$ be any execution of $L(M)$.
Observe that a process may be stuck indefinitely only in Lines~\ref{line:tm} and \ref{line:entryunlock} as it performs
the \emph{while} loop.

Since $M$ is strongly progressive, in every execution $E$ that contains an invocation of $\lit{Enter}$
by process $p_i$, some process returns $\true$ from the invocation of $\ms{func}()$ in Line~\ref{line:tm}.

Now consider a process $p_i$ that returns successfully from the \emph{while} loop in Line~\ref{line:tm}.
Suppose that $p_i$ is stuck indefinitely as it performs the \emph{while} loop in Line~\ref{line:entryunlock}.
Thus, no process has \emph{unlocked} the register $\ms{Lock}[p_i][\ms{prev.pid}]$ by writing to it in the $\lit{Exit}$ section.
Recall that since $[p_i,\ms{face}_i]$ has reached the \emph{while} loop in Line~\ref{line:entryunlock}, $[p_i,\ms{face}_i]$
necessarily has a predecessor, say $[p_j,\ms{face}_j]$, and has set itself to be $p_j$'s successor by writing
$p_i$ to register $\ms{Succ}[p_j,\ms{face}_j]$ in Line~\ref{line:entrysucc2}.
Consider the possible two cases: the predecessor of $[p_j,\ms{face}_j$ is some process $p_k$;$k \neq i$ or
the predecessor of $[p_j,\ms{face}_j$ is the process $p_i$ itself.

(Case $1$) Since by assumption, process $p_j$ takes infinitely many steps in $E$, the only reason that
$p_j$ is stuck without entering the critical section is that $[p_k,\ms{face}_k]$ is also stuck in
the \emph{while} loop in Line~\ref{line:entryunlock}. Note that it is possible for us to iteratively extend this execution
in which $p_k$'s predecessor is a process that is not $p_i$ or $p_j$ that is also stuck in the
\emph{while} loop in Line~\ref{line:entryunlock}. But then the last such process must eventually
read the corresponding $\ms{Lock}$ to be \emph{unlocked} and enter the critical section.
Thus, in every extension of $E$ in which every process takes infinitely many steps, some process will enter the critical section.

(Case $2$) Suppose that the predecessor of $[p_j,\ms{face}_j$ is the process $p_i$ itself.
Thus, as $[p_i,\ms{face}_]$ is stuck in the \emph{while} loop waiting for $\ms{Lock}[p_i,p_j]$ to be \emph{unlocked}
by process $p_j$, $p_j$ leaves the critical section, \emph{unlocks} $\ms{Lock}[p_i,p_j]$ in Line~\ref{line:exitunlock}
and prior to the read of $\ms{Lock}[p_i,p_j]$, $p_j$ re-starts the $\lit{Entry}$ operation,
writes $\false$ to $\ms{Done}[p_j,1-\ms{face}_j]$ and sets itself to be the successor of $[p_i,\ms{face}_i]$
and spins on the register $\ms{Lock}[p_j,p_i]$. However, observe that process $p_i$, which takes infinitely many steps by our assumption
must eventually read that $\ms{Lock}[p_i,p_j]$ is \emph{unlocked} and enter the critical section, thus establishing
deadlock-freedom.
\end{proof}
We say that a TM implementation $M$ \emph{accesses a single t-object} if in every execution $E$ of $M$ and
every transaction $T\in \ms{txns}(E)$, $|\Dset(T)| \leq 1$. We can now prove the following theorem:
\begin{theorem}
\label{th:mutex-tm}
Any strictly serializable, strongly progressive TM implementation $M$ that accesses a single t-object 
implies a \emph{deadlock-free}, 
\emph{finite exit} mutual exclusion
implementation $L(M)$ such that the RMR complexity of $M$ is within a constant factor of the RMR complexity of $L(M)$.
\end{theorem}
\begin{proof}
(Mutual-exclusion)
Follows from Lemma~\ref{lm:mutex}.

(Finite-exit)
The proof is immediate since the $\lit{Exit}$ operation contains no unbounded
loops or waiting statements.

(Deadlock-freedom)
Follows from Lemma~\ref{lm:dead}.

(RMR complexity)
First, let us consider the CC model. Observe that every event not on $M$ performed by a process $p_i$ 
as it performs the $\lit{Entry}$
or $\lit{Exit}$ operations incurs $O(1)$ RMR cost clearly, possibly barring the \emph{while} loop executed in
Line~\ref{line:entryunlock}. During the execution of this \emph{while} loop, process $p_i$ spins on the register
$\ms{Lock}[p_i][p_j]$, where $p_j$ is the predecessor of $p_i$.  Observe that $p_i$'s cached copy of $\ms{Lock}[p_i][p_j]$
may be invalidated only by process $p_j$ as it \emph{unlocks} the register in Line~\ref{line:exitunlock}.
Since no other process may write to this register and $p_i$ terminates the \emph{while} loop
immediately after the write to $\ms{Lock}[p_i][p_j]$ by $p_j$, $p_i$ incurs $O(1)$ RMR's.
Thus, the overall RMR cost incurred by $M$ is within a constant factor of the RMR cost of $L(M)$.

Now we consider the DSM model. As with the reasoning for the CC model, every event not on $M$ performed by a process $p_i$
as it performs the $\lit{Entry}$
or $\lit{Exit}$ operations incurs $O(1)$ RMR cost clearly, possibly barring the \emph{while} loop executed in
Line~\ref{line:entryunlock}. During the execution of this \emph{while} loop, process $p_i$ spins on the register
$\ms{Lock}[p_i][p_j]$, where $p_j$ is the predecessor of $p_i$. Recall that $\ms{Lock}[p_i][p_j]$
is a register that is local to $p_i$ and thus, $p_i$ does not incur any RMR cost on account of executing this loop.
It follows that $p_i$ incurs $O(1)$ RMR cost in the DSM model. Thus, the overall RMR cost of $M$ is within
a constant factor of the RMR cost of $L(M)$ in the DSM model.
\end{proof}
\begin{theorem}
\label{th:mutex-rmr}
(\cite{rmr-mutex})
Any deadlock-free, finite-exit mutual exclusion implementation from read, write and 
conditional primitives has an execution whose RMR complexity is $\Omega(n \log{n})$.
\end{theorem}
Theorems~\ref{th:mutex-rmr} and \ref{th:mutex-tm} imply:
\begin{theorem}
\label{th:mutex-tm-rmr}
Any strictly serializable, strongly progressive TM implementation with wait-free TM-liveness from read, write and 
conditional primitives that accesses a single t-object has an execution whose RMR complexity is $\Omega(n \log{n})$.
\end{theorem}
\subsection{A constant expensive synchronization opaque TM}
\label{sec:p3c2s3s1}
In this section, we describe a strongly progressive opaque TM implementation providing starvation-free TM-liveness 
from read-write base objects with constant RAW/AWAR cost.
For our implementation, we define and implement a \emph{starvation-free multi-trylock} object.

\vspace{1mm}\noindent\textbf{Starvation-free multi-trylock.}
A \textit{multi-trylock} provides exclusive write-access to a set $Q$ of t-objects. 
Specifically, a \textit{multi-trylock} exports the following operations
\begin{itemize}
\item \textit{acquire(Q)} returns {\it true} or {\it false}
\item \textit{release(Q)} releases the lock and returns \textit{ok}
\item \textit{isContended($X_j$)}, $X_j \in Q$ returns \emph{true} or \emph{false}
\end{itemize}
We assume that processes are well-formed: they never invoke a new
operation on the multi-trylock before receiving response from the
previous invocation. 

We say that a process $p_i$ \emph{holds a lock on $X_j$ after an execution $\pi$} if
$\pi$ contains the invocation of \textit{acquire($Q$)}, $X_j\in Q$ by
$p_i$ that returned \emph{true}, but does not contain a subsequent
invocation of \textit{release($Q'$)}, $X_j\in Q'$, by $p_i$ in $\pi$. 
We say that $X_j$ is \emph{locked after $\pi$} by process $p_i$ if $p_i$
holds a lock on $X_j$ after $\pi$.

We say that $X_j$ is \emph{contended by $p_i$ after an execution $\pi$} if
$\pi$ contains the invocation of \textit{acquire($Q$)}, $X_j \in Q$,
by $p_i$ but does not
contain a subsequent return \emph{false} or return of
\textit{release($Q'$)}, $X_j \in Q'$, by $p_i$ in $\pi$. 
 
Let an execution $\pi$ contain the invocation $i_{op}$ of an operation $op$
followed by a corresponding response $r_{op}$ (we say that $\pi$
\emph{contains} $op$). 
We say that \emph{$X_j$ is uncontended (resp., locked) during the execution of $op$ in
  $\pi$} if $X_j$ is uncontended (resp., locked) after every prefix of $\pi$ that
contains $i_{op}$ but does not contain $r_{op}$.

We implement a \emph{multi-trylock} object whose operations are \emph{starvation-free}.
The algorithm is inspired by the \emph{Black-White Bakery Algorithm}~\cite{Gadi-Bakery} and uses a 
finite number of bounded registers.

A starvation-free multi-trylock implementation satisfies the following properties:
\begin{itemize}
\item \emph{Mutual-exclusion}: For any object $X_j$, and any execution
  $\pi$, there exists at most one process that \emph{holds} a lock on
  $X_j$ after $\pi$. 

\item \emph{Progress}:  Let $\pi$ be any execution that contains
  $\textit{acquire}(Q)$ by process $p_i$. If no other process $p_k, k\neq i$ contends infinitely long on some 
$X_j \in Q$, then $\textit{acquire}(Q)$ returns \emph{true} in $\pi$. 

\item Let $\pi$ be any execution that contains $\textit{isContended}(X_j)$ invoked by $p_i$.

\begin{itemize}
\item If $X_j$ is locked by $p_{\ell};\ell \neq i$ during the complete execution of
$\textit{isContended}(X_j)$ in $\pi$, then
$\textit{isContended}(X_j)$ returns $\true$.

\item If $\forall \ell \neq i$, $X_j$ is never contended by $p_{\ell}$ during the execution of
$\textit{isContended}(X_j)$ in $\pi$, then
$\textit{isContended}(X_j)$ returns $\false$.
\end{itemize}
\end{itemize}
Our starvation-free multi-trylock in Algorithm~\ref{alg:mul2} uses the following shared variables: registers $r_{ij}$ for each 
process $p_i$ and object $X_j$, a shared bit $\textit{color} \in \{B,W\}$, 
registers $LA_i \in \{0, \ldots , N\}$ for each $p_i$ that denote a \emph{Label} and $MC_i \in \{B,W\}$ for each $p_i$. 

We say $(LA_i,i) < (LA_k,k)$ \emph{iff} $LA_i < LA_k$ or $LA_i=LA_k$ and $i<k$.
We now prove the following invariant about the multi-trylock implementation.
\begin{algorithm}[t]
\caption{Starvation-free multi-trylock invoked by process $p_i$}\label{alg:mul2}
  \begin{algorithmic}[1]
  	{\footnotesize
	\Part{Shared variables}{
		\State $LA_i$, for each process $p_i$, initially $0$
		\State \textit{$MC_i \in \{B,W\}$} for each process $p_i$, initially $W$
		\State $\textit{color} \in \{B,W\}$, initally $W$
	  	\State $r_{ij}$, for each process $p_i$ and each t-object $X_j$, initially $0$
	}\EndPart	
	\Statex
	\Part{\textit{acquire}($Q$)}{
   		\ForAll{$X_j \in Q$}	
			\State \Write$(r_{ij},1)$ \label{line:lwrite}
		\EndFor
		\State $c_i:=color$ \label{line:cread}
		\State \Write$(MC_i,c_i)$ \label{line:color}
		\State \Write($LA_i, 1+max(\{LA_k)|MC_k=MC_i$\}) \label{line:label1}
		\While{$\exists j:\exists k \neq i$: $isContended(X_j)$ $\&\&$ (($LA_k \neq 0$;~($MC_k=MC_i$);~$(LA_k,k) < (LA_i,i)$) $||$ \\~~~~($LA_k \neq 0;~(MC_k \neq MC_i$);~$MC_i=color$)) } \label{line:label2}
			\State no op
		\EndWhile
		\State return $\true$ 
	}\EndPart		
	 \Statex
	 \Part{\textit{release}($Q$)}{
  		\ForAll{$X_j \in Q$} \label{line:rwrite}	
 			\State \Write$(r_{ij},0)$
		\EndFor
		\If{$MC_i=B$} \label{line:color2}
			\State $\Write(color,W)$
		\Else
			\State $\Write(color,B)$
		\EndIf
		\State \Write($LA_i,0$)
		\State return \ok
	}\EndPart
	\Statex
	\Part{{\it \textit{isContended}($X_j$)}}{
		\If{$\exists p_t: r_{tj} \neq 0, t \neq i$} \label{line:isl}
			\State return $\true$
		\EndIf
		\State return $\false$
	}\EndPart			
	}
  \end{algorithmic}
\end{algorithm}

\begin{lemma}
\label{lm:mutexsp}
In every execution $\pi$ of Algorithm~\ref{alg:mul2}, if $p_i$ \emph{holds} a lock on some object $X_j$ after $\pi$, then one of the following conditions must hold:
\begin{enumerate}
\item[(1)]
for some $k\neq i$; $LA_k \neq 0$, if $MC_k=MC_i$, then $(LA_k,k) > (LA_i,i)$
\item[(2)]
for some $k\neq i$; $LA_k \neq 0$, if $MC_k \neq MC_i$, then $MC_i \neq color$
\end{enumerate}
\end{lemma}
\begin{proof}
In order to hold the lock on $X_j$, some process $p_i$ writes $1$ to $r_{ij}$, writes a value, say $W$ to $MC_i$ and reads the Labels
of other processes that have obtained the same color as itself and generates a Label greater by one than the maximum Label read (Line~\ref{line:label1}). 
Observe that until the value of the \emph{color} bit is changed, all processes read the same value $W$. The first process $p_i$
to hold the lock on $X_j$ changes the \emph{color} bit to $B$ when releasing the lock and hence the value read by all subsequent processes will be $B$ until it is changed again. Now consider two cases:  
\begin{enumerate}
\item[(1)]
Assume that there exists a process $p_k$, $k \neq i$, $LA_k \neq 0$ and $MC_k=MC_i$ such that $(LA_k,k) < (LA_i,i)$, but $p_i$
holds a lock on $X_j$ after $\pi$. Thus, $isContended(X_j)$ returns \emph{true} to $p_i$ because $p_k$ writes to $r_{kj}$ (Line~\ref{line:lwrite}) before writing to $LA_k$ (Line~\ref{line:label1}). 
By assumption, $(LA_k,k) < (LA_i,i);LA_k>0$ and $MC_i=MC_k$, but the conditional in Line~\ref{line:label2} returned \emph{true} to $p_i$ 
without waiting for $p_k$ to stop contending on $X_j$---contradiction. 
\item[(2)]
Assume that there exists a process $p_k$, $k \neq i$, $LA_k \neq 0$ and $MC_k \neq MC_i$ such that $MC_i=color$, but $p_i$ holds a lock on $X_j$ after $\pi$. Again, since $LA_k >0$, $isContended(X_j)$ returns \emph{true} to $p_i$, $MC_k \neq MC_i$ and $MC_i=color$, but the conditional in Line~\ref{line:label2} returned \emph{true} to $p_i$ without waiting for $p_k$ to stop contending on $X_j$---contradiction.
\end{enumerate}
\end{proof}
We can thus prove the following theorem:
\begin{theorem}
\label{th:lock2}
Algorithm~\ref{alg:mul2} is an implementation of multi-trylock object
in which every operation is starvation-free and incurs at most four RAWs.
\end{theorem}
\begin{proof}
Denote by $L$ the shared object implemented by
Algorithm~\ref{alg:mul2}.

Assume, by contradiction, that $L$ does not provide mutual-exclusion:
there exists an execution $\pi$ after which processes $p_i$ and $p_k$, $k\neq i$
\emph{hold} a lock on the same object, say $X_j$. Since both $p_i$ and $p_k$ have performed the write to $LA_i$ and $LA_k$ resp. in Line~\ref{line:label1}, $LA_i, LA_k >0$. Consider two cases:
\begin{enumerate}
\item[(1)]
If $MC_k=MC_i$, then from Condition $1$ of Lemma~\ref{lm:mutexsp}, we have $(LA_k,k) < (LA_i,i)$ and $(LA_k,k) > (LA_i,i)$---contradiction.
\item[(2)]
If $MC_k \neq MC_i$, then from Condition $2$ of Lemma~\ref{lm:mutexsp}, we have $MC_i \neq color$ and $MC_k \neq color$ which
implies $MC_k=MC_i$---contradiction.
\end{enumerate}

$L$ also ensures progress. If process $p_i$ wants to hold the lock on an object $X_j$ i.e. invokes $\textit{acquire}(Q), X_j \in Q$, it checks if any other process $p_k$ holds the lock on $X_j$. If such a process $p_k$ exists and $MC_k=MC_i$, then clearly $isContended(X_j)$ returns \emph{true} for $p_i$ and $(LA_k,k) < (LA_i,i)$. Thus, $p_i$ fails the conditional in Line~\ref{line:label2} and waits until $p_k$ releases the lock on $X_j$ to return \emph{true}. However, if $p_k$ contends infinitely long on  $X_j$, $p_i$ is also forced to wait indefinitely to be returned \emph{true} from the invocation of $\textit{acquire}(Q)$. The same argument works when $MC_k \neq MC_i$ since when $p_k$ stops contending on $X_j$, $isContended(X_j)$ eventually returns \emph{false} for $p_i$ if $p_k$ does not contend infinitely long on $X_j$.

All operations performed by $L$ are starvation-free. Each process $p_i$ that successfully holds the lock on
an object $X_j$ in an execution $\pi$ invokes $\textit{acquire}(Q), X_j \in Q$, obtains a color and 
chooses a value for $LA_i$ since there is no way to be blocked while writing to $LA_i$. 
The response of operation $\textit{acquire}(Q)$ by $p_i$ is only delayed if there exists a 
concurrent invocation of $\textit{acquire}(Q'),X_j \in Q'$ by $p_k$ in $\pi$. 
In that case, process $p_i$ waits until $p_k$ invokes $\textit{release}(Q)$ and 
writes $0$ to $r_{kj}$ and eventually holds the lock on $X_j$. The implementation of \emph{release} 
and \emph{isContended} are wait-free operations (and hence starvation-free) since they 
contains no unbounded loops or waiting statements.

The implementation of $\textit{isContended}(X_j)$ only reads base
objects. 
The implementation of $\textit{release}(Q)$ writes to a series of base objects (Line~\ref{line:rwrite}) and 
then reads a base object (Line~\ref{line:color2}) incurring a single RAW. 
The implementation of $\textit{acquire}(Q)$ writes to base objects (Line~\ref{line:lwrite}), reads the 
shared bit $color$ (Line~\ref{line:cread})---one RAW, writes to a base object (Line~\ref{line:color}), 
reads the Labels (Line~\ref{line:label1})---one RAW, writes to its own Label and finally performs a 
sequence of reads when evaluating the conditional in Line~\ref{line:label2}---one RAW. 

Thus, Algorithm~\ref{alg:mul2} incurs at most four RAWs.
\end{proof}
\vspace{1mm}\noindent\textbf{Strongly progressive TM from starvation-free multi-trylock.}
We now use the starvation-free multi-trylock to implement a starvation-free strongly progressive opaque TM implementation
with constant expensive synchronization (Algorithm~\ref{alg:gp}).
The implementation is almost identical to the progressive TM implementation $LP$ in Algorithm~\ref{alg:ic}, except
that the function calls to $\textit{acquire}$ and $\textit{release}$ the transaction's write set are replaced
with analogous calls to a multi-trylock object.
\begin{algorithm}[t]
\caption{Strongly progressive, opaque TM: the implementation of $T_k$ executed by $p_i$}\label{alg:gp}
\begin{algorithmic}[1]
  	\begin{multicols}{2}
  	{\footnotesize
	\Part{Shared variables}{
		\State $v_j$, for each t-object $X_j$
		\State $L$, a starvation-free multi-trylock object
	}\EndPart
	\Part{Local variables}{
		\State $\ms{Rset}_k,\ms{Wset}_k$ for every transaction $T_k$;
		\State ~~~~dictionaries storing $\{X_m, v_m\}$
	}\EndPart	

	\Statex
	\Part{\Read$_k(X_j)$}{
		\If{$X_j \not\in \Rset(T_k)$}
		
		\State $[\textit{ov}_j,k_j ] := \Read(v_j)$ 
		\State $\Rset(T_k) := \Rset(T_k)\cup\{X_j,[\textit{ov}_j,k_j]\}$ 
		\If{$\lit{isAbortable}$()}
		      \Return $A_k$ \EndReturn
		\EndIf
		\Return $\textit{ov}_j$ \EndReturn
		
		\Else
		    
		\State $[\textit{ov}_j, \bot] :=\Rset(T_k).\lit{locate}(X_j)$
		\Return $\textit{ov}_j$ \EndReturn
		
		\EndIf
   	}\EndPart
	\Statex
	\Part{\Write$_k(X_j,v)$}{
		\State $\textit{nv}_j := v$
		\State $\Wset(T_k) := \Wset(T_k)\cup\{X_j\}$
		\Return $\ok$ \EndReturn
		
   	}\EndPart
	
	\newpage
	\Part{\TryC$_k$()}{
		\If{$|\Wset(T_k)|= \emptyset$}
			\Return $C_k$ \EndReturn \label{line:spreturn}
		\EndIf
		\State locked $:= L.\textit{acquire}(\Wset(T_k))$\label{line:spacq} 
		\If{isAbortable()} \label{line:spabort3}
			\State $L.\textit{release}(\Wset(T_k))$ 
			\Return $A_k$ \EndReturn
		\EndIf
		\ForAll{$X_j \in \Wset(T_k)$}
	 		 \State $\Write(v_j,(\textit{nv}_j,k))$ \label{line:spwrite}
	 	\EndFor		
		\State $L.\textit{release}(\Wset(T_k))$   		
   		\State return $C_k$
   	 }\EndPart	
   	 \Statex
   	 \Part{Function: {\bf isAbortable()}}{
		\If{$\exists X_j \in \Rset(T_k): X_j\not\in \Wset(T_k)\wedge L.\textit{isContended}(X_j)$}
			\Return $\true$ \EndReturn
		\EndIf
		\If{$\lit{validate}()$} 
			\Return $\true$ \EndReturn
		\EndIf
		\Return $\false$ \EndReturn
	}\EndPart
	 
	}\end{multicols}
  \end{algorithmic}
\end{algorithm}

\begin{theorem}
\label{th:sprogop}
Algorithm~\ref{alg:gp} implements a strongly progressive opaque TM
implementation with starvation-free t-operations that uses invisible reads and employs
at most four RAWs per transaction.
\end{theorem}
\begin{proof}
\textit{(Opacity)}
Since Algorithm~\ref{alg:gp} is similar to the opaque progressive TM implementation in Algorithm~\ref{alg:ic}, 
it is easy to adapt the proof of Lemma~\ref{lm:icopaque} to prove opacity for this implementation.

\textit{(TM-progress and TM-liveness)}
Every transaction $T_k$ in a TM $M$ whose t-operations are defined by Algorithm~\ref{alg:gp} 
can be aborted in the following scenarios:
\begin{itemize}
\item Read-validation failed in \textit{$read_k$} or $\TryC_k$
\item \textit{$read_k$} or \textit{tryC$_k$} returned $A_k$ because $X_j \in \Rset(T_k)$ is 
locked (belongs to write set of a concurrent transaction)  
\end{itemize}
Since in each of these cases, a transaction is aborted only because of a read-write conflict with a concurrent transaction,
it is easy to see that $M$ is progressive.

To show Algorithm~\ref{alg:gp} also implements a strongly progressive TM, we 
need to show that for every set of transactions that concurrently contend on a single t-object, 
at least one of the transactions is not aborted. 

Consider transactions $T_i$ and $T_k$ that concurrently attempt to execute $\TryC_i$ and $\TryC_k$ 
such that $X_j \in \Wset_i \cup \Wset_k$. Consequently, they both invoke the \emph{acquire} operation 
of the multi-trylock  (Line~\ref{line:spacq}) and thus, from Theorem~\ref{th:lock2}, both $T_i$ and $T_k$ 
must commit eventually.
Also, if validation of a t-read in $T_k$ fails, it means that the t-object is overwritten by some 
transaction $T_i$ such that $T_i$ precedes $T_k$, implying at least one of the transactions commit. 
Otherwise, if some t-object $X_j \in \Rset(T_k)$ is locked and returns \emph{abort} since the t-object is 
in the write set of a concurrent transaction $T_i$. While it may still be possible that $T_i$ returns $A_i$ 
after acquiring the lock on $\Wset_i$, strong progressiveness only guarantees progress for transactions that 
conflict on at most one t-object. Thus, in either case, for every set of transactions that conflict on at most
one t-object, at least one transaction is not forcefully aborted.

Starvation-free TM-liveness follows from the fact that the multi try-lock we use in the implementation of $M$
provides starvation-free \emph{acquire} and \emph{release} operations.

\textit{(Complexity)}
Any process executing a transaction $T_k$ holds the lock on $\Wset(T_k)$ only once during \emph{tryC$_k$}. If
$|\Wset(T_k)|=\emptyset$, then the transaction simply returns $C_k$
incurring no RAW's. Thus, from Theorem~\ref{th:lock2}, Algorithm~\ref{alg:gp} incurs at most four RAWs per updating transaction and no RAW's are performed in read-only transactions.
\end{proof}
%
%
%
\section{On the cost of permissive opaque TMs}
\label{sec:perm}
We have shown that (strongly) progressive TMs that allow a transaction to be aborted only on read-write conflicts
have constant RAW/AWAR complexity. However, not aborting on conflicts may not necessarily affect TM-correctness.
Ideally, we would like to derive TM implementations that are \emph{permissive}, in the sense that a transaction
is aborted only if committing it would violate TM-correctness.
\begin{definition}[Permissiveness]
\label{def:perm}
A TM implementation $M$ is \emph{permissive with respect to TM-correctness $C$} if
for every history $H$ of $M$ such that $H$ ends with a response
$r_k$ and 
replacing $r_k$ with some $r_k\neq A_k$ gives a history that satisfies $C$, 
we have $r_k \neq A_k$.
\end{definition}
Therefore, permissiveness does not allow a transaction to abort, unless committing 
it would violate the execution's correctness. 

We first show that a transaction in a permissive opaque implementation can only be forcefully aborted if it  
tries to commit:
\begin{lemma}
\label{lem:perm-tryC}
Let a TM implementation $M$ be permissive with respect to opacity. 
If a transaction $T_i$ is forcefully aborted executing 
a t-operation $op_i$, then $op_i$ is $\TryC_i$.
\end{lemma}
\begin{proof}
Suppose, by contradiction, that there exists a history $H$ of $M$ such that 
some $op_i\in\{\Read_i,\Write_i\}$ executed within a transaction $T_i$ returns $A_i$.
Let $H_0$ be the shortest prefix of $H$ that ends just before $op_i$ returns.
By definition, $H_0$ is opaque and any history $H_0\cdot r_i$ where $r_i\neq A_i$ is not opaque.
Let $H_0'$ be the serialization of $H_0$.

If  $op_i$ is a write, then $H_0\cdot\ok_i$ is also opaque - no write operation of the incomplete transaction $T_i$ appears in $H_0'$ and, thus, $H_0'$ is also a serialization of $H_0\cdot\ok_i$. 

If  $op_i$ is a $\Read(X)$ for some t-object $X$, then we can construct a serialization of $H_0\cdot v$ where $v$ is the value of $X$ written by the last committed transaction in $H_0'$ preceding $T_i$ or the initial value of $X$ if there is no such transaction. 
It is easy to see that $H_0"$ obtained from $H_0'$ by adding $\Read(X)\cdot v$ at the end of $T_i$ is a serialization of 
$H_0\cdot\Read(X)$.
In both cases, there exists a non-$A_i$ response $r_i$ to $op_i$ that preserves opacity of  $H_0\cdot r_i$, and, thus,
the only t-operation that can be forcefully aborted in an execution of $M$ is $\TryC$.  
\end{proof}
We now show that an execution of a transaction 
in a \emph{permissive opaque} TM implementation (providing starvation-free TM-liveness) may require to perform
at least one RAW/AWAR pattern \emph{per} t-read.
\begin{theorem}
\label{th:perm}
Let $M$ be a permissive opaque TM implementation providing starvation-free TM-liveness.
Then, for any $m\in\Nat$, $M$ has an execution in which some transaction
performs $m$ t-reads such that the execution of each t-read contains 
at least one RAW or AWAR.
\end{theorem}
\begin{proof}
Consider an execution $E$ of $M$ consisting of 
transactions $T_1$, $T_2$, $T_3$ as shown in Figure~\ref{fig:H}:
$T_3$ performs a t-read of $X_1$, then $T_2$ performs a t-write on $X_1$ and commits, 
and finally $T_1$ performs a series of reads from objects $X_1,\ldots,X_m$.
Since the implementation is permissive, no transaction can be
forcefully aborted in $E$, and 
the only valid serialization of this execution is $T_3$, $T_2$, $T_1$.  
Note also that the execution generates a sequential history: 
each invocation of a t-operation is immediately followed by a matching response.
Thus, since we assume starvation-freedom as a liveness property, such an execution exists. 

We consider $\Read_1(X_k)$, $2 \leq k \leq m$ in execution $E$.
Imagine that we modify the execution $E$ as follows. Immediately after $\Read_1(X_k)$ executed 
by $T_1$ we add $\Write_3(X,v)$, and $\TryC_3$ executed by $T_3$
(let $TC_3(X_k)$ denote the complete execution of $W_3(X_k,v)$ followed by $\TryC_3$).
Obviously, $TC_3(X_k)$ must return abort: neither $T_3$ can be serialized before $T_1$ nor
$T_1$ can be serialized before $T_3$.
On the other hand if $TC_3(X_k)$ takes place just before $\Read_1(X_k)$, then 
$TC_3(X_k)$ must return commit but $\Read_1(X_k)$ must return the value written by $T_3$.
In other words, $\Read_1(X_k)$ and $TC_3(X_k)$ are \emph{strongly non-commutative}~\cite{AGK11-popl}:
both of them see the difference when ordered differently.    
As a result, intuitively, $\Read_1(X_k)$ 
needs to perform a RAW or AWAR to make sure that the order
of these two ``conflicting'' operations is properly maintained. We formalize this argument below.

Consider a modification $E'$ of $E$, in which $T_3$ performs $\Write_3(X_k)$ immediately after $\Read_1(X_k)$
and then tries to commit.
In any serialization of $E'$, $T_3$ must precede $T_2$ ($\Read_3(X_1)$ returns the initial value of $X_1$)
and $T_2$ must precede $T_1$ to respect the real-time order of transactions.
The execution of $\Read_1(X_k)$ does not modify base objects, hence, $T_3$
does not observe $\Read_1(X_k)$ in $E'$.
Since $M$ is permissive, $T_3$ must commit in $E'$. 
But since $T_1$ performs $\Read_1(X_k)$ before $T_3$ commits and $T_3$ 
updates $X_k$, we also have $T_1$ must precede $T_3$ in any serialization.
Thus, $T_3$ cannot precede $T_1$ in any serialization---contradiction. 
Consequently, each $\Read_1(X_k)$ must perform a write to a base object.

Let $\pi$ be the execution fragment that represents the complete execution of $\Read_1(X_k)$
and $E^k$, the prefix of $E$ up to (but excluding) the invocation of $\Read_1(X_k)$.

Clearly, $\pi$ contains a write to a base object. Let $\pi_w$ be 
the first write to a base object in $\pi$. Thus, $\pi$ can be represented 
as $\pi_s\cdot \pi_w\cdot \pi_f$.
Suppose that $\pi$ does not contain a RAW or AWAR. 
Consider the execution fragment $E^k\cdot \pi_s\cdot \rho$, where $\rho$ is the complete execution 
of $TC_3(X_k)$ by $T_3$. Such an execution of $M$ exists since $\pi_s$ 
does not perform any base object write, hence, $E^k\cdot \pi_s \cdot\rho$ 
is indistinguishable to $T_3$ from $E^k\cdot \rho$. 

Since, by our assumption, $\pi_w\cdot \pi_f$ contains no RAW, 
any read performed in $\pi_w\cdot \pi_f$ 
can only be applied to base objects previously written in $\pi_w\cdot \pi_f$. 
Since $\pi_w$ is not an AWAR, $E^k\cdot \pi_s \cdot\rho \cdot \pi_w\cdot \pi_f$ 
is an execution of $M$ since 
it is indistinguishable to $T_1$ from $E^k\cdot \pi$.
In $E^k\cdot \pi_s \cdot\rho \cdot \pi_w\cdot \pi_f$, $T_3$ commits (as in $\rho$) but 
$T_1$ ignores the value written by $T_3$ to $X_k$.  
But there exists no serialization that justifies this execution---contradiction to the assumption that $M$ is opaque.
Thus, each $\Read_1(X_k)$, $2\leq k \leq m$ must contain a RAW/AWAR.

Note that since all t-reads of $T_1$ are executed sequentially, all these RAW/AWAR patterns are
pairwise non-overlapping, which completes the proof. 
\end{proof}
\begin{figure}[t]
\begin{center}
\scalebox{0.95}{
\begin{tikzpicture}

\node (r1) at (2.3,0) {};
\node (rm) at (9,0) {};
\node (w1) at (0,-1) {};

\node (r3) at (-1.8,-2) {};

\draw (r1) node [below] {\tiny {$R_1(X_1)\rightarrow nv$}}
   (rm) node [below] {\tiny {$R_1(X_{m})$}};
   
\draw (w1) node [below] {\tiny {$W_2(X_1,nv)$}};
   
\draw 
(r3) node [below] {\tiny {$R_3(X_1)\rightarrow v$}};

\begin{scope}
\draw[loosely dashed] (4,-0.3) to (7.5,-0.3);
\end{scope}
\begin{scope}   
\draw [|-,thick] (1.5,0) node[left] {{\tiny $T_1$}} to (10,0);
\draw [-|,thick] (1.5,0) to (3,0);
\draw [|-|,thick] (8.5,0) to (10,0);
\end{scope}

\begin{scope}
\draw [|-|,thick] (-.8,-1) node[left] {{\tiny $T_2$}} to (.8,-1) node[right] {{\tiny $C_2$}} ; 
\end{scope}  

\begin{scope}
\draw [|-,thick] (-2.5,-2) node[left] {{\tiny $T_3$}} to (11,-2) ; 
\draw [-|,thick] (-2.5,-2) to (-1,-2);
\end{scope} 
\end{tikzpicture}
}
\end{center}
\caption{\small {Execution $E$ of a permissive, opaque TM: $T_2$ and $T_3$ force $T_1$ to perform 
a RAW/AWAR in each $R_1(X_k)$, $2 \leq k \leq m$}}
\label{fig:H}
\end{figure}
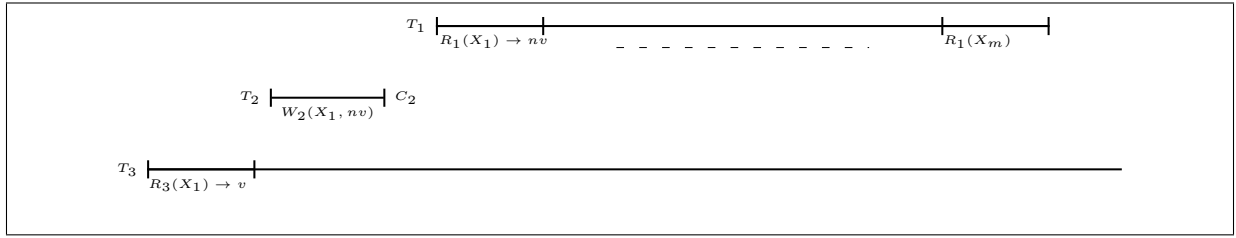
%
%
\section{Related work and Discussion}
\label{sec:p3c2disc}
In this section, we summarize the complexity bounds for blocking TMs presented in this chapter and identify some
open questions.

\vspace{1mm}\noindent\textbf{Sequential TMs.}
Theorem~\ref{th:iclb} improves the read-validation step-complexity lower bound~\cite{GK09-progressiveness,tm-book} 
derived for \emph{strict-data
partitioning} (a very strong version of DAP) and \emph{invisible reads}.
In a \emph{strict data partitioned} TM, the set of base objects used by the TM is split into 
disjoint sets, each storing information only about a single data item.
Indeed, every TM implementation that is strict data-partitioned satisfies weak DAP, but not vice-versa (cf. Section~\ref{sec:dap}).
The definition of invisible reads assumed in \cite{GK09-progressiveness,tm-book}
requires that a t-read operation does not apply nontrivial events in any execution.
Theorem~\ref{th:iclb} however, assumes \emph{weak}
invisible reads, stipulating that t-read operations of a transaction $T$ do not apply nontrivial events only when $T$
is not concurrent with any other transaction.
We believe that the TM-progress and TM-liveness restrictions
as well as the definitions of DAP and invisible reads we consider for this result are the weakest possible assumptions
that may be made. To the best of our knowledge, these assumptions cover every TM implementation that is subject
to the validation step-complexity~\cite{HLM+03,norec,DSS06}.

\begin{table}[h]
      \begin{center}
      \scalebox{0.85}[0.85]{
      
     \begin{tabular}{l|l|l|l|l|l|l}
	TM-correctness & TM-liveness & DAP & Invisible reads & Read-write & Complexity\\ \hline
	Opacity & ICF  & weak & yes & yes & $\Theta (|\Rset|^2)$ step-complexity \\ \hline
	Strict serializability & ICF  & weak & yes & yes & $\Theta (|\Rset|)$ step-complexity for tryCommit \\ \hline
	Opacity & WF  & strict & yes & yes & $O(1)$ RAW/AWAR, $O(1)$ stalls for t-reads \\ \hline
	Opacity & starvation-free  & strict &  &  & $\Theta(|\Wset|)$ protected data
	
     \end{tabular}
     }
     \end{center}
     \caption{Complexity bounds for progressive TMs.\label{table:prog}}
\end{table}
\begin{table}[h]
      \begin{center}
      \scalebox{0.85}[0.85]{
      
     \begin{tabular}{l|l|l|l|l|l|l}
	TM-correctness & TM-liveness & Invisible reads & rmw primitives & Complexity\\ \hline
	Strict serializability & WF  &  & read-write & Impossible \\ \hline
	Strict serializability &   &  & read-write, conditional & $\Omega (n\log n)$ RMRs \\ \hline
	Opacity & starvation-free   & yes & read-write & $O(1)$ RAW/AWAR
	
     \end{tabular}
     }
     \end{center}
     \caption{Complexity bounds for strongly progressive TMs.\label{table:sprog}}
\end{table}
\vspace{1mm}\noindent\textbf{Progressive TMs.}
We summarize the known complexity bounds for progressive (and resp. strongly progressive) TMs 
in Table~\ref{table:prog} (and resp. Table~\ref{table:sprog}).
Some questions remain open. Can the tight bounds on step complexity for progressive TMs in 
Corollaries~\ref{cr:pg1} and \ref{cr:pg2} be extended to strongly progressive TMs?

Guerraoui and Kapalka~\cite{tm-book} proved that it is impossible to implement strictly serializable strongly progressive TMs
that provide \emph{wait-free} TM-liveness (every t-operation returns a matching response within a finite number of steps)
using only read and write primitives. Algorithm~\ref{alg:gp} describes one means to circumvent this impossibility
result by describing an opaque strongly progressive TM implementation from read-write base objects
that provides starvation-free TM-liveness.

We conjecture that the lower bound of Theorem~\ref{th:mutex-tm-rmr} on the RMR complexity is tight. 
Proving this remains an interesting open question.

\vspace{1mm}\noindent\textbf{Permissive TMs.}
Crain et al.~\cite{michel-permissive} proved that a
permissive opaque TM implementation cannot maintain invisible reads, which
inspired the derivation of our lower bound on RAW/AWAR complexity in Section~\ref{sec:perm}.
Furthermore, \cite{michel-permissive} described a permissive VWC TM implementation that ensures
that t-read operations do not perform nontrivial primitives, but the tryCommit invoked by a read-only transaction
perform a linear (in the size of the transaction's data set) number of RAW/AWARs.
Thus, an open question is whether there exists a linear lower bound on RAW/AWAR complexity for weaker (than opacity)
TM-correctness properties of VWC and TMS1.
%
%
\chapter{Complexity bounds for non-blocking TMs}
\label{ch:p3c3}
\section{Overview}
\label{sec:p3c3s0}
In the previous chapter, we presented complexity bounds for \emph{lock-based} blocking TMs.
Early TM implementations such as the popular \emph{DSTM}~\cite{HLM+03} however avoid using locks and 
provide non-blocking TM-progress.
In this chapter, we present several complexity bounds for non-blocking TMs exemplified by
\emph{obstruction-freedom}, possibly the weakest non-blocking
progress condition~\cite{HLM03,HS11-progress}.

We first establish that it is impossible to
implement a strictly serializable obstruction-free TM that provides both weak DAP and \emph{read invisibility}.
Indeed, popular obstruction-free TMs like \emph{DSTM}~\cite{HLM+03} and \emph{FSTM}~\cite{fraser} are weak DAP, but
use \emph{visible} reads for aborting pending writing transactions.
Secondly, we show that a t-read operation in a $n$-process strictly serializable obstruction-free TM
implementation may incur $\Omega(n)$ stalls.
Specifically, we prove that every such TM implementation has a $(n-1)$-stall execution for an invoked t-read operation.
Thirdly, we prove that any RW DAP opaque
obstruction-free TM implementation has an execution in which a
read-only transaction incurs $\Omega(n)$ non-overlapping \emph{RAWs} or \emph{AWARs}. 
Finally, we show that there exists a considerable complexity gap between blocking (\emph{i.e.}, progressive) and 
non-blocking (\emph{i.e.}, obstruction-free) TM implementations. We use the progressive opaque TM implementation
$LP$ described in Algorithm~\ref{alg:ic} (Chapter~\ref{ch:p3c2}) to establish a linear separation in memory stall
and RAW/AWAR complexity between blocking and non-blocking TMs.

Formally, let $\mathcal{OF}$ denote the class of TMs that provide OF TM-progress and OF TM-liveness.

\vspace{1mm}\noindent\textbf{Roadmap of Chapter~\ref{ch:p3c3}.}
In Section~\ref{sec:p3c3ss1}, we show that no strictly serializable TM in $\mathcal{OF}$ can be
weak DAP and have invisible reads.
In Section~\ref{sec:p3c3ss2}, we determine stall complexity bounds for
strictly serializable TMs in $\mathcal{OF}$, and in Section~\ref{sec:p3c3ss3}, we
present a linear (in $n$) lower bound on the RAW/AWAR complexity for RW DAP opaque TMs in
$\mathcal{OF}$. 
In Section~\ref{sec:oftmalgos}, we describe two obstruction-free algorithms: 
a RW DAP opaque TM and a weak DAP (but not RW DAP) opaque TM.
In Section~\ref{sec:p3c3s3}, we present complexity gaps between blocking and non-blocking TM implementations.
We conclude this chapter with a discussion on related work and open questions concerning obstruction-free TMs.
%
\section{Impossibility of weak DAP and invisible reads}
\label{sec:p3c3ss1}
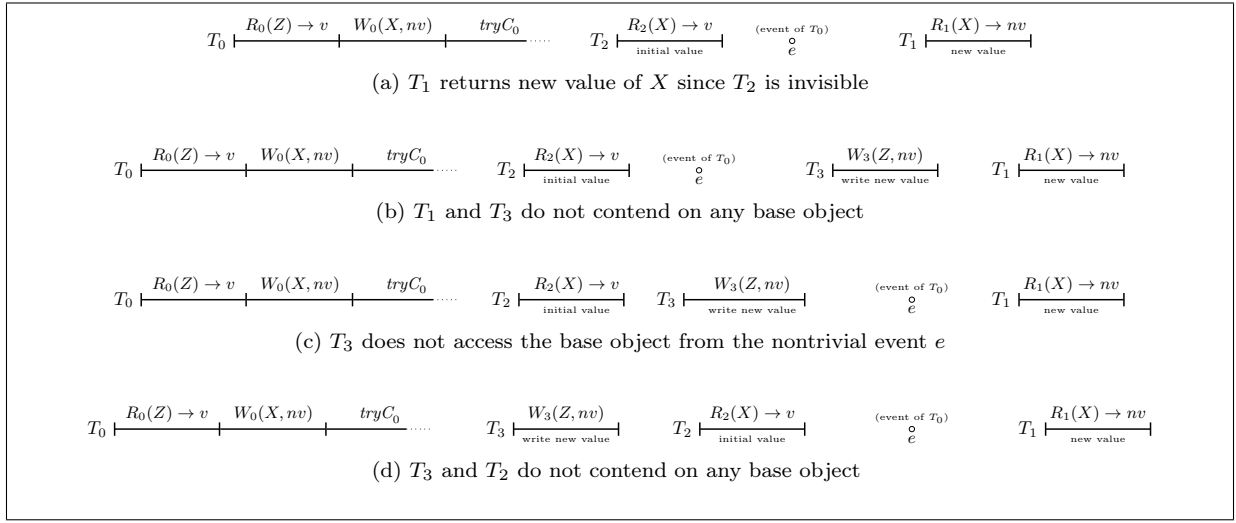
\begin{figure*}[t]
\begin{center}
	\subfloat[$T_1$ returns new value of $X$ since $T_2$ is invisible\label{sfig:ofinv-1}]{\scalebox{0.7}[0.7]{\begin{tikzpicture}
\node (r1) at (1,0) [] {};
\node (w1) at (3,0) [] {};
\node (c1) at (5,0) [] {};

\node (r2) at (8.2,0) [] {};

\node (e) at (10.5,0) [] {};

\node (r3) at (14,0) [] {};

\draw (r1) node [above] {\small {$R_0(Z)\rightarrow v$}};
\draw (w1) node [above] {\small {$W_0(X,nv)$}};
\draw (c1) node [above] {\small {$\TryC_0$}};

\draw (r2) node [above] {\small {$R_2(X)\rightarrow v$}};
\draw (r2) node [below] {\tiny {initial value}};

\draw (e) node [above] {\tiny {(event of $T_0$)}};
\draw (e) node [below] {\small {$e$}};

\draw (r3) node [above] {\small {$R_1(X)\rightarrow nv$}};
\draw (r3) node [below] {\tiny {new value}};

\begin{scope}   
\draw [|-|,thick] (0,0) node[left] {$T_0$} to (2,0);
\draw [-|,thick] (2,0) node[left] {} to (4,0);
\draw [-,thick] (4,0) node[left] {} to (5.5,0);
\draw [-,dotted] (5.5,0) node[left] {} to (6,0);
\draw [|-|,thick] (7.2,0) node[left] {$T_2$} to (9.2,0);
\draw  (10.5,0) circle [fill, radius=0.05]  (10.5,0);
\draw [|-|,thick] (13,0) node[left] {$T_1$} to (15,0);
\end{scope}
\end{tikzpicture}}}
        \\
        \vspace{2mm}
	\subfloat[$T_1$ and $T_3$ do not contend on any base object\label{sfig:ofinv-2}]{\scalebox{0.7}[0.7]{\begin{tikzpicture}
\node (r1) at (1,0) [] {};
\node (w1) at (3,0) [] {};
\node (c1) at (5,0) [] {};

\node (r2) at (8.2,0) [] {};

\node (e) at (10.5,0) [] {};

\node (r3) at (17.5,0) [] {};

\node (w3) at (14,0) [] {};

\draw (r1) node [above] {\small {$R_0(Z)\rightarrow v$}};
\draw (w1) node [above] {\small {$W_0(X,nv)$}};
\draw (c1) node [above] {\small {$\TryC_0$}};

\draw (r2) node [above] {\small {$R_2(X)\rightarrow v$}};
\draw (r2) node [below] {\tiny {initial value}};

\draw (e) node [above] {\tiny {(event of $T_0$)}};
\draw (e) node [below] {\small {$e$}};

\draw (r3) node [above] {\small {$R_1(X)\rightarrow nv$}};
\draw (r3) node [below] {\tiny {new value}};

\draw (w3) node [above] {\small {$W_3(Z,nv)$}};
\draw (w3) node [below] {\tiny {write new value}};

\begin{scope}   
\draw [|-|,thick] (0,0) node[left] {$T_0$} to (2,0);
\draw [-|,thick] (2,0) node[left] {} to (4,0);
\draw [-,thick] (4,0) node[left] {} to (5.5,0);
\draw [-,dotted] (5.5,0) node[left] {} to (6,0);
\draw [|-|,thick] (7.2,0) node[left] {$T_2$} to (9.2,0);
\draw  (10.5,0) circle [fill, radius=0.05]  (10.5,0);
\draw [|-|,thick] (13,0) node[left] {$T_3$} to (15,0);
\draw [|-|,thick] (16.5,0) node[left] {$T_1$} to (18.5,0);
\end{scope}
\end{tikzpicture}}}
	\\
	\vspace{2mm}
	\subfloat[$T_3$ does not access the base object from the nontrivial event $e$\label{sfig:ofinv-3}]{\scalebox{0.7}[0.7]{\begin{tikzpicture}
\node (r1) at (1,0) [] {};
\node (w1) at (3,0) [] {};
\node (c1) at (5,0) [] {};

\node (r2) at (8.2,0) [] {};

\node (w3) at (11.5,0) [] {};

\node (r3) at (17.5,0) [] {};

\node (e) at (14.5,0) [] {};

\draw (r1) node [above] {\small {$R_0(Z)\rightarrow v$}};
\draw (w1) node [above] {\small {$W_0(X,nv)$}};
\draw (c1) node [above] {\small {$\TryC_0$}};

\draw (r2) node [above] {\small {$R_2(X)\rightarrow v$}};
\draw (r2) node [below] {\tiny {initial value}};

\draw (e) node [above] {\tiny {(event of $T_0$)}};
\draw (e) node [below] {\small {$e$}};

\draw (r3) node [above] {\small {$R_1(X)\rightarrow nv$}};
\draw (r3) node [below] {\tiny {new value}};

\draw (w3) node [above] {\small {$W_3(Z,nv)$}};
\draw (w3) node [below] {\tiny {write new value}};

\begin{scope}   
\draw [|-|,thick] (0,0) node[left] {$T_0$} to (2,0);
\draw [-|,thick] (2,0) node[left] {} to (4,0);
\draw [-,thick] (4,0) node[left] {} to (5.5,0);
\draw [-,dotted] (5.5,0) node[left] {} to (6,0);
\draw [|-|,thick] (7.1,0) node[left] {$T_2$} to (9.1,0);
\draw  (14.5,0) circle [fill, radius=0.05]  (14.5,0);
\draw [|-|,thick] (10.2,0) node[left] {$T_3$} to (12.5,0);
\draw [|-|,thick] (16.5,0) node[left] {$T_1$} to (18.5,0);
\end{scope}
\end{tikzpicture}}}
	\\
	\vspace{2mm}
	\subfloat[$T_3$ and $T_2$ do not contend on any base object \label{sfig:ofinv-4}]{\scalebox{0.7}[0.7]{\begin{tikzpicture}
\node (r1) at (1,0) [] {};
\node (w1) at (3,0) [] {};
\node (c1) at (5,0) [] {};

\node (r2) at (12,0) [] {};

\node (e) at (15,0) [] {};

\node (r3) at (18.5,0) [] {};

\node (w3) at (8.5,0) [] {};

\draw (r1) node [above] {\small {$R_0(Z)\rightarrow v$}};
\draw (w1) node [above] {\small {$W_0(X,nv)$}};
\draw (c1) node [above] {\small {$\TryC_0$}};

\draw (r2) node [above] {\small {$R_2(X)\rightarrow v$}};
\draw (r2) node [below] {\tiny {initial value}};

\draw (e) node [above] {\tiny {(event of $T_0$)}};
\draw (e) node [below] {\small {$e$}};

\draw (r3) node [above] {\small {$R_1(X)\rightarrow nv$}};
\draw (r3) node [below] {\tiny {new value}};

\draw (w3) node [above] {\small {$W_3(Z,nv)$}};
\draw (w3) node [below] {\tiny {write new value}};

\begin{scope}   
\draw [|-|,thick] (0,0) node[left] {$T_0$} to (2,0);
\draw [-|,thick] (2,0) node[left] {} to (4,0);
\draw [-,thick] (4,0) node[left] {} to (5.5,0);
\draw [-,dotted] (5.5,0) node[left] {} to (6,0);
\draw [|-|,thick] (7.5,0) node[left] {$T_3$} to (9.5,0);
\draw [|-|,thick] (11,0) node[left] {$T_2$} to (13,0);
\draw  (15,0) circle [fill, radius=0.05]  (15,0);
\draw [|-|,thick] (17.5,0) node[left] {$T_1$} to (19.5,0);
\end{scope}
\end{tikzpicture}}}
	\caption{Executions in the proof of Theorem~\ref{th:ir}; execution in \ref{sfig:ofinv-4} is not strictly serializable
          \label{fig:ofindis}} 
\end{center}
\end{figure*}
In this section, we prove that it is impossible to combine weak DAP and invisible reads for strictly serializable TMs
in $\mathcal{OF}$.

Here is a proof sketch:
suppose, by contradiction, that such a TM implementation $M$ exists.
Consider an execution $E$ of $M$ 
in which a transaction $T_0$ performs a t-read of t-object $Z$
(returning the initial value $v$), 
writes $nv$ (new value) to t-object $X$, and commits.
Let $E'$ denote the longest prefix of $E$ that cannot be extended
with
the t-complete step contention-free execution of any transaction that 
reads $nv$ in $X$ and commits.

Thus if $T_0$ takes one more step, then 
the resulting execution $E'\cdot e$ can be extended with the t-complete step contention-free
execution of a transaction $T_1$ that reads $nv$ in $X$ and commits.

Since $M$ uses invisible reads, the following execution exists: $E'$ can be extended with the t-complete step contention-free
execution of a transaction $T_2$ that reads the initial value $v$
in $X$ and commits, followed by the step $e$ of $T_0$
after which transaction $T_1$ running step contention-free reads $nv$
in $X$ and commits.
Moreover, this execution is indistinguishable to $T_1$ and $T_2$ from an execution in which 
the read set of $T_0$ is empty. Thus, 
we can modify this execution by inserting the step contention-free
execution of a committed transaction
$T_3$ that writes a new value to $Z$ after $E'$, but preceding $T_2$ in real-time order.
Intuitively, by weak DAP,  transactions $T_1$ and $T_2$ cannot distinguish this
execution from the original one in which $T_3$ does not participate.

Thus, we can show that the following execution exists: $E'$ is extended with the t-complete step contention-free
execution of $T_3$ that writes $nv$ to $Z$ and commits, followed by the t-complete step contention-free
execution of $T_2$ that reads the initial value $v$ in $X$ and commits,
followed by the step $e$ of $T_0$, 
after which $T_1$ reads $nv$ in $X$ and commits.

This execution is, however, not strictly serializable:
$T_0$ must appear in any serialization ($T_1$ reads a value
written by $T_0$). 
Transaction $T_2$ must precede $T_0$, since the t-read of $X$ by $T_2$ returns the initial value of $X$. 
To respect real-time order, $T_3$ must precede $T_2$.
Finally, $T_0$ must precede $T_3$ since the t-read of $Z$ returns the initial value of $Z$.
The cycle $T_0\rightarrow T_3 \rightarrow T_2 \rightarrow T_0$ implies a contradiction.

The formal proof follows.
\begin{theorem}
\label{th:ir}
There does not exist a weak DAP strictly serializable TM
implementation in $\mathcal{OF}$ that uses invisible reads.
\end{theorem}
\begin{proof}
By contradiction, assume that such an implementation $M\in\mathcal{OF}$ exists.
Let $v$ be the initial value of t-objects $X$ and $Z$.
Consider an execution $E$ of $M$ 
in which 
a transaction $T_0$ performs $\Read_0(Z) \rightarrow v$ (returning $v$), 
writes $nv\neq v$ to $X$, and commits.
Let $E'$ denote the longest prefix of $E$ that cannot be extended with
the t-complete step contention-free execution of
any transaction performing a t-read $X$ that returns $nv$ and commits.

Let $e$ be the enabled event of transaction $T_0$ in the configuration after $E'$.
Without loss of generality, assume that $E'\cdot e$ can be extended with the t-complete step contention-free
execution of a committed transaction $T_1$ that reads $X$ and returns $nv$.
Let $E' \cdot e \cdot E_1$ be such an execution, where 
$E_1$ is the t-complete step contention-free execution fragment of transaction $T_1$ that 
performs $\Read_1(X) \rightarrow nv$ and commits.

We now prove that $M$ has an execution of the form $E' \cdot E_2 \cdot e \cdot E_1$, where
$E_2$ is the t-complete step contention-free execution fragment of transaction $T_2$ that 
performs $\Read_2(X) \rightarrow v$ and commits.

We observe that $E'\cdot E_2$ is an execution of $M$.
Indeed, by OF TM-progress and OF TM-liveness, $T_2$ must return a matching response that is not $A_2$ in $E'\cdot E_2$,
and by the definition of $E'$, this response must be the initial value $v$ of $X$.

By the assumption of invisible reads,
$E_2$ does not contain any nontrivial events.
Consequently, $E'\cdot E_2 \cdot e \cdot E_1$ is indistinguishable to transaction $T_1$
from the execution $E'\cdot e  \cdot E_1$. 
Thus, $E'\cdot E_2 \cdot e \cdot E_1$ is also an execution of $M$ (Figure~\ref{sfig:ofinv-1}).
\begin{claim}
\label{cl:invdapof2}
$M$ has an execution of the form $E'\cdot E_2 \cdot  E_3  \cdot e \cdot E_1$ where
$E_3$ is the t-complete step contention-free execution fragment of transaction $T_{3}$ that 
writes $nv \neq v$ to $Z$ and commits.
\end{claim}
\begin{proof}
The proof is through a sequence of indistinguishability arguments to construct the execution.

We first claim that $M$ has an execution of the form $E' \cdot E_2 \cdot e \cdot E_1 \cdot E_3$.
Indeed, by OF TM-progress and OF TM-liveness, $T_3$ must be committed in $E' \cdot E_2 \cdot e \cdot E_1 \cdot E_3$.

Since $M$ uses invisible reads,
the execution $E' \cdot E_2 \cdot e \cdot E_1 \cdot E_3$ is indistinguishable
to transactions $T_1$ and $T_3$ from the execution ${\hat E}\cdot E_1 \cdot E_3$, where ${\hat E}$ is
the t-incomplete step contention-free execution of transaction $T_0$ with
$\Wset_{\hat E}(T_0)=\{X\}$; $\Rset_{\hat E}(T_0)=\emptyset$ that writes $nv$ to $X$.

Observe that the execution $E'\cdot E_2 \cdot e \cdot E_1 \cdot E_3$ is indistinguishable
to transactions $T_1$ and $T_3$ from the execution ${\hat E} \cdot E_1 \cdot E_3$, in which
transactions $T_3$ and $T_1$ are disjoint-access. Consequently, by Lemma~\ref{lm:dap}, $T_1$ and $T_3$ 
do not contend on any base object in ${\hat E}  \cdot E_1 \cdot E_3$.
Thus, $M$ has an execution of the form $E' \cdot E_2 \cdot e \cdot E_3 \cdot E_1$ (Figure~\ref{sfig:ofinv-2}).

By definition of $E'$, $T_0$ applies a nontrivial primitive to some base object, say $b$, in event $e$ that
$T_1$ must access in $E_1$.
Thus, the execution fragment $E_3$ does not contain any nontrivial event on $b$ in
the execution $E'\cdot E_2 \cdot e \cdot E_1 \cdot E_3$.
In fact, since $T_3$ is disjoint-access with $T_0$ in the execution ${\hat E} \cdot  E_3 \cdot E_1$,
by Lemma~\ref{lm:dap}, it cannot access the base object $b$ to which $T_0$ applies a nontrivial primitive
in the event $e$. Thus, transaction
$T_3$ must perform the same sequence of events $E_3$ immediately after $E'$, implying that
$M$ has an execution of the form 
$E'\cdot E_2 \cdot  E_3 \cdot e \cdot E_1$ (Figure~\ref{sfig:ofinv-3}).
\end{proof} 
Finally, we observe that the execution $E' \cdot E_2 \cdot E_3 \cdot e \cdot E_1$ established in Claim~\ref{cl:invdapof2}
is indistinguishable
to transactions $T_2$ and $T_3$ from an execution ${\tilde E}\cdot E_2 \cdot E_3 \cdot e \cdot E_1$, 
where $\Wset(T_0)=\{X\}$ and $\Rset(T_0)=\emptyset$ in $\tilde E$.
But transactions $T_3$ and $T_2$ are disjoint-access in ${\tilde E} \cdot  E_2 \cdot E_3 \cdot e \cdot E_1$
and by Lemma~\ref{lm:dap}, $T_2$ and $T_3$ do not contend on any base object in this execution.
Thus, $M$ has an execution of the form $E' \cdot E_3 \cdot E_2 \cdot e \cdot E_1$ (Figure~\ref{sfig:ofinv-4})
in which $T_3$ precedes $T_2$ in real-time order.

However, the execution $E' \cdot E_3 \cdot E_2 \cdot e \cdot E_1$ is not strictly serializable:
$T_0$ must be committed in any serialization and transaction $T_2$
must precede $T_0$ since $\Read_2(X)$ returns the initial value of $X$. 
To respect real-time order, $T_3$ must precede $T_2$, while $T_0$ must precede $T_1$ since 
$\Read_1(X)$ returns $nv$, the value of $X$ updated by $T_0$.
Finally, $T_0$ must precede $T_3$ since $\Read_0(Z)$ returns the initial value of $Z$.
But there exists no such serialization---contradiction.
\end{proof}
%
%
\section{A linear lower bound on memory stall complexity}
\label{sec:p3c3ss2}
We prove a linear (in $n$) lower bound for strictly serializable TM implementations in $\mathcal{OF}$ 
on the total number of \emph{memory stalls}
incurred by a single t-read operation.

Inductively, for each $k\leq n-1$, we construct a specific \emph{$k$-stall execution}~\cite{G05} in which some 
t-read operation by a process $p$ incurs $k$ stalls.
In the $k$-stall execution, $k$ processes are partitioned into disjoint subsets
$S_1,\ldots,S_{i}$. 
The execution can be represented as $\alpha\cdot \sigma_1\cdots\sigma_{i}$; $\alpha$ is $p$-free, where in each $\sigma_j$,
$j=1,\ldots,i$, $p$ first runs by itself, then each process in
$S_j$ applies a \emph{nontrivial} event on a base object $b_j$, and then $p$
applies an event on $b_j$.
Moreover, $p$ does not detect step contention in this execution and,
thus, must return a non-abort value in its t-read and commit in the
solo extension of it. 
Additionally, it is guaranteed that in any extension of $\alpha$ by
the processes other than $\{p\}\cup S_1\cup S_2\cup \ldots \cup  S_{i}$, 
no nontrivial primitive is applied on a base object accessed in $\sigma_1\cdots\sigma_{i}$.

Assuming that $k\leq n-2$, we introduce a not previously used process executing an updating
transaction immediately after $\alpha$, so that the subsequent t-read
operation executed by $p$ is ``perturbed'' (must return another value). 
This will help us to construct a $(k+k')$-stall execution
$\alpha\cdot\alpha'\cdot\sigma_1\cdots\sigma_{i}\cdot\sigma_{i+1}$, where $k'>0$.

The formal proof follows:
\begin{theorem}
\label{th:oftmstalls} 
Every strictly serializable TM implementation $M\in \mathcal{OF}$ has a $(n-1)$-stall execution $E$ for a t-read operation
performed in $E$.
\end{theorem}
\begin{proof}
We proceed by induction.
Observe that the empty
execution is a $0$-stall execution since it vacuously satisfies the invariants of Definition~\ref{def:stalls}. 

Let $v$ be the initial value of t-objects $X$ and $Z$.
Let $\alpha={\alpha}_1\cdots { \alpha}_{n-2}$ be a step contention-free execution of a strictly serializable TM implementation
$M\in \mathcal{OF}$, where for all $j\in \{1,\ldots, n-2\}$,
$\alpha_j$ is the longest prefix of the execution fragment ${\bar \alpha}_j$ that denotes the t-complete step-contention
free execution of committed transaction $T_j$ (invoked by process $p_j$)
that performs $\Read_j(Z)\rightarrow v$, writes value $nv \neq v$ to $X$ 
in the execution ${\alpha}_1\cdots { \alpha}_{j-1} \cdot {\bar \alpha}_j$
such that
\begin{itemize}
\item 
$\TryC_j()$ is incomplete in $\alpha_j$,
\item
$\alpha_1\cdots \alpha_{j}$ cannot be extended with the t-complete step contention-free execution fragment
of any transaction $T_{n-1}$ or $T_n$ that performs exactly one t-read of $X$ that returns $nv$ and commits.
\end{itemize}
Assume, inductively, that $\alpha \cdot \sigma_1\cdots \sigma_i$ is a $k$-stall execution
for $\Read_n(X)$ executed by process $p_n$, where $0\leq k \leq n-2$. 
By Definition~\ref{def:stalls}, there are distinct base objects $b_1,\ldots b_i$ 
accessed by disjoint sets of processes $S_1\ldots S_i$ in the execution fragment $\sigma_1\cdots \sigma_i$,
where $|S_1\cup \ldots \cup S_i|=k$ and 
$\sigma_1\cdots \sigma_i$ contains no events of processes not in $S_1 \cup \ldots  \cup S_i \cup \{p_n\}$.
We will prove that there exists a $(k+k')$-stall execution for $\Read_n(X)$, for some $k'\geq 1$.

By Lemma~\ref{lm:stalls}, 
$\alpha \cdot \sigma_1\cdots \sigma_i$ is indistinguishable to $T_n$ from a step contention-free execution.
Let $\sigma$ be the finite step contention-free execution fragment that extends $\alpha \cdot \sigma_1\cdots \sigma_i$
in which $T_n$ performs events by itself: completes $\Read_n(X)$ and returns a response. 
By OF TM-progress and OF TM-liveness, $\Read_n(X)$ and the subsequent $\TryC_k$ must each
return non-$A_n$ responses in $\alpha \cdot \sigma_1\cdots \sigma_i\cdot \sigma$.
By construction of $\alpha$ and strict serializability of $M$, $\Read_n(X)$
must return the response $v$ or $nv$ in this execution. 
We prove that there exists an execution fragment $\gamma$ performed 
by some process $p_{n-1} \not\in (\{p_n\}\cup S_1\cup \cdots \cup S_i)$
extending $\alpha$ that contains a nontrivial event on some
base object that must be accessed by $\Read_n(X)$ in $\sigma_1\cdots \sigma_i\cdot \sigma$.

Consider the case that $\Read_n(X)$
returns the response $nv$ in $\alpha \cdot \sigma_1\cdots \sigma_i\cdot \sigma$.
We define a step contention-free fragment $\gamma$ extending $\alpha$ that is the t-complete step contention-free 
execution of transaction $T_{n-1}$ executed by some process $p_{n-1} \not\in (\{p_n\}\cup S_1\cup \cdots \cup S_i)$ 
that performs $\Read_{n-1}(X)\rightarrow v$, writes 
$nv\neq v$ to $Z$ and commits.
By definition of $\alpha$, OF TM-progress and OF TM-liveness, $M$ has an execution of the form
$\alpha\cdot \gamma$.
We claim that the execution fragment $\gamma$ must contain a nontrivial event on some base object that
must be accessed by $\Read_n(X)$ in $\sigma_1\cdots \sigma_i\cdot \sigma$.
Suppose otherwise. 
Then, $\Read_n(X)$ must return the response $nv$ in $\sigma_1\cdots \sigma_i\cdot \sigma$.
But the execution $\alpha\cdot \sigma_1\cdots \sigma_i\cdot \sigma$ is not strictly serializable.
Since $\Read_n(X)\rightarrow nv$, there exists a transaction $T_q\in \ms{txns}(\alpha)$ that must be committed
and must precede $T_n$ in any serialization.
Transaction $T_{n-1}$ must precede $T_n$ in any serialization to respect the real-time order and
$T_{n-1}$ must precede $T_q$ in any serialization. Also, $T_q$ must precede $T_{n-1}$ in any serialization.
But there exists no such serialization.

Consider the case that $\Read_n(X)$
returns the response $v$ in $\alpha \cdot \sigma_1\cdots \sigma_i\cdot \sigma$.
In this case, we define the step contention-free fragment $\gamma$ extending $\alpha$ as the t-complete step contention-free 
execution of transaction $T_{n-1}$ executed by some process $p_{n-1} \not\in (\{p_n\}\cup S_1\cup \cdots \cup S_i)$ 
that writes $nv\neq v$ to $X$ and commits.
By definition of $\alpha$, OF TM-progress and OF TM-liveness, $M$ has an execution of the form
$\alpha\cdot \gamma$.
By strict serializability of $M$, the execution fragment $\gamma$ must contain a nontrivial event on some base object that
must be accessed by $\Read_n(X)$ in $\sigma_1\cdots \sigma_i\cdot \sigma$. Suppose otherwise.
Then, $\sigma_1\cdots \sigma_i \cdot \gamma \cdot \sigma$ is an execution of $M$
in which $\Read_n(X)\rightarrow v$. But this execution is not strictly
serializable: every transaction $T_q\in \ms{txns}(\alpha)$ must be aborted or must be preceded by $T_n$
in any serialization, but committed transaction $T_{n-1}$ must precede $T_n$ in any serialization to
respect the real-time ordering of transactions. But then $\Read_n(X)$ must return the new value $nv$ of $X$ that is
updated by $T_{n-1}$---contradiction.

Since, by Definition~\ref{def:stalls},
the execution fragment $\gamma$ executed by some process $p_{n-1} \not\in (\{p_n\}\cup S_1\cup \cdots \cup S_i)$
contains no nontrivial events to any base object accessed in $\sigma_1 \cdots \sigma_i$,
it must contain a nontrivial event to some base object $b_{i+1}\not\in \{b_1,\ldots , b_i\}$ that is
accessed by $T_n$ in the execution fragment $\sigma$.

Let $\mathcal{A}$ denote the set of all finite $(\{p_n\}\cup S_1 \ldots \cup S_i)$-free execution fragments that extend $\alpha$.
Let $b_{i+1} \not\in \{b_1,\ldots , b_i\}$ be the first base object accessed by $T_n$ in the execution fragment
$\sigma$ to which some transaction applies a nontrivial event in the 
execution fragment $\alpha'\in \mathcal{A}$.
Clearly, some such execution $\alpha \cdot \alpha'$ exists that contains a nontrivial event in $\alpha'$ to some
distinct base object $b_{i+1}$ not accessed in the execution fragment $\sigma_1 \cdots \sigma_i$.
We choose the execution $\alpha\cdot \alpha' \in \mathcal{A}$ that maximizes the number of transactions
that are poised to apply nontrivial events on $b_{i+1}$ in the configuration after $\alpha\cdot \alpha'$.
Let $S_{i+1}$ denote the set of processes executing these transactions and $k'=|S_{i+1}|$ ($k'>0$ as already proved).

We now construct a $(k+k')$-stall execution
$\alpha\cdot \alpha' \cdot \sigma_1 \cdots \sigma_i \cdot \sigma_{i+1}$ for $\Read_n(X)$,
where in $\sigma_{i+1}$, $p_n$ applies events by itself, then each of the processes in $S_{i+1}$ applies a nontrivial event
on $b_{i+1}$, and finally, $p_n$ accesses $b_{i+1}$.

By construction, $\alpha\cdot \alpha'$ is $p_n$-free.
Let $\sigma_{i+1}$ be the prefix of $\sigma$ not including $T_n$'s first access to $b_{i+1}$, concatenated with
the nontrivial events on $b_{i+1}$ by each of the $k'$ transactions executed by processes in $S_{i+1}$ 
followed by the access of $b_{i+1}$
by $T_n$. Observe that $T_n$ performs exactly one t-operation $\Read_n(X)$ in the execution fragment $\sigma_1\cdots \sigma_{i+1}$
and $\sigma_1\cdots \sigma_{i+1}$ contains no events of processes not in $(\{p_n\}\cup S_1\cup \cdots \cup S_i \cup S_{i+1})$.

To complete the induction, we need to show that in every $(\{p_n\}\cup S_1\cup \cdots \cup S_i \cup S_{i+1})$-free extension of 
$\alpha \cdot \alpha'$,
no transaction applies a nontrivial event to any base object accessed 
in the execution fragment $\sigma_1 \cdots \sigma_i \cdot \sigma_{i+1}$.
Let $\beta$ be any such execution fragment that extends $\alpha\cdot \alpha' $.
By our construction, $\sigma_{i+1}$ is the execution fragment that consists of events by $p_n$ on base objects accessed in
$\sigma_1\cdots \sigma_i$, nontrivial events on $b_{i+1}$ by transactions in $S_{i+1}$ and finally, an access to $b_{i+1}$
by $p_n$. 
Since $\alpha\cdot \sigma_1\cdots \sigma_i$ is a $k$-stall execution by our induction hypothesis,
$\alpha'\cdot \beta$ is $(\{p_n\}\cup S_1 \ldots \cup S_i\})$-free and thus, $\alpha'\cdot \beta$
does not contain nontrivial events on any base object accessed in $\sigma_1\cdots \sigma_i$.
We now claim that $\beta$ does not contain nontrivial events to $b_{i+1}$. Suppose otherwise.
Thus, there exists some transaction $T'$ that has an enabled nontrivial event to $b_{i+1}$ in the
configuration after $\alpha\cdot \alpha' \cdot \beta'$, where $\beta'$ is some prefix of $\beta$.
But this contradicts the choice of $\alpha \cdot \alpha'$ as the extension of $\alpha$ that maximizes $k'$.

Thus, $\alpha\cdot \alpha'\cdot \sigma_1 \cdots \sigma_i \cdot \sigma_{i+1}$ is indeed a $(k+k')$-stall execution for $T_n$
where $1< k< (k+k') \leq (n-1)$. 
%
\end{proof}
Since there are at most $n$ processes that are concurrent at any prefix of an execution, the lower bound of
Theorem~\ref{th:oftmstalls} is tight.
%
\section{A linear lower bound on expensive synchronization for RW DAP}
\label{sec:p3c3ss3}
We prove that opaque, RW DAP TM implementations in $\mathcal{OF}$ have executions in which
some read-only transaction performs a linear (in $n$) number of non-overlapping RAWs or AWARs.

Prior to presenting the formal proof, we present an overview (the executions used in the proof 
are depicted in Figure~\ref{fig:rw}).

We first construct an execution of the form ${\bar \rho}_1\cdots {\bar \rho}_m$, where for all $j\in \{1,\ldots, m\}$; $m=n-3$, 
${\bar \rho}_j$ 
denotes the t-complete step contention-free execution of transaction $T_j$
that reads the initial value $v$ in a distinct t-object $Z_j$, writes
a new value $nv$ to a distinct t-object $X_j$ and commits.
Observe that since any two transactions that participate in this execution
are mutually read-write disjoint-access, they cannot contend on the
same base object and, thus, the execution appears solo to each of them.

Let each of two new transactions $T_{n-1}$ and $T_n$ perform $m$
t-reads on objects  $X_1,\ldots,X_m$.  
For $j\in \{1,\ldots, m\}$,
we now define $\rho_j$ to be  
the longest prefix of ${\bar \rho}_j$ such that 
$\rho_1\cdots \rho_{j}$ cannot be extended the complete step
contention-free execution fragment of $T_{n-1}$ or $T_n$ where the
t-read of $X_j$ returns $nv$ (Figure~\ref{sfig:rw-1}).
Let $e_j$ be the event by $T_j$ enabled after $\rho_1\cdots \rho_j$.
Let us count the number of indices $j\in \{1,\ldots, m\}$ such that 
$T_{n-1}$ (resp., $T_n$) reads the new value
$nv$ in $X_j$ when it runs after $\rho_1\cdots \rho_j\cdot e_j$. 
Without loss of generality, assume that $T_{n-1}$ has more such
indices $j$ than $T_{n}$.
We are going to show that, in the worst-case, 
$T_n$ must perform $\lceil \frac{m}{2} \rceil$ non-overlapping RAW/AWARs
in the course of performing $m$ t-reads of $X_1,\ldots , X_m$ immediately after $\rho_1\cdots \rho_m$.

Consider any $j\in \{1,\ldots, m\}$ such that
$T_{n-1}$, when it runs step contention-free after $\rho_1\cdots \rho_{j}\cdot e_j$, reads
$nv$ in $X_j$. 
We claim that, in $\rho_1\cdots \rho_{m}$ extended with the step contention-free execution of 
$T_n$ performing $j$ t-reads $\Read_n(X_1)\cdots \Read_n(X_j)$,
the t-read of $X_j$ must contain a RAW or an AWAR. 

Suppose not. Then we are going to schedule a specific execution of $T_j$ and $T_{n-1}$
concurrently with $\Read_n(X_j)$ so that $T_n$ cannot detect the concurrency.
By the definition of $\rho_j$ and the fact that the TM is RW DAP, $T_n$, when it runs
step contention-free after $\rho_1\cdots \rho_{m}$, must read $v$ (the initial
value) in $X_j$ (Figure~\ref{sfig:rw-2}).
Then the following execution exists: $\rho_1\cdots \rho_{m}$ is extended with the
t-complete step contention-free execution of $T_{n-2}$ writing $nv$ to
$Z_j$ and committing, after which $T_n$ runs step contention-free and
reads $v$ in $X_j$ (Figure~\ref{sfig:rw-3}).
%
Since, by the assumption, $\Read_n(X_j)$ contains no RAWs or AWARs, 
we show that we can run $T_{n-1}$ performing $j$ t-reads concurrently with the execution of $\Read_n(X_j)$ so that
$T_n$ and $T_{n-1}$ are unaware of step contention and
$\Read_{n-1}(X_j)$ still reads the value $nv$ in $X_j$.

To understand why this is possible, consider the following:
we take the execution depicted in Figure~\ref{sfig:rw-3}, but without the execution of $\Read_n(X_j)$, \emph{i.e},
$\rho_1\cdots \rho_{m}$ is extended with the
step contention-free execution of committed transaction $T_{n-2}$ writing $nv$ to
$Z_j$, after which $T_n$ runs step contention-free performing $j-1$ t-reads.
This execution can be extended with the step $e_j$ by $T_j$, followed by the step contention-free execution
of transaction $T_{n-1}$ in which it reads $nv$ in $X_j$.
Indeed, by RW DAP and the definition of $\rho_j\cdot e_j$, there exists such an execution (Figure~\ref{sfig:rw-5}).

Since $\Read_n(X_j)$ contains no RAWs or AWARs,
we can reschedule the execution fragment $e_j$  followed by the execution of $T_{n-1}$
so that it is concurrent with the execution of $\Read_n(X_j)$
and neither $T_n$ nor $T_{n-1}$ see a difference (Figure~\ref{sfig:rw-4}).
Therefore, in this execution, $\Read_n(X_j)$ still returns $v$, while $\Read_{n-1}(X_j)$
returns $nv$.

However, the resulting execution (Figure~\ref{sfig:rw-4})
is not opaque.
In any serialization the following must hold. 
Since $T_{n-1}$ reads the value written by $T_j$ in $X_j$,
$T_j$ must be committed.  
Since $\Read_{n}(X_j)$ returns the initial value $v$, 
$T_{n}$ must precede $T_j$.
The committed transaction $T_{n-2}$, which writes a new value to $Z_j$, must precede 
$T_{n}$ to respect the real-time order on transactions. 
However, $T_j$ must precede $T_{n-2}$ since $\Read_j(Z_j)$ returns the
initial value and the implementation is opaque.
The cycle $T_j\rightarrow T_{n-2} \rightarrow T_{n} \rightarrow T_j$
implies a contradiction.

Thus, we can show that transaction $T_n$ must perform $\Omega(n)$ RAW/AWARs during the execution of $m$ t-reads
immediately after $\rho_1\cdots \rho_m$.
\begin{figure*}[t]
\begin{center}
	\subfloat[Transactions in $\{T_1,\ldots , T_m\}$;$m=n-3$ are mutually read-write disjoint-access and concurrent; they are
	poised to apply a nontrivial primitive\label{sfig:rw-1}]{\scalebox{0.6}[0.6]{\begin{tikzpicture}
\node (r1) at (1,0) [] {};
\node (w1) at (3,0) [] {};
\node (c1) at (5,0) [] {};

\draw (r1) node [above] {\small {$R_1(Z_1)\rightarrow v$}};
\draw (w1) node [above] {\small {$W_1(X_1,nv)$}};
\draw (c1) node [above] {\small {$\TryC_1$}};


\node (rj) at (8,0) [] {};
\node (wj) at (10.2,0) [] {};
\node (cj) at (12,0) [] {};

\draw (rj) node [above] {\small {$R_m(Z_m)\rightarrow v$}};
\draw (wj) node [above] {\small {$W_m(X_m,nv)$}};
\draw (cj) node [above] {\small {$\TryC_m$}};


\begin{scope}   
\draw [|-,thick] (0,0) node[left] {$T_1$} to (5,0);
\draw [|-|,thick] (0,0) node[left] {} to (2,0);
\draw [-|,thick] (2,0) node[left] {} to (4,0);
\draw [-,dotted] (5,0) node[left] {} to (6,0);
\end{scope}
\begin{scope}   
\draw [|-|,thick] (7,0) node[left] {$T_m$} to (9,0);
\draw [-|,thick] (9,0) node[left] {} to (11,0);
\draw [-,thick] (11,0) node[left] {} to (12,0);
\draw [-,dotted] (11,0) node[left] {} to (13,0);
\end{scope}

\end{tikzpicture}}}
        \\
        \vspace{3mm}
        \subfloat[$T_{n}$ performs $m$ reads; each $\Read_{n}(X_j)$ returns initial value $v$\label{sfig:rw-2}]{\scalebox{0.6}[0.6]{\begin{tikzpicture}
\node (r1) at (1,0) [] {};
\node (w1) at (3,0) [] {};
\node (c1) at (5,0) [] {};

\draw (r1) node [above] {\small {$R_1(Z_1)\rightarrow v$}};
\draw (w1) node [above] {\small {$W_1(X_1,nv)$}};
\draw (c1) node [above] {\small {$\TryC_1$}};


\node (rj) at (8,0) [] {};
\node (wj) at (10.2,0) [] {};
\node (cj) at (12,0) [] {};

\draw (rj) node [above] {\small {$R_m(Z_m)\rightarrow v$}};
\draw (wj) node [above] {\small {$W_m(X_m,nv)$}};
\draw (cj) node [above] {\small {$\TryC_m$}};


\node (p1) at (15,0) [] {};
\node (pj) at (19,0) [] {};

\draw (p1) node [above] {\small {$R_{n}(X_1)\rightarrow v$}};
\draw (pj) node [above] {\small {$R_{n}(X_j)\rightarrow v$}};

\begin{scope}   
\draw [|-,thick] (0,0) node[left] {$T_1$} to (5,0);
\draw [|-|,thick] (0,0) node[left] {} to (2,0);
\draw [-|,thick] (2,0) node[left] {} to (4,0);
\draw [-,dotted] (5,0) node[left] {} to (6,0);
\end{scope}
\begin{scope}   
\draw [|-|,thick] (7,0) node[left] {$T_m$} to (9,0);
\draw [-|,thick] (9,0) node[left] {} to (11,0);
\draw [-,thick] (11,0) node[left] {} to (12,0);
\draw [-,dotted] (11,0) node[left] {} to (13,0);
\end{scope}
\begin{scope}   
\draw [|-,dotted] (14,0) node[left] {$T_{n}$} to (21,0);
\draw [|-|,thick] (14,0) to (16,0);
\draw [|-|,thick] (18,0) to (20,0);
\end{scope}
\end{tikzpicture}}}
        \\
        \vspace{3mm}
        \subfloat[$T_{n-2}$ commits; $T_{n}$ is read-write disjoint-access with $T_{n-2}$\label{sfig:rw-3}]{\scalebox{0.6}[0.6]{\begin{tikzpicture}
\node (r1) at (1,0) [] {};
\node (w1) at (3,0) [] {};
\node (c1) at (5,0) [] {};

\draw (r1) node [above] {\small {$R_1(Z_1)\rightarrow v$}};
\draw (w1) node [above] {\small {$W_1(X_1,nv)$}};
\draw (c1) node [above] {\small {$\TryC_1$}};


\node (rj) at (8,0) [] {};
\node (wj) at (10.2,0) [] {};
\node (cj) at (12,0) [] {};

\draw (rj) node [above] {\small {$R_m(Z_m)\rightarrow v$}};
\draw (wj) node [above] {\small {$W_m(X_m,nv)$}};
\draw (cj) node [above] {\small {$\TryC_m$}};


\node (p1) at (18,0) [] {};
\node (pj) at (22,0) [] {};

\draw (p1) node [above] {\small {$R_{n}(X_1)\rightarrow v$}};
\draw (pj) node [above] {\small {$R_{n}(X_j)\rightarrow v$}};
\node (z1) at (15,0) [] {};

\draw (z1) node [above] {\small {$W_{n-2}(Z_j,nv)$}};

\begin{scope}   
\draw [|-,thick] (0,0) node[left] {$T_1$} to (5,0);
\draw [|-|,thick] (0,0) node[left] {} to (2,0);
\draw [-|,thick] (2,0) node[left] {} to (4,0);
\draw [-,dotted] (5,0) node[left] {} to (6,0);
\end{scope}
\begin{scope}   
\draw [|-|,thick] (7,0) node[left] {$T_m$} to (9,0);
\draw [-|,thick] (9,0) node[left] {} to (11,0);
\draw [-,thick] (11,0) node[left] {} to (12,0);
\draw [-,dotted] (11,0) node[left] {} to (13,0);
\end{scope}
\begin{scope}   
\draw [|-,dotted] (17,0) node[left] {$T_{n}$} to (23.5,0);
\draw [|-|,thick] (17,0) to (19,0);
\draw [|-|,thick] (21,0) to (23,0);
\end{scope}
\begin{scope}   
\draw [|-|,thick] (14,0) node[left] {$T_{n-2}$} to (16,0);
\end{scope}
%

\end{tikzpicture}}}
        \\
        \vspace{3mm}
        \subfloat[$T_{n-1}$ is read-write disjoint-access with $T_{n-2}$; $\Read_{n-1}(X_j)$ returns the value $nv$\label{sfig:rw-5}]{\scalebox{0.5}[0.5]{\begin{tikzpicture}
\node (r1) at (1,0) [] {};
\node (w1) at (3,0) [] {};
\node (c1) at (5,0) [] {};

\draw (r1) node [above] {\small {$R_1(Z_1)\rightarrow v$}};
\draw (w1) node [above] {\small {$W_1(X_1,nv)$}};
\draw (c1) node [above] {\small {$\TryC_1$}};


\node (rj) at (8,0) [] {};
\node (wj) at (10.2,0) [] {};
\node (cj) at (12,0) [] {};

\draw (rj) node [above] {\small {$R_m(Z_m)\rightarrow v$}};
\draw (wj) node [above] {\small {$W_m(X_m,nv)$}};
\draw (cj) node [above] {\small {$\TryC_m$}};


%
\node (z1) at (15,0) [] {};
\node (p1) at (19,0) [] {};

\draw (z1) node [above] {\small {$W_{n-2}(Z_j,nv)$}};
\node (e) at (21.1+.5,-1) [] {};
\node (l1) at (24+.8,-1) [] {};
\node (lj) at (28+.8,-1) [] {};

\draw (e) node [above] {\large {(event of $T_j$)}};
\draw (e) node [below] {\small {$e_j$}};
\draw (l1) node [above] {\large {$R_{n-1}(X_1)$}};
\draw (lj) node [above] {\large {$R_{n-1}(X_j)\rightarrow nv$}};

\draw (p1) node [above] {\large {$R_{n}(X_1)\cdots R_n(X_{j-1})$}};

\begin{scope}   
\draw [|-,thick] (0,0) node[left] {$T_1$} to (5,0);
\draw [|-|,thick] (0,0) node[left] {} to (2,0);
\draw [-|,thick] (2,0) node[left] {} to (4,0);
\draw [-,dotted] (5,0) node[left] {} to (6,0);
\end{scope}
\begin{scope}   
\draw [|-|,thick] (7,0) node[left] {$T_m$} to (9,0);
\draw [-|,thick] (9,0) node[left] {} to (11,0);
\draw [-,thick] (11,0) node[left] {} to (12,0);
\draw [-,dotted] (11,0) node[left] {} to (13,0);
\end{scope}
\begin{scope}   
\draw [|-,dotted] (17,0) node[left] {$T_{n}$} to (30,0);
\draw [|-|,thick] (17,0) to (21,0);
\end{scope}
\begin{scope}   
\draw [|-|,thick] (14,0) node[left] {$T_{n-2}$} to (16,0);
\end{scope}
\begin{scope}
\draw  (21.2+.5,-1) circle [fill, radius=0.05]   (21.2+.5,-1);
\draw [-,dotted] (23.2+.8,-1) node[left] {\large $T_{n-1}$} to (29+.8,-1);
\draw [|-|,thick] (23.2+.8,-1) node[left] {} to (25.1+.8,-1);
\draw [|-|,thick] (27+.8,-1) node[left] {} to (29+.8,-1);
\end{scope}

\end{tikzpicture}}}
	\\
        \vspace{3mm}
        \subfloat[Suppose $\Read_{n}(X_j)$ does not perform a RAW/AWAR, 
        $T_{n}$ and $T_{n-1}$ are unaware of step contention and $T_n$ misses the event of $T_j$, but $R_{n-1}(X_j)$ returns the value of $X_j$ 
        that is updated by $T_j$\label{sfig:rw-4}]{\scalebox{0.5}[0.5]{\begin{tikzpicture}
\node (r1) at (1,0) [] {};
\node (w1) at (3,0) [] {};
\node (c1) at (5,0) [] {};

\draw (r1) node [above] {\small {$R_1(Z_1)\rightarrow v$}};
\draw (w1) node [above] {\small {$W_1(X_1,nv)$}};
\draw (c1) node [above] {\small {$\TryC_1$}};


\node (rj) at (8,0) [] {};
\node (wj) at (10.2,0) [] {};
\node (cj) at (12,0) [] {};

\draw (rj) node [above] {\small {$R_m(Z_m)\rightarrow v$}};
\draw (wj) node [above] {\small {$W_m(X_m,nv)$}};
\draw (cj) node [above] {\small {$\TryC_m$}};


\node (p1) at (19,0) [] {};
\node (pj) at (25,0) [] {};

\draw (p1) node [above] {\small {$R_{n}(X_1)\cdots R_n(X_{j-1})$}};
\draw (pj) node [above] {\large {$R_{n}(X_j)\rightarrow v$}};
\node (z1) at (15,0) [] {};

\draw (z1) node [above] {\small {$W_{n-2}(Z_j,nv)$}};
\node (e) at (21.1+.5,-1) [] {};
\node (l1) at (24+.8,-1) [] {};
\node (lj) at (28+.8,-1) [] {};

\draw (e) node [above] {\large {(event of $T_j$)}};
\draw (e) node [below] {\small {$e_j$}};
\draw (l1) node [above] {\large {$R_{n-1}(X_1)$}};
\draw (lj) node [above] {\large {$R_{n-1}(X_j)\rightarrow nv$}};

\begin{scope}   
\draw [|-,thick] (0,0) node[left] {$T_1$} to (5,0);
\draw [|-|,thick] (0,0) node[left] {} to (2,0);
\draw [-|,thick] (2,0) node[left] {} to (4,0);
\draw [-,dotted] (5,0) node[left] {} to (6,0);
\end{scope}
\begin{scope}   
\draw [|-|,thick] (7,0) node[left] {$T_m$} to (9,0);
\draw [-|,thick] (9,0) node[left] {} to (11,0);
\draw [-,thick] (11,0) node[left] {} to (12,0);
\draw [-,dotted] (11,0) node[left] {} to (13,0);
\end{scope}
\begin{scope}   
\draw [|-,dotted] (17,0) node[left] {$T_{n}$} to (30,0);
\draw [|-|,thick] (17,0) to (21,0);
\draw [|-|,thick] (21,0) to (30,0);
\end{scope}
\begin{scope}   
\draw [|-|,thick] (14,0) node[left] {$T_{n-2}$} to (16,0);
\end{scope}
\begin{scope}
\draw  (21.2+.5,-1) circle [fill, radius=0.05]   (21.2+.5,-1);
\draw [-,dotted] (23.2+.8,-1) node[left] {\large $T_{n-1}$} to (29+.8,-1);
\draw [|-|,thick] (23.2+.8,-1) node[left] {} to (25+.8,-1);
\draw [|-|,thick] (27+.8,-1) node[left] {} to (29+.8,-1);
\end{scope}

\end{tikzpicture}}}
                
	\caption{Executions in the proof of Theorem~\ref{th:oftriv}; execution in \ref{sfig:rw-4} is not opaque
          \label{fig:rw}} 
\end{center}
\end{figure*}
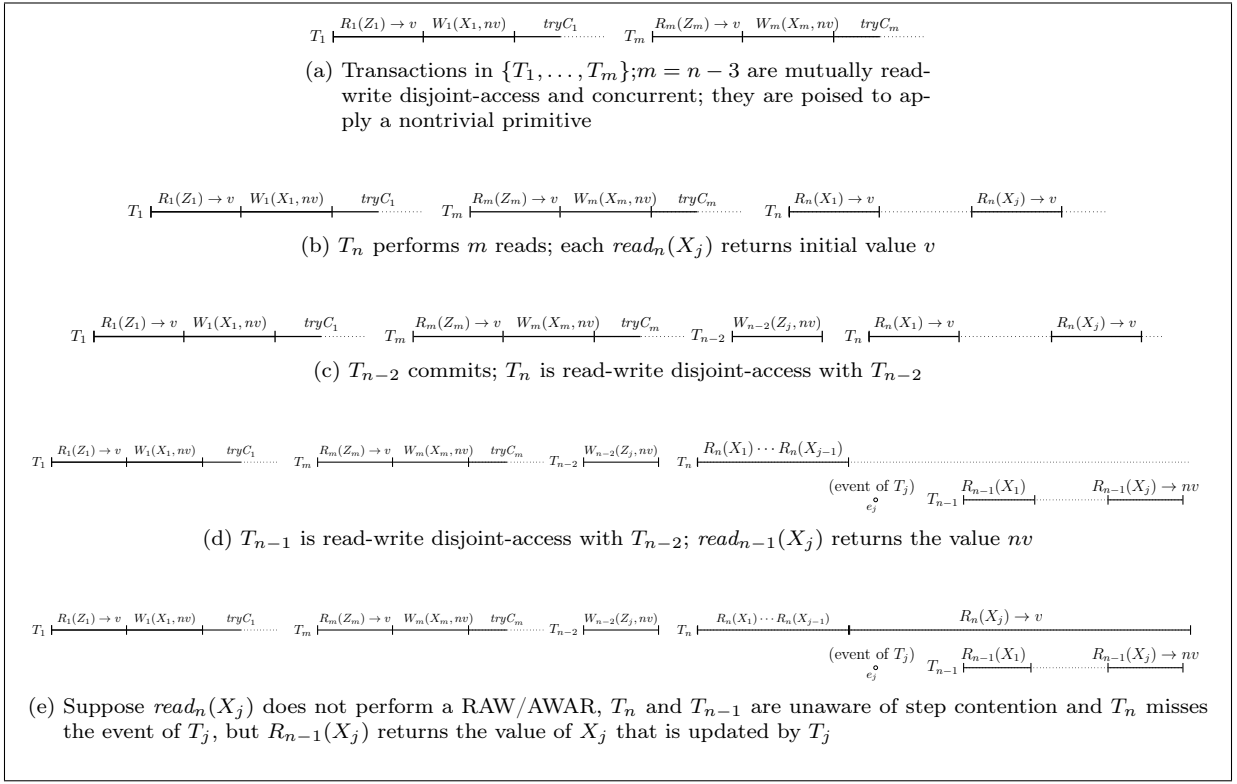
\begin{theorem}
\label{th:oftriv}
Every RW DAP opaque TM implementation $M\in \mathcal{OF}$ has
an execution $E$ in which some read-only transaction $T \in \ms{txns}(E)$ performs $\Omega(n)$ non-overlapping RAW/AWARs.
\end{theorem}
\begin{proof}
For all $j\in \{1,\ldots, m\}$; $m=n-3$, let $v$ be the initial value of t-objects $X_j$ and $Z_j$.
Throughout this proof, we assume that, for all $i\in \{1,\ldots , n\}$, transaction $T_i$ is invoked by process $p_i$.

By OF TM-progress and OF TM-liveness, any opaque and RW DAP TM implementation $M\in \mathcal{OF}$ 
has an execution of the form ${\bar \rho}_1\cdots {\bar \rho}_m$, where for all $j\in \{1,\ldots, m\}$, ${\bar \rho}_j$ 
denotes the t-complete step contention-free execution of transaction $T_j$
that performs $\Read_j(Z_j)\rightarrow v$, writes value $nv \neq v$ to 
$X_j$ and commits.

By construction, any two transactions that participate in ${\bar \rho}_1\cdots {\bar \rho}_n$
are mutually read-write disjoint-access and cannot contend on the same base object.
It follows that for all $1\leq j\leq m$,
${\bar \rho}_j$ is an execution of $M$.

For all $j\in \{1,\ldots, m\}$, we iteratively define an execution $\rho_j$ of $M$ as follows:  
it is the longest prefix of ${\bar \rho}_j$ such that 
$\rho_1\cdots \rho_{j}$ cannot be extended with the complete step contention-free execution fragment
of transaction $T_{n}$ that performs $j$ t-reads: $\Read_n(X_1)\cdots \Read_n(X_j)$
in which $\Read_{n}(X_j) \rightarrow nv$ nor with the 
complete step contention-free execution fragment
of transaction $T_{n-1}$ that performs $j$ t-reads: $\Read_{n-1}(X_1)\cdots \Read_{n-1}(X_j)$ 
in which $\Read_{n-1}(X_j) \rightarrow nv$ (Figure~\ref{sfig:rw-1}).

For any $j\in \{1,\ldots, m\}$, let $e_j$ be the event
transaction $T_j$ is poised to apply in the configuration after $\rho_1\cdots \rho_j$.
Thus, the execution $\rho_1\cdots \rho_j\cdot e_j$ can be extended with the complete step contention-free
executions of at least one of transaction $T_{n}$ or $T_{n-1}$ that performs $j$ t-reads of $X_1, \ldots ,X_j$
in which the t-read of $X_j$ returns the new value $nv$. 
Let $T_{n-1}$ be the transaction that must return the new value for the maximum number of $X_j$'s when
$\rho_1\cdots \rho_j\cdot e_j$ is extended with the t-reads of $X_1,\ldots , X_j$.
We show that, in the worst-case, transaction $T_n$ must perform $\lceil \frac{m}{2} \rceil$ non-overlapping RAW/AWARs
in the course of performing $m$ t-reads of $X_1,\ldots , X_m$ immediately after $\rho_1\cdots \rho_m$.
Symmetric arguments apply for the case when $T_{n}$
must return the new value for the maximum number of $X_j$'s when
$\rho_1\cdots \rho_j\cdot e_j$ is extended with the t-reads of $X_1,\ldots , X_j$.

\vspace{1mm}\noindent\textbf{Proving the RAW/AWAR lower bound.}
We prove that transaction $T_n$ must perform
$\lceil \frac{m}{2} \rceil$ non-overlapping RAWs or AWARs
in the course of performing $m$ t-reads of $X_1,\ldots , X_m$ immediately after the execution
$\rho_1\cdots \rho_m$.
Specifically, we prove that $T_n$ must perform a RAW or an AWAR during the execution of the t-read of each $X_j$
such that $\rho_1\cdots \rho_j\cdot e_j$ can be extended with the complete step contention-free
execution of $T_{n-1}$ as it performs $j$ t-reads of $X_1\ldots X_j$
in which the t-read of $X_j$ returns the new value $nv$.
Let $\mathbb{J}$ denote the of all $j\in \{1,\ldots , m\}$ such that
$\rho_1\cdots \rho_j\cdot e_j$ extended with the complete step contention-free
execution of $T_{n-1}$ performing $j$ t-reads of $X_1\ldots X_j$ must
return the new value $nv$ during the t-read of $X_j$.

We first prove that, for all $j \in \mathbb{J}$, 
$M$ has an execution of the form $\rho_1\cdots \rho_m \cdot \delta_j$ 
(Figures~\ref{sfig:rw-1} and \ref{sfig:rw-2}),
where $\delta_j$ is the complete step contention-free
execution fragment of $T_{n}$ that performs $j$ t-reads: $\Read_{n}(X_1)  \cdots \Read_{n}(X_j)$, each of which
return the initial value $v$.

By definition of $\rho_j$, OF TM-progress and OF TM-liveness, 
$M$ has an execution of the form $\rho_1\cdots \rho_j \cdot \delta_j$. 
By construction, transaction $T_{n}$ is read-write disjoint-access with each transaction $T \in \{T_{j+1},\ldots , T_m\}$ 
in $\rho_1\cdots \rho_j \cdots \rho_m \cdot \delta_j$. Thus, $T_n$ cannot contend with any of the transactions
in $\{T_{j+1},\ldots , T_m\}$, implying that, for all $j\in \{1,\ldots , m\}$, $M$ has an execution of the form
$\rho_1\cdots \rho_m \cdot \delta_j$ (Figure~\ref{sfig:rw-2}).

We claim that, for each $j\in \mathbb{J}$, the t-read of $X_j$ performed by $T_n$ must perform a RAW or an AWAR in the course
of performing $j$ t-reads of $X_1, \ldots , X_j$ immediately after $\rho_1\cdots \rho_m$.
Suppose by contradiction that $\Read_n(X_j)$ does not perform a RAW or an AWAR in 
$\rho_1\cdots \rho_m \cdot \delta_m$.
\begin{claim}
\label{cl:oftmex0}
For all $j \in \mathbb{J}$, $M$ has an execution of the form
$\rho_1\cdots \rho_{j} \cdots \rho_m  \cdot  \delta_{j-1} \cdot e_j \cdot \beta$ where,
$\beta$ is the complete step contention-free execution fragment of transaction $T_{n-1}$ that
performs $j$ t-reads: $\Read_{n-1}(X_1)\cdots \Read_{n-1}(X_{j-1}) \cdot \Read_{n-1}(X_j)$
in which $\Read_{n-1}(X_j)$ returns $nv$.
\end{claim}
\begin{proof}
We observe that transaction $T_{n}$ is read-write disjoint-access with every transaction $T\in \{T_j,T_{j+1},\ldots , T_m\}$ in
$\rho_1\cdots \rho_{j}\cdots \rho_m \cdot  \delta_{j-1}$.
By RW DAP, it follows that $M$ has an execution of the form
$\rho_1\cdots \rho_{j} \cdots \rho_m  \cdot  \delta_{j-1} \cdot e_j$
since $T_n$ cannot perform a nontrivial event on the base object accessed by $T_j$ in the event $e_j$.

By the definition of $\rho_j$, 
transaction $T_{n-1}$ must access the base object to which
$T_j$ applies a nontrivial primitive in $e_j$ to return the value $nv$ of $X_j$ as it performs $j$ t-reads
of $X_1,\ldots , X_j$ immediately after the execution $\rho_1\cdots \rho_{j} \cdots \rho_m  \cdot  \delta_{j-1} \cdot e_j$.
Thus, $M$ has an execution of the form $\rho_1\cdots \rho_{j}  \cdot  \delta_{j-1} \cdot e_j \cdot \beta$.

By construction, transactions $T_{n-1}$ is read-write disjoint-access with every
transaction $T\in \{T_{j+1},\ldots , T_m\}$ in
$\rho_1\cdots \rho_{j} \cdots \rho_m  \cdot  \delta_{j-1} \cdot e_j \cdot \beta$.
It follows that $M$ has an execution of the form 
$\rho_1\cdots \rho_j \cdots \rho_m \cdot \delta_{j-1} \cdot e_j \cdot \beta$.
\end{proof}
\begin{claim}
\label{cl:oftmex1}
For all $j \in \{1,\ldots , m\}$, 
$M$ has an execution of the form 
$\rho_1\cdots \rho_{j} \cdots \rho_m  \cdot \gamma\cdot  \delta_{j-1} \cdot e_j \cdot \beta$,
where $\gamma$ is the t-complete step contention-free execution fragment
of transaction $T_{n-2}$ that writes $nv \neq v$ to $Z_j$ and commits.
\end{claim}
\begin{proof}
Observe that $T_{n-2}$ precedes transactions $T_{n}$ and $T_{n-1}$ in real-time order in the above execution.

By OF TM-progress and OF TM-liveness, transaction $T_{n-2}$ must be committed in
$\rho_1\cdots \rho_{j} \cdots \rho_m  \cdot \gamma$.

Since transaction $T_{n-1}$ is read-write disjoint-access
with $T_{n-2}$ in $\rho_1\cdots \rho_{j} \cdots \rho_m  \cdot \gamma \cdot  \delta_{j-1} \cdot e_j \cdot \beta$,
$T_{n-1}$ does not contend with $T_{n-2}$ on any base object (recall that we associate an edge
with t-objects in the conflict graph only if they are both contained in the write set of some transaction).
Since the execution fragment $\beta$ contains an access to the base object to which $T_j$ performs
a nontrivial primitive in the event $e_j$, $T_{n-2}$ cannot perform a nontrivial event on this base object
in $\gamma$.
It follows that $M$ has an execution of the form 
$\rho_1\cdots \rho_{j} \cdots \rho_m  \cdot \gamma\cdot  \delta_{j-1} \cdot e_j \cdot \beta$
since, it is indistinguishable to $T_{n-1}$
from the execution 
$\rho_1\cdots \rho_{j} \cdots \rho_m  \cdot  \delta_{j-1} \cdot e_j \cdot \beta$ (the existence of which is already
established in Claim~\ref{cl:oftmex0}).
\end{proof}
Recall that transaction $T_n$ is read-write disjoint-access with $T_{n-2}$ in 
$\rho_1\cdots \rho_{j} \cdots \rho_m  \cdot \gamma \cdot  \delta_{j}$.
Thus, $M$ has an execution of the form 
$\rho_1\cdots \rho_{j} \cdots \rho_m  \cdot \gamma\cdot  \delta_{j}$ (Figure~\ref{sfig:rw-3}).

\vspace{1mm}\noindent\textbf{Deriving a contradiction.}
For all $j \in \{1,\ldots , m\}$, we represent the execution fragment 
$\delta_j$ as $\delta_{j-1} \cdot \pi^j$, where $\pi^j$ is the complete execution
fragment of the $j^{th}$ t-read $\Read_{n}(X_j) \rightarrow v$.
By our assumption, $\pi^j$ does not contain a RAW or an AWAR.

For succinctness, let $\alpha=\rho_1\cdots \rho_m \cdot \gamma \cdot \delta_{j-1}$.
We now prove that if $\pi^j$ does not contain a RAW or an AWAR, we can define $\pi^j_{1} \cdot \pi^j_{2}=\pi^j$
to construct an execution of the form
$\alpha  \cdot \pi^j_{1} \cdot e_j \cdot \beta \cdot \pi^j_{2}$ (Figure~\ref{sfig:rw-4}) such that
\begin{itemize}
\item
no event in $\pi^j_1$ is the application of a nontrivial primitive
\item 
$\alpha \cdot \pi^j_{1} \cdot e_j \cdot \beta \cdot \pi^j_{2}$
is indistinguishable to $T_{n}$ from the step contention-free execution
$\alpha \cdot \pi^j_{1} \cdot \pi^j_{2}$
\item
$\alpha \cdot \pi^j_{1}  \cdot e_j \cdot \beta \cdot \pi^j_{2}$
is indistinguishable to $T_{n-1}$ from the step contention-free execution
$\alpha \cdot e_j \cdot \beta$.
\end{itemize}
The following claim defines $\pi^j_1$ and $\pi^j_2$ to construct this execution.
\begin{claim}
\label{cl:ofraw}
For all $j \in \{1,\ldots , m\}$,
$M$ has an execution of the form
$\alpha  \cdot \pi^j_{1} \cdot e_j \cdot \beta \cdot \pi^j_{2}$.
\end{claim}
\begin{proof}
Let $t$ be the first event containing a write to a base object in the execution fragment $\pi^j$.
We represent $\pi^j$ as the execution fragment $\pi^j_{1}\cdot t \cdot \pi^j_{f}$.
Since $\pi^j_1$ does not contain nontrivial events that write to a base object, $\alpha \cdot \pi^j_{1}  \cdot e_j\cdot \beta$
is indistinguishable to transaction $T_{n-1}$ from the step contention-free execution
$\alpha  \cdot e_j\cdot \beta$ (as already proven in Claim~\ref{cl:oftmex1}).
Consequently, $\alpha \cdot \pi^j_{1} \cdot e_j\cdot \beta$
is an execution of $M$.

Since $t$ is not an atomic-write-after-read, $M$ has an execution of the form
$\alpha\cdot \gamma\cdot  \pi^j_{1} \cdot e_j\cdot \beta \cdot t$.
Secondly, since $\pi^j$ does not contain a read-after-write, any read of a base object performed in $\pi^j_{f}$ 
may only be performed to base objects previously written in $t \cdot \pi^j_{f}$.
Thus, $\alpha \cdot \pi^j_{1} \cdot e_j\cdot \beta \cdot t \cdot \pi^j_{f}$
is indistinguishable to $T_{n}$ from the step contention-free execution
$\alpha \cdot \pi^j_1 \cdot t \cdot \pi^j_{f}$. 
But, as already proved, $\alpha \cdot \pi^{j}$ is an execution of $M$.

Choosing $\pi^j_{2}=t \cdot \pi^j_{f}$, it follows that
$M$ has an execution of the form $\alpha  \cdot \pi^j_{1} \cdot  e_j\cdot \beta \cdot \pi^j_{2}$.
\end{proof}
We have now proved that, for all $j \in \{1,\ldots , m\}$,
$M$ has an execution of the form
$\rho_1\cdots \rho_m \cdot \gamma \cdot \delta_{j-1} \cdot \pi^j_{1} \cdot  e_j\cdot \beta \cdot \pi^j_{2}$ (Figure~\ref{sfig:rw-4}). 

The execution in Figure~\ref{sfig:rw-4}
is not opaque.
Indeed, in any serialization the following must hold. 
Since $T_{n-1}$ reads the value written by $T_j$ in $X_j$,
$T_j$ must be committed.  
Since $\Read_{n}(X_j)$ returns the initial value $v$, 
$T_{n}$ must precede $T_j$.
The committed transaction $T_{n-2}$, which writes a new value to $Z_j$, must precede 
$T_{n}$ to respect the real-time order on transactions. 
However, $T_j$ must precede $T_{n-2}$ since $\Read_j(Z_j)$ returns the initial value of $Z_j$.
The cycle $T_j\rightarrow T_{n-2} \rightarrow T_{n} \rightarrow T_j$
implies that there exists no such a serialization.

Thus, for each $j\in \mathbb{J}$, transaction $T_n$ must perform a RAW or an AWAR during 
the t-read of $X_j$ in the course
of performing $m$ t-reads of $X_1, \ldots , X_m$ immediately after $\rho_1\cdots \rho_m$.
Since $|\mathbb{J}|\geq \lceil \frac{(n-3)}{2} \rceil$, in the worst-case, $T_n$ must perform $\Omega(n)$ RAW/AWARs during the execution of $m$ t-reads
immediately after $\rho_1\cdots \rho_m$.
\end{proof}
%
%
%
\section{Algorithms for obstruction-free TMs}
\label{sec:oftmalgos}
In this section, we present two opaque obstruction-free TM implementations: the first one satisfies RW DAP, but not strict DAP
while the second one satisfies weak DAP, but not RW DAP.
\subsection{An opaque RW DAP TM implementation}
\label{sec:rwoftm}
In this section, we describe a RW DAP TM implementation in $\mathcal{OF}$ (based on \emph{DSTM}~\cite{HLM+03}).

Every t-object $X_m$ maintains a base object $\ms{tvar}[m]$ and every transaction $T_k$
maintains a $\ms{status}[k]$ base object. Both base objects support 
the \emph{read}, \emph{write} and \emph{compare-and-swap} (\emph{cas}) primitives.

The object $\ms{tvar}[m]$ stores a triple:
the \emph{owner} of $X_m$ is an updating transaction that performs the latest write to $X_m$, 
the \emph{old value} and \emph{new value} of $X_m$ represent two latest versions
of $X_m$. The base object $\ms{status}[k]$ denotes if $T_k$
is \emph{live} (i.e. t-incomplete), \emph{committed} or \emph{aborted}.
Intuitively, if $\ms{status}[k]$ is \emph{committed}, then other transactions can safely read the value
of the t-objects updated by $T_k$.

Implementation of $\Read_k(X_m)$ first reads $\ms{tvar}[m]$ and checks if the \emph{owner} of $X_m$ is \emph{live}; if so, it
forcefully aborts the owning transaction and returns the \emph{old value} of $X_m$. Otherwise, if
the owner is \emph{committed}, it returns the \emph{new value} of $X_m$. In both cases, it only returns
a non-abort value if no t-object previously read has been updated since.
The $\Write_k(X_m,v)$
works similar to the $\Read_k(X_m)$ implementation; but additionally, if the \emph{owner} of $X_m$ is \emph{live}, it
forcefully aborts the owning transaction, assumes ownership of $X_m$, sets $v$ as the new value
of $X_m$ and leaves the \emph{old value} of $X_m$ unchanged. Otherwise, if the \emph{owner} of $X_m$ 
is a committed transaction, it updates the \emph{old value} of $X_m$ to be the value of $X_m$ updated by its previous \emph{owner}.
The $\TryC_k$ implementation sets $\ms{status}[k]$ to \emph{committed} if it has not been set to
aborted by a concurrent transaction, otherwise $T_k$ is deemed aborted.
Since any t-read operation performs at most two AWARs and the \emph{tryC} performs only a single AWAR, 
any read-only transaction $T$
performs at most $O(|\Rset(T)|)$ AWARs.
The pseudocode is described in Algorithm~\ref{alg:oftm}.
\begin{algorithm}[!h]
\caption{RW DAP opaque implementation $M\in \mathcal{OF}$; code for $T_k$
}\label{alg:oftm}
  \begin{algorithmic}[1]
  	\begin{multicols}{2}
  	{\footnotesize
	\Part{Shared base objects}{
		\State $\ms{tvar}[m]$, storing $[\ms{owner}_m,\ms{oval}_m,\ms{nval}_m]$
		\State ~~~~for each t-object $X_m$, supports read, write, cas
		\State ~~~~$\ms{owner}_m$, a transaction identifier 
		\State ~~~~$\ms{oval}_m\in V$
		\State ~~~~$\ms{nval}_m\in V$
		\State $\ms{status}[k] \in \{\ms{live},\ms{aborted},\ms{committed}\}$,  
		\State ~~~~for each $T_k$; supports read, write, cas
	}\EndPart	
	\Part{Local variables}{
		\State $\ms{Rset}_k,\ms{Wset}_k$ for every transaction $T_k$;
		\State ~~~~dictionaries storing $\{X_m$, $\ms{Tvar}[m]\}$
	}\EndPart	
	\Statex
	\Part{\Read$_k(X_m)$}{
		\State $[\ms{owner}_m,\ms{oval}_m,\ms{nval}_m]$ $\gets$ $\ms{tvar}[m].\lit{read}()$ \label{line:linr}
		
		\If{$\ms{owner}_m \neq k$}
			
			\State $s_m\gets \ms{status}[\ms{owner}_m].\lit{read}()$ \label{line:status1}
			\If{$s_m=\ms{committed}$} \label{line:rcurr}
				\State $\ms{curr}=\ms{nval}_m$
			\ElsIf{$s_m=\ms{aborted}$} 
			  \State $\ms{curr}=\ms{oval}_m$
			  
			\Else
				\If{$\ms{status}[\ms{owner}_m].\lit{cas}(\ms{live},\ms{aborted})$} \label{line:awar}
				  \State $\ms{curr}=\ms{oval}_m$
				\Else
				\Return $A_k$ \EndReturn
				\EndIf
				
			\EndIf
			\If{$ \ms{status}[k]=\ms{live} \wedge \neg \lit{validate}()$} \label{line:rc1}
				\State $\Rset(T_k).\lit{add}(\{X_m,[\ms{owner}_m,\ms{oval}_m,\ms{nval}_m]\})$
				\Return $\ms{curr}$ \EndReturn
			\EndIf
			\Return $A_k$ \EndReturn \label{line:of2}
			
		\Else
			
			\Return $\Rset(T_k).\lit{locate}(X_m)$ \EndReturn
				
		\EndIf

   	 }\EndPart
	\Statex
	\Part{Function: $\lit{validate}()$}{
		\If{$\exists \{X_j,[\ms{owner}_j,\ms{oval}_j,\ms{nval}_j]\} \in \Rset(T_k)$:\\
		~~~~~~~($[\ms{owner}_j,\ms{oval}_j,\ms{nval}_j]\neq \ms{tvar}[j].\lit{read}())$}
			\Return $\true$ \EndReturn
		\EndIf
		\Return $\false$ \EndReturn
	}\EndPart
	
	\newpage
	\Part{\Write$_k(X_m,v)$}{
		\State $[\ms{owner}_m,\ms{oval}_m,\ms{nval}_m]$ $\gets$ $\ms{tvar}[m].\lit{read}()$ \label{line:writeread}
		\If{$\ms{owner}_m \neq k$}
			
			\State $s_m\gets \ms{status}[\ms{owner}_m].\lit{read}()$ \label{line:status2}
			\If{$s_m=\ms{committed}$} \label{line:nval}
				\State $\ms{curr}=\ms{nval}_m$
			\ElsIf{$s_m=\ms{aborted}$} 
			  \State $\ms{curr}=\ms{oval}_m$
			\Else
				\If{$\ms{status}[\ms{owner}_m].\lit{cas}(\ms{live},\ms{aborted})$} \label{line:writeabort}
				  \State $\ms{curr}=\ms{oval}_m$
				
				\Else
				\Return $A_k$ \EndReturn
			\EndIf
				
		\EndIf
		\State $o_m\gets \ms{tvar}[m].\lit{cas}([\ms{owner}_m,\ms{oval}_m,\ms{nval}_m],[k,\ms{curr},v])$ \label{line:linw}
		\If{$o_m\wedge \ms{status}[k]= \ms{live}$} \label{line:wc1}
			\State $\ms{Wset}_k.\lit{add}(\{X_m,[k,\ms{curr},v]\})$
				\Return $ok$ \EndReturn
		\Else
			\Return $A_k$ \EndReturn \label{line:of4}
		\EndIf
		\Else
				\State $[\ms{owner}_m,\ms{oval}_m,\ms{nval}_m]=\ms{Wset}_k.\lit{locate}(X_m)$
				\State $s= \ms{tvar}[m].\lit{cas}([\ms{owner}_m,\ms{oval}_m,\ms{nval}_m],[k,\ms{oval}_m,v])$
				\If{$s$}
					\State $\Wset(T_k).\lit{add}(\{X_m,[k,\ms{oval}_m,v]\})$
					\Return $ok$ \EndReturn
				\Else
					\Return $A_k$ \EndReturn \label{line:of5}
		
				\EndIf
		\EndIf
   	}\EndPart	
   	
   	\Statex
	\Part{\TryC$_k$()}{
		\If{$\lit{validate}()$}
			\Return $A_k$ \EndReturn  \label{line:of1}
		\EndIf
		\If{$\ms{status}[k].\lit{cas}(\ms{live},\ms{committed})$} \label{line:tryc}
			\Return $C_k$ \EndReturn
		\EndIf
		\Return $A_k$ \EndReturn 
   	 }\EndPart

	}\end{multicols}
  \end{algorithmic}
\end{algorithm}
\begin{lemma}
\label{lm:oftmopaque}
Algorithm~\ref{alg:oftm} implements an opaque TM.
\end{lemma}
\begin{proof}
Since opacity is a safety property, we only consider finite executions~\cite{icdcs-opacity}.
Let $E$ by any finite execution of Algorithm~\ref{alg:oftm}. 
Let $<_E$ denote a total-order on events in $E$.

Let $H$ denote a subsequence of $E$ constructed by selecting
\emph{linearization points} of t-operations performed in $E$.
The linearization point of a t-operation $op$, denoted as $\ell_{op}$ is associated with  
a base object event or an event performed during 
the execution of $op$ using the following procedure. 

\vspace{1mm}\noindent\textbf{Completions.}
First, we obtain a completion of $E$ by removing some pending
invocations and adding responses to the remaining pending invocations
involving a transaction $T_k$ as follows:
every incomplete $\Read_k$, $\Write_k$, $\TryC_k$ operation is removed from $E$;
an incomplete $\Write_k$ is removed from $E$.

\vspace{1mm}\noindent\textbf{Linearization points.}
We now associate linearization points to
t-operations in the obtained completion of $E$ as follows:
\begin{itemize}
\item For every t-read $op_k$ that returns a non-A$_k$ value, $\ell_{op_k}$ is chosen as the event in Line~\ref{line:linr}
of Algorithm~\ref{alg:oftm}, else, $\ell_{op_k}$ is chosen as invocation event of $op_k$
\item For every t-write $op_k$ that returns a non-A$_k$ value, $\ell_{op_k}$ is chosen as the event in Line~\ref{line:writeread}
of Algorithm~\ref{alg:oftm}, else, $\ell_{op_k}$ is chosen as invocation event of $op_k$
\item For every $op_k=\TryC_k$ that returns $C_k$, $\ell_{op_k}$ is associated with Line~\ref{line:of1}.
\end{itemize}
$<_H$ denotes a total-order on t-operations in the complete sequential history $H$.

\vspace{1mm}\noindent\textbf{Serialization points.}
The serialization of a transaction $T_j$, denoted as $\delta_{T_j}$ is
associated with the linearization point of a t-operation 
performed during the execution of the transaction.

We obtain a t-complete history ${\bar H}$ from $H$ as follows: 
for every transaction $T_k$ in $H$ that is complete, but not t-complete, 
we insert $\textit{tryC}_k\cdot A_k$ after $H$. 

${\bar H}$ is thus a t-complete sequential history.
A t-complete t-sequential history $S$ equivalent to ${\bar H}$ is obtained by associating 
serialization points to transactions in ${\bar H}$ as follows:
\begin{itemize}
\item If $T_k$ is an update transaction that commits, then $\delta_{T_k}$ is $\ell_{tryC_k}$
\item If $T_k$ is an aborted or read-only transaction in $\bar H$,
then $\delta_{T_k}$ is assigned to the linearization point of the last t-read that returned a non-A$_k$ value in $T_k$
\end{itemize}
$<_S$ denotes a total-order on transactions in the t-sequential history $S$.
\begin{claim}
\label{cl:oseq}
If $T_i \prec_{H}^{RT} T_j$, then $T_i <_S T_j$.
\end{claim}
\begin{proof}
This follows from the fact that for a given transaction, its
serialization point is chosen between the first and last event of the transaction
implying if $T_i \prec_{H} T_j$, then $\delta_{T_i} <_{E} \delta_{T_j}$ implies $T_i <_S T_j$ 
\end{proof}
\begin{claim}
\label{cl:ofclaim0}
If transaction $T_i$ returns $C_i$ in $E$, then \emph{status[i]=committed} in $E$.
\end{claim}
\begin{proof}
Transaction $T_i$ must perform the event in Line~\ref{line:tryc}
before returning $T_i$ i.e. the \emph{cas} on its own \emph{status} to change the
value to \emph{committed}. The proof now follows from the fact that any other transaction
may change the \emph{status} of $T_i$ only if it is \emph{live} (Lines~\ref{line:writeabort} and \ref{line:awar}).
\end{proof}
\begin{claim}
\label{cl:oreadfrom}
$S$ is legal.
\end{claim}
\begin{proof}
Observe that for every $\Read_j(X) \rightarrow v$, there exists some transaction $T_i$
that performs $\Write_i(X,v)$ and completes the event in Line~\ref{line:linw} to write $v$ as the \emph{new value} of $X$ such that
$\Read_j(X) \not\prec_H^{RT} \Write_i(X,v)$. 
For any updating
committing transaction $T_i$, $\delta_{T_i}=\ell_{\TryC_i}$. 
Since $\Read_j(X)$ returns a response $v$, the event in Line~\ref{line:linr} must succeed
the event in Line~\ref{line:tryc} when $T_i$ changes \emph{status[i]} to \emph{committed}.
Suppose otherwise, then $\Read_j(X)$ subsequently forces $T_i$ to abort by writing
\emph{aborted} to \emph{status[i]} and must return the \emph{old value} of $X$ 
that is updated by the previous \emph{owner} of $X$, which must be committed in $E$ (Line~\ref{line:nval}).
Since $\delta_{T_{i}}=\ell_{\TryC_{i}}$ precedes the event in Line~\ref{line:tryc},
it follows that $\delta_{T_{i}} <_E \ell_{\Read_{j}(X)}$.

We now need to prove that $\delta_{T_{i}} <_E \delta_{T_{j}}$. Consider the following cases:
\begin{itemize}
\item
if $T_j$ is an updating committed transaction, then $\delta_{T_{j}}$ is assigned to $\ell_{\TryC_{j}}$.
But since $\ell_{\Read_{j}(X)} <_E \ell_{\TryC_{j}}$, it follows that $ \delta_{T_{i}} <_E \delta_{T_{j}}$.
\item
if $T_j$ is a read-only or aborted transaction, then $\delta_{T_{j}}$ is assigned to
the last t-read that did not abort. Again, it follows that $ \delta_{T_{i}} <_E \delta_{T_{j}}$.
\end{itemize}
To prove that $S$ is legal, we need to show that,
there does not exist any
transaction $T_k$ that returns $C_k$ in $S$ and performs $\Write_k(X,v')$; $v'\neq v$ such that $T_i <_S T_k <_S T_j$. 
Now, suppose by contradiction that there exists a committed transaction $T_k$, $X \in \Wset(T_k)$ that writes $v'\neq v$ to $X$ 
such that $T_i <_S T_k <_S T_j$.
Since $T_i$ and $T_k$ are both updating transactions that commit,
\begin{center}
($T_i <_S T_k$) $\Longleftrightarrow$ ($\delta_{T_i} <_{E} \delta_{T_k}$) \\
($\delta_{T_i} <_{E} \delta_{T_k}$) $\Longleftrightarrow$ ($\ell_{\TryC_i} <_{E} \ell_{\TryC_k}$) 
\end{center}
Since, $T_j$ reads the value of $X$ written by $T_i$, one of the following is true:
$\ell_{\TryC_i} <_{E} \ell_{\TryC_k} <_{E} \ell_{\Read_j(X)}$ or
$\ell_{\TryC_i} <_{E} \ell_{\Read_j(X)} <_{E} \ell_{\TryC_k}$.

If $\ell_{\TryC_i} <_{E} \ell_{\TryC_k} <_{E} \ell_{\Read_j(X)}$, then the event in Line~\ref{line:tryc}
performed by $T_k$ when it changes the status field to \emph{committed}
precedes the event in Line~\ref{line:linr} performed by $T_j$.
Since $\ell_{\TryC_i} <_{E} \ell_{\TryC_k}$ and both $T_i$ and $T_k$ are committed in $E$,
$T_k$ must perform the event in Line~\ref{line:writeread} after $T_i$ changes \emph{status[i]}
to \emph{committed} since otherwise, $T_k$ would perform the event in Line~\ref{line:writeabort}
and change \emph{status[i]} to \emph{aborted}, thereby forcing $T_i$ to return $A_i$.
However, $\Read_j(X)$ observes that the \emph{owner} of $X$ is $T_k$
and since the \emph{status} of $T_k$ is committed at this point in the execution,
$\Read_j(X)$ must return $v'$ and not $v$---contradiction.

Thus, $\ell_{\TryC_i} <_{E} \ell_{\Read_j(X)} <_{E} \ell_{\TryC_k}$.
We now need to prove that $\delta_{T_{j}}$ indeed precedes $\delta_{T_{k}}=\ell_{\TryC_k}$ in $E$.

Now consider two cases:
\begin{itemize}
\item
Suppose that $T_j$ is a read-only transaction. 
Then, $\delta_{T_j}$ is assigned to the last t-read performed by $T_j$ that returns a non-A$_j$ value. 
If $\Read_j(X)$ is not the last t-read that returned a non-A$_j$ value, then there exists a $read_j(X')$ such that 
$\ell_{\Read_j(X)} <_{E} \ell_{\TryC_k} <_E \ell_{read_j(X')}$.
But then this t-read of $X'$ must abort since the value of $X$ has been updated by $T_k$ since $T_j$ first 
read $X$---contradiction.
\item
Suppose that $T_j$ is an updating transaction that commits, then $\delta_{T_j}=\ell_{\TryC_j}$ which implies that
$\ell_{read_j(X)} <_{E} \ell_{\TryC_k} <_E \ell_{\TryC_j}$. Then, $T_j$ must neccesarily perform the validation
of its read set in Line~\ref{line:of1} and return $A_j$---contradiction.
\end{itemize}
\end{proof}
Claims~\ref{cl:oseq} and \ref{cl:oreadfrom} establish that
Algorithm~\ref{alg:oftm} is opaque.
\end{proof}

\begin{theorem}
\label{th:ofraw}
Algorithm~\ref{alg:oftm} describes a RW DAP, progressive opaque TM implementation $M\in \mathcal{OF}$ such that in every
execution $E$ of $M$, 
\begin{itemize}
\item 
the total number of stalls incurred by a t-read operation invoked in $E$ is $O(n)$,
\item
every read-only transaction $T \in \ms{txns}(E)$ performs $O(|\Rset(T)|)$ AWARs in $E$, and
\item
every complete t-read operation invoked by transaction $T_k \in \ms{txns}(E)$ performs $O(|\Rset_E(T_k)|$ steps.
\end{itemize}
\end{theorem}
\begin{proof}
\textit{(Opacity)}
Follows from Lemma~\ref{lm:oftmopaque}

\textit{(TM-liveness and TM-progress)}
Since none of the implementations of the t-operations in Algorithm~\ref{alg:oftm}
contain unbounded loops or waiting statements, every
t-operation $op_k$ returns a matching response after taking a finite number of steps.
Thus, Algorithm~\ref{alg:oftm} provides wait-free TM-liveness.

To prove OF TM-progress, we proceed by enumerating the cases under which a transaction $T_k$ may be aborted in any execution.
\begin{itemize}
\item
Suppose that there exists a $\Read_k(X_m)$ performed by $T_k$ that returns $A_k$.
If $\Read_k(X_m)$ returns $A_k$ in Line~\ref{line:of2}, then there exists a concurrent transaction
that updated a t-object in $\Rset(T_k)$ or changed \emph{status[k]} to \emph{aborted}.
In both cases, $T_k$ returns $A_k$ only because there is step contention.
\item
Suppose that there exists a $\Write_k(X_m,v)$ performed by $T_k$ that returns $A_k$ in Line~\ref{line:of4}.
Thus, either a concurrent transaction has changed \emph{status[k]} to \emph{aborted} or the value
in $\ms{tvar}[m]$ has been updated since the event in Line~\ref{line:writeread}.
In both cases, $T_k$ returns $A_k$ only because of step contention with another transaction.
\item
Suppose that a $\Read_k(X_m)$ or $\Write_k(X_m,v)$ return $A_k$ in Lines~\ref{line:awar} and \ref{line:writeabort} respectively.
Thus, a concurrent transaction has takes steps concurrently by updating the 
\emph{status} of $\ms{owner}_m$ since the read by $T_k$ in Lines~\ref{line:linr} and \ref{line:writeread} respectively.
\item
Suppose that $\TryC_k()$ returns $A_k$ in Line~\ref{line:of5}. This is because
there exists a t-object in $\Rset(T_k)$ that has been updated by a concurrent transaction since,
\emph{i.e.}, $\TryC_k()$ returns $A_k$ only on encountering step contention.
\end{itemize}
It follows that in any step contention-free execution of a transaction $T_k$
from a $T_k$-free execution, $T_k$ must return $C_k$ after taking a finite number of steps.

The enumeration above also proves that $M$ implements a progressive TM.

\textit{(Read-write disjoint-access parallelism)}
Consider any execution $E$ of Algorithm~\ref{alg:oftm} and let
$T_i$ and $T_j$ be any two transactions
that contend on a base object $b$ in $E$.
We need to prove that there is a path between a t-object in $\Dset(T_i)$ and a t-object in $\Dset(T_j)$ 
in ${\tilde G}(T_i,T_j,E)$ or there exists $X \in \Dset(T_i) \cap \Dset(T_j)$.
Recall that there exists an edge between t-objects $X$ and $Y$ in ${\tilde G}(T_i,T_j,E)$
only if there exists a transaction $T\in \ms{txns}(E)$ such that $\{X,Y\} \in \Wset(T)$.
\begin{itemize}
\item
Suppose that $T_i$ and $T_j$ contend on base object \ms{tvar}[m] belonging to t-object $X_m$ in $E$.
By Algorithm~\ref{alg:oftm}, a transaction accesses $X_m$ only if $X_m$ is contained in $\Dset(T_m)$.
Thus, both $T_i$ and $T_j$ must access $X_m$.
\item
Suppose that $T_i$ and $T_j$ contend on base object \ms{status}[i] in $E$ (the case when $T_i$ and $T_j$ contend
on \ms{status}[j] is symmetric).
$T_j$ accesses \emph{status[i]} while performing a t-read of some t-object $X$ in Lines~\ref{line:status1} and \ref{line:awar}
only if $T_i$ is the \emph{owner} of $X$.
Also, $T_j$ accesses \emph{status[i]} while performing a t-write to $X$ in Lines~\ref{line:status2} and \ref{line:writeabort}
only if $T_i$ is the \emph{owner} of $X$.
But if $T_i$ is the \emph{owner} of $X$, then $X \in \Wset(T_i)$.
\item
Suppose that $T_i$ and $T_j$ contend on base object \ms{status}[m] belonging to some transaction $T_m$ in $E$.
Firstly, observe that $T_i$ or $T_j$ access \emph{status[m]} only if 
there exist t-objects $X$ and $Y$ in $\Dset(T_i)$ and $\Dset(T_j)$ respectively such that $\{X,Y\} \in \Wset(T_m)$.
This is because $T_i$ and $T_j$ would both read \emph{status[m]} in Lines~\ref{line:status1} (during t-read) and 
\ref{line:status2} (during t-write)
only if $T_m$ was the previous \emph{owner} of $X$ and $Y$.
Secondly, one of $T_i$ or $T_j$ applies a nontrivial primitive to \emph{status[m]}
only if $T_i$ and $T_j$ read \emph{status[m]=live} in Lines~\ref{line:status1} (during t-read) and 
\ref{line:writeread} (during t-write).
Thus, at least one of $T_i$ or $T_j$ is concurrent to $T_m$ in $E$.
It follows that there exists a path between $X$ and $Y$ 
in ${\tilde G}(T_i,T_j,E)$.
\end{itemize}
\textit{(Complexity)} 
Every t-read operation performs at most one AWAR in an execution $E$ (Line~\ref{line:awar}) of Algorithm~\ref{alg:oftm}.
It follows that any read-only transaction $T_k \in \ms{txns}(E)$ performs at most $|\Rset(T_k)|$ AWARs in $E$.

The linear step-complexity is immediate from the fact that during the t-read operations, the transaction validates its entire read
set (Line~\ref{line:rc1}). All other t-operations incur $O(1)$ step-complexity since they involve no iteration statements
like \emph{for} and \emph{while} loops.

Since at most $n-1$ transactions may be t-incomplete at any point in an execution $E$, it follows that
$E$ is at most a $(n-1)$-stall execution for any t-read $op$ and every $T\in \ms{txns}(E)$
incurs $O(n)$ stalls on account of any event performed in $E$. More specifically, consider the following
execution $E$: for all $i\in \{1,\ldots , n-1\}$, each transaction $T_i$ performs $\Write_i(X_m,v)$
in a step-contention free execution until it is poised to apply a nontrivial event on $\ms{tvar}[m]$ (Line~\ref{line:linw}).
By OF TM-progress, we construct $E$ such that each of the $T_i$ is poised to apply a nontrivial event on $\ms{tvar}[m]$
after $E$. Consider the execution fragment of $\Read_n(X_m)$ that is poised to perform an event $e$
that reads $\ms{tvar}[m]$ (Line~\ref{line:linr}) immediately after $E$.
In the constructed execution, $T_n$ incurs $O(n)$ stalls on account of $e$ and thus, produces the desired $(n-1)$-stall execution
for $\Read_n(X)$.
\end{proof}
\subsection{An opaque weak DAP TM implementation}
\label{sec:woftm}
In this section, we describe a weak DAP TM implementation in $\mathcal{OF}$ with constant step-complexity t-read operations.
\begin{algorithm}[!h]
\caption{Weak DAP opaque implementation $M\in \mathcal{OF}$; code for $T_k$
}\label{alg:oftm2}
  \begin{algorithmic}[1]
  	{\footnotesize
	
	\Part{\Read$_k(X_m)$}{
		\State $[\ms{owner}_m,\ms{oval}_m,\ms{nval}_m]$ $\gets$ $\ms{tvar}[m].\lit{read}()$ \label{line:ofread1}
		
		\If{$\ms{owner}_m \neq k$}
			
			\State $s_m\gets \ms{status}[\ms{owner}_m].\lit{read}()$ \label{line:ownerread}
			\If{$s_m=\ms{committed}$} 
				\State $\ms{curr}=\ms{nval}_m$
			\ElsIf{$s_m=\ms{aborted}$} 
			  \State $\ms{curr}=\ms{oval}_m$

			\Else
				\If{$\ms{status}[\ms{owner}_m].\lit{cas}(\ms{live},\ms{aborted})$} \label{line:ownerwrite}
				  \State $\ms{curr}=\ms{oval}_m$
				\EndIf
				\Return $A_k$ \EndReturn
			\EndIf
				
			\State $o_m\gets \ms{tvar}[m].\lit{cas}([\ms{owner}_m,\ms{oval}_m,\ms{nval}_m],[k,\ms{oval}_m,\ms{nval}_m])$ \label{line:readowner}
			\If{$o_m\wedge \ms{status}[k]= \ms{live}$} 
			    \State $\Rset(T_k).\lit{add}(\{X_m,[\ms{owner}_m,\ms{oval}_m,\ms{nval}_m]\})$
				\Return $\ms{curr}$ \EndReturn
			\EndIf
			
		\Else
			
			\Return $\Rset(T_k).\lit{locate}(X_m)$ \EndReturn
				
		\EndIf

   	 }\EndPart
	\Statex
	\Part{\TryC$_k$()}{
		
		\If{$\ms{status}[k].\lit{cas}(\ms{live},\ms{committed})$} 
			\Return $C_k$ \EndReturn
		\EndIf
		\Return $A_k$ \EndReturn 
   	 }\EndPart

	}
  \end{algorithmic}
\end{algorithm}

Algorithm~\ref{alg:oftm2} describes a weak DAP implementation in $\mathcal{OF}$ that does not satisfy read-write DAP.
The code for the t-write operations is identical to Algorithm~\ref{alg:oftm}.
During the t-read of t-object $X_m$ by transaction $T_k$, $T_k$ becomes the \emph{owner} of $X_m$
thus eliminating the per-read validation step-complexity inherent to Algorithm~\ref{alg:oftm}.
Similarly, $\TryC_k$ also not involve performing the validation of the $T_k$'s read set; the implementation
simply sets $\ms{status}[k]=\ms{committed}$ and returns $C_k$.
\begin{theorem}
\label{th:ofweakdap}
Algorithm~\ref{alg:oftm2} describes a weak TM implementation $M\in \mathcal{OF}$ such that
in any execution $E$ of $M$, for every transaction $T \in \ms{txns}(E)$,
$T$ performs $O(1)$ steps during the execution of any t-operation in $E$.
\end{theorem}
\begin{proof}
The proofs of opacity, TM-liveness and TM-progress are almost identical to the analogous proofs
for Algorithm~\ref{alg:oftm}.

\textit{(Weak disjoint-access parallelism)}
Consider any execution $E$ of Algorithm~\ref{alg:oftm2} and let
$T_i$ and $T_j$ be any two transactions
that contend on a base object $b$ in $E$.
We need to prove that there is a path between a t-object in $\Dset(T_i)$ and a t-object in $\Dset(T_j)$ 
in ${\tilde G}(T_i,T_j,E)$ or there exists $X \in \Dset(T_i) \cap \Dset(T_j)$.
Recall that there exists an edge between t-objects $X$ and $Y$ in ${G}(T_i,T_j,E)$
only if there exists a transaction $T\in \ms{txns}(E)$ such that $\{X,Y\} \in \Dset(T)$.
\begin{itemize}
\item
Suppose that $T_i$ and $T_j$ contend on base object \ms{tvar}[m] belonging to t-object $X_m$ in $E$.
By Algorithm~\ref{alg:oftm2}, a transaction accesses $X_m$ only if $X_m$ is contained in $\Dset(T_m)$.
Thus, both $T_i$ and $T_j$ must access $X_m$.
\item
Suppose that $T_i$ and $T_j$ contend on base object \ms{status}[i] in $E$ (the case when $T_i$ and $T_j$ contend
on \ms{status}[j] is symmetric).
$T_j$ accesses \emph{status[i]} while performing a t-read of some t-object $X$ in 
Lines~\ref{line:ownerread} and \ref{line:ownerwrite}
only if $T_i$ is the \emph{owner} of $X$.
Also, $T_j$ accesses \emph{status[i]} while performing a t-write to $X$ in Lines~\ref{line:status2} and \ref{line:writeabort}
only if $T_i$ is the \emph{owner} of $X$.
But if $T_i$ is the \emph{owner} of $X$, then $X \in \Dset(T_i)$.
\item
Suppose that $T_i$ and $T_j$ contend on base object \ms{status}[m] belonging to some transaction $T_m$ in $E$.
Firstly, observe that $T_i$ or $T_j$ access \emph{status[m]} only if 
there exist t-objects $X$ and $Y$ in $\Dset(T_i)$ and $\Dset(T_j)$ respectively such that $\{X,Y\} \in \Dset(T_m)$.
This is because $T_i$ and $T_j$ would both read \emph{status[m]} in Lines~\ref{line:ownerread} (during t-read) and 
\ref{line:status2} (during t-write)
only if $T_m$ was the previous \emph{owner} of $X$ and $Y$.
Secondly, one of $T_i$ or $T_j$ applies a nontrivial primitive to \emph{status[m]}
only if $T_i$ and $T_j$ read \emph{status[m]=live} in Lines~\ref{line:ownerread} (during t-read) and 
\ref{line:writeread} (during t-write).
Thus, at least one of $T_i$ or $T_j$ is concurrent to $T_m$ in $E$.
It follows that there exists a path between $X$ and $Y$ 
in ${\tilde G}(T_i,T_j,E)$.
\end{itemize}
\textit{(Complexity)}
Since no implementation of any of the t-operation contains any iteration statements like \emph{for} and \emph{while} loops), 
the proof follows.
\end{proof}

%
%
\section{Why Transactional memory should not be obstruction-free}
\label{sec:p3c3s3}
%
\begin{figure}[h]
      
     \scalebox{1}[1]{
     \begin{tabularx}{\textwidth}{c|c|c}
	~~~~~ & Obstruction-free TMs & Progressive TM $LP$ \\ \hline
	strict DAP & No~\cite{OFTM} & Yes \\ \hline
	invisible reads+weak DAP & No  & Yes\\ \hline
	stall complexity of t-reads & $\Omega(n)$  
        & $O(1)$  \\ \hline
	RAW/AWAR complexity & $\Omega(n)$
         & $O(1)$  \\ \hline
	read-write base objects, wait-free TM-liveness & No~\cite{tm-book} & Yes
   \end{tabularx}
\caption{Complexity gap between blocking and non-blocking TMs}\label{fig:main2}    
}
\end{figure}
As a synchronization abstraction, TM came as an alternative to conventional lock-based
synchronization, and it therefore appears natural that early TM
implementations~\cite{HLM+03,astm,nztm,fraser}, 
avoided using locks. 
Instead, early TM designs relied on non-blocking
synchronization, where 
a prematurely halted transaction cannot prevent
all other transactions from committing. 
Possibly the weakest
progress condition elucidating non-blocking TM-progress is obstruction-freedom.

However, in $2005$, Ennals~\cite{Ennals05} argued
that obstruction-free TMs inherently yield poor performance, because
they require  transactions to forcefully abort each other. 
Ennals further described a \emph{lock-based} TM
implementation~\cite{Ennals-code} satisfying progressiveness
that he claimed to outperform \emph{DSTM}~\cite{HLM+03},
the most referenced obstruction-free TM implementation at the time.
Inspired by~\cite{Ennals05}, more recent lock-based progressive TMs, such as \emph{TL}~\cite{DStransaction06}, 
\emph{TL2}~\cite{DSS06} and \emph{NOrec}~\cite{norec},
demonstrate better performance than obstruction-free TMs on most workloads. 

There is a considerable amount of empirical evidence on the
performance gap between non-blocking
(obstruction-free) and blocking (progressive) TM implementations but
no analytical result explains it.
We present complexity lower and upper bounds that provide
such an explanation.

To exhibit a complexity gap between blocking and non-blocking TMs, we go back to the 
the progressive opaque TM implementation $LP$ (Algorithm~\ref{alg:ic}) that beats the impossibility result
and the lower bounds we established for obstruction-free TMs.     
Recall that our implementation $LP$, (1)~uses only read-write base objects and
provides wait-free TM-liveness, (2)~ensures strict DAP, (3)~has invisible reads, 
(4)~performs $O(1)$ non-overlapping RAWs/AWARs per transaction, and
(5)~incurs $O(1)$ memory stalls
for read operations (Theorem~\ref{th:ic}).
In contrast, from prior work and our lower bounds we know that
(i) no OF TM that provides wait-free transactional operations can be implemented
using only read-write base objects~\cite{tm-book}; 
(ii) no OF TM can provide strict DAP~\cite{OFTM};
(iii)~no weak DAP OF TM has invisible reads
(Section~\ref{sec:p3c3ss1}) and (iv)~no OF TM ensures a constant
number of stalls incurred by a t-read operation (Section~\ref{sec:p3c3ss2}). 
Finally, (v)~no RW DAP \emph{opaque} OF TM has constant  RAW/AWAR complexity
(Section~\ref{sec:p3c3ss3}).
In fact, (iv) and (v) exhibit a linear separation between blocking and non-blocking TMs w.r.t 
expensive synchronization and memory stall complexity, respectively.

Altogether, our results exhibit a considerable complexity gap between progressive and obstruction-free TMs, as summarized
in Figure~\ref{fig:main2}, that seems to
justify the shift in TM practice (circa. $2005$) from non-blocking to blocking TMs.

Overcoming our lower bounds for obstruction-free TMs individually is
comparatively easy. 
Say, TL~\cite{DStransaction06} combines strict DAP with invisible
reads, but it is not read-write, and it does not provide
constant RAW/AWAR and stall complexities. 

Coming out with a single algorithm that beats all these lower bounds is
quite nontrivial. Our algorithm $LP$ incurs the cost of \emph{incremental
  validation}, \emph{i.e.}, checking that the current read set has
not changed per every new read operation. 
This is, however, unavoidable for invisible read algorithms (cf. Theorem~\ref{th:iclb}), and
is, in fact,  believed to yield better performance in practice than
``visible'' reads~\cite{Ennals-code,DStransaction06, norec}, and we show that
it enables constant stall and RAW/AWAR complexity.     
%
%
\section{Related work and Discussion}
\label{sec:p3c3disc}
In this section, we summarize the results presented in this chapter and identify some unresolved questions.

\vspace{1mm}\noindent\textbf{Lower bounds for non-blocking TMs.}
Complexity of obstruction-free TMs was first studied by Guerraoui and Kapalka~\cite{OFTM,tm-book} who proved
that they cannot provide strict DAP. However, as we show in Section~\ref{sec:oftmalgos}, 
it is possible to realize weaker than strict DAP variants
of obstruction-free opaque TMs.
Bushkov et al.~\cite{BDFG14} improved on the impossibility result in \cite{OFTM} and showed that a variant of strict DAP
cannot be combined with obstruction-free TM-progress, even if a weaker (than strictly serializability) TM-correctness 
property is assumed.
In the thesis, we do not consider relaxations of strict serializability.

Guerraoui and Kapalka~\cite{OFTM,tm-book} also proved that a strict serializable TM that provides OF TM-progress and
wait-free TM-liveness cannot be implemented using only read and write primitives.
An interesting open question is whether we can implement strict serializable TMs in $\mathcal{OF}$ 
using only read and write primitives.

Observe that, since there are at most $n$ concurrent transactions,
we cannot do better than $(n-1)$ stalls (cf. Definition~\ref{def:stalls}).
Thus, the lower bound of Theorem~\ref{th:oftmstalls} is tight. 

Moreover, we conjecture that the linear (in $n$) lower bound of Theorem~\ref{th:oftriv} for RW DAP opaque obstruction-free TMs 
can be strengthened to be linear in the size of the transaction's read set.  
Then, Algorithm~\ref{alg:oftm}, which proves a linear upper bound in the size of 
the transaction's read set, would allow us to
establish a linear tight bound (in the size of the transaction's read set) for RW DAP opaque obstruction-free TMs.

\vspace{1mm}\noindent\textbf{Blocking versus non-blocking TMs.}
As highlighted in \cite{Ennals05,DStransaction06}, 
obstruction-free TMs typically must forcefully abort pending conflicting transactions.
This observation inspires the impossibility of invisible reads (Theorem~\ref{th:ir}). 
Typically, to detect the presence of a conflicting transaction and
abort it, the reading transaction must employ a RAW or a read-modify-write primitive like \emph{compare-and-swap},
motivating the linear lower bound on expensive synchronization
(Theorem~\ref{th:oftriv}).
Also, in obstruction-free TMs, a transaction may not wait for a
concurrent inactive transaction to complete and, as a result, we may
have an execution in which a transaction incurs a distinct stall due to a
transaction run by each other process, hence the linear stall
complexity (Theorem~\ref{th:oftmstalls}).    
Intuitively, since transactions in progressive TMs may abort themselves in case
of conflicts, they can employ invisible reads and maintain
constant stall and RAW/AWAR complexities.

The lower bound and the proof technique in Theorem~\ref{th:oftmstalls} 
is inspired by an analogous lower bound on \emph{linearizable}
\emph{solo-terminating} implementations~\cite{G05,AGHK09}
of a wide class of ``perturbable'' objects that include
\emph{counters}, \emph{compare-and-swap} and \emph{single-writer snapshots}~\cite{G05,AGHK09}.
Informally,
the definition of solo-termination (adapted to the TM context) says that for every finite execution $E$, 
and every transaction $T$ that is t-incomplete in $E$, there is a finite step
contention-free extension in which $T$ eventually commits. Observe that, under this definition, $T$
is guaranteed to commit even in some executions that are not step contention-free for $T$. However,
the definition of OF TM-progress used in the thesis ensures that $T$ is guaranteed to commit only if all its events
are issued in the absence of step contention.
Moreover, \cite{AGHK09} described a single-lock 
(only the process holding the lock can invoke an operation)
implementation of these objects
that incurs $O(\log n)$ stalls, thus establishing a separation between the worst-case operation 
stall complexity of non-blocking
and blocking (\emph{i.e.}, lock-based) implementations of these objects.
In this chapter, we presented a linear separation in memory stall complexity between obstruction-free TMs and
lock-based TMs characterized by progressiveness, which is a strictly stronger (than single-lock) progress guarantee,
thus establishing the inherent cost of non-blocking progress in the TM context.

Some benefits of obstruction-free TMs, namely their ability to make progress even if some
transactions prematurely fail, are not provided by progressive TMs. 
However, several papers~\cite{DStransaction06,DSS06,Ennals05} argued that lock-based TMs tend to outperform
obstruction-free ones by allowing for simpler algorithms with lower
overhead,
and their inherent progress issues may be resolved using timeouts and \emph{contention-managers}~\cite{dstm-contention}.
%
This chapter explains the empirically observed performance gap between
blocking and non-blocking TMs via a series of lower bounds on obstruction-free TMs and a progressive TM
algorithm that beats all of them.
%
\chapter{Lower bounds for partially non-blocking TMs}
\label{ch:p3c4}
\section{Overview}
\label{sec:p3c4s0}
It is easy to see that \emph{dynamic} TMs where
the patterns in which transactions access t-objects 
are not known in advance do not allow for
\emph{wait-free} TMs~\cite{tm-book}, \emph{i.e.}, every transaction must commit in a finite number of steps of the process
executing it,  regardless of the behavior of concurrent
processes.
Suppose that a transaction
$T_1$ reads t-object $X$, then a concurrent transaction $T_2$ reads
t-object $Y$, writes to $X$ and commits, and finally $T_2$ writes to
$Y$. 
Since $T_1$ has read the ``old'' value in $X$ and $T_2$ has read the
``old'' value  in $Y$, there is no way to commit $T_1$ and order the
two transactions in a sequential execution.   
As this scenario can be repeated arbitrarily often, even the weaker guarantee of \emph{local progress} that only requires that
each transaction \emph{eventually} commits if repeated
sufficiently often, cannot be ensured by \emph{any} strictly
serializable TM implementation,
regardless of the base objects it
uses~\cite{bushkov2012}.\footnote{Note that the counter-example would
  not work if we imagine that 
  the data sets accessed by a transaction can be known in
  advance. However, in the thesis, we consider the conventional
  dynamic TM programming model.}              
 
But can we ensure that at least \emph{some} transactions commit wait-free and what are the inherent costs?
%
It is often argued that many realistic workloads are
\emph{read-dominated}: the proportion of read-only transactions is
higher than that of updating ones, or read-only transactions have much
larger data sets than updating ones~\cite{stmbench7, Attiya09-tmread}. 
Therefore, it seems natural to require that read-only transactions
commit wait-free. Since we are interested in complexity lower bounds, we require that
updating transaction provide only sequential TM-progress. 

First, we focus on strictly serializable TMs with the above TM-progress conditions 
that use invisible reads.
We show that this requirement results in maintaining unbounded sets of versions
for every data item, \emph{i.e.}, such implementations 
may not be practical due to their space complexity. 
Secondly, we prove that strictly serializable TMs with these progress conditions
cannot ensure strict DAP. 
Thus, two transactions that access mutually disjoint data sets may prevent each other from committing.
Thirdly, for weak DAP TMs, we show that a read-only transaction (with
an arbitrarily large read
set) must sometimes perform at least one
expensive synchronization pattern~\cite{AGK11-popl} per t-read
operation, \emph{i.e.},
the expensive synchronization complexity of a read-only transaction is linear in
the size of its data set. 

Formally, we denote by $\mathcal{RWF}$ the class of partially non-blocking TMs.
\begin{definition}
\label{def:rwf} 
(The class $\mathcal{RWF}$) 
A TM implementation $M \in \mathcal{RWF}$ \emph{iff}
in its every execution:
\begin{itemize}
\item 
(\emph{wait-free TM-progress for read-only transactions}) 
every read-only transaction 
commits
in a finite number of its steps, and
\item 
(\emph{sequential TM-progress and sequential TM-liveness for updating transactions}) 
\emph{i.e.}, every transaction running step contention-free from a t-quiescent
configuration, commits in a finite number of its steps.
\end{itemize}
\end{definition}
\vspace{1mm}\noindent\textbf{Roadmap of Chapter~\ref{ch:p3c4}.}
Section~\ref{sec:p3c4s1} presents a lower bound on the inherent space complexity of TMs in $\mathcal{RWF}$.
Section~\ref{sec:p3c4s2} proves the impossibility of strict DAP TMs in $\mathcal{RWF}$
while in Section~\ref{sec:p3c4s3}, assuming weak DAP, we prove a linear, in the size of the transaction's read set, lower bound
on expensive synchronization complexity.
We conclude this chapter with a discussion of the related work and open questions concerning TMs in $\mathcal{RWF}$.
\section{The space complexity of invisible reads}
\label{sec:p3c4s1}
We prove that every strictly serializable TM implementation $M\in \mathcal{RWF}$ that uses invisible reads 
must keep unbounded sets of values for every t-object.
To do so, for every $c \in \mathbb{N}$, we construct an execution of $M$ that \emph{maintains at least $c$ distinct values
for every t-object}. We require the following technical definition:
\begin{figure*}[t]
\begin{center}
	\subfloat[for all $i\in \{1,\ldots , c-1 \}$, $T_{2i-1}$ writes $v_{i_{\ell}}$ to each $X_{\ell}$; $\Read_{2i}(X_1)$ must return $v_{i_{1}}$ \label{sfig:res-0}]{\scalebox{0.7}[0.7]{\begin{tikzpicture}
\node (r1) at (1.4,0) [] {};
\node (r2) at (8.4,-1) [] {};
\node (ri) at (14.9,-3) [] {};

\node (w1) at (4.5,-1) [] {};
\node (wi) at (11,-3) [] {};

\draw (r1) node [above] {\normalsize {$R_0(X_1) \rightarrow v_{0_{1}}$}};
\draw (r2) node [above] {\normalsize {$R_2(X_1) \rightarrow v_{1_{1}}$}};
\draw (ri) node [above] {\normalsize {$R_{2i}(X_1) \rightarrow v_{i_{1}}$}};

\draw (w1) node [above] {\normalsize {$\forall X_{\ell}\in \mathcal{X}$: write $v_{1_{\ell}}$}}; 
\draw (w1) node [below] {\normalsize {$T_1$ commits}};

\draw (wi) node [above] {\normalsize {$\forall X_{\ell}\in \mathcal{X}$: write $v_{i_{\ell}}$}}; 
\draw (wi) node [below] {\normalsize {$T_{2i-1}$ commits}};

\begin{scope}   
\draw [|-,dotted] (0,0) node[left] {$\ms{Phase}~ 0|$ $T_0$} to (20,0);
\draw [|-|,thick] (0,0) node[left] {} to (2,0);
\end{scope}
\begin{scope}   
\draw [|-,dotted] (7,-1) node[left] {  $T_2$} to (20,-1);
\draw [|-|,thick] (3,-1) node[left] {$\ms{Phase}~ 1|$ ${T}_{1}$} to (6,-1);
\draw [|-|,thick] (7,-1) node[left] {} to (9,-1);
\end{scope}
\begin{scope}   
\draw [-,dotted] (10,-2) node[left] {} to (20,-2);
\end{scope}
\begin{scope}   
\draw [|-,dotted] (13.5,-3) node[left] {} to (20,-3);
\draw [|-|,thick] (9.5,-3) node[left] { $\ms{Phase}~ i|$ ${T}_{2i-1}$} to (12.5,-3);
\draw [|-|,thick] (13.5,-3) node[left] {$T_{2i}$} to (15.5,-3);
\end{scope}
\begin{scope}   
\draw [] (16,-4) node[above] {$\downarrow$} to (16,-4);
\draw [-,dotted] (15.5,-5) node[above] {extend to $c-1$ phases} to (18,-5);
\end{scope}
\end{tikzpicture}}}
	\\
	\vspace{2mm}
	\subfloat[extend every read-only transaction $T_{2i}$ in phase $i$ with t-reads of $X_2,\ldots X_{\ell}, \ldots $; each $\Read_{2i}(X_{\ell})$ must return $v_{i_{\ell}}$ \label{sfig:res-1}]{\scalebox{0.7}[0.7]{\begin{tikzpicture}
\node (r1) at (1.2,0) [] {};
\node (r2) at (8.4,-1) [] {};
\node (ri) at (14.9,-3) [] {};
\node (rf1) at (19.5,-0) [] {};
\node (rf2) at (19.5,-1) [] {};
\node (rfi) at (19.5,-3) [] {};

\node (w1) at (4.5,-1) [] {};
\node (wi) at (11,-3) [] {};

\draw (r1) node [above] {\normalsize {$R_0(X_1) \rightarrow v_{0_{1}}$}};
\draw (r2) node [above] {\normalsize {$R_2(X_1) \rightarrow v_{1_{1}}$}};
\draw (ri) node [above] {\normalsize {$R_{2i}(X_1) \rightarrow v_{i_{1}}$}};

\draw (rf1) node [above] {\normalsize {$R_0(X_2) \rightarrow v_{0_{2}} \cdots R_0(X_\ell)\rightarrow v_{0_{\ell}} \cdots$}};
\draw (rf2) node [above] {\normalsize {$R_2(X_2) \rightarrow v_{1_{2}} \cdots R_2(X_\ell)\rightarrow v_{1_{\ell}} \cdots$}};
\draw (rfi) node [above] {\normalsize {$R_{2i}(X_2) \rightarrow v_{i_{2}} \cdots R_{2i}(X_\ell)\rightarrow v_{i_{\ell}} \cdots$}};

\draw (w1) node [above] {\normalsize {$\forall X_{\ell} \in \mathcal{X}$: write $v_{1_{\ell}}$}}; 
\draw (w1) node [below] {\normalsize {$T_1$ commits}};

\draw (wi) node [above] {\normalsize {$\forall X_{\ell} \in \mathcal{X}$: write $v_{i_{\ell}}$}}; 
\draw (wi) node [below] {\normalsize {$T_{2i-1}$ commits}};

\begin{scope}   
\draw [|-,dotted] (0,0) node[left] {$T_0$} to (22,0);
\draw [|-|,thick] (0,0) node[left] {} to (2,0);
\draw [|-|,thick] (17,0) node[left] {} to (22,0);
\end{scope}
\begin{scope}   
\draw [|-,dotted] (7,-1) node[left] {$T_2$} to (22,-1);
\draw [|-|,thick] (3,-1) node[left] {${T}_{1}$} to (6,-1);
\draw [|-|,thick] (7,-1) node[left] {} to (9,-1);
\draw [|-|,thick] (17,-1) node[left] {} to (22,-1);
\end{scope}
\begin{scope}   
\draw [-,dotted] (10,-2) node[left] {} to (22,-2);
\end{scope}
\begin{scope}   
\draw [|-,dotted] (13.5,-3) node[left] {} to (22,-3);
\draw [|-|,thick] (9.5,-3) node[left] {${T}_{2i-1}$} to (12.5,-3);
\draw [|-|,thick] (13.5,-3) node[left] {$T_{2i}$} to (15.5,-3);
\draw [|-|,thick] (17,-3) node[left] {} to (22,-3);
\end{scope}
\begin{scope}   
\draw [] (16,-4) node[above] {$\downarrow$} to (16,-4);
\draw [-,dotted] (15.5,-5) node[above] {extend to $c-1$ phases} to (18,-5);
\end{scope}
\end{tikzpicture}}}
        
	\caption{Executions in the proof of Theorem~\ref{th:inv}; execution in \ref{sfig:res-0} must maintain $c$ distinct
	values of every t-object
        \label{fig:invdap}} 
\end{center}
\end{figure*}
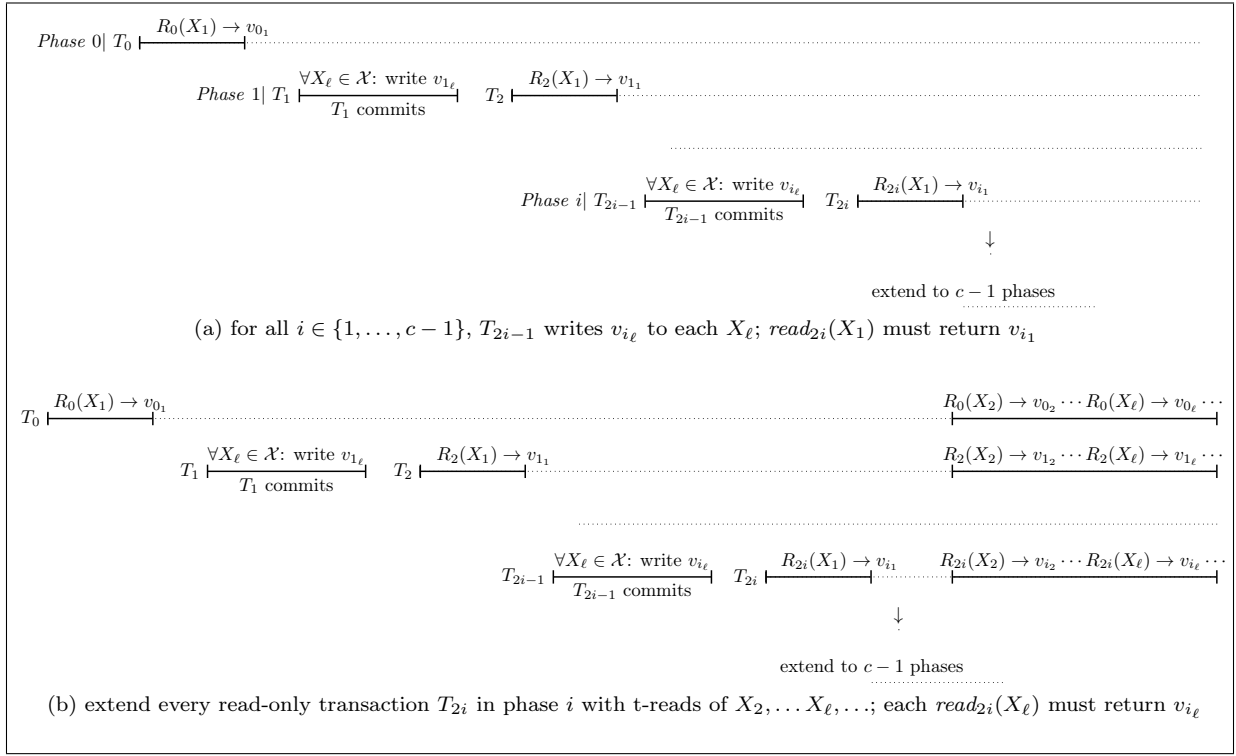
\begin{definition}
Let $E$ be any execution of a TM implementation $M$.
We say that $E$ \emph{maintains $c$ distinct values $\{v_1,\ldots , v_c\}$ of t-object $X$},
if there exists an execution $E\cdot E'$ of $M$ such that
\begin{itemize}
\item
$E'$ contains the complete executions of $c$ t-reads of $X$ and,
\item
for all $i\in \{1,\ldots , c\}$, the response of the $i^{th}$ t-read
of $X$ in $E'$ is $v_i$. 
\end{itemize}
\end{definition}
\begin{theorem}
\label{th:inv}
Let $M$ be any strictly serializable TM implementation in $\mathcal{RWF}$ that uses invisible reads, and
$\mathcal{X}$, any set of t-objects.
Then, for every $c \in \mathbb{N}$, there exists an execution $E$ of $M$
such that $E$ maintains at least $c$ distinct values of each t-object $X\in \mathcal{X}$.
\end{theorem}
\begin{proof}
Let $v_{0_{\ell}}$ be the initial value of t-object $X_{\ell}\in \mathcal{X}$.
For every $c\in \mathbb{N}$, we iteratively construct an execution $E$ of $M$ of the form depicted in Figure~\ref{sfig:res-0}.
The construction of $E$ proceeds in phases: there are at most $c-1$ phases. 
For all $i\in \{0,\ldots c-1\}$, we denote the execution after phase $i$ as $E_i$ which
is defined as follows:
\begin{itemize}
\item
$E_0$ is the complete step contention-free execution fragment $\alpha_0$ of read-only transaction $T_0$ that performs
$\Read_0(X_1)\rightarrow v_{0_{1}}$
\item
for all $i \in \{1,\ldots ,c-1\}$, $E_i$ is defined to be an execution of the form 
$\alpha_0\cdot \rho_1\cdot \alpha_1 \cdots \rho_i\cdot \alpha_i$ such that for all $j\in \{1,\ldots, i\}$,
\begin{itemize}
\item
$\rho_j$ is the t-complete step contention-free execution fragment of an updating transaction ${T}_{2j-1}$ that,
for all $X_{\ell} \in \mathcal{X}$ writes the value $v_{j_{\ell}}$ and commits
\item
$\alpha_j$ is the complete step contention-free execution fragment of a read-only transaction $T_{2j}$ that performs
$\Read_{2j}(X_1)\rightarrow v_{j_{1}}$
\end{itemize}
\end{itemize}
%
Since read-only transactions are invisible, for all $i\in \{0,\ldots, c-1\}$, the execution fragment $\alpha_i$ does not
contain any nontrivial events.
Consequently, for all $i<j\leq c-1$, the configuration after $E_i$ is indistinguishable to transaction ${T}_{2j-1}$
from a t-quiescent configuration and it must be committed in $\rho_{j}$ (by sequential progress for updating transactions).
Observe that, for all $1\leq j < i$, ${T}_{2j-1} \prec_{E}^{RT} {T}_{2i-1}$.
Strict serializability of $M$ now stipulates that, for all $i \in \{1,\ldots, c-1\}$, 
the t-read of $X_1$ performed by transaction $T_{2i}$ in the 
execution fragment $\alpha_{i}$ must return the value $v_{i_{1}}$ of $X_1$ as written by transaction ${T}_{2i-1}$ in
the execution fragment $\rho_i$ (in any serialization, 
${T}_{2i-1}$ is the latest committed transaction writing to $X_1$ that precedes $T_{2i}$). 
Thus, $M$ indeed has an execution $E$ of the form depicted in Figure~\ref{sfig:res-0}.

Consider the execution fragment $E'$ that extends $E$ in which, for all $i\in \{0,\ldots ,c-1 \}$, read-only transaction $T_{2i}$ 
is extended with the complete execution of the t-reads of every t-object 
$X_{\ell} \in \mathcal{X}\setminus \{X_1\}$ (depicted in Figure~\ref{sfig:res-1}).

We claim that, for all $i\in \{0,\ldots ,c-1 \}$, and for all $X_{\ell} \in \mathcal{X}\setminus \{X_1\}$,
$\Read_{2i}(X_{\ell})$ performed by transaction $T_{2i}$ must return the value $v_{i_{\ell}}$ of $X_{\ell}$ 
written by transaction ${T}_{2i-1}$ in the execution fragment $\rho_i$.
Indeed, by wait-free progress, $\Read_i(X_{\ell})$ must return a non-abort response in such an extension of $E$.
Suppose by contradiction that $\Read_i(X_{\ell})$ returns a response that is not $v_{i_{\ell}}$.
There are two cases: 
\begin{itemize}
\item
$\Read_{2i}(X_{\ell})$ returns the value $v_{j_{\ell}}$ written by transaction ${T}_{2j-1}$; $j<i$.
However, since for all $j < i$, ${T}_{2j} \prec_{E}^{RT} {T}_{2i}$, the execution is not strictly serializable---contradiction.
\item
$\Read_{2i}(X_{\ell})$ returns the value $v_{j_{\ell}}$ written by transaction ${T}_{2j}$; $j>i$.
Since $\Read_i(X_1)$ returns the value $v_{i_{1}}$ and ${T}_{2i} \prec_{E}^{RT} {T}_{2j}$,
there exists no such serialization---contradiction.
\end{itemize}
Thus, $E$ maintains at least $c$ distinct values of every t-object $X\in \mathcal{X}$.
\end{proof}
%
\section{Impossibility of strict DAP}
\label{sec:p3c4s2}
In this section, we prove that it is impossible to derive strictly serializable TM implementations in $\mathcal{RWF}$ which ensure
that any two transactions accessing pairwise disjoint data sets can execute without contending on the same base object.
\begin{figure*}[t]
\begin{center}
	\subfloat[By strict DAP, $T_0$ and $T_3$ do not contend on any base object \label{sfig:dap-0}]{\scalebox{0.7}[0.7]{\begin{tikzpicture}
\node (r1) at (3,0) [] {};

\node (e) at (5.5,-1) [] {};

\node (w1) at (-2,-1) [] {};
\node (w2) at (0,-1) [] {};
\node (c) at (1.5,-1) [] {};

\node (r3) at (9,-2) [] {};

\draw (r1) node [above] {\normalsize{$R_0(X_1) \rightarrow v$}};

\draw (e) node [above] {\normalsize {(event of $T_1$)}};

\draw (w1) node [above] { \normalsize{$W_1(X_1,nv)$}};
\draw (w2) node [above] { \normalsize{$W_1(X_3,nv)$}};
\draw (c) node [above] { \normalsize{$\TryC_1$}};

\draw (r3) node [above] { \normalsize{$R_3(X_3) \rightarrow nv$}};
\draw (r3) node [below] { \normalsize{$T_3$ commits}};

\begin{scope}   
\draw [-,dotted] (2,0) node[left] {$T_0$} to (18,0);
\draw [|-|,thick] (2,0) node[left] {} to (4,0);
\end{scope}
\begin{scope}   
\draw [-,thick] (-3,-1) node[left] {$T_1$} to (1.5,-1);
\draw [|-|,thick] (-3,-1) node[left] {} to (-1,-1);
\draw [|-|,thick] (-1,-1) node[left] {} to (1,-1);
\draw [|-,dotted] (1,-1) node[left] {} to (1.5,-1);
\draw  (5.5,-1) circle [fill, radius=0.05]  (5.5,-1);
\draw [-,dotted] (-3,-1) to  (5,-1);
\end{scope}
\begin{scope}   
\draw [|-|,thick] (8,-2) node[left] {$T_3$} to (10,-2);
\end{scope}
\end{tikzpicture}}}
 	\\ 
	\vspace{2mm}
	\subfloat[$\Read_0(X_2)$ must return $nv$ \label{sfig:dap-1}]{\scalebox{0.7}[0.7]{\begin{tikzpicture}
\node (r1) at (3,0) [] {};
\node (r2) at (13,0) [] {};

\node (e) at (5.5,-1) [] {};

\node (w1) at (-2,-1) [] {};
\node (w2) at (0,-1) [] {};
\node (c) at (1.5,-1) [] {};

\node (r3) at (9,-2) [] {};
\node (w3) at (-6,-2) [] {};

\draw (r1) node [above] {\normalsize{$R_0(X_1) \rightarrow v$}};
\draw (r2) node [above] {\normalsize{$R_0(X_2)\rightarrow nv$}};

\draw (e) node [above] {\normalsize{(event of $T_1$)}};

\draw (w1) node [above] { \normalsize{$W_1(X_1,nv)$}};
\draw (w2) node [above] { \normalsize{$W_1(X_3,nv)$}};
\draw (c) node [above] { \normalsize{$\TryC_1$}};

\draw (r3) node [above] { \normalsize{$R_3(X_3) \rightarrow nv$}};
\draw (r3) node [below] { \normalsize{$T_3$ commits}};

\draw (w3) node [above] { \normalsize{$W_2(X_2,nv)$}};
\draw (w3) node [below] { \normalsize{$T_2$ commits}};

\begin{scope}   
\draw [-,dotted] (2,0) node[left] {$T_0$} to (14,0);
\draw [|-|,thick] (2,0) node[left] {} to (4,0);
\draw [|-|,thick] (12,0) node[left] {} to (14,0);
\end{scope}
\begin{scope}   
\draw [-,thick] (-3,-1) node[left] {$T_1$} to (1.5,-1);
\draw [|-|,thick] (-3,-1) node[left] {} to (-1,-1);
\draw [|-|,thick] (-1,-1) node[left] {} to (1,-1);
\draw [|-,dotted] (1,-1) node[left] {} to (1.5,-1);
\draw  (5.5,-1) circle [fill, radius=0.05]  (5.5,-1);
\draw [-,dotted] (-3,-1) to  (5,-1);
\end{scope}
\begin{scope}   
\draw [|-|,thick] (8,-2) node[left] {$T_3$} to (10,-2);
\draw [|-|,thick] (-7,-2) node[left] {$T_2$} to (-5,-2);
\end{scope}
\end{tikzpicture}}}
        \\ 
        \vspace{2mm}
	\subfloat[By strict DAP, $T_0$ cannot distinguish this execution from the execution in \ref{sfig:dap-1} \label{sfig:dap-2}]{\scalebox{0.7}[0.7]{\begin{tikzpicture}
\node (r1) at (3,0) [] {};
\node (r2) at (17,0) [] {};

\node (e) at (5.5,-1) [] {};

\node (w1) at (-2,-1) [] {};
\node (w2) at (0,-1) [] {};
\node (c) at (1.5,-1) [] {};

\node (r3) at (9,-2) [] {};
\node (w3) at (13,-2) [] {};

\draw (r1) node [above] {\normalsize{$R_0(X_1) \rightarrow v$}};
\draw (r2) node [above] {\normalsize{$R_0(X_2)\rightarrow nv$}};

\draw (e) node [above] {\normalsize {(event of $T_1$)}};

\draw (w1) node [above] { \normalsize{$W_1(X_1,nv)$}};
\draw (w2) node [above] { \normalsize{$W_1(X_3,nv)$}};
\draw (c) node [above] { \normalsize{$\TryC_1$}};

\draw (r3) node [above] { \normalsize{$R_3(X_3) \rightarrow nv$}};
\draw (r3) node [below] { \normalsize{$T_3$ commits}};

\draw (w3) node [above] { \normalsize{$W_2(X_2,nv)$}};
\draw (w3) node [below] { \normalsize{$T_2$ commits}};

\begin{scope}   
\draw [-,dotted] (2,0) node[left] {$T_0$} to (18,0);
\draw [|-|,thick] (2,0) node[left] {} to (4,0);
\draw [|-|,thick] (16,0) node[left] {} to (18,0);
\end{scope}
\begin{scope}   
\draw [-,thick] (-3,-1) node[left] {$T_1$} to (1.5,-1);
\draw [|-|,thick] (-3,-1) node[left] {} to (-1,-1);
\draw [|-|,thick] (-1,-1) node[left] {} to (1,-1);
\draw [|-,dotted] (1,-1) node[left] {} to (1.5,-1);
\draw  (5.5,-1) circle [fill, radius=0.05]  (5.5,-1);
\draw [-,dotted] (-3,-1) to  (5,-1);
\end{scope}
\begin{scope}   
\draw [|-|,thick] (8,-2) node[left] {$T_3$} to (10,-2);
\draw [|-|,thick] (12,-2) node[left] {$T_2$} to (14,-2);
\end{scope}
\end{tikzpicture}}}
	
	\caption{Executions in the proof of Theorem~\ref{th:lpdap}; execution in \ref{sfig:dap-2} is not strictly serializable
        \label{fig:indisdap}} 
\end{center}
\end{figure*}
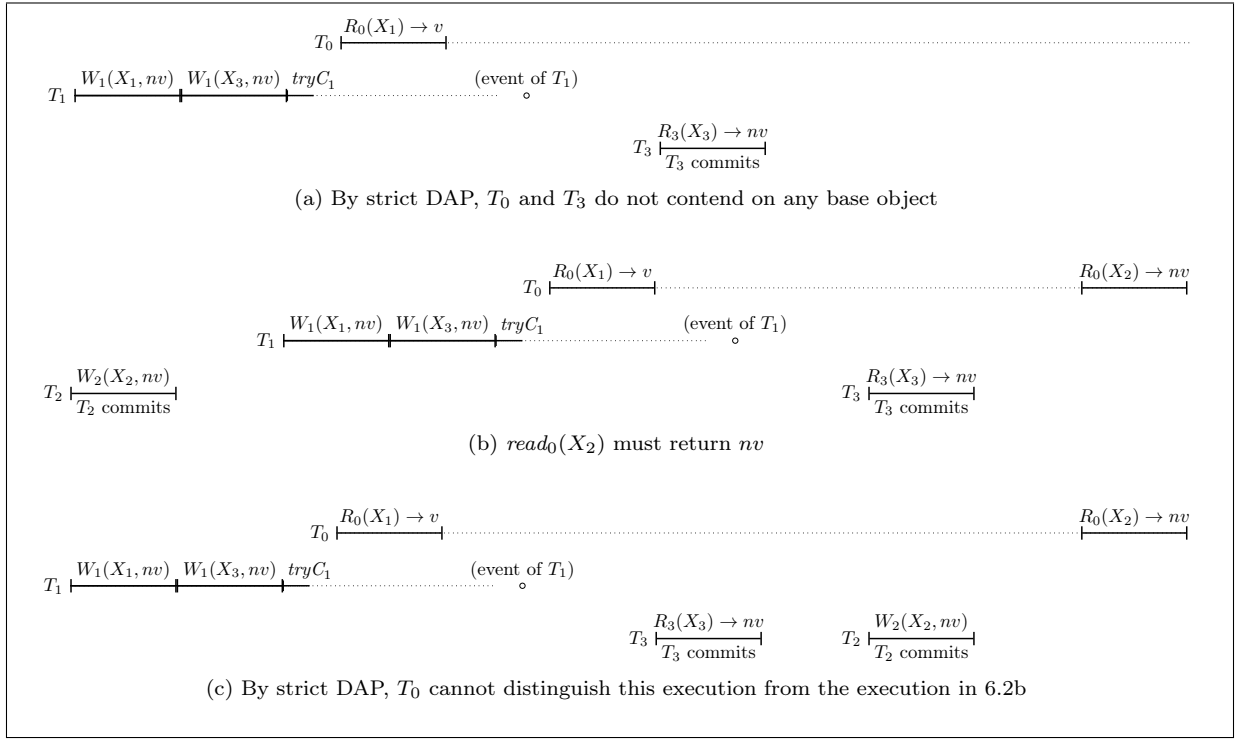
\begin{theorem}
\label{th:lpdap}
There exists no strictly serializable strict DAP TM implementation in $\mathcal{RWF}$.
\end{theorem}
\begin{proof}
Suppose by contradiction that there exists a strict DAP TM implementation $M\in \mathcal{RWF}$.

Let $v$ be the initial value of t-objects $X_1$, $X_2$ and $X_3$.
Let $\pi$ be the t-complete step contention-free execution of transaction $T_1$ that
writes the value $nv \neq v$ to t-objects $X_1$ and $X_3$.
By sequential progress for updating transactions, $T_1$ must be committed in $\pi$.

Note that any read-only transaction that runs step contention-free after some prefix of $\pi$ must return a non-abort value.
Since any such transaction reading $X_1$ or $X_3$ must return $v$ after the empty prefix of $\pi$ and $nv$ when it starts from $\pi$,
there exists $\pi'$, the longest prefix of $\pi$ that cannot be extended with the
t-complete step contention-free execution of any transaction that performs a t-read of $X_1$ and
returns $nv$ nor with the t-complete step contention-free execution of any transaction that performs a t-read of $X_3$
and returns $nv$.

Consider the execution fragment $\pi'\cdot \alpha_1$, where $\alpha_1$ is
the complete step contention-free execution of transaction $T_0$ that performs $\Read_0(X_1)\rightarrow v$.
Indeed, by definition of $\pi'$ and wait-free progress (assumed for read-only transactions),
$M$ has an execution of the form $\pi'\cdot \alpha_1$.

Let $e$ be the enabled event of transaction $T_1$ in the configuration after $\pi'$.
Without loss of generality, assume that $\pi'\cdot e$ can be extended with the t-complete
step contention-free execution of a transaction that reads $X_3$ and returns $nv$.

We now prove that $M$ has an execution of the form $\pi'\cdot \alpha_1 \cdot e \cdot \beta \cdot \gamma$, where
\begin{itemize}
\item
$\beta$ is the t-complete step contention-free execution fragment of transaction $T_3$ that performs $\Read_3(X_3)\rightarrow nv$
and commits
\item
$\gamma$ is the t-complete step contention-free execution fragment of transaction $T_2$ that
writes $nv$ to $X_2$ and commits.
\end{itemize}
Observe that, by definition of $\pi'$, $M$ has an execution of the form $\pi' \cdot e \cdot \beta$.
By construction, transaction $T_1$ applies a nontrivial primitive to a base object, say $b$ in the event $e$
that is accessed by transaction $T_3$ in the execution fragment $\beta$.
Since transactions $T_0$ and $T_3$ access mutually disjoint data sets in $\pi'\cdot \alpha_1 \cdot e \cdot \beta$,
$T_3$ does not access any base object in $\beta$ to which transaction $T_0$ applies a nontrivial primtive
in the execution fragment $\alpha_1$ (assumption of strict DAP).
Thus, $\alpha_1$ does not contain a nontrivial primitive to $b$ and $\pi'\cdot \alpha_1 \cdot e \cdot \beta$
is indistinguishable to $T_3$ from the execution $\pi' \cdot e \cdot \beta$.
This proves that $M$ has an execution of the form $\pi'\cdot \alpha_1 \cdot e \cdot \beta$ (depicted in Figure~\ref{sfig:dap-0}).

Since transaction $T_2$ writes to t-object $\Dset(T_2)=X_2 \not\in \{\Dset(T_1)\cup \Dset(T_0)\cup \Dset(T_3)\}$,
by strict DAP,
the configuration after $\pi'\cdot \alpha_1 \cdot e \cdot \beta$ is indistinguishable to $T_2$
from a t-quiescent configuration. 
Indeed, transaction $T_2$ does not contend with
any of the transactions $T_1$, $T_0$ and $T_3$ on any base object in $\pi'\cdot \alpha_1 \cdot e \cdot \beta \cdot \gamma$.
Sequential progress of $M$ requires that
$T_2$ must be committed in $\pi'\cdot \alpha_1 \cdot e \cdot \beta \cdot \gamma$.
Thus, $M$ has an execution of the form $\pi'\cdot \alpha_1 \cdot e \cdot \beta \cdot \gamma$.

By the above arguments, the execution $\pi'\cdot \alpha_1 \cdot e \cdot \beta \cdot \gamma$
is indistinguishable to each of the transactions $T_1$, $T_0$, $T_2$ and $T_3$ from
$\gamma \cdot \pi'\cdot \alpha_1 \cdot e \cdot \beta$ in which transaction $T_2$ precedes $T_1$ in real-time ordering.
Thus, $\gamma \cdot \pi'\cdot \alpha_1 \cdot e \cdot \beta$ is also an execution of $M$.

Consider the extension of the execution $\gamma \cdot \pi'\cdot \alpha_1 \cdot e \cdot \beta$
in which transaction $T_0$ performs $\Read_0(X_2)$ and commits (depicted in Figure~\ref{sfig:dap-1}). 
Strict serializability of $M$ stipulates that $\Read_0(X_2)$ must return $nv$
since $T_2$ (which writes $nv$ to $X_2$ in $\gamma$) precedes $T_0$ in this execution.

Similarly, we now extend the execution $\pi'\cdot \alpha_1 \cdot e \cdot \beta \cdot \gamma$
with the complete step contention-free execution fragment of the t-read of $X_2$ by transaction $T_0$. 
Since $T_0$ is a read-only transaction,
it must be committed in this extension.
However, as proved above, this execution is indistinguishable to $T_0$ from the execution
depicted in Figure~\ref{sfig:dap-1} in which $\Read_0(X_2)$ must return $nv$.
Thus, $M$ has an execution of the form $\pi'\cdot \alpha_1 \cdot e \cdot \beta \cdot \gamma \cdot \alpha_2$,
where $T_0$ performs $\Read_0(X_2) \rightarrow nv$ in $\alpha_2$ and commits.

However, the execution $\pi'\cdot \alpha_1 \cdot e \cdot \beta \cdot \gamma \cdot \alpha_2$ (depicted in Figure~\ref{sfig:dap-2})
is not strictly serializable.
Transaction $T_1$ must be committed in any serialization and must precede transaction $T_3$
since $\Read_3(X_3)$ returns the value of $X_3$ written by $T_m$. However, transaction $T_0$ must
must precede $T_1$ since $\Read_0(X_1)$ returns the initial the value of $X_1$.
Also, transaction $T_2$ must precede $T_0$ since $\Read_0(X_2)$ returns the value of $X_2$ written by $T_2$.
But transaction $T_3$ must precede $T_2$ to respect real-time ordering of transactions.
Thus, $T_1$ must precede $T_0$ in any serialization. But there exists no such serialization: a contradiction to the
assumption that $M$ is strictly serializable.
\end{proof}
%
%
\section{A linear lower bound on expensive synchronization for weak DAP}
\label{sec:p3c4s3}
In this section, we prove a linear lower bound (in the size of the transaction's read set) on the number of RAWs or AWARs
for weak DAP TM implementations in $\mathcal{RWF}$.
To do so, we construct an execution in which each t-read
operation of an arbitrarily long read-only transaction contains a RAW or an AWAR. 
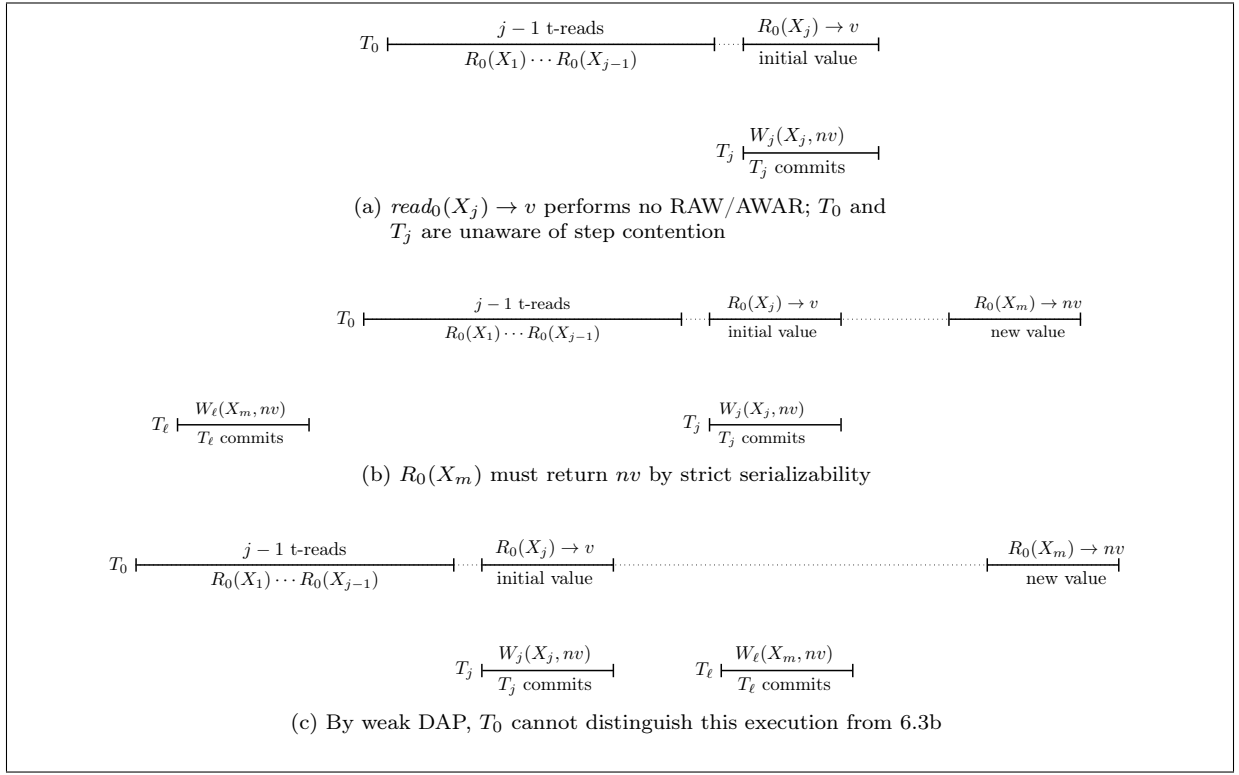
\begin{figure*}[t]
\begin{center}
	\subfloat[$\Read_0(X_j)\rightarrow v$ performs no RAW/AWAR; $T_0$ and $T_j$ are unaware of step contention \label{sfig:inv-1}]{\scalebox{0.72}[0.72]{\begin{tikzpicture}
\node (r1) at (3,0) [] {};
\node (r2) at (7.7,0) [] {};

\node (w1) at (7.5,-2) [] {};

\draw (r1) node [below] {\normalsize {$R_0(X_1) \cdots R_0(X_{j-1})$}};
\draw (r1) node [above] {\normalsize {$j-1$ t-reads}};

\draw (r2) node [above] {\normalsize {$R_0(X_j)\rightarrow v$}};
\draw (r2) node [below] {\normalsize {initial value}};

\draw (w1) node [above] {\normalsize {$W_j(X_j,nv)$}}; 
\draw (w1) node [below] {\normalsize {$T_j$ commits}};

\begin{scope}   
\draw [|-|,thick] (0,0) node[left] {$T_0$} to (6,0);
\draw [|-|,thick] (6.5,0) node[left] {} to (9,0);
\draw [-,dotted] (0,0) node[left] {} to (9,0);
\end{scope}
\begin{scope}   
\draw [|-|,thick] (6.5,-2) node[left] {$T_j$} to (9,-2);
\end{scope}
\end{tikzpicture}}}
        \\
        \vspace{2mm}
	\subfloat[$R_0(X_{m})$ must return $nv$ by strict serializability \label{sfig:inv-2}]{\scalebox{0.7}[0.7]{\begin{tikzpicture}
\node (r1) at (3,0) [] {};
\node (r2) at (7.7,0) [] {};
\node (r3) at (12.5,0) [] {};

\node (w1) at (7.5,-2) [] {};

\node (w2) at (-2.3,-2) [] {};

\draw (r1) node [below] {\small {$R_0(X_1) \cdots R_0(X_{j-1})$}};
\draw (r1) node [above] {\small {$j-1$ t-reads}};

\draw (r2) node [above] {\small {$R_0(X_j)\rightarrow v$}};
\draw (r2) node [below] {\small {initial value}};

\draw (w1) node [above] {\small {$W_j(X_j,nv)$}}; 
\draw (w1) node [below] {\small {$T_j$ commits}};

\draw (w2) node [above] {\small {$W_{\ell}(X_{m},nv)$}}; 
\draw (w2) node [below] {\small {$T_{\ell}$ commits}};

\draw (r3) node [above] {\small {$R_0(X_{m})\rightarrow nv$}};
\draw (r3) node [below] {\small {new value}};

\begin{scope}   
\draw [|-|,thick] (0,0) node[left] {$T_0$} to (6,0);
\draw [|-|,thick] (6.5,0) node[left] {} to (9,0);
\draw [|-|,thick] (-3.5,-2) node[left] {$T_{\ell}$} to (-1,-2);
\draw [|-|,thick] (11,0) node[left] {} to (13.5,0);
\draw [-,dotted] (0,0) to (13.5,0);
\end{scope}
\begin{scope}   
\draw [|-|,thick] (6.5,-2) node[left] {$T_j$} to (9,-2);
\end{scope}
\end{tikzpicture}}}
	\\ 
	\vspace{2mm}
	\subfloat[By weak DAP, $T_0$ cannot distinguish this execution from \ref{sfig:inv-2} \label{sfig:inv-3}]{\scalebox{0.7}[0.7]{\begin{tikzpicture}
\node (r1) at (3,0) [] {};
\node (r2) at (7.7,0) [] {};
\node (r3) at (17.5,0) [] {};

\node (w1) at (7.7,-2) [] {};

\node (w2) at (12.2,-2) [] {};

\draw (r1) node [below] {\normalsize {$R_0(X_1) \cdots R_0(X_{j-1})$}};
\draw (r1) node [above] {\normalsize {$j-1$ t-reads}};

\draw (r2) node [above] {\normalsize {$R_0(X_j)\rightarrow v$}};
\draw (r2) node [below] {\normalsize {initial value}};

\draw (w1) node [above] {\normalsize {$W_j(X_j,nv)$}}; 
\draw (w1) node [below] {\normalsize {$T_j$ commits}};

\draw (w2) node [above] {\normalsize {$W_{{\ell}}(X_{m},nv)$}}; 
\draw (w2) node [below] {\normalsize {$T_{\ell}$ commits}};

\draw (r3) node [above] {\normalsize {$R_0(X_{m})\rightarrow nv$}};
\draw (r3) node [below] {\normalsize {new value}};

\begin{scope}   
\draw [|-|,thick] (0,0) node[left] {$T_0$} to (6,0);
\draw [|-|,thick] (6.5,0) node[left] {} to (9,0);
\draw [|-|,thick] (11,-2) node[left] {$T_{\ell}$} to (13.5,-2);
\draw [|-|,thick] (16,0) node[left] {} to (18.5,0);
\draw [-,dotted] (0,0) to (18.5,0);
\end{scope}
\begin{scope}   
\draw [|-|,thick] (6.5,-2) node[left] {$T_j$} to (9,-2);
\end{scope}
\end{tikzpicture}}}
	\caption{Executions in the proof of Theorem~\ref{th:rwf}; execution in \ref{sfig:inv-3} is not strictly serializable
        \label{fig:indis}} 
\end{center}
\end{figure*}
\begin{theorem}
\label{th:rwf}
Every strictly serializable weakly DAP TM implementation $M \in \mathcal{RWF}$
has, for all $m\in \mathbb{N}$, an
execution in which some read-only transaction $T_0$ 
with $m=|\Rset(T_0)|$ performs $\Omega(m)$ RAWs/AWARs.
\end{theorem}
\begin{proof}
Let $v$ be the initial value of each of the t-objects $X_1,\ldots , X_m$.
Consider the t-complete step contention-free execution
of transaction $T_0$ that performs $m$ t-reads $\Read_0(X_1)$, $\Read_0(X_1)$,$\ldots \Read_0(X_m)$ and commits.
We prove that each of the first $m-1$ t-reads must perform a RAW or an AWAR.

For all $j\in \{1,\ldots , m-1\}$, $M$ has an execution of the form $\alpha_1\cdot \alpha_2 \cdots \alpha_j$, 
where for all $i \in \{1,\ldots , j\}$, $\alpha_i$ is the complete step contention-free execution fragment
of $\Read_0(X_j) \rightarrow v$.
Assume inductively that each of the first $j-1$ t-reads performs a RAW or an AWAR in this execution.
We prove that $\Read_0(X_j)$ must perform a RAW or an AWAR in the execution fragment $\alpha_j$.
Suppose by contradiction that $\alpha_j$ does not contain a RAW or an AWAR.

The following claim shows that we can schedule a committed transaction $T_j$ that writes a new value to $X_j$ 
concurrent to $\Read_0(X_j)$ such that
the execution is indistinguishable to both $T_0$ and $T_j$ from a 
step contention-free execution (depicted in Figure~\ref{sfig:inv-1}).
\begin{claim}
\label{cl:one}
For all $j\in \{1,\ldots , m-1\}$, $M$ has an execution of the form 
$\alpha_1 \cdots \alpha_{j-1} \cdot \alpha^1_j \cdot \delta_j \cdot \alpha^2_j$
where,
\begin{itemize}
\item
$\delta_{j}$ is the t-complete step contention-free execution fragment of transaction $T_{j}$ that
writes $nv \neq v$ and commits
\item
$\alpha^1_j\cdot \alpha^2_j=\alpha_j$ is the complete execution fragment of the $j^{th}$ t-read $\Read_0(X_j) \rightarrow v$
such that
\begin{itemize}
\item
$\alpha^1_j$ does not contain any nontrivial events
\item
$\alpha_1 \cdots \alpha_{j-1} \cdot \alpha^1_j \cdot \delta_j \cdot \alpha^2_j$ is indistinguishable 
to $T_0$ from the step contention-free
execution fragment $\alpha_1 \cdots \alpha_{j-1} \cdot \alpha^1_j \cdot \alpha^2_j$
\end{itemize}
\end{itemize}
Moreover, $T_j$ does not access any base object to which $T_0$ applies a nontrivial event in
$\alpha_1 \cdots \alpha_{j-1} \cdot \alpha^1_j \cdot \delta_j$.
\end{claim}
\begin{proof}
By wait-free progress (for read-only transactions) and strict serializability, $M$ has an execution of the form 
$\alpha_1 \cdots \alpha_{j-1}$ in which each of the t-reads performed by $T_0$ must return the initial value of the t-objects.

Since $T_j$ is an updating transaction, by sequential progress, there exists an execution of $M$
of the form $\delta_j\cdot \alpha_1 \cdots \alpha_{j-1}$.
Since $T_0$ and $T_j$ are disjoint-access in the $\delta_j \cdot \alpha_1 \cdots \alpha_{j-1}$, by Lemma~\ref{lm:dap},
$T_0$ and $T_j$ do not contend on any base object in $\delta_j \cdot \alpha_1 \cdots \alpha_{j-1}$.
Thus, $\alpha_1 \cdots \alpha_{j-1}  \cdot \delta_j$ is indistinguishable to $T_j$ from the execution $\delta_j$
and $\alpha_1 \cdots \alpha_{j-1}  \cdot \delta_j$ is also an execution of $M$.

Let $e$ be the first event that contains a write to a base object in $\alpha_j$.
If there exists no such write event to a base object in $\alpha_j$, then $\alpha^1_j= \alpha_j$
and $\alpha^2_j$ is empty.
Otherwise, we represent the execution fragment $\alpha_j$ as $\alpha^1_{j}\cdot e \cdot \alpha^f_{j}$.

Since $\alpha^s_{j}$ does not contain any nontrivial events that write to a base object, 
$\alpha_1 \cdots \alpha_{j-1} \cdot \alpha^s_j \cdot \delta_j$ is indistinguishable to transaction $T_j$
from the execution $\alpha_1 \cdots \alpha_{j-1} \cdot \delta_j$.
Thus, $\alpha_1 \cdots \alpha_{j-1} \cdot \alpha^s_j \cdot \delta_j$ is an execution of $M$.
Since $e$ is not an atomic-write-after-read, $\alpha_1 \cdots \alpha_{j-1} \cdot \alpha^s_j \cdot \delta_j\cdot e$
is an execution of $M$.
Since $\alpha_j$ does not contain a RAW, any read performed in $\alpha^f_{j}$ may only be performed
to base objects previously written in $e \cdot \alpha^f_{j}$.
Thus, $\alpha_1 \cdots \alpha_{j-1} \cdot \alpha^s_j \cdot \delta_j \cdot e\cdot \alpha^f_j$
is indistinguishable to transaction $T_0$ from the step contention-free execution
$\alpha_1 \cdots \alpha_{j-1} \cdot \alpha^s_j \cdot e\cdot \alpha^f_j$ in which $\Read_0(X_j) \rightarrow v$.

Choosing $\alpha^2_j= e\cdot \alpha^f_j$,
it follows that $M$ has an execution of the form 
$\alpha_1 \cdots \alpha_{j-1} \cdot \alpha^1_j \cdot \delta_j \cdot \alpha^2_j$
that is indistinguishable to $T_j$ and $T_0$ from a step contention-free execution.
The proof follows. 
\end{proof}
We now prove that, for all $j\in \{1,\ldots , m-1\}$, 
$M$ has an execution of the form 
$\delta_{m}\cdot \alpha_1 \cdots \alpha_{j-1} \cdot \alpha^1_j \cdot \delta_j \cdot \alpha^2_j$
such that
\begin{itemize}
\item
$\delta_{m}$ is the t-complete step contention-free execution of transaction $T_{\ell}$
that writes $nv\neq v$ to $X_m$ and commits
\item
$T_{\ell}$ and $T_0$ do not contend on any base object in 
$\delta_{m}\cdot \alpha_1 \cdots \alpha_{j-1} \cdot \alpha^1_j \cdot \delta_j \cdot \alpha^2_j$
\item
$T_{\ell}$ and $T_j$ do not contend on any base object in
$\delta_{m}\cdot \alpha_1 \cdots \alpha_{j-1} \cdot \alpha^1_j \cdot \delta_j \cdot \alpha^2_j$.
\end{itemize}

By sequential progress for updating transactions, $T_{\ell}$ which writes the value $nv$ to $X_{m}$ must be committed in
$\delta_{m}$ since it is running in the absence of step-contention from the initial configuration.
Observe that $T_{\ell}$ and $T_0$ are disjoint-access in 
$\delta_{m}\cdot \alpha_1 \cdots \alpha_{j-1} \cdot \alpha^1_j\cdot \delta_j \cdot \alpha^2_j$.
By definition of $\alpha^1_j$ and $\alpha^2_j$, 
$\delta_{m}\cdot \alpha_1 \cdots \alpha_{j-1} \cdot \alpha^1_j \cdot \delta_j \cdot \alpha^2_j$
is indistinguishable to $T_0$ from $\delta_{m} \cdot \alpha_1 \cdots \alpha_{j-1} \cdot \alpha^1_j \cdot \alpha^2_j$.
By Lemma~\ref{lm:dap}, $T_{\ell}$ and $T_0$ do not contend on any base object in 
$\delta_{m} \cdot \alpha_1 \cdots \alpha_{j-1} \cdot \alpha^1_j \cdot \alpha^2_j$.

By Claim~\ref{cl:one}, $\delta_{m}\cdot \alpha_1 \cdots \alpha_{j-1} \cdot \alpha^1_j \cdot \delta_j$
is indistinguishable to $T_j$ from $\delta_{m}\cdot \delta_j$.
But transactions $T_{\ell}$ and $T_j$ are disjoint-access in $\delta_{m}\cdot \delta_j$, and by Lemma~\ref{lm:dap},
$T_j$ and $T_{\ell}$ do not contend on any base object in $\delta_{m}\cdot \delta_j$.

Since strict serializability of $M$ stipulates that each of the $j$ t-reads performed by $T_0$ return the initial
values of the respective t-objects, $M$ has an execution of the form 
$\delta_{m}\cdot \alpha_1 \cdots \alpha_{j-1} \cdot \alpha^1_j \cdot \delta_j \cdot \alpha^2_j$.

Consider the extension of $\delta_{m}\cdot \alpha_1 \cdots \alpha_{j-1} \cdot \alpha^1_j \cdot \delta_j \cdot \alpha^2_j$
in which $T_0$ performs $(m-j)$ t-reads
of $X_{j+1},\cdots , X_m$ step contention-free and commits (depicted in Figure~\ref{sfig:inv-2}). 
By wait-free progress of $M$ and since $T_0$ is a read-only transaction, there exists such an execution.
Notice that the $m^{th}$ t-read, $\Read_0(X_{m})$
must return the value $nv$ by strict serializability since $T_{\ell}$ precedes $T_0$ in real-time order in this execution.

Recall that neither pairs of transactions $T_{\ell}$ and $T_{j}$ nor $T_{\ell}$ and $T_0$ contend on any base object in
the execution $\delta_{m}\cdot \alpha_1 \cdots \alpha_{j-1} \cdot \alpha^1_j \cdot \delta_j \cdot \alpha^2_j$.
It follows that for all $j\in \{1,\ldots , m-1\}$, 
$M$ has an execution of the form
$\alpha_1 \cdots \alpha_{j-1} \cdot \alpha^1_j \cdot \delta_j \cdot \alpha^2_j\cdot \delta_{m}$
in which $T_{j}$ precedes $T_{\ell}$ in real-time order.

Let $\alpha'$ be the execution fragment that extends
$ \alpha_1 \cdots \alpha_{j-1} \cdot \alpha^1_j \cdot \delta_j \cdot \alpha^2_j\cdot \delta_{m}$
in which $T_0$ performs $(m-j)$ t-reads
of $X_{j+1},\cdots , X_m$ step contention-free and commits (depicted in Figure~\ref{sfig:inv-3}).
Since
$\alpha_1 \cdots \alpha_{j-1} \cdot \alpha^1_j \cdot \delta_j \cdot \alpha^2_j \cdot \delta_{m}$
is indistinguishable to $T_0$ from the execution
$\delta_{m}\cdot \alpha_1 \cdots \alpha!_{j-1} \cdot \alpha^1_j \cdot \delta_j \cdot \alpha^2_j$,
$\Read_0(X_{m})$ must return the response value $nv$ in $\alpha'$. 

The execution 
$ \alpha_1 \cdots \alpha_{j-1} \cdot \alpha^1_j \cdot \delta_j \cdot \alpha^2_j\cdot \delta_{m} \cdot \alpha'$
is not strictly serializable. 
In any serialization, $T_j$ must precede $T_{\ell}$ to respect the real-time ordering of transactions,
while $T_{\ell}$ must precede $T_0$ since $\Read_j(X_{m})$ returns the value of $X_{m}$ updated by $T_{\ell}$. Also,
transaction $T_0$ must precede $T_j$ since $\Read_0(X_j)$ returns the initial value of $X_j$. 
But there exists no such serialization: a contradiction to the assumption that $M$ is strict serializable.

Thus, for all $j\in \{1,\ldots , m-1\}$, transaction $T_0$ must perform a RAW or an AWAR during the execution of $\Read_0(X_j)$,
completing the proof. \qed
\end{proof}
Since Theorem~\ref{th:rwf} implies that read-only transactions must perform nontrivial events, we have the following corollary
that was proved directly in \cite{AHM09}.
\begin{corollary}[\cite{AHM09}]
\label{cr:inv}
There does not exist any strictly serializable weak DAP TM implementation $M\in \mathcal{RWF}$
that uses invisible reads.
\end{corollary}
%
%
\section{Related work and Discussion}
\label{sec:p3c4disc}
Attiya \emph{et al.}~\cite{AHM09} showed that it is impossible to implement 
weak DAP strictly serializable TMs in $\mathcal{RWF}$ 
if read-only transactions may only apply trivial primitives to base objects.
Attiya et al.~\cite{AHM09} also considered a stronger ``disjoint-access'' property, 
called simply DAP, referring to the
original definition proposed Israeli and
Rappoport~\cite{israeli-disjoint}.
In DAP,  two transactions are allowed to \emph{concurrently access} (even
for reading) the same base object only if they are disjoint-access.
For an $n$-process DAP TM implementation,  it is shown in~\cite{AHM09}
that a read-only transaction must
perform at least $n-3$ writes. 
Our lower bound in Theorem~\ref{th:rwf} is strictly stronger than the one in~\cite{AHM09}, as
it assumes only weak DAP,
considers a more precise RAW/AWAR metric, 
and does not depend on the number of processes in the system.
(Technically, the last point follows from the fact that the execution
constructed in the proof of Theorem~\ref{th:rwf}
uses only  $3$ concurrent processes.)
Thus, the theorem subsumes the two lower bounds of~\cite{AHM09} within a single proof.

Perelman \emph{et al.}~\cite{PFK10} considered the closely related (to $\mathcal{RWF}$) class of 
\emph{mv-permissive} TMs:
a transaction can only be aborted if it is an updating transaction that conflicts with another
updating transaction.
$\mathcal{RWF}$ is incomparable
with the class of mv-permissive TMs. On the one hand,
mv-permissiveness guarantees that read-only transactions never abort, but does not imply that they commit
in a wait-free manner.  
On the other hand, $\mathcal{RWF}$ allows an
updating transaction to abort in the presence of a concurrent
read-only transaction, which is disallowed by mv-permissive TMs.       
Observe that, technically, mv-permissiveness is a blocking TM-progress condition, although when used in conjunction with
wait-free TM-liveness, it is a partially non-blocking TM-progress condition that is strictly stronger than $\mathcal{RWF}$.

Assuming starvation-free TM-liveness,
\cite{PFK10} showed that 
implementing a weak DAP strictly serializable mv-permissive TM is impossible.
In the thesis, we showed that strictly serializable TMs in $\mathcal{RWF}$ cannot provide strict DAP, but proving
the impossibility result assuming weak DAP remains an interesting open question.

\cite{PFK10} also proved that mv-permissive TMs cannot be \emph{online
  space optimal}, \emph{i.e.}, no mv-permissive TM can keep the minimum number of old object
versions for any TM history. 
Our result on the space complexity of implementations in $\mathcal{RWF}$ that use invisible reads (Theorem~\ref{th:inv})
is different since it proves that the implementation must maintain an unbounded number of versions
of every t-object. 
Our proof technique can however be used to show that 
mv-permissive TMs considered in \cite{PFK10} should also maintain unbounded number of versions.


\chapter{Hybrid transactional memory (HyTM)}
\label{ch:p4c4}
\epigraph{\textbf{HAL}:
The $9000$ series is the most reliable computer ever made. 
No $9000$ computer has ever made a mistake or distorted information. 
We are all, by any practical definition of the words, foolproof and incapable of error.\\
$\ldots$ \\
\textbf{HAL}: I've just picked up a fault in the AE35 unit. It's going to go $100\%$ failure in $72$ hours.\\
\textbf{HAL}: It can only be attributable to human error.}
{\textit{Stanley Kubrick}-$2001$: A Space Odyssey}
\section{Overview}
\label{sec:p4c4intro}
\vspace{1mm}\noindent\textbf{Hybrid transactional memory.}
The TM abstraction, in its original manifestation from the proposal by Herlihy and Moss~\cite{HM93}, 
augmented the processor's \emph{cache-coherence protocol} and extended the CPU's instruction set with
instructions to indicate which memory accesses must be transactional~\cite{HM93}.
Most popular TM designs, subsequent to the original proposal in \cite{HM93} 
have implemented all the functionality in software~\cite{norec, ST95,HLM+03, astm, fraser} (cf. software TM model in
Chapter~\ref{ch:tm-model}).
More recently, CPUs have included hardware extensions to 
support \emph{short}, \emph{small} hardware transactions~\cite{Rei12, asf, bluegene}.

Early experience with programming \emph{Hardware transactional memory (HTM)}, 
\emph{e.g.}~\cite{DiceLMN09, DragojevicMLM11, AlistarhEMMS14}, paints an interesting picture: 
if used carefully, HTM can be an extremely useful construct, 
and can significantly speed up and simplify concurrent implementations. 
At the same time, this powerful tool is not without its limitations: 
since HTMs are usually implemented on top of the cache coherence mechanism, 
hardware transactions have inherent \emph{capacity constraints} on the number of distinct memory locations 
that can be accessed inside a single transaction.  
Moreover, all current proposals are \emph{best-effort}, as they may abort under imprecisely specified 
conditions (cache capacity overflow, interrupts \emph{etc}). 
In brief, the programmer should not solely rely on HTMs.

Several \emph{Hybrid Transactional Memory (HyTM)} schemes~\cite{hybridnorec,damronhytm, kumarhytm,phasedtm} have been
proposed to complement the fast, but best-effort nature of HTM 
with a slow, reliable software transactional memory (STM) backup. 
These proposals have explored a wide range of trade-offs between the
overhead on hardware transactions, concurrent execution of hardware and
software, and the provided progress guarantees. 

Early proposals for HyTM implementations~\cite{damronhytm, kumarhytm} 
shared some interesting features.
First, transactions that do not conflict are expected to run
concurrently, regardless of their types (software or hardware).
This property is referred to as \emph{progressiveness}~\cite{tm-theory}
and is believed to allow for increased parallelism.
Second, in addition to
changing the values of transactional objects, hardware transactions usually employ \emph{code instrumentation} techniques.
Intuitively, instrumentation is used by hardware transactions to
detect concurrency scenarios and abort in the
case of contention.
The number of instrumentation steps performed by these implementations within a hardware
transaction is usually proportional to the size of the transaction's data set. 

Recent work by Riegel \emph{et al.}~\cite{hynorecriegel} surveyed the various HyTM algorithms to date, focusing on techniques to reduce instrumentation overheads in the frequently executed hardware fast-path. 
However, it is not clear whether there are fundamental limitations when building a HyTM with non-trivial concurrency between
hardware and software transactions. In particular, what are the inherent instrumentation costs of building a HyTM, and what are the trade-offs between these
costs and the provided \emph{concurrency}, \emph{i.e.}, the ability of the HyTM
system to run software and hardware transactions in parallel?

\vspace{1mm}\noindent\textbf{Modelling HyTM.}
To address these questions, 
the thesis proposes the first model for hybrid TM systems which formally captures the notion of
\emph{cached} accesses provided by hardware transactions, and
precisely defines instrumentation costs in a quantifiable way.

We model a hardware transaction as a series of
memory accesses that operate on locally cached copies of the variables, followed by a \emph{cache-commit} operation.
In case a concurrent transaction performs a (read-write or write-write) conflicting access
to a cached object, the cached copy is invalidated and the hardware transaction aborts.   

Our model for instrumentation is motivated by recent experimental evidence 
which suggests that the overhead on hardware transactions imposed by code which 
detects concurrent software transactions is a significant performance bottleneck~\cite{MS13}. 
In particular, we say that a HyTM implementation imposes a logical partitioning of shared memory 
into \emph{data} and \emph{metadata} locations. 
Intuitively, metadata is used by transactions to exchange information about contention and 
conflicts while data locations only store the \emph{values} of data items read and updated within transactions. 
We quantify instrumentation cost by measuring the number of accesses to \emph{metadata objects} which transactions perform. 
Our framework captures all known HyTM proposals which combine
HTMs with an STM fallback~\cite{hybridnorec,riegel-thesis,phasedtm,damronhytm, kumarhytm}.

\vspace{1mm}\noindent\textbf{The cost of instrumentation.}
Once this general model is in place, we derive two lower bounds on the cost of implementing a HyTM.  
First, we show that some instrumentation is necessary in a HyTM implementation even if we only
intend to provide \emph{sequential} progress,
where a transaction is only guaranteed to commit if it runs in the
absence of concurrency. 

Second, we prove that any progressive HyTM implementation providing \emph{obstruction-free liveness} (every operation
running \emph{solo} returns some response) and  has
executions in which an arbitrarily long read-only hardware transaction running in the
absence of concurrency \emph{must} access a number of 
distinct metadata objects proportional to the size of its data set.
Our proof technique is interesting in its own right.
Inductively,  we start with a 
sequential execution in which a ``large'' set $S_m$ of
read-only hardware transactions, each accessing $m$
distinct data items and $m$ distinct metadata memory locations,  
run after an execution $E_m$.
We then construct execution $E_{m+1}$, an extension of $E_m$,
which forces at least half of the transactions in $S_m$ to access a \emph{new} metadata base object when reading a new 
$(m+1)^{\textit{th}}$ data item, running after $E_{m+1}$. 
The technical challenge, and the key departure from prior work on 
STM lower bounds, \emph{e.g.}~\cite{OFTM, tm-book, AHM09}, is that hardware transactions 
practically possess ``automatic'' conflict detection, aborting on contention. 
This is in contrast to STMs, which must take steps to detect contention on memory locations. 

We match this lower bound with an HyTM
algorithm that, additionally, allows for uninstrumented writes and
\emph{invisible reads} and is provably \emph{opaque}~\cite{tm-book}. 
To the best of our knowledge, this is the first formal proof of
correctness of a HyTM algorithm.

\vspace{1mm}\noindent\textbf{Low-instrumentation HyTM.}
The high instrumentation costs of early HyTM designs, 
which we show to be inherent, stimulated more recent HyTM schemes~\cite{phasedtm,hybridnorec,hynorecriegel,MS13}  
to sacrifice progressiveness for \emph{constant} instrumentation cost (\emph{i.e.}, not
depending on the size of the transaction). 
In the past two years, Dalessandro \emph{et al.}~\cite{hybridnorec} and 
Riegel \emph{et al.}~\cite{hynorecriegel} have proposed HyTMs based on the efficient \emph{NOrec STM}~\cite{norec}. 
These HyTMs schemes do not guarantee any parallelism among transactions; only
sequential progress is ensured. 
Despite this, they are among the best-performing HyTMs to date due to
the limited instrumentation in hardware transactions. 

Starting from this observation, we provide a more precise upper bound
for \emph{low-instrumentation} HyTMs
by presenting a HyTM algorithm with invisible reads \emph{and} uninstrumented
hardware writes which guarantees that a hardware transaction accesses at most one metadata object in the course of its execution.
Software transactions in this implementation remain progressive, while hardware transactions are guaranteed to commit
only if they do not run concurrently with an updating software transaction (or exceed capacity).
Therefore, the cost of avoiding the linear lower bound for progressive implementations is that hardware
transactions may be aborted by non-conflicting software ones.

\vspace{1mm}\noindent\textbf{Roadmap of Chapter~\ref{ch:p4c4}.}
In Section~\ref{sec:p4c1}, we introduce the model of HyTMs and
Section~\ref{sec:p4c2s1} studies the inherent cost of concurrency in progressive HyTMs by presenting a linear
lower bound on the cost of instrumentation while Section~\ref{sec:p4c4s2} presents a matching upper bound.
In Section~\ref{sec:p4c4s3} discusses providing partial concurrency with instrumentation cost and in
Section~\ref{sec:p4disc}, we elaborate on prior work related to HyTMs.
%
\section{Modelling HyTM}
\label{sec:p4c1}
In this chapter, we introduce the model of HyTMs, extending the TM model from Chapter~\ref{ch:tm-model}, 
that intuitively captures the cache-coherence protocols
employed in shared memory systems.
\subsection{Direct and cached accesses}
\label{sec:p4c1s1}
We now describe the operation of a \emph{Hybrid Transactional Memory  (HyTM)} implementation.
%
In our model, every base object can be accessed with two kinds of
primitives, \emph{direct} and \emph{cached}.

In a direct access, the rmw primitive operates on the memory state:
the direct-access event atomically reads the value of the object in
the shared memory and, if necessary, modifies it.

In a cached access performed by a process $i$, the rmw primitive operates on the \emph{cached}
state recorded in process $i$'s \emph{tracking set} $\tau_i$. 
One can think of $\tau_i$ as the \emph{L1 cache} of process $i$.
A \emph{hardware transaction} is a series of cached rmw primitives performed on $\tau_i$ followed by
a \emph{cache-commit} primitive. 
 
More precisely, $\tau_i$ is a set of triples $(b, v, m)$ where $b$ is a base object identifier, $v$ is a value, 
and $m \in \{\shared, \exclusive\}$ is an access \emph{mode}. 
The triple $(b, v, m)$ is added to the tracking set when $i$ performs a cached
rmw access of $b$, where $m$ is set to $\exclusive$ if the access is
nontrivial, and to $\shared$ otherwise.  
We assume that there exists some constant $\TS$ (representing the size of the L1 cache)
such that the condition $|\tau_i| \leq \TS$ must always hold; this
condition will be enforced by our model.
A base object $b$ is \emph{present} in $\tau_i$ with mode $m$ if $\exists v, (b,v,m) \in \tau_i$.

A trivial (resp.\ nontrivial) 
cached primitive $\langle g,h \rangle$ applied to $b$ 
by process $i$ first checks the condition $|\tau_i|=\TS$ and if so, it
sets $\tau_i=\emptyset$ and immediately returns $\bot$ (we call this event a
\emph{capacity abort}). 
We assume that $\TS$ is large enough so that no transaction 
with data set of size $1$ can incur a capacity abort.
If the transaction does not incur a capacity abort, the process checks whether $b$ is present in exclusive
(resp.\ any) mode in $\tau_j$ 
for any $j\neq i$. If so, $\tau_i$ is set to $\emptyset$ and the
primitive returns $\bot$. 
Otherwise, the triple $(b, v, \shared)$ (resp. $(b, g(v), \exclusive)$)
is added to $\tau_i$,  where $v$ is the most recent cached value of $b$ in $\tau_i$
(in case $b$ was previously accessed by $i$ within the current
hardware transaction) or the value of $b$ in the current
memory configuration, and finally $h(v)$ is returned.

A tracking set can be \emph{invalidated} by a concurrent process: 
if, in a configuration $C$ where  $(b,v,\exclusive)\in\tau_i$
(resp.\ $(b,v,\shared)\in\tau_i)$,  a process $j\neq i$ applies any primitive 
(resp.\ any \emph{nontrivial} primitive) to $b$, then $\tau_i$ becomes
\emph{invalid} and any subsequent cached primitive invoked by $i$
sets $\tau_i$ to $\emptyset$ and returns $\bot$. We refer to this event as a \emph{tracking set abort}.

Finally, the \emph{cache-commit} primitive issued by process $i$ with
a valid $\tau_i$ does the following: for each base object $b$ such that $(b,v,\exclusive) \in \tau_i$, the value of $b$ in $C$ is updated to $v$. 
Finally, $\tau_i$ is set to $\emptyset$ and the primitive 
returns $\textit{commit}$. 

Note that HTM may also abort spuriously, or because of unsupported operations~\cite{Rei12}. 
The first cause can be modelled probabilistically in the above
framework, which would not however significantly affect our claims and proofs, except for a more cumbersome presentation. 
Also, our lower bounds are based exclusively on executions containing t-reads and t-writes. 
Therefore, in the following, we only consider contention and capacity aborts.  
%
\subsection{Slow-path and fast-path transactions}
\label{sec:p4c1s2}
In the following, we partition HyTM transactions into \emph{fast-path
  transactions} and \emph{slow-path transactions}.
Practically,  two separate algorithms (fast-path one and slow-path one) 
are provided for each t-operation. 

A slow-path transaction models a regular software transaction.
An event of a slow-path transaction is either an invocation or response of a t-operation, or
a  rmw primitive on a base object. 

A fast-path transaction essentially encapsulates a hardware transaction. 
An event of a fast-path transaction is either an invocation or response of a t-operation, 
a cached primitive on a base object, or a \emph{cache-commit}:
\textit{t-read} and \emph{t-write} are only allowed to contain cached
primitives, and \textit{tryC} consists of invoking \emph{cache-commit}.  
Furthermore, we assume that a fast-path transaction $T_k$ returns $A_k$
as soon an underlying cached primitive or \emph{cache-commit} returns $\bot$. 
Figure~\ref{fig:tracking-set} depicts such a scenario illustrating a tracking set abort: 
fast-path transaction $T_2$ executed by process $p_2$
accesses a base object $b$ in shared (and resp. exclusive) mode and it is added to its tracking set $\tau_2$. 
Immediately after the access of $b$ by $T_2$, a concurrent transaction $T_1$ applies a nontrivial primitive to $b$ 
(and resp. accesses $b$). Thus, the tracking of $p_2$ is invalidated and $T_2$ must be aborted in any extension of this execution.
\begin{figure*}[t]
\begin{center}
	\subfloat[$\tau_2$ is invalidated by (fast-path or slow-path) transaction $T_1$'s access of base object $b$ \label{sfig:hinv-1}]{\scalebox{0.6}[0.6]{\begin{tikzpicture}
\node (e) at (13,-2) [] {};

\node[draw,align=left] at (10,1) {Fast-path};
\draw (e) node [above] {\small {(access of $b$)}};

\begin{scope}   
\draw [|-,thick] (9,0) node[left] {$T_2$} to (11,0);
\draw [-|,dashed] (11,0)  to (13.4,0) node[right] {$A_2$} ;
\draw [|-,thick] (11.5,-2) node[left] {$T_1$} to (13,-2);

\draw [-,dotted] (11.2,-1.7)  to (11.2,0.5) node[right] {$E$};
\end{scope}
\draw  (13,-2) circle [fill, radius=0.05]  (13,-2);

\node[draw,align=right] at (12,2) { $(b,v,exclusive) \in \tau_2$ after $E$};

\end{tikzpicture}}}
        \hspace{50mm}
	\subfloat[$\tau_2$ is invalidated by (fast-path or slow-path) transaction $T_1$'s write to base object $b$ \label{sfig:hinv-2}]{\scalebox{0.6}[0.6]{\begin{tikzpicture}
\node (e) at (13+6,-2) [] {};

\node[draw,align=left] at (10+6,1) {Fast-path};
\draw (e) node [above] {\small {(write to $b$)}};

\begin{scope}   
\draw [|-,thick] (9+6,0) node[left] {$T_2$} to (11+6,0);
\draw [-|,dashed] (11+6,0)  to (13.4+6,0) node[right] {$A_2$} ;
\draw [|-,thick] (11.5+6,-2) node[left] {$T_1$} to (13+6,-2);
\draw [|-,thick] (11.5+6,-2) node[left] {$T_1$} to (13+6,-2);

\draw [-,dotted] (11.2+6,-1.7)  to (11.2+6,0.5) node[right] {$E$};
\end{scope}
\draw  (13+6,-2) circle [fill, radius=0.05]  (13+6,-2);

\node[draw,align=right] at (10+8,2) { $(b,v,shared) \in \tau_2$ after $E$};

\end{tikzpicture}}}
	
\end{center}

\caption{Tracking set aborts in fast-path transactions;
we denote a fast-path (and resp. slow-path) transaction by $F$ (and resp. $S$)
\label{fig:tracking-set}} 
\end{figure*}
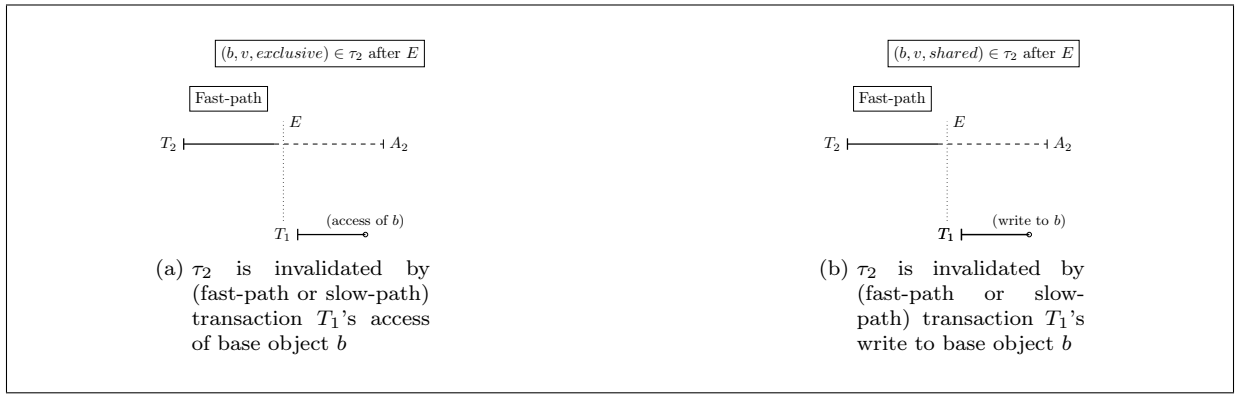
%

We provide two key observations on this model regarding the interactions of non-committed fast path transactions 
with other transactions. 
Let $E$ be any execution of a HyTM implementation $\mathcal{M}$ in
which a fast-path transaction $T_k$ is either
t-incomplete or aborted. 
Then the sequence of events $E'$ derived by removing all events of $E|k$
from $E$ is an execution  $\mathcal{M}$. Moreover: 
%
\begin{observation} 
\label{ob:one}
To every slow-path transaction $T_m \in \ms{txns}(E)$, $E$ is indistinguishable 
from $E'$. 
\end{observation}
\begin{figure*}[t]
\begin{center}
	\subfloat[\label{sfig:ob-01}]{\scalebox{0.65}[0.65]{\begin{tikzpicture}
\node (w2) at (10,0) [] {};
\node (w1) at (10,-2) [] {};
\node (w3) at (18,-2) [] {};

\draw (w2) node [above] {\small {$W_2(X,v)$}};

\draw (w1) node [above] {\small {$W_1(X,v)$}};

\node[draw,align=left] at (10,1) {Fast-path};
\node[draw,align=left] at (10,-1) {Slow-path};

\begin{scope}   
\draw [|-,thick] (9,0) node[left] {$T_2$} to (11,0);
\draw [-,dashed] (11,0) to (11.7,0);
\draw [|-|,thick] (8,-2) node[left] {$T_1$} to (12,-2);
\draw [-,dotted] (12.2,-2.5)  to (12.2,1) node[right] {$E$};
\end{scope}
\node[draw,align=right] at (15,.5) {Aborted or incomplete\\ fast-path transaction $T_2$};

\end{tikzpicture}}}
	\hspace{10mm}
	\subfloat[\label{sfig:ob-02}]{\scalebox{0.65}[0.65]{\begin{tikzpicture}

\node (w1) at (11,-2) [] {};

\draw (w1) node [above] {\small {$W_1(X,v)$}};

\node[draw,align=left] at (11,-1) {Slow-path};

\begin{scope}   
\draw [|-|,thick] (9,-2) node[left] {$T_1$} to (13,-2);
\draw [-,dotted] (13.2,-2.5)  to (13.2,1) node[right] {$E'$};
\end{scope}
\end{tikzpicture}}}
	 
\end{center}
\caption{
 \label{fig:ob1}
 Execution $E$ in Figure~\ref{sfig:ob-01} is indistinguishable
to $T_1$ from the execution $E'$ in Figure~\ref{sfig:ob-02}}
\end{figure*}
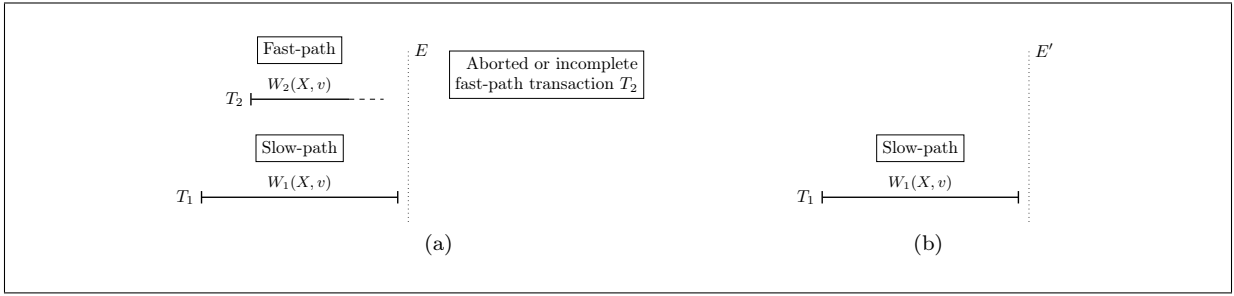
\begin{observation} 
\label{ob:two}
If a fast-path transaction $T_m\in \ms{txns}(E) \setminus \{T_k\}$ does not incur a tracking set abort in $E$, 
then $E$ is indistinguishable to $T_m$ from $E'$.
\end{observation}
Intuitively, these observations say that fast-path transactions which are not yet committed are 
invisible to slow-path transactions, and can communicate with other
fast-path transactions only by incurring their tracking-set aborts.
Figure~\ref{fig:ob1} illustrates Observation~\ref{ob:one}: a fast-path transaction $T_2$ is concurrent to a slow-path transaction
$T_1$ in an execution $E$. Since $T_2$ is t-incomplete or aborted in this execution, $E$ is indistinguishable to $T_1$
from an execution $E'$ derived by removing all events of $T_2$ from $E$.
Analogously, to illustrate Observation~\ref{ob:two}, if $T_1$ is a fast-path transaction that does not incur a tracking set abort in $E$, then
$E$ is indistinguishable to $T_1$ from $E'$.
%
\subsection{Instrumentation}
\label{sec:p4c1s3}
Now we define the notion of \emph{code instrumentation} in fast-path transactions.
Intuitively, instrumentation characterizes the number of extra ``metadata'' accesses performed by a fast-path transaction.

We start with the following technical definition.
An execution $E$ of a HyTM $\mathcal{M}$ \emph{appears t-sequential}
to a transaction $T_{k} \in \ms{txns}(E)$ if there exists an execution
$E'$ of $\mathcal{M}$ such that:
\begin{itemize}
\item $\ms{txns}(E')\subseteq \ms{txns}(E)\setminus \{T_k\}$ and the configuration after $E'$ is t-quiescent,
\item every transaction $T_m\in \ms{txns}(E)$
that precedes $T_k$ in real-time order is included in $E'$ such that $E|m=E'|m$,
\item for every transaction $T_m\in \ms{txns}(E')$, 
$\Rset_{E'}(T_m)\subseteq \Rset_E(T_m)$ and $\Wset_{E'}(T_m)\subseteq \Wset_E(T_m)$, and
\item $E'\cdot E|k$ is an execution of $\mathcal{M}$.
\end{itemize}
%
%
%

\begin{definition}[Data and metadata base objects] 
\label{def:metadata}
Let $\mathcal{X}$ be the set of t-objects operated by a HyTM implementation $\mathcal{M}$. 
Now we partition the set of base objects used by $\mathcal{M}$ into a set $\mathbb{D}$ of \emph{data} objects and 
a set $\mathbb{M}$ of \emph{metadata} objects ($\mathbb{D}\cap \mathbb{M} = \emptyset$). We further partition
$\mathbb{D}$ into sets $\mathbb{D}_X$ associated with each t-object $X
\in \mathcal{X}$:  $\mathbb{D} = \bigcup\limits_{X\in\mathcal{X}} \mathbb{D}_X$,
for all $X\neq Y$ in $\X$, $\mathbb{D}_X \cap \mathbb{D}_Y =
\emptyset$,
such that:
%
\begin{enumerate}
\item In every execution $E$, each fast-path transaction $T_k \in \ms{txns}(E)$ only 
accesses base objects in $\bigcup\limits_{X\in DSet(T_k)}
\mathbb{D}_X$ or $\mathbb{M}$.
\item
Let $E\cdot\rho$ and $E\cdot E'\cdot\rho'$ be two t-complete
executions, such that $E$ and $E\cdot E'$ are t-complete, 
$\rho$ and $\rho'$ are complete executions of a transaction
$T_k\notin\ms{txns}(E\cdot E')$, $H_{\rho}=H_{\rho'}$, and $\forall
T_m\in \ms{txns}(E')$, $\Dset(T_m)\cap \Dset(T_k)=\emptyset$. 
Then the states of the base objects $\bigcup\limits_{X\in DSet(T_k)} \mathbb{D}_X$ 
in the configuration after $E\cdot \rho$ and $E\cdot E' \cdot {\rho'}$
are the same.    


\item 
Let execution $E$ appear t-sequential to a transaction $T_k$ and let
the enabled event $e$ of $T_k$ after $E$ be a primitive on a base
object $b\in \mathbb{D}$. Then, unless $e$ returns $\bot$, $E\cdot e$
also appears t-sequential to $T_k$. 
 
%
\end{enumerate}
\end{definition}
Intuitively, the first condition says that a transaction is only
allowed to access data objects based on its data set. The second
condition says that transactions with disjoint data sets can
communicate only via metadata objects.
Finally, the last condition means that base objects in
$\mathbb{D}$ may only contain the ``values'' of t-objects, and cannot be used to detect concurrent transactions. Note that our results will lower bound the number of metadata objects that must be accessed under particular assumptions, thus from a cost perspective, $\mathbb{D}$ should be made as large as possible.

All HyTM proposals we aware of, such as 
\emph{HybridNOrec}~\cite{hybridnorec,riegel-thesis}, \emph{PhTM}~\cite{phasedtm} and others~\cite{damronhytm, kumarhytm},
conform to our definition of instrumentation in fast-path transactions.
For instance, HybridNOrec~\cite{hybridnorec,riegel-thesis} employs a distinct base object in $\mathbb{D}$ for each
t-object and a global \emph{sequence lock} as the metadata that is accessed by fast-path transactions to detect
concurrency with slow-path transactions.
Similarly, the HyTM implementation by \emph{Damron et al.}~\cite{damronhytm} also 
associates a distinct base object in $\mathbb{D}$ for each
t-object and additionally, a \emph{transaction header} and \emph{ownership record} as metadata base objects.

%
%
\begin{definition}[Uninstrumented HyTMs]
\label{def:ins}
A HyTM implementation $\mathcal{M}$ provides \emph{uninstrumented writes (resp.\ reads)} 
if in every execution $E$ of $\mathcal{M}$, for every write-only (resp.\ read-only) 
fast-path
 transaction $T_k$, 
all primitives in $E|k$ are performed on base objects in $\mathbb{D}$.
A HyTM is uninstrumented if both its reads and writes are uninstrumented. 
\end{definition}
%
%
\begin{observation}
\label{ob:ins}
Consider any execution $E$ of a HyTM implementation $\mathcal{M}$ which provides uninstrumented reads (resp. writes). 
For any fast-path read-only (resp.\ write-only) transaction $T_k \not\in \ms{txns}(E)$, 
that runs step-contention free after $E$, 
the execution $E$ appears t-sequential to $T_k$.
\end{observation}
%
\subsection{Impossibility of uninstrumented HyTMs}
\label{sec:p4c1s4}
\begin{figure*}[t]
\begin{center}
	\subfloat[$T_y$ must return the new value\label{sfig:hinv-0}]{\scalebox{0.65}[0.65]{\begin{tikzpicture}
\node (r1) at (-1,0) [] {};
\node (w0) at (1,0) [] {};
\node (w1) at (3,0) [] {};
\node (c1) at (5,0) [] {};

\node (e) at (7.5,0) [] {};

\node (r3) at (10,0) [] {};

\draw (r1) node [above] {\small {$R_0(Z)\rightarrow v$}};
\draw (w1) node [above] {\small {$W_0(X,nv)$}};
\draw (c1) node [above] {\small {$\TryC_0$}};
\draw (w0) node [above] {\small {$W_0(Y,nv)$}};

\draw (e) node [above] {\tiny {(event of $T_0$)}};
\draw (e) node [below] {\small {$e$}};

\draw (r3) node [above] {\small {$R_y(Y)\rightarrow nv$}};
\draw (r3) node [below] {\tiny {returns new value}};

\node[draw,align=left] at (2.5,1) {S};
\node[draw,align=left] at (10,1) {F};
\begin{scope}   
\draw [|-|,thick] (-2,0) node[left] {$T_0$} to (0,0);
\draw [-|,thick] (0,0) node[left] {} to (2,0);
\draw [-|,thick] (2,0) node[left] {} to (4,0);
\draw [-,thick] (4,0) node[left] {} to (5.5,0);
\draw  (7.5,0) circle [fill, radius=0.05]  (7.5,0);
\draw [|-|,thick] (9,0) node[left] {$T_y$} to (11,0);
\end{scope}
\end{tikzpicture}}}
	\\
	\subfloat[Since $T_z$ is uninstrumented, by Observation~\ref{ob:ins} and sequential TM-progress, $T_z$ must commit\label{sfig:hinv-1}]{\scalebox{0.65}[0.65]{\begin{tikzpicture}
\node (r1) at (-1,0) [] {};
\node (w0) at (1,0) [] {};
\node (w1) at (3,0) [] {};
\node (c1) at (5,0) [] {};

\node (w3) at (8.5,0) [] {};

\draw (r1) node [above] {\small {$R_0(Z)\rightarrow v$}};
\draw (w1) node [above] {\small {$W_0(X,nv)$}};
\draw (c1) node [above] {\small {$\TryC_0$}};
\draw (w0) node [above] {\small {$W_0(Y,nv)$}};

\draw (w3) node [above] {\small {$W_z(Z,nv)$}};
\draw (w3) node [below] {\tiny {write new value}};

\node[draw,align=left] at (2.5,1) {S};
\node[draw,align=left] at (8.5,1) {F};

\begin{scope}   
\draw [|-|,thick] (-2,0) node[left] {$T_0$} to (0,0);
\draw [-|,thick] (0,0) node[left] {} to (2,0);
\draw [-|,thick] (2,0) node[left] {} to (4,0);
\draw [-,thick] (4,0) node[left] {} to (5.5,0);
\draw [|-|,thick] (7.5,0) node[left] {$T_z$} to (9.5,0);
\end{scope}
\end{tikzpicture}}}
        \\
        \vspace{2mm}
	\subfloat[Since $T_x$ does not access any metadata, by Observation~\ref{ob:ins}, it cannot abort and must 
	return the initial value value of $X$\label{sfig:hinv-2}]{\scalebox{0.65}[0.65]{\begin{tikzpicture}
\node (r1) at (-1,0) [] {};
\node (w0) at (1,0) [] {};
\node (w1) at (3,0) [] {};
\node (c1) at (5,0) [] {};

\node (r2) at (12,0) [] {};

\node (w3) at (8.5,0) [] {};

\draw (r1) node [above] {\small {$R_0(Z)\rightarrow v$}};
\draw (w1) node [above] {\small {$W_0(X,nv)$}};
\draw (c1) node [above] {\small {$\TryC_0$}};
\draw (w0) node [above] {\small {$W_0(Y,nv)$}};

\draw (r2) node [above] {\small {$R_x(X)\rightarrow v$}};
\draw (r2) node [below] {\tiny {returns initial value}};

\draw (w3) node [above] {\small {$W_z(Z,nv)$}};
\draw (w3) node [below] {\tiny {write new value}};

\node[draw,align=left] at (2.5,1) {S};
\node[draw,align=left] at (8.5,1) {F};
\node[draw,align=left] at (12,1) {F};

\begin{scope}   
\draw [|-|,thick] (-2,0) node[left] {$T_0$} to (0,0);
\draw [-|,thick] (0,0) node[left] {} to (2,0);
\draw [-|,thick] (2,0) node[left] {} to (4,0);
\draw [-,thick] (4,0) node[left] {} to (5.5,0);
\draw [|-|,thick] (7.5,0) node[left] {$T_z$} to (9.5,0);
\draw [|-|,thick] (11,0) node[left] {$T_x$} to (13,0);
\end{scope}
\end{tikzpicture}}}
	\\
	\vspace{2mm}
	\subfloat[$T_y$ does not contend with $T_x$ or $T_z$ on any base object \label{sfig:hinv-3}]{\scalebox{0.65}[0.65]{\begin{tikzpicture}
\node (r1) at (-1,0) [] {};
\node (w0) at (1,0) [] {};
\node (w1) at (3,0) [] {};
\node (c1) at (5,0) [] {};

\node (r2) at (12,0) [] {};

\node (e) at (15,0) [] {};

\node (r3) at (18.5,0) [] {};

\node (w3) at (8.5,0) [] {};

\draw (r1) node [above] {\small {$R_0(Z)\rightarrow v$}};
\draw (w1) node [above] {\small {$W_0(X,nv)$}};
\draw (c1) node [above] {\small {$\TryC_0$}};
\draw (w0) node [above] {\small {$W_0(Y,nv)$}};

\draw (r2) node [above] {\small {$R_x(X)\rightarrow v$}};
\draw (r2) node [below] {\tiny {returns initial value}};

\draw (e) node [above] {\tiny {(event of $T_0$)}};
\draw (e) node [below] {\small {$e$}};

\draw (r3) node [above] {\small {$R_y(Y)\rightarrow nv$}};
\draw (r3) node [below] {\tiny {returns new value}};

\draw (w3) node [above] {\small {$W_z(Z,nv)$}};
\draw (w3) node [below] {\tiny {write new value}};

\node[draw,align=left] at (2.5,1) {S};
\node[draw,align=left] at (8.5,1) {F};
\node[draw,align=left] at (18.5,1) {F};
\node[draw,align=left] at (12,1) {F};

\begin{scope}   
\draw [|-|,thick] (-2,0) node[left] {$T_0$} to (0,0);
\draw [-|,thick] (0,0) node[left] {} to (2,0);
\draw [-|,thick] (2,0) node[left] {} to (4,0);
\draw [-,thick] (4,0) node[left] {} to (5.5,0);
\draw [|-|,thick] (7.5,0) node[left] {$T_z$} to (9.5,0);
\draw [|-|,thick] (11,0) node[left] {$T_x$} to (13,0);
\draw  (15,0) circle [fill, radius=0.05]  (15,0);
\draw [|-|,thick] (17.5,0) node[left] {$T_y$} to (19.5,0);
\end{scope}
\end{tikzpicture}}}
	\caption{Executions in the proof of Theorem~\ref{instrumentation}; execution in \ref{sfig:hinv-3} is not strictly serializable
          \label{fig:indis}} 
\end{center}
\end{figure*}
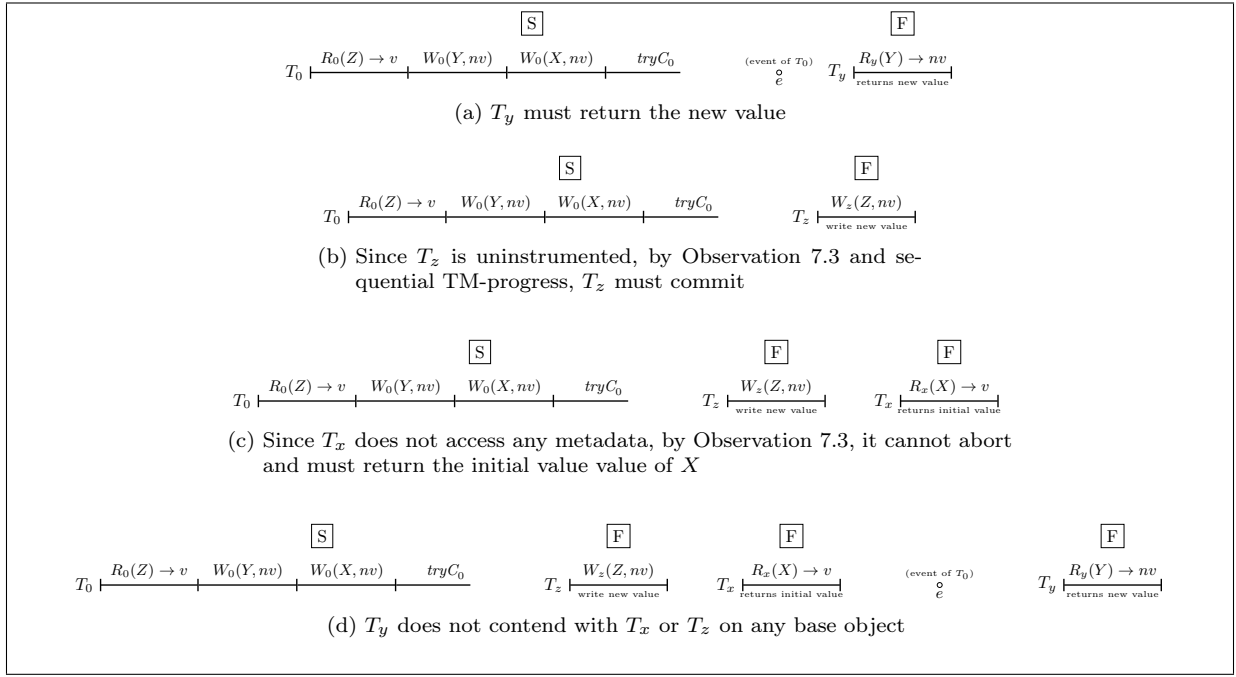
In this section, we show that any strictly serializable HyTM must be
instrumented, even under a very weak progress assumption by which  
a transaction is guaranteed to commit only when run t-sequentially:

\begin{definition}[Sequential TM-progress for HyTMs]
\label{def:seqtmprogress}
A HyTM implementation $\mathcal{M}$ provides \emph{sequential TM-progress for fast-path transactions (and resp. slow-path)} if 
in every execution $E$ of $\mathcal{M}$, a fast-path 
(and resp. slow-path) 
transaction $T_k$ returns $A_k$ in $E$ only if 
$T_k$ incurs a capacity abort or $T_k$ is concurrent to another transaction. 
We say that $\mathcal{M}$ provides sequential TM-progress if it provides sequential TM-progress for fast-path and slow-path
transactions. 
\end{definition}

\begin{theorem}
\label{instrumentation}
There does not exist a strictly serializable uninstrumented HyTM
implementation 
that ensures sequential TM-progress and TM-liveness.
\end{theorem}
\begin{proof}
Suppose by contradiction that such a HyTM $\M$ exists. 
For simplicity, assume that $v$ is the initial value of t-objects $X$, $Y$ and $Z$.
Let $E$ be the t-complete step contention-free execution of a slow-path transaction $T_0$
that performs $\Read_0(Z) \rightarrow v$,  $\Write_0(X, nv)$,
$\Write_0(Y, nv)$ ($nv\neq v$), and commits. 
Such an execution exists since $\M$ ensures sequential TM-progress. 

By Observation~\ref{ob:ins}, any transaction that 
runs step contention-free starting from a prefix of $E$ must return a
non-abort value.
Since any such transaction reading $X$ or $Y$ must return $v$ 
when it starts from the
empty prefix of $E$ and $nv$ when it starts from $E$.

Thus, there exists  $E'$, the longest prefix of $E$ that cannot be extended with the 
t-complete step contention-free execution of a \emph{fast-path} transaction 
reading $X$ or $Y$ and returning $nv$.
Let $e$ is the enabled event of $T_0$ in the configuration after $E'$.
Without loss of generality, suppose that there exists an execution
$E'\cdot e\cdot E_y$ where $E_y$ is the t-complete step contention-free
execution fragment of some fast-path transaction $T_y$ that reads $Y$ is returns $nv$ (Figure~\ref{sfig:hinv-0}). 
%
%
\begin{claim}\label{claim:concat}
$\M$ has an execution $E' \cdot E_z \cdot  E_x$, where
\begin{itemize}
\item
$E_z$ is the t-complete step contention-free execution fragment of a fast-path transaction $T_{z}$ that 
writes $nv \neq v$ to $Z$ and commits
\item
$E_x$ is the t-complete step contention-free execution fragment of a fast-path transaction $T_x$ that performs
a single t-read $\Read_x(X) \rightarrow v$ and commits.
\end{itemize}
\end{claim}
\begin{proof}
%
By Observation~\ref{ob:ins}, the extension of $E'$ in which $T_z$ writes to
$Z$ and tries to commit appears t-sequential to $T_z$.
By sequential TM-progress, $T_z$ complets the write and commits. 
Let  $E' \cdot E_z$ (Figure~\ref{sfig:hinv-1})
be the resulting execution of $\M$.

Similarly, the extension of $E'$ in which $T_x$ reads $X$ 
and tries to commit appears t-sequential to $T_x$.
By sequential TM-progress, $T_x$ commits and let  $E' \cdot E_x$
be the resulting execution of $\M$.
By the definition of $E'$, 
$\Read_x(X)$ must return $v$ in $E' \cdot E_x$.

Since $\M$ is uninstrumented and the data sets of $T_x$
and $T_z$ are disjoint, 
the sets of base objects accessed in the execution fragments $E_x$ and
$E_y$ are also disjoint.
Thus, $E' \cdot E_z \cdot E_x $
is indistinguishable to $T_x$ from the execution $E' \cdot E_x$, which
implies that $E' \cdot E_z \cdot E_x$ is an execution of $\M$ (Figure~\ref{sfig:hinv-2}).
\end{proof} 
Finally, we prove that the sequence of events, ${E' \cdot E_z \cdot E_x \cdot e \cdot E_y}$ is an execution of $\M$.

Since the transactions $T_x$, $T_y$, $T_z$ have
pairwise disjoint data sets in ${E' \cdot E_z \cdot E_x \cdot e \cdot E_y}$,
no base object accessed ib $E_y$ can be accessed in $E_x$ and $E_z$.
The read operation on $X$ performed by $T_y$ in $E'\cdot e\cdot E_y$
returns $nv$ and, by the definition of $E'$ and $e$, $T_y$ must have accessed
the base object $b$ modified in the event $e$ by $T_0$.
Thus, $b$ is not accessed in $E_x$ and $E_z$ and 
$E' \cdot E_z \cdot E_x \cdot e$ is an execution of $\M$.
Summing up, $E' \cdot E_z \cdot E_x \cdot e \cdot E_y$
is indistinguishable to $T_y$ from
$E'  \cdot e \cdot E_y$, which implies that 
$E' \cdot E_z \cdot E_x \cdot e \cdot E_y$ is an execution of $\M$
(Figure~\ref{sfig:hinv-3}).

But the resulting execution is not strictly serializable.
Indeed, suppose that a serialization exists.  
As the value written by $T_0$ is returned by a committed transaction
$T_y$, $T_0$ must be committed and precede $T_y$ in the serialization.
Since $T_x$ returns the initial value of $X$, $T_x$ must precede
$T_0$. 
Since $T_0$ reads the initial value of $Z$, $T_0$ must
precede $T_z$.
Finally, $T_z$ must precede $T_x$ to respect the real-time order. 
The cycle in the serialization establishes a contradiction.
\end{proof}
%
\section{A linear lower bound on instrumentation for progressive HyTMs}
\label{sec:p4c2s1}
In this section, we show that giving HyTM the ability to run and commit
transactions in parallel brings 
considerable 
instrumentation costs.   
We focus on a 
natural
progress condition called
progressiveness~\cite{GK08-opacity,GK09-progressiveness,tm-theory} that allows a
transaction to abort only if it experiences a read-write or write-write
conflict with a concurrent transaction: 
\begin{definition}[Progressiveness for HyTMs]
\label{def:prog}
We say that transactions $T_i$ and $T_j$ \emph{conflict} in an execution $E$ 
on a t-object $X$ if
$X\in\Dset(T_i)\cap\Dset(T_j)$ and $X\in\Wset(T_i)\cup\Wset(T_j)$.

A HyTM implementation $\M$ 
is \emph{fast-path} (resp. \emph{slow-path}) \emph{progressive} 
if in every execution $E$ of $\M$ and for every fast-path (and resp. slow-path) transaction 
$T_i $ that aborts in $E$, 
either $A_i$ is a capacity abort or  $T_i$ conflicts with some transaction $T_j$ that is concurrent to $T_i$ in $E$.  
We say $\M$ is \emph{progressive} if it is both fast-path and slow-path progressive.
%
\end{definition}
We show that for every opaque fast-path progressive HyTM that provides
obstruction-free TM-liveness, an arbitrarily long read-only
transaction might access a number of distinct
metadata base objects that is linear in the size of its read set or
experience a capacity abort.

The following auxiliary results will be crucial in proving our lower
bound.
We observe first that a fast path transaction in a progressive HyTM can 
contend on a base object only with a conflicting transaction.
\begin{lemma}
\label{lm:pgtwo}
Let $\mathcal{M}$ be any fast-path progressive HyTM implementation.
Let $E\cdot E_1 \cdot E_2$ be an execution of $\mathcal{M}$ where
$E_1$ (and resp. $E_2$) is the step contention-free
execution fragment of transaction $T_1 \not\in \ms{txns}(E)$ (and
resp. $T_2 \not\in \ms{txns}(E)$),
$T_1$ (and resp. $T_2$) does not conflict with any transaction in $E\cdot E_1 \cdot E_2$, and
at least one of $T_1$ or $T_2$ is a fast-path transaction. 
Then, $T_1$ and $T_2$ do not contend on any base object in $E\cdot E_1 \cdot E_2$.
\end{lemma}
\begin{proof}
Suppose, by contradiction that $T_1$ or $T_2$ 
contend on the same base object in $E\cdot E_1\cdot E_2$.

If in $E_1$, $T_1$ performs a nontrivial event on a base object on which they contend, let $e_1$ be the last
event in $E_1$ in which $T_1$ performs such an event to some base object $b$ and $e_2$, the first event
in $E_2$ that accesses $b$.
Otherwise, $T_1$
only performs trivial events in $E_1$ to base objects on which it contends with $T_2$ in $E\cdot E_1\cdot E_2$:
let $e_2$ be the first event in $E_2$ in which $E_2$ performs a nontrivial event to some base object $b$
on which they contend and $e_1$, the last event of $E_1$ in $T_1$ that accesses $b$.

Let $E_1'$ (and resp. $E_2'$) be the longest prefix of $E_1$ (and resp. $E_2$) that does not include
$e_1$ (and resp. $e_2$).
Since before accessing $b$, the execution is step contention-free for $T_1$, $E \cdot
E_1'\cdot E_2'$ is an execution of $\mathcal{M}$.
By construction, $T_1$ and $T_2$ do not conflict in $E \cdot E_1'\cdot E_2'$.
Moreover, $E\cdot E_1 \cdot E_2'$ is indistinguishable to $T_2$ from $E\cdot E_1' \cdot E_2'$.
Hence, $T_1$ and
$T_2$ are poised to apply contending events $e_1$ and $e_2$ on $b$ in the execution
$\tilde E=E\cdot E_1' \cdot E_2'$.
Recall that at least one event of $e_1$ and $e_2$ must be nontrivial.  

Consider the execution $\tilde E\cdot e_1 \cdot e_2'$ where $e_2'$ is the
event of $p_2$ in which it applies the primitive of $e_2$ to the
configuration after $\tilde E \cdot e_1$. 
After $\tilde E\cdot e_1$, $b$ is contained in the tracking set of process
$p_1$.
If $b$ is contained in $\tau_1$ in the shared mode, then $e_2'$ is a
nontrivial primitive on $b$, which invalidates $\tau_1$ in $\tilde E\cdot
e_1\cdot e_2'$.
If  $b$ is contained in $\tau_1$ in the exclusive mode, then any
subsequent access of $b$ invalidates $\tau_1$ in  $\tilde E\cdot
e_1\cdot e_2'$.
In both cases,   $\tau_1$ is invalidated and $T_1$ incurs a tracking set abort.
Thus, transaction $T_1$ must return $A_1$ in any extension of $E\cdot
e_1\cdot e_2$---a contradiction
to the assumption that $\mathcal{M}$ is progressive.   
\end{proof}
Iterative application of Lemma~\ref{lm:pgtwo} implies the following:  
\begin{corollary}
\label{lm:pgthree}
Let $\mathcal{M}$ be any fast-path progressive HyTM implementation.
Let $E\cdot E_1 \cdots E_i \cdot E_{i+1} \cdots E_m$ be any execution of $\mathcal{M}$ where
for all $i\in \{1,\ldots , m\}$, $E_i$ is the step contention-free
execution fragment of transaction $T_i \not\in \ms{txns}(E)$
and any two transactions in $E_1 \cdots E_m$ do not conflict.
For all $i,j=1,\ldots,m$, $i\neq j$, if $T_i$ is fast-path, then $T_i$
and $T_j$ do not contend on a
base object in  $E\cdot E_1 \cdots E_{i} \cdots E_m$   
\end{corollary}
\begin{proof}
Let $T_i$ be a fast-path transaction.
By Lemma~\ref{lm:pgtwo}, in $E\cdot E_1 \cdots E_{i} \cdots E_m$, 
$T_i$ does not contend with $T_{i-1}$ (if $i>1$) or
$T_{i+1}$ (if $i<m$) on any base object
and, thus, $E_i$ commutes with $E_{i-1}$ and $E_{i+1}$.  
Thus, 
$E\cdot E_1 \cdots E_{i-2}\cdot E_{i}\cdot E_{i-1}\cdot E_{i+1}\cdots E_m$  (if $i>1$) and 
$E\cdot E_1 \cdots E_{i-1}\cdot E_{i+1}\cdot E_{i}\cdot E_{i+2}\cdots E_m$  (if $i<m$) are executions of
$\M$.
By iteratively applying Lemma~\ref{lm:pgtwo}, we derive that $T_i$
does not contend with any $T_j$, $j\neq i$.
\end{proof}
%
%
%
%
We say that execution fragments $E$ and $E'$ are \emph{similar} if
they export equivalent histories, \emph{i.e.}, no process can see the
difference between them by looking at the invocations and responses of
t-operations.   
We now use Corollary~\ref{lm:pgthree} to show that t-operations
only accessing data base objects cannot detect contention with
non-conflicting transactions.
\begin{lemma}
\label{lm:finallm}
Let $E$ be any t-complete execution of 
a progressive  HyTM implementation 
$\M$ that provides OF TM-liveness.
For any $m\in \mathbb{N}$, consider a set of $m$ executions of $\M$ of the form
$E\cdot E_{i} \cdot \gamma_i \cdot \rho_i$ where
$E_{i}$ is the t-complete step contention-free execution fragment of
a transaction $T_{m+i}$,
$\gamma_i$ is a complete step contention-free execution fragment of
a \emph{fast-path} transaction $T_i$ such that
$\Dset(T_i)\cap \Dset(T_{m+i})=\emptyset$ in $E\cdot E_{i} \cdot \gamma_i$, and
$\rho_i$ is the execution fragment of a t-operation by $T_i$ 
that does not contain accesses to any metadata base object.
If, for all $i,j\in \{1,\ldots , m\}$, $i\neq j$,
$\Dset(T_i)\cap \Dset(T_{m+j})=\emptyset$,
$\Dset(T_i)\cap \Dset(T_{j})=\emptyset$ and
$\Dset(T_{m+i})\cap \Dset(T_{m+j})=\emptyset$, then
there exists a t-complete step 
contention-free execution fragment $E'$ that is similar to $E_{1}\cdots E_{m}$ such that 
for all $i\in \{1,\ldots, m\}$, $E\cdot E' \cdot \gamma_i\cdot  \rho_i$ is 
an execution of $\M$.
\end{lemma}
\begin{proof}
Observe that any two transactions in the execution fragment
$E_{1}\cdots E_{m}$ access mutually disjoint data sets.
Since $\M$ is progressive and provides OF TM-liveness, there exists a t-sequential execution fragment $E'=E'_{1}\cdots E'_{m}$
such that, for all $i\in \{1,\ldots , m\}$, the execution fragments $E_{i}$ and $E'_{i}$ are similar and
$E\cdot E'$ is an execution of $\M$.
Corollary~\ref{lm:pgthree} implies that, for all for all $i\in \{1,\ldots , m\}$,
$\M$ has an execution of the form $E\cdot E'_{1}\cdots E'_{i}\cdots  E'_{m} \cdot \gamma_i$.
More specifically, 
$\M$ has an execution of the form
$E\cdot \gamma_i \cdot E'_{1} \cdots E'_{i}\cdots  E'_{m}$.
Recall that the execution fragment $\rho_i$ of fast-path transaction $T_i$ that extends $\gamma_i$ contains accesses only to
base objects in $\bigcup\limits_{X\in DSet(T_i)} \mathbb{D}_X$.
Moreover, for all $i,j\in \{1,\ldots , m\}$; $i\neq j$,
$\Dset(T_i)\cap \Dset(T_{m+j})=\emptyset$ and $\Dset(T_{m+i})\cap \Dset(T_{m+j})=\emptyset$.

It follows that
$\M$ has an execution of the form
$E\cdot \gamma_i \cdot E'_{1} \cdots E'_{i}\cdot \rho_i \cdot E'_{i+1} \cdots  E'_{m}$.
and the states of each of the base objects $\bigcup\limits_{X\in DSet(T_i)} \mathbb{D}_X$ accessed by $T_i$
in the configuration after 
$E\cdot \gamma_i \cdot E'_{1} \cdots E'_{i}$
and $E\cdot \gamma_i \cdot E_{i}$ are the same.
But $E\cdot \gamma_i \cdot E_{i}\cdot \rho_i$ is an execution of $\M$.
Thus, for all $i\in \{1,\ldots , m\}$, $\M$ has an execution of the form
$E\cdot E'\cdot \gamma_i \cdot \rho_i$.
\end{proof}
Finally, we are now ready to derive our lower bound.
\begin{theorem}
\label{linearlowerbound}
Let $\M$ be any progressive, opaque HyTM implementation that provides OF TM-liveness.
For every $m\in\Nat$, there
exists an execution $E$ in which some fast-path read-only
transaction $T_k\in \ms{txns}(E)$ satisfies either (1) $\Dset(T_k) \leq m$ and $T_k$ incurs a capacity abort in $E$ or 
(2) $\Dset(T_k) = m$ and $T_k$ accesses $\Omega(m)$ distinct metadata base objects in $E$.
\end{theorem}
Here is a high-level overview of the proof technique. 
Let $\kappa$ be the smallest integer such that some fast-path transaction
running step contention-free after a t-quiescent configuration performs $\kappa$ t-reads and incurs a capacity abort. 

We prove that, for all $m \leq \kappa-1$,
there exists a t-complete execution $E_{m}$ and a set $S_m$ with $|S_m|=2^{\kappa-m}$ of
read-only fast-path transactions that access mutually disjoint data sets such that each
transaction in $S_m$
that runs step contention-free from $E_{_m}$ and 
performs t-reads of $m$ distinct t-objects accesses at least one distinct metadata base object
within the execution of each t-read operation. 

We proceed by induction. Assume that the induction statement holds for all $m <kappa -1$.
We prove that a set $S_{m+1}$; $|S_{m+1}| = 2^{\kappa-(m+1)}$ of fast-path transactions, each of which run 
step contention-free after the same t-complete execution $E_{m+1}$, perform
$m+1$ t-reads of distinct t-objects so that at least one distinct metadata base object is accessed within the execution
of each t-read operation.
In our construction, we pick any two new transactions from the set $S_m$ and show that one of them
running step contention-free from a t-complete execution that extends $E_m$ performs
$m+1$ t-reads of distinct t-objects so that at least one distinct metadata base object is accessed within the execution
of each t-read operation.
In this way, the set of transactions is reduced by half in each step of the induction 
until one transaction remains which must have accessed a distinct metadata base object in every one of its $m+1$ t-reads.

Intuitively, since all the transactions that we use in our construction access mutually disjoint data sets, 
we can apply Lemma~\ref{lm:pgtwo} to construct a t-complete execution $E_{m+1}$ such that
each of the fast-path transactions in $S_{m+1}$ when running step contention-free after $E_{m+1}$ perform
$m+1$ t-reads so that at least one distinct metadata base object is accessed within the execution of each t-read operation.

We now present the formal proof:
\begin{proof}
In the constructions which follow, every fast-path transaction executes at most $m+1$ t-reads. 
Let $\kappa$ be the smallest integer such that some fast-path transaction
running step contention-free after a t-quiescent configuration performs $\kappa$ t-reads and incurs a capacity abort. 
We proceed by induction.

\vspace{1mm}\noindent\textbf{Induction statement.}
We prove that, for all $m \leq \kappa-1$,
there exists a t-complete execution $E_{m}$ and a set $S_m$ with $|S_m|=2^{\kappa-m}$ of
read-only fast-path transactions that access mutually disjoint data sets such that each
transaction $T_{f_i}\in S_m$
that runs step contention-free from $E_{_m}$ and 
performs t-reads of $m$ distinct t-objects accesses at least one distinct metadata base object
within the execution of each t-read operation. 
Let $E_{f_i}$ be the step contention-free execution of $T_{f_i}$ after $E_{m}$ and 
let $\Dset(T_{f_i}) = \{X_{i,1}, \dots, X_{i,m}\}$. 

\vspace{1mm}\noindent\textbf{The induction.}
Assume that the induction statement holds for all $m \leq \kappa-1$.
The statement is trivially true for the base case $m=0$ for every $\kappa \in \mathbb{N}$.

We will prove that a set
$S_{m+1}$; $|S_{m+1}| = 2^{\kappa-(m+1)}$ of fast-path transactions, each of which run 
step contention-free from the same t-quiescent configuration $E_{m+1}$, perform
$m+1$ t-reads of distinct t-objects so that at least one distinct metadata base object is accessed within the execution
of each t-read operation.

The construction proceeds in \emph{phases}: there are exactly $\frac{|S_m|}{2}$ phases.
In each phase, we pick any two new transactions from the set $S_m$ and show that one of them
running step contention-free after a t-complete execution that extends $E_m$ performs
$m+1$ t-reads of distinct t-objects so that at least one distinct metadata base object is accessed within the execution
of each t-read operation.

Throughout this proof, we will assume that any two transactions (and resp. execution fragments) with distinct subscripts
represent distinct identifiers.

For all $i\in \{0,\ldots , \frac{|S_m|}{2} -1\}$, 
let $X_{2i+1},X_{2i+2}\not\in \displaystyle\bigcup_{i=0}^{|S_m|-1}\{X_{i,1},\ldots, X_{i,m}\}$ be distinct t-objects and 
let $v$ be the value of $X_{2i+1}$ and $X_{2i+2}$ after $E_m$. 
Let $T_{s_i}$ denote a slow-path transaction which writes $nv \neq v$ to $X_{2i+1}$ and $X_{2i+2}$. 
Let $E_{s_i}$ be the t-complete step contention-free execution fragment of $T_{s_i}$ running immediately after $E_m$.

Let $E'_{s_i}$ be the longest prefix of the execution $E_{s_i}$ such that 
$E_m\cdot E'_{s_i}$ can be extended neither with the 
complete step contention-free execution fragment of transaction 
$T_{f_{2i+1}}$ that performs its $m$ t-reads of $X_{2i+1,1},\ldots , X_{2i+1,m}$ and then 
performs $\Read_{f_{2i+1}}(X_{2i+1})$ and
returns $nv$, nor with the complete step contention-free execution fragment of some transaction $T_{f_{2i+2}}$ that performs 
t-reads of $X_{{2i+2}_{1}},\ldots , X_{2i+2,m}$ and then performs $\Read_{f_{2i+2}}(X_{2i+2})$ and returns $nv$.
Progressiveness and OF TM-liveness of $\M$ stipulates that such an execution exists.

Let $e_{i}$ be the enabled event of $T_{s_{i}}$ in the configuration after $E_m\cdot E'_{s_{i}}$.
By construction, the execution $E_m\cdot E'_{s_{i}}$ can be
extended with at least one of the complete step contention-free executions of transaction
$T_{f_{2i+1}}$ performing $(m+1)$ t-reads of $X_{{2i+1,1}},\ldots , X_{{2i+1,m}},X_{2i+1}$ such that
$\Read_{f_{2i+1}}(X_{2i+1})\rightarrow nv$ or
transaction $T_{f_{2i+2}}$ performing 
t-reads of $X_{{2i+2,1}},\ldots , X_{{2i+2,m}},  X_{2i+2}$ such that $\Read_{f_{2i+2}}(X_{2i+2})\rightarrow nv$.
Without loss of generality, suppose that $T_{f_{2i+1}}$ reads the value of $X_{2i+1}$ to be $nv$ after 
$E_m\cdot E'_{0_{i}} \cdot e_{i}$.

For any $i\in \{0,\ldots , \frac{|S_m|}{2} -1\}$, we will denote by $\alpha_i$ the execution fragment which we will
construct in phase $i$. 
For any $i\in \{0,\ldots , \frac{|S_m|}{2} -1\}$, we prove that
$\M$ has an execution of the form $E_m\cdot \alpha_i$ in which
$T_{f_{2i+1}}$ (or $T_{f_{2i+2}}$) running step contention-free after
a t-complete execution that extends $E_m$ performs $m+1$ t-reads of distinct t-objects so that at least one distinct metadata
base object is accessed within the execution of each first $m$ t-read operations and $T_{f_{2i+1}}$ (or $T_{f_{2i+2}}$) 
is poised to apply an event
after $E_m\cdot \alpha_i$
that accesses a distinct metadata base object during the $(m+1)^{th}$ t-read. 
Furthermore, we will show that $E_m\cdot \alpha_i$ appears t-sequential to $T_{f_{2i+1}}$ (or $T_{f_{2i+2}}$). 

\vspace{1mm}\noindent\textit{(Construction of phase $i$)}

Let $E_{f_{2i+1}}$ (and resp. $E_{f_{2i+2}}$) be the complete step contention-free execution of the t-reads of 
$X_{{2i+1},1},\dots,X_{{2i+1},m}$
(and resp. $X_{{2i+2},1},\ldots , X_{{2i+2},m}$)
running after $E_{m}$ by $T_{f_{2i+1}}$ (and resp. $T_{f_{2i+2}}$). 
By the inductive hypothesis, transaction $T_{f_{2i+1}}$ (and resp. $T_{f_{2i+2}}$)
accesses $m$ distinct metadata objects in the execution $E_m\cdot E_{f_{2i+1}}$ (and resp. $E_m\cdot E_{f_{2i+2}}$).  
Recall that transaction $T_{f_{2i+1}}$ does not conflict with transaction $T_{s_i}$. 
Thus, by Corollary~\ref{lm:pgthree},
$\M$ has an execution of the form
$E_m\cdot E'_{s_{i}} \cdot e_{i} \cdot E_{f_{2i+1}}$ 
(and resp. $E_m\cdot E'_{s_{i}} \cdot e_{i} \cdot E_{f_{2i+2}}$). 

Let $E_{rf_{2i+1}}$ be the complete step contention-free execution fragment of $\Read_{f_{2i+1}}(X_{2i+1})$ that extends
$E_{2i+1}=E_m\cdot E'_{s_{i}} \cdot e_{i} \cdot E_{f_{2i+1}} $. 
By OF TM-liveness, $\Read_{f_{2i+1}}(X_{2i+1})$ must return a matching 
response in $E_{2i+1}\cdot E_{rf_{2i+1}}$.
We now consider two cases.

\vspace{1mm}\noindent\textit{\textbf{Case \RNum{1}}: Suppose $E_{rf_{2i+1}}$ accesses at least one metadata base object 
$b$ not previously accessed by $T_{f_{2i+1}}$.}

Let $E'_{rf_{2i+1}}$ be the longest prefix of $E_{rf_{2i+1}}$ which does not apply 
any primitives to any metadata base object $b$ not previously accessed by $T_{f_{2i+1}}$. 
The execution $E_m\cdot E'_{s_{i}} \cdot e_{i} \cdot E_{f_{2i+1}} \cdot E'_{rf_{2i+1}}$ 
appears t-sequential to $T_{f_{2i+1}}$ because $E_{f_{2i+1}}$ does not contend with $T_{s_{i}}$ on any base object
and any common base object accessed in the execution fragments $E'_{rx_{2i+1}}$ and $E_{s_{i}}$
by $T_{f_{2i+1}}$ and $T_{s_{i}}$ respectively must be data objects contained in $\mathbb{D}$. 
Thus, we have that
$|\Dset(T_{f_{2i+1}})| = m+1$ and that 
$T_{f_{2i+1}}$ accesses $m$ distinct metadata base objects within each of its first $m$ t-read operations
and is poised to access a distinct metadata base object during the execution of the $(m+1)^{th}$ t-read.
In this 
case, let $\alpha_i = E_m\cdot E'_{s_{i}} \cdot e_{i} \cdot E_{f_{2i+1}} \cdot E'_{rf_{2i+1}}$.

\vspace{1.5mm}\noindent\textit{\textbf{Case \RNum{2}}: Suppose 
$E_{rf_{2i+1}}$ does not access any metadata base object not previously accessed by $T_{f_{2i+1}}$.}

In this case, we will first prove the following:
\begin{claim}
\label{cl:iterationone}
$\M$ has an execution of the form 
$E_{2i+2} = E_m\cdot E'_{s_{i}}\cdot e_{i}\cdot {\bar E}_{f_{2i+1}} \cdot E_{f_{2i+2}}$ 
where ${\bar E}_{f_{2i+1}}$ is the t-complete step contention-free execution of $T_{f_{2i+1}}$ in which 
$\Read_{f_{2i+1}}(X_{2i+1})\rightarrow nv$,
$T_{f_{2i+1}}$ invokes $\TryC_{f_{2i+1}}$ and returns a matching response.
\end{claim}
\begin{proof}
Since $E_{rf_{2i+1}}$ does not contain accesses to any distinct metadata base objects,
the execution $E_m\cdot E'_{s_{i}}\cdot e_{i} \cdot E_{f_{2i+1}}\cdot E_{rf_{2i+1}}$ appears t-sequential to $T_{f_{2i+1}}$.
By definition of the event $e_{i}$, $\Read_{f_{2i+1}}(X_{2i+1})$ must access the base object to which the event $e_i$ 
applies a nontrivial
primitive and return the response $nv$ in $E'_{s_{i}}\cdot e_{i}\cdot E_{f_{2i+1}}\cdot E_{rf_{2i+1}}$.
By OF TM-liveness, it follows that $E_m\cdot E'_{s_{i}}\cdot e_{i}\cdot {\bar E}_{f_{2i+1}}$ is an execution of $\M$.

Now recall that $E_m\cdot E'_{s_i} \cdot e_{i}\cdot E_{f_{2i+2}}$ is an execution of $\mathcal{M}$ because
transactions $T_{f_{2i+2}}$ and $T_{s_{i}}$ do not conflict in this execution and thus, cannot contend on any base object.
Finally, because $T_{f_{2i+1}}$ and $T_{f_{2i+2}}$ access disjoint data sets in 
$E_m\cdot E'_{s_{i}}\cdot e_{i}\cdot {\bar E}_{f_{2i+1}} \cdot E_{f_{2i+2}}$,
by Lemma~\ref{lm:pgtwo} again, 
we have that $E_m\cdot E'_{s_{i}} \cdot e_{i}\cdot \bar{E}_{f_{2i+1}} \cdot E_{f_{2i+2}}$ is an execution of $\mathcal{M}$.
\end{proof}
Let $E_{rf_{2i+2}}$ be the complete step contention-free execution fragment of $\Read_{f_{2i+2}}(X_{2i+2})$ after 
$E_m\cdot E'_{s_{i}}\cdot e_{i}\cdot {\bar E}_{f_{2i+1}} \cdot E_{f_{2i+2}}$. 
By the induction hypothesis and Claim~\ref{cl:iterationone}, 
transaction $T_{f_{2i+2}}$ must access $m$ distinct metadata base objects
in the execution $E_m\cdot E'_{s_{i}}\cdot e_{i}\cdot {\bar E}_{f_{2i+1}} \cdot E_{f_{2i+2}}$. 

If $E_{rf_{2i+2}}$ accesses some metadata base object, then
by the argument given in Case~\RNum{1} applied to transaction $T_{f_{2i+2}}$, we get that
$T_{f_{2i+2}}$ accesses $m$ distinct metadata base objects within each of the first $m$ t-read operations
and is poised to access a distinct metadata base object during the execution of the $(m+1)^{th}$ t-read.

Thus, suppose that $E_{rf_{2i+2}}$ does not access any metadata base object previously accessed by $T_{f_{2i+2}}$. 
We claim that this is impossible and proceed to derive a contradiction. 
In particular, $E_{rf_{2i+2}}$ does not contend with $T_{s_i}$ on any metadata base object.
Consequently, the execution $E_m\cdot  E'_{s_{i}} \cdot e_{i} \cdot {\bar E}_{f_{2i+1}} \cdot E_{f_{2i+2}}$ 
appears t-sequential to $T_{x_{2i+2}}$ since 
$E_{rx_{2i+2}}$ only contends with $T_{s_{i}}$ on base objects in $\mathbb{D}$. 
It follows that $E_{2i+2}\cdot E_{rf_{2i+2}}$ must also appear t-sequential to $T_{f_{2i+2}}$ and so 
$E_{rf_{2i+2}}$ cannot abort.
Recall that the base object, say $b$, to which $T_{s_{i}}$ applies a nontrivial primitive in the event $e_{i}$
is accessed by $T_{f_{2i+1}}$ in $E_m\cdot E'_{s_{i}}\cdot e_{i}\cdot {\bar E}_{f_{2i+1}} \cdot E_{f_{2i+2}}$; 
thus, $b \in \mathbb{D}_{X_{2i+1}}$. 
Since $X_{2i+1} \not\in \Dset(T_{f_{2i+2}})$,
$b$ cannot be accessed by $T_{f_{2i+2}}$. Thus, the execution 
$E_m\cdot E'_{s_{i}}\cdot e_{i}\cdot {\bar E}_{f_{2i+1}} \cdot E_{f_{2i+2}} \cdot E_{rf_{2i+2}}$ is indistinguishable to 
$T_{f_{2i+2}}$ from the execution ${\hat E}_i\cdot E'_{s_{i}} \cdot E_{f_{2i+2}}\cdot E_{rf_{2i+2}}$
in which $\Read_{f_{2i+2}}(X_{2i+2})$ must return the response $v$ (by construction of $E'_{s_i}$).

But we observe now that the execution 
$E_m\cdot E'_{s_{i}}\cdot e_{i}\cdot {\bar E}_{f_{2i+1}} \cdot E_{f_{2i+2}}\cdot E_{rf_{2i+2}}$ is not opaque. 
In any serialization corresponding to this execution, 
$T_{s_{i}}$ must be committed and must precede $T_{f_{2i+1}}$ because $T_{f_{2i+1}}$ read $nv$ from $X_{2i+1}$.
Also, transaction $T_{f_{2i+2}}$ must precede $T_{s_{i}}$ 
because $T_{f_{2i+2}}$ read $v$ from $X_{2i+2}$. 
However $T_{f_{2i+1}}$ must precede $T_{f_{2i+2}}$ to respect real-time ordering of transactions.
Clearly, there exists no such serialization---contradiction. 

Letting $E'_{rf_{2i+2}}$ be the longest prefix of $E_{rf_{2i+2}}$ which does not access a base object $b\in \mathbb{M}$ not previously accessed by
 $T_{f_{2i+2}}$, we can let 
 $\alpha_i=E'_{s_{i}}\cdot e_{i}\cdot {\bar E}_{f_{2i+1}} \cdot E_{f_{2i+2}}\cdot E'_{rf_{2i+2}}$ in this case.

Combining Cases~\RNum{1} and~\RNum{2}, the following claim holds.
\begin{claim}
\label{cl:final}
For each $i\in \{0,\ldots , \frac{|S_m|}{2} -1\}$, $\M$ has an execution of the form 
$E_m\cdot \alpha_i$ in which
\begin{enumerate}
\item[(1)]
some fast-path transaction 
$T_{i}\in \ms{txns}(\alpha_i)$ performs t-reads of $m+1$ distinct t-objects so that at least one distinct metadata base object
is accessed within the execution of each of the first $m$ t-reads, 
$T_i$ is poised to access a distinct metadata base object
after $E_m\cdot \alpha_i$ during the execution of the $(m+1)^{th}$ t-read and the execution appears t-sequential to $T_i$,
\item[(2)] 
the two fast-path transactions in the execution fragment $\alpha_i$ do not contend on the same base object.
\end{enumerate}
\end{claim}

\vspace{1mm}\noindent\textit{(Collecting the phases)}

We will now describe how we can construct the set $S_{m+1}$ of fast-path transactions
from these $\frac{|S_m|}{2}$ phases and force each of them to access $m+1$ distinct metadata base objects when
running step contention-free after the same t-complete execution.

For each $i\in \{0,\ldots , \frac{|S_m|}{2} -1\}$,
let $\beta_i$ be the subsequence of the execution $\alpha_i$ consisting of all the events of the fast-path transaction
that is poised to access a $(m+1)^{th}$ distinct metadata base object.
Henceforth, we denote by $T_i$ the fast-path transaction that participates in $\beta_i$.
Then, from Claim~\ref{cl:final}, it follows that, for each $i\in \{0,\ldots , \frac{|S_m|}{2} -1\}$, 
$\M$ has an execution of the form $E_m\cdot E'_{s_i}\cdot e_i \cdot \beta_i$ in which
the fast-path transaction $T_i$ performs t-reads of $m+1$ distinct t-objects 
so that at least one distinct metadata base object
is accessed within the execution of each of the first $m$ t-reads, 
$T_i$ is poised to access a distinct metadata base object
after $E_m\cdot E'_{s_i}\cdot e_i \cdot \beta_i$ during the execution of the $(m+1)^{th}$ t-read and the 
execution appears t-sequential to $T_i$. 

The following result is a corollary to the above claim that is obtained by applying the definition of ``appears t-sequential''.
Recall that $E'_{s_i}\cdot e_i$ is the t-incomplete execution of slow-path transaction $T_{s_i}$ that accesses
t-objects $X_{2i+1}$ and $X_{2i+2}$.
\begin{corollary}
\label{cr:pgone}
For all $i\in \{0,\ldots , \frac{|(S_m|}{2} -1\}$, $\M$ has an execution of the form
$E_m\cdot E_{{i}}\cdot \beta_i$ such that the configuration after $E_m\cdot E_i$ is t-quiescent,
$\ms{txns}(E_i)\subseteq \{T_{s_i}\}$
and $\Dset(T_{s_i})\subseteq \{X_{2i+1},X_{2i+2}\}$ in $E_i$.
\end{corollary}
We can represent the execution $\beta_i=\gamma_{i}\cdot \rho_{i}$ where
fast-path transaction $T_i$ performs complete t-reads of $m$ distinct t-objects in $\gamma_{i}$
and then performs an incomplete t-read of the $(m+1)^{th}$ t-object in $\rho_{i}$ in which $T_i$ only
accesses base objects in $\displaystyle\bigcup_{X\in DSet(T_i)}\{X\}$.
Recall that $T_i$ and $T_{s_i}$ do not contend on the same base object
in the execution $E_m\cdot E_i\cdot \gamma_i$.
Thus, for all $i\in \{0,\ldots , \frac{|S_m|}{2} -1\}$,
$\M$ has an execution of the form $E_m\cdot \gamma_i\cdot E_i \cdot \rho_i$.

Observe that the fast-path transaction $T_i\in \gamma_i$ does not access any t-object
that is accessed by any slow-path transaction in the execution fragment $E_0\cdots E_{\frac{|S_m|}{2} -1}$.
By Lemma~\ref{lm:finallm}, there exists a t-complete step contention-free execution fragment $E'$
that is similar to $E_0\cdots E_{\frac{|S_m|}{2} -1}$
such that
for all $i\in \{0,\ldots ,  \frac{|S_m|}{2} -1\}$, $\M$ has an execution of the form
$E_{m}\cdot E'\cdot  \gamma_i \cdot \rho_i$.
By our construction, the enabled event of each fast-path transaction $T_i\in \beta_i$ in this execution
is an access to a distinct metadata base object.

Let $S_{m+1}$ denote the set of all fast-path transactions that participate in the execution fragment
$\beta_0\cdots \beta_{ \frac{|(S_m|}{2} -1}$ and $E_{m+1}=E_m\cdot E'$.
Thus, $|S_{m+1}|$ fast-path transactions, each of which run 
step contention-free from the same t-quiescent configuration, perform
$m+1$ t-reads of distinct t-objects so that at least one distinct metadata base object is accessed within the execution
of each t-read operation. This completes the proof.
\end{proof}
%
\section{Instrumentation-optimal progressive HyTM}
\label{sec:p4c4s2}
\begin{algorithm}[!h]
\caption{Progressive opaque HyTM implementation that provides uninstrumented writes and invisible reads; code for process $p_i$
executing transaction $T_k$}
\label{alg:inswrite}
\begin{algorithmic}[1]
  	\begin{multicols}{2}
  	{
  	\footnotesize
	\Part{Shared objects}{
		\State $v_j \in \mathbb{D}$, for each t-object $X_j$ 
		\State ~~~~~allows reads, writes and cas
		\State $r_{j} \in \mathbb{M}$, for each t-object $X_j$
		\State ~~~~~allows reads, writes and cas
	}\EndPart	
	\Statex
	\Part{Local objects}{
		\State $\ms{Lset}(T_k) \subseteq \Wset(T_k)$, initially empty
		\State $\ms{Oset}(T_k) \subseteq \Wset(T_k)$, initially empty
	}\EndPart
	\Statex
	\textbf{Code for slow-path transactions}
	\Statex
	\Part{\Read$_k(X_j)$}\quad\Comment{slow-path}{
		\If{$X_j \not\in \Rset_k$}
		  \State $[\textit{ov}_j,k_j] := \Read(v_j)$ \label{line:hread2}
		  \State $\Rset(T_k) := \Rset(T_k)\cup\{X_j,[\textit{ov}_j,k_j]\}$ \label{line:hrset}
		  \If{$r_j\neq 0$} \label{line:habort0}
		    \Return $A_k$ \EndReturn
		  \EndIf
		  \If{$\exists X_j \in Rset(T_k)$:$(\textit{ov}_j,k_j)\neq \Read(v_j)$} \label{line:hvalid}
			\Return $A_k$ \EndReturn
		  \EndIf
%
		  \Return $\textit{ov}_j$ \EndReturn
		\Else
		    
		  \State $\textit{ov}_j :=\Rset(T_k).\lit{locate}(X_j)$
		  \Return $\textit{ov}_j$ \EndReturn
		\EndIf
   	 }\EndPart
	\Statex
	\Part{\Write$_k(X_j,v)$}\quad\Comment{slow-path}{
		
			\State $(\textit{ov}_j,k_j) := \Read(v_j)$
			\State $\textit{nv}_j := v$
			\State $\Wset(T_k) := \Wset(T_k)\cup\{X_j,[\textit{ov}_j,k_j]\}$
			\Return $\ok$ \EndReturn
		
   	}\EndPart
	\Statex
	
	\Part{\TryC$_k$()}\quad\Comment{slow-path}{
		\If{$\Wset(T_k)= \emptyset$}
			\Return $C_k$ \EndReturn \label{line:hreturn}
		\EndIf
		\State locked := $\lit{acquire}(\Wset(T_k))$\label{line:hacq} 
		\If{$\neg$ locked} \label{line:habort2} 
	 		\Return $A_k$ \EndReturn
	 	\EndIf
		\If{$\lit{isAbortable}()$} \label{line:habort3}
			\State $\lit{release}( \ms{Lset}(T_k))$ 
			\Return $A_k$ \EndReturn
		\EndIf
		\ForAll{$X_j \in \Wset(T_k)$}
	 		 \If{ $v_j.\lit{cas}([ov_j,k_j],[\textit{nv}_j,k])$} \label{line:hwrite}
			      \State $\ms{Oset}(T_k):=\ms{Oset}(T_k)\cup \{X_j\}$
			 \Else
			 
			      \State $\lit{undo}(\ms{Oset}(T_k))$
			 
			 \EndIf
			 
	 	\EndFor		
		\State $\lit{release}(\Wset(T_k))$   \label{line:hrel}		
   		\Return $C_k$ \EndReturn
   	 }\EndPart		
	 
 	\newpage
 	\Part{Function: $\lit{acquire}(Q$)}{
   		\ForAll{$X_j \in Q$}	
			\If{$r_j.\lit{cas}(0,1)$} \label{line:hacq1}
			  \State $\ms{Lset}(T_k) := \ms{Lset}(T_k)\cup \{X_j\}$
			\Else
			  \State $\lit{release}(\ms{Lset}(T_k))$
			  \Return $\false$ \EndReturn
			\EndIf
			
		\EndFor
		
		\Return $\true$ \EndReturn 
	}\EndPart		
	 \Statex
	 \Part{Function: $\lit{release}(Q)$}{
  		\ForAll{$X_j \in Q$}	
 			\State $r_j.\Write(0)$ \label{line:hrel1}
		\EndFor
		\Return $ok$ \EndReturn
	}\EndPart
	
	\Statex
	
	\Part{Function: $\lit{undo}(\ms{Oset}(T_k))$}{
		\ForAll{$X_j \in \ms{Oset}(T_k)$}
		    \State $v_j.\lit{cas}([nv_j,k],[ov_j,k_j])$
		 \EndFor
		\State $\lit{release}( \ms{Wset}(T_k))$ 
		\Return $A_k$ \EndReturn
	 }\EndPart
	
	\Statex
	 \Part{Function: $\lit{isAbortable()}$ }{
		\If{$\exists X_j \in \Rset(T_k)$: $X_j\not\in \Wset(T_k)\wedge \Read(r_{j}) \neq 0$}
			\Return $\true$ \EndReturn
		\EndIf
		\If{$\exists X_j \in Rset(T_k)$:$[\textit{ov}_j,k_j]\neq \Read(v_j)$} \label{line:hvalid}
			\Return $\true$ \EndReturn
		\EndIf
		\Return $\false$ \EndReturn
	}\EndPart
	\Statex
	\Statex
	\textbf{Code for fast-path transactions}
	\Statex
	\Part{$\textit{read}_k(X_j)$}\quad\Comment{fast-path}{
		\State $[\textit{ov}_j,k_j] := \Read(v_j)$ \Comment{cached read} \label{line:hlin1}
		
		\If{$\Read(r_{j}) \neq 0$}  
		\label{line:hread}
			\Return $A_k$ \EndReturn
		\EndIf
		
		\Return $\textit{ov}_j$ \EndReturn
		
   	 }\EndPart
	\Statex
	\Part{$\textit{write}_k(X_j,v)$}{\quad\Comment{fast-path}
		\State $\Write(v_j,[\textit{nv}_j,k])$ \Comment{cached write} \label{line:hlin2}
		\Return $\ok$ \EndReturn
		
   	}\EndPart
	\Statex
	
	\Part{$\textit{tryC}_k$()}{\quad\Comment{fast-path}
		\State $\ms{commit-cache}_i$ \label{line:hlin3} \Comment{returns $C_k$ or $A_k$}
   	 }\EndPart		

	}
	\end{multicols}
  \end{algorithmic}
\end{algorithm}
We prove that the lower bound in Theorem~\ref{linearlowerbound} is
tight by describing an `instrumentation-optimal'' 
HyTM implementation (Algorithm~\ref{alg:inswrite}) that is opaque, progressive, provides wait-free TM-liveness,
uses \emph{invisible reads}.

\vspace{1mm}\noindent\textbf{Base objects.}
For every t-object $X_j$, our implementation maintains a base object $v_j\in \mathbb{D}$ that stores the value of $X_j$
and a metadata base object $r_{j}$, which is a \emph{lock bit} that stores $0$ or $1$.

\vspace{1mm}\noindent\textbf{Fast-path transactions.}
For a fast-path transaction $T_k$, the $\Read_k(X_j)$ implementation first reads $r_j$ 
to check if $X_j$ is locked by a concurrent updating transaction. 
If so, it returns $A_k$,
else it returns the value of $X_j$.
Updating fast-path transactions use uninstrumented writes:
$\Write (X_j,v)$ simply stores the cached state of $X_j$ along with its value $v$ and
if the cache has not been invalidated, updates the shared memory
during $\TryC_k$ by invoking the $\ms{commit-cache}$ primitive.

\vspace{1mm}\noindent\textbf{Slow-path read-only transactions.}
Any $\Read_k(X_j)$ invoked by a slow-path transaction first reads the value of the object from $v_j$, 
checks if $r_j$ is set
and then performs \emph{value-based validation} on its entire read set to check if any of them have been modified. 
If either of these conditions is true,
the transaction returns $A_k$. Otherwise, it returns the value of $X_j$. 
A read-only transaction simply returns $C_k$ during the tryCommit.

\vspace{1mm}\noindent\textbf{Slow-path updating transactions.}
The $\Write_k(X,v)$ implementation of a slow-path transaction stores
$v$ and the current value of $X_j$ locally, 
deferring the actual update in shared memory to tryCommit. 

During $\TryC_k$, an updating slow-path transaction $T_k$ attempts to obtain exclusive write access to its 
entire write set as follows:
for every t-object $X_j \in \Wset(T_k)$, it writes $1$ to each base
object $r_{j}$ by performing a \emph{compare-and-set} (\emph{cas})
primitive that checks if the value of $r_j$ is not $1$ 
and, if so, replaces it with $1$.   
If the \emph{cas} fails, then $T_k$ releases the locks on all objects $X_{\ell}$ 
it had previously acquired 
by writing $0$ to $r_{\ell}$ and then returns $A_k$. Intuitively, if the \emph{cas} fails, some concurrent transaction
is performing a t-write to a t-object in $\Wset(T_k)$.
If all the locks on the write set were acquired successfully,
$T_k$ checks if any t-object in $\Rset(T_k)$ is concurrently being updated by another transaction
and then performs value-based validation of the read set. If a conflict is detected from the these checks,
the transaction is aborted.
Finally, $\TryC_k$ attempts to write the values of the t-objects via \emph{cas} operations.
If any \emph{cas} on the individual base objects fails, there must be a concurrent fast-path writer, and so $T_k$ rolls back the
state of the base objects that were updated, releases locks on its write set and returns $A_k$. 
The roll backs are performed with \emph{cas} operations,
skipping any which fail to allow for concurrent fast-path writes to
locked locations. Note that if a concurrent read operation of a
fast-path transaction $T_{\ell}$ finds an ``invalid'' value in $v_j$ that was
written by such transaction $T_k$ but has not
been rolled back yet, then $T_{\ell}$ either incurs a tracking set
abort later because $T_k$ has updated $v_j$ or finds $r_j$ to be $1$. 
In both cases, the read operation of $T_{\ell}$ aborts.     

The implementation uses invisible reads (no nontrivial primitives are applied by reading transactions).
Every t-operation returns a matching response within a finite number of its steps.

\vspace{1mm}\noindent\textbf{Complexity.}
Every t-read operation performed by a fast-path transaction accesses a metadata base object
once (the lock bit corresponding to the t-object), 
which is the price to pay for detecting conflicting updating slow-path
transactions. Write operations of fast-path transactions are uninstrumented. 
%
\begin{lemma}
\label{lm:opacityh1}
Algorithm~\ref{alg:inswrite} implements an opaque TM.
\end{lemma}
\begin{proof}
Let $E$ by any execution of Algorithm~\ref{alg:inswrite}. 
Since opacity is a safety property, it is sufficient to prove that every finite execution is opaque~\cite{icdcs-opacity}.
Let $<_E$ denote a total-order on events in $E$.

Let $H$ denote a subsequence of $E$ constructed by selecting
\emph{linearization points} of t-operations performed in $E$.
The linearization point of a t-operation $op$, denoted as $\ell_{op}$ is associated with  
a base object event or an event performed during 
the execution of $op$ using the following procedure. 

\vspace{1mm}\noindent\textbf{Completions.}
First, we obtain a completion of $E$ by removing some pending
invocations or adding responses to the remaining pending invocations
as follows:
\begin{itemize}
\item
incomplete $\Read_k$, $\Write_k$ operation performed by a slow-path transaction $T_k$ is removed from $E$;
an incomplete $\TryC_k$ is removed from $E$ if $T_k$ has not performed any write to a base object $r_j$; $X_j \in \Wset(T_k)$
in Line~\ref{line:hwrite}, otherwise it is completed by including $C_k$ after $E$.
\item
every incomplete $\Read_k$, $\TryA_k$, $\Write_k$ and $\TryC_k$ performed by a fast-path transaction $T_k$ is removed from $E$.
\end{itemize}
\vspace{1mm}\noindent\textbf{Linearization points.}
Now a linearization $H$ of $E$ is obtained by associating linearization points to
t-operations in the obtained completion of $E$.
For all t-operations performed a slow-path transaction $T_k$, linearization points as assigned as follows:
\begin{itemize}
\item For every t-read $op_k$ that returns a non-A$_k$ value, $\ell_{op_k}$ is chosen as the event in Line~\ref{line:hread2}
of Algorithm~\ref{alg:inswrite}, else, $\ell_{op_k}$ is chosen as invocation event of $op_k$
\item For every $op_k=\Write_k$ that returns, $\ell_{op_k}$ is chosen as the invocation event of $op_k$
\item For every $op_k=\TryC_k$ that returns $C_k$ such that $\Wset(T_k)
  \neq \emptyset$, $\ell_{op_k}$ is associated with the first write to a base object performed by $\lit{release}$
  when invoked in Line~\ref{line:hrel}, 
  else if $op_k$ returns $A_k$, $\ell_{op_k}$ is associated with the invocation event of $op_k$
\item For every $op_k=\TryC_k$ that returns $C_k$ such that $\Wset(T_k) = \emptyset$, 
$\ell_{op_k}$ is associated with Line~\ref{line:hreturn}
\end{itemize}
For all t-operations performed a fast-path transaction $T_k$, linearization points as assigned as follows:
\begin{itemize}
\item For every t-read $op_k$ that returns a non-A$_k$ value, $\ell_{op_k}$ is chosen as the event in Line~\ref{line:hlin1}
of Algorithm~\ref{alg:inswrite}, else, $\ell_{op_k}$ is chosen as invocation event of $op_k$
\item
For every $op_k$ that is a $\TryC_k$, $\ell_{op_k}$ is the $\ms{commit-cache}_k$ primitive invoked by $T_k$
\item
For every $op_k$ that is a $\Write_k$, $\ell_{op_k}$ is the event in Line~\ref{line:hlin2}.
\end{itemize}
$<_H$ denotes a total-order on t-operations in the complete sequential history $H$.

\vspace{1mm}\noindent\textbf{Serialization points.}
The serialization of a transaction $T_j$, denoted as $\delta_{T_j}$ is
associated with the linearization point of a t-operation 
performed by the transaction.

We obtain a t-complete history ${\bar H}$ from $H$ as follows. 
A serialization $S$ is obtained by associating serialization points to transactions in ${\bar H}$ as follows:
for every transaction $T_k$ in $H$ that is complete, but not t-complete, 
we insert $\textit{tryC}_k\cdot A_k$ immediately 
after the last event of $T_k$ in $H$. 
\begin{itemize}
\item If $T_k$ is an updating transaction that commits, then $\delta_{T_k}$ is $\ell_{\TryC_k}$
\item If $T_k$ is a read-only or aborted transaction,
then $\delta_{T_k}$ is assigned to the linearization point of the last t-read that returned a non-A$_k$ value in $T_k$
\end{itemize}
$<_S$ denotes a total-order on transactions in the t-sequential history $S$.
\begin{claim}
\label{cl:hseq}
If $T_i \prec_{H}T_j$, then $T_i <_S T_j$
\end{claim}
\begin{proof}
This follows from the fact that for a given transaction, its
serialization point is chosen between the first and last event of the transaction
implying if $T_i \prec_{H} T_j$, then $\delta_{T_i} <_{E} \delta_{T_j}$ implies $T_i <_S T_j$.
\end{proof}
\begin{claim}
\label{cl:hreadfrom}
$S$ is legal.
\end{claim}
\begin{proof}
We claim that for every $\Read_j(X_m) \rightarrow v$, there exists some slow-path transaction $T_i$ (or resp. fast-path)
that performs $\Write_i(X_m,v)$ and completes the event in Line~\ref{line:hwrite} (or resp. Line~\ref{line:hlin2}) such that
$\Read_j(X_m) \not\prec_H^{RT} \Write_i(X_m,v)$.

Suppose that $T_i$ is a slow-path transaction:
since $\Read_j(X_m)$ returns the response $v$, the event in Line~\ref{line:hread2}
succeeds the event in Line~\ref{line:hwrite} performed by $\TryC_i$. 
Since $\Read_j(X_m)$ can return a non-abort response only after $T_i$ writes $0$ to $r_m$ in
Line~\ref{line:hrel1}, $T_i$ must be committed in $S$.
Consequently,
$\ell_{\TryC_i} <_E \ell_{\Read_j(X_m)}$.
Since, for any updating
committing transaction $T_i$, $\delta_{T_i}=\ell_{\TryC_i}$, it follows that
$\delta_{T_{i}} <_E \delta_{T_{j}}$.

Otherwise if $T_i$ is a fast-path transaction, then clearly $T_i$ is a committed transaction in $S$.
Recall that $\Read_j(X_m)$ can read $v$ during the event in Line~\ref{line:hread2}
only after $T_i$ applies the $\ms{commit-cache}$ primitive.
By the assignment of linearization points, 
$\ell_{\TryC_i} <_E \ell_{\Read_j(X_m)}$ and thus, $\delta_{T_{i}} <_E \ell_{\Read_j(X_m)}$.

Thus, to prove that $S$ is legal, it suffices to show that  
there does not exist a
transaction $T_k$ that returns $C_k$ in $S$ and performs $\Write_k(X_m,v')$; $v'\neq v$ such that $T_i <_S T_k <_S T_j$. 

$T_i$ and $T_k$ are both updating transactions that commit. Thus, 
\begin{center}
($T_i <_S T_k$) $\Longleftrightarrow$ ($\delta_{T_i} <_{E} \delta_{T_k}$) \\
($\delta_{T_i} <_{E} \delta_{T_k}$) $\Longleftrightarrow$ ($\ell_{\TryC_i} <_{E} \ell_{\TryC_k}$) 
\end{center}
Since, $T_j$ reads the value of $X$ written by $T_i$, one of the following is true:
$\ell_{\TryC_i} <_{E} \ell_{\TryC_k} <_{E} \ell_{\Read_j(X_m)}$ or
$\ell_{\TryC_i} <_{E} \ell_{\Read_j(X_m)} <_{E} \ell_{\TryC_k}$.

Suppose that $\ell_{\TryC_i} <_{E} \ell_{\TryC_k} <_{E} \ell_{\Read_j(X_m)}$.

(\textit{Case \RNum{1}:}) $T_i$ and $T_k$ are slow-path transactions.

Thus, $T_k$ returns a response from the event in Line~\ref{line:hacq} 
before the read of the base object associated with $X_m$ by $T_j$ in Line~\ref{line:hread2}. 
Since $T_i$ and $T_k$ are both committed in $E$, $T_k$ returns \emph{true} from the event in
Line~\ref{line:hacq} only after $T_i$ writes $0$ to $r_{m}$ in Line~\ref{line:hrel1}.

If $T_j$ is a slow-path transaction, 
recall that $\Read_j(X_m)$ checks if $X_j$ is locked by a concurrent transaction, 
then performs read-validation (Line~\ref{line:habort0}) before returning a matching response. 
We claim that $\Read_j(X_m)$ must return $A_j$ in any such execution.

Consider the following possible sequence of events: 
$T_k$ returns \emph{true} from \emph{acquire} function invocation, 
updates the value of $X_m$ to shared-memory (Line~\ref{line:hwrite}), 
$T_j$ reads the base object $v_m$ associated with $X_m$, 
$T_k$ releases $X_m$ by writing $0$ to $r_{m}$ and finally $T_j$ performs the check in Line~\ref{line:habort0}. 
But in this case, $\Read_j(X_m)$ is forced to return the value $v'$ written by $T_m$--- 
contradiction to the assumption that $\Read_j(X_m)$ returns $v$. 

Otherwise suppose that $T_k$ acquires exclusive access to $X_m$ by writing $1$ to $r_{m}$ and returns \emph{true}
from the invocation of \emph{acquire}, updates $v_m$ in Line~\ref{line:hwrite}), 
$T_j$ reads $v_m$, $T_j$ performs the check in Line~\ref{line:habort0} and finally $T_k$ 
releases $X_m$ by writing $0$ to $r_{m}$. 
Again, $\Read_j(X_m)$ must return $A_j$ since $T_j$ reads that $r_{m}$ is $1$---contradiction.

A similar argument applies to the case that $T_j$ is a fast-path transaction.
Indeed, since every \emph{data} base object read by $T_j$ is contained in its tracking set, if any concurrent
transaction updates any t-object in its read set, $T_j$ is aborted immediately by our model(cf. Section~\ref{sec:p4c1s2}).

Thus, $\ell_{\TryC_i} <_E \ell_{\Read_j(X)} <_{E} \ell_{\TryC_k}$.

(\textit{Case \RNum{2}:}) $T_i$ is a slow-path transaction and $T_k$ is a fast-path transaction.
Thus, $T_k$ returns $C_k$ 
before the read of the base object associated with $X_m$ by $T_j$ in Line~\ref{line:hread2}, but after the response
of \emph{acquire} by $T_i$ in Line~\ref{line:hacq}.
Since $\Read_j(X_m)$ reads the value of $X_m$ to be $v$ and not $v'$, $T_i$ performs the \emph{cas}
to $v_m$ in Line~\ref{line:hwrite} after the $T_k$ performs the $\ms{commit-cache}$ primitive (since if
otherwise, $T_k$ would be aborted in $E$).
But then the \emph{cas} on $v_m$ performed by $T_i$ would return $\false$ and $T_i$ would return $A_i$---contradiction.

(\textit{Case \RNum{3}:}) $T_k$ is a slow-path transaction and $T_i$ is a fast-path transaction.
This is analogous to the above case.

(\textit{Case \RNum{4}:}) $T_i$ and $T_k$ are fast-path transactions.
Thus, $T_k$ returns $C_k$ 
before the read of the base object associated with $X_m$ by $T_j$ in Line~\ref{line:hread2}, but before $T_i$
returns $C_i$ (this follows from Observations~\ref{ob:one} and \ref{ob:two}).
Consequently, $\Read_j(X_m)$ must read the value of $X_m$ to be $v'$ and return $v'$---contradiction.

We now need to prove that $\delta_{T_{j}}$ indeed precedes $\ell_{\TryC_k}$ in $E$.

Consider the two possible cases:
\begin{itemize}
\item
Suppose that $T_j$ is a read-only transaction. 
Then, $\delta_{T_j}$ is assigned to the last t-read performed by $T_j$ that returns a non-A$_j$ value. 
If $\Read_j(X_m)$ is not the last t-read that returned a non-A$_j$ value, then there exists a $\Read_j(X')$ such that 
$\ell_{\Read_j(X_m)} <_{E} \ell_{\TryC_k} <_E \ell_{read_j(X')}$.
But then this t-read of $X'$ must abort by performing the checks in Line~\ref{line:habort0} or incur a tracking set abort---contradiction.
\item
Suppose that $T_j$ is an updating transaction that commits, then $\delta_{T_j}=\ell_{\TryC_j}$ which implies that
$\ell_{read_j(X)} <_{E} \ell_{\TryC_k} <_E \ell_{\TryC_j}$. Then, $T_j$ must neccesarily perform the checks
in Line~\ref{line:habort3} and return $A_j$ or incur a tracking set abort---contradiction to the assumption that $T_j$ is a committed transaction.
\end{itemize}
The proof follows.
\end{proof}
The conjunction of Claims~\ref{cl:hseq} and \ref{cl:hreadfrom} establish that Algorithm~\ref{alg:inswrite} is opaque.
\end{proof}
\begin{theorem}
\label{th:inswrite}
There exists an opaque HyTM implementation that provides uninstrumented writes, invisible reads, progressiveness
and wait-free TM-liveness such that
in its every execution $E$, every read-only fast-path transaction $T\in \ms{txns}(E)$
accesses $O(|\Rset(T)|)$ distinct metadata base objects.
\end{theorem}
\begin{proof}
\textit{(Opacity)} Follows from Lemma~\ref{lm:opacityh1}.

\textit{(TM-liveness and TM-progress)}
Since none of the implementations of the t-operations in Algorithm~\ref{alg:inswrite}
contain unbounded loops or waiting statements, Algorithm~\ref{alg:inswrite} provides wait-free TM-liveness,
\emph{i.e.}, every t-operation returns a matching response after taking a finite number of steps.

Consider the cases under which a slow-path transaction $T_k$ may be aborted in any execution.
\begin{itemize}
\item
Suppose that there exists a $\Read_k(X_j)$ performed by $T_k$ that returns $A_k$
from Line~\ref{line:habort0}.
Thus, there exists a transaction
that has written $1$ to $r_{j}$ in Line~\ref{line:hacq1}, but has not yet written
$0$ to $r_{j}$ in Line~\ref{line:hrel1} or
some t-object in $\Rset(T_k)$ has been updated since its t-read by $T_k$.
In both cases, there exists a concurrent transaction performing a 
t-write to some t-object in $\Rset(T_k)$, thus forcing a read-write conflict.
\item
Suppose that $\TryC_k$ performed by $T_k$ that returns $A_k$
from Line~\ref{line:habort2}.
Thus, there exists a transaction
that has written $1$ to $r_{j}$ in Line~\ref{line:hacq1}, but has not yet written
$0$ to $r_{j}$ in Line~\ref{line:hrel1}. Thus, $T_k$ encounters write-write conflict with another
transaction that concurrently attempts to update a t-object in $\Wset(T_k)$.
\item
Suppose that $\TryC_k$ performed by $T_k$ that returns $A_k$
from Line~\ref{line:habort3}.
Since $T_k$ returns $A_k$ from Line~\ref{line:habort3} for the same reason it
returns $A_k$ after Line~\ref{line:habort0}, the proof follows.
\end{itemize}
Consider the cases under which a fast-path transaction $T_k$ may be aborted in any execution $E$.
\begin{itemize}
\item
Suppose that a $\Read_k(X_m)$ performed by $T_k$ returns $A_k$
from Line~\ref{line:hread}.
Thus, there exists a concurrent slow-path transaction that is pending in its tryCommit and 
has written $1$ to $r_m$, but not released the lock on $X_m$ i.e. $T_k$ conflicts with another transaction in $E$.
\item
Suppose that $T_k$ returns $A_k$
while performing a cached access of some base object $b$ via a trivial (and resp. nontrivial) primitive. 
Indeed, this is possible only if some concurrent transaction writes (and resp. reads or writes) to $b$.
However, two transactions $T_k$ and $T_m$ may contend on $b$ in $E$
only if there exists $X\in\Dset(T_i)\cap\Dset(T_j)$ and $X\in\Wset(T_i)\cup\Wset(T_j)$.
from Line~\ref{line:habort2}.
The same argument applies for the case when $T_k$ returns $A_k$
while performing $\ms{commit-cache}_k$ in $E$.
\end{itemize}
\textit{(Complexity)}
The implementation uses uninstrumented writes since each $\Write_k(X_m)$ simply writes to $v_m \in \mathbb{D}_{X_{m}}$
and does not access any metadata base object.
The complexity of each $\Read_k(X_m)$ is a single access to a metadata base object $r_m$ in Line~\ref{line:hread}
that is not accessed any other transaction $T_i$ unless $X_m \in \Dset(T_i)$.
while the $\TryC_k$ just calls $\ms{cache-commit}_k$ that returns $C_k$.
Thus, each read-only transaction $T_k$ accesses $O(|\Rset(T_k)|)$ distinct metadata base objects in any execution.
\end{proof}
%
\section{Providing partial concurrency at low cost}
\label{sec:p4c4s3}
\begin{algorithm}[!h]
\caption{Opaque HyTM implementation with progressive slow-path and sequential fast-path TM-progress; code for $T_k$ by process $p_i$}
\label{alg:inswrite2}
\begin{algorithmic}[1]
  	\begin{multicols}{2}
  	{
  	\footnotesize
	\Part{Shared objects}{
		\State $v_j \in \mathbb{D}$, for each t-object $X_j$ 
		\State ~~~~~allows reads, writes and cas
		\State $r_{j} \in \mathbb{M}$, for each t-object $X_j$
		\State ~~~~~allows reads, writes and cas
		\State $\ms{fa}$, fetch-and-add object 
	}\EndPart	
	
	\Statex	
	\textbf{Code for slow-path transactions}
	\Part{\TryC$_k$()}{\quad\Comment{slow-path}
		\If{$\Wset(T_k)= \emptyset$}
			\Return $C_k$  \EndReturn
		\EndIf
				
		\State locked := $\lit{acquire}(\Wset(T_k))$ \label{line:hacq2}
		\If{$\neg$ locked} 
	 		\Return $A_k$ \EndReturn
	 	\EndIf
	 	\State $\ms{fa}.\lit{add}(1)$ \label{line:hinc}
		\If{$\lit{isAbortable}()$} 
			\State $\lit{release}( \ms{Lset}(T_k))$ 
			\Return $A_k$ \EndReturn
		\EndIf
		\ForAll{$X_j \in \Wset(T_k)$}
	 		 \If{ $v_j.\lit{cas}((ov_j, k_j),(\textit{nv}_j,k))$} 
			      \State $\ms{Oset}(T_k):=\ms{Oset}(T_k)\cup \{X_j\}$
			 \Else
			      \Return $\lit{undo}(\ms{Oset}(T_k))$ \EndReturn
			 \EndIf
			 
	 	\EndFor		
		\State $\lit{release}(\Wset(T_k))$ \label{line:hrel2}
   		\Return $C_k$ \EndReturn
   	 }\EndPart		
	 
 	\newpage
 	
	 \Part{Function: $\lit{release}(Q)$}{
  		\ForAll{$X_j \in Q$}	
 			\State $r_j.\Write(0)$
		\EndFor
		\State $\ms{fa}.\lit{add}(-1)$ \label{line:hdec}
		\Return \ok \EndReturn
	}\EndPart
	\Statex
	\Statex

	\textbf{Code for fast-path transactions}	
	
	\Part{$\textit{read}_k(X_j)$}{\quad\Comment{fast-path}
		\If{$Rset(T_k) = \emptyset$}
			\State $l \gets \Read(\ms{fa})$ \Comment{cached read} \label{line:hfread}
		
			\If{$\ms{l}\neq 0$}
			    \Return $A_k$ \EndReturn
			\EndIf
		\EndIf
		\State $(\textit{ov}_j,k_j) := \Read(v_j)$ \Comment{cached read}
		
		\Return $\textit{ov}_j$ \EndReturn
		
   	 }\EndPart
	\Statex
	\Part{$\textit{write}_k(X_j,v)$}{\quad\Comment{fast-path}
		\State $v_j.\Write(\textit{nv}_j,k)$ \Comment{cached write} 
		\Return $\ok$ \EndReturn
		
   	}\EndPart
	\Statex
	
	\Part{$\textit{tryC}_k$()}{\quad\Comment{fast-path}
		\State $\ms{commit-cache}_i$ \Comment{returns $C_k$ or $A_k$}

   	 }\EndPart		

	}
	\end{multicols}
  \end{algorithmic}
\end{algorithm}

We showed that allowing fast-path transactions to run concurrently in
HyTM results in an instrumentation cost that is
proportional to the read-set size of a fast-path transaction.   
But can we run at least \emph{some} transactions 
concurrently with constant instrumentation cost, while still keeping invisible reads?  

Algorithm~\ref{alg:inswrite2} 
implements a \emph{slow-path progressive} opaque HyTM
with invisible reads and wait-free TM-liveness. 
To fast-path transactions, it only provides \emph{sequential}
TM-progress (they are only guaranteed to commit in the absence
of concurrency), but in return the algorithm is only using a single
metadata base object $\ms{fa}$ 
that is read once by a fast-path transaction and accessed twice with a \emph{fetch-and-add}
primitive by an updating slow-path transaction.
Thus, the instrumentation cost of the algorithm is constant.   

Intuitively, $\ms{fa}$ allows fast-path transactions to detect the
existence of concurrent updating slow-path transactions.
Each time an updating slow-path updating transaction tries to commit, it increments
$\ms{fa}$ and once all writes to data base objects are completed (this
part of the algorithm is identical to Algorithm~\ref{alg:inswrite})
or the transaction is aborted,
it decrements $\ms{fa}$. Therefore, $\ms{fa}\neq 0$ means
that at least one slow-path updating transaction is incomplete.  
A fast-path transaction simply checks if $\ms{fa}\neq 0$ in the
beginning and aborts if so, 
otherwise, its code is identical to that in
Algorithm~\ref{alg:inswrite}.
Note that this way, any update of $\ms{fa}$ automatically causes a tracking set abort of any
incomplete fast-path transaction.

%
\begin{theorem}
\label{th:inswrite2}
There exists an opaque HyTM implementation that provides uninstrumented writes, invisible reads,
progressiveness for slow-path transactions, sequential TM-progress for fast-path transactions and wait-free TM-liveness
such that in every its execution $E$, every fast-path transaction
accesses at most one metadata base object.
\end{theorem}
\begin{proof}
The proof of opacity is almost identical to the analogous proof for Algorithm~\ref{alg:inswrite} in Lemma~\ref{lm:opacityh1}.

As with Algorithm~\ref{alg:inswrite}, enumerating the cases under which a slow-path transaction $T_k$
returns $A_k$ proves that Algorithm~\ref{alg:inswrite2} satisfies progressiveness for slow-path transactions.
Any fast-path transaction $T_k$; $\Rset(T_k) \neq \emptyset$ reads the metadata base object $\ms{fa}$
and adds it to the process's tracking set (Line~\ref{line:hfread}).
If the value of $\ms{fa}$ is not $0$, indicating that there exists a concurrent slow-path transaction pending in its
tryCommit, $T_k$ returns $A_k$. Thus, the implementation provides sequential TM-progress for fast-path transactions.

Also, in every execution $E$ of $\mathcal{M}$, no fast-path write-only transaction accesses any metadata base object
and a fast-path reading transaction accesses the metadata base object $\ms{fa}$ exactly once, during the first t-read.
\end{proof}
%
\section{Related work and Discussion}
\label{sec:p4disc}
\vspace{1mm}\noindent\textbf{HyTM model.}
Our HyTM model is a natural extension of the model we specified for Software Transactional memory (cf. Chapter~\ref{ch:tm-model}), 
and has the advantage of being relatively simple.
The term \emph{instrumentation} was originally used in the context of HyTMs~\cite{hybridnorec,riegel-thesis,phasedtm}
to indicate the overhead a hardware transaction induces in order to detect pending software transactions.
The impossibility of designing HyTMs without any code instrumentation
was intuitively suggested in~\cite{hybridnorec}, we present a formal proof in this paper.

In~\cite{attiyaH13}, Attiya and Hillel considered the instrumentation cost of \emph{privatization}, \emph{i.e.}, 
allowing transactions to isolate data items by making them private to a process so that no other process is allowed to
modify the privatized item.
Just as we capture a tradeoff between the cost of hardware instrumentation and the amount
of concurrency allowed between hardware and software transactions, \cite{attiyaH13} 
captures a tradeoff between the cost of privatization and the number of transactions guaranteed to make
progress concurrently in $\ell$-progressive STMs. 
The model we consider is fundamentally different to \cite{attiyaH13}, in that we 
model hardware transactions at the level of cache coherence, 
and do not consider non-transactional accesses, \emph{i.e.}, neither data nor meta-data base objects are private in our HyTM model.
The proof techniques we employ are also different. 

Uninstrumented HTMs may be viewed as being  \emph{disjoint-access parallel (DAP)}~\cite{israeli-disjoint, AHM09}. 
As such, some of the techniques used in the proof of Theorem~\ref{instrumentation} 
resemble those used in~\cite{OFTM, tm-book, AHM09}.

We have proved that it is impossible to completely
forgo instrumentation in a 
HyTM even if only sequential TM-progress is required,
and that any opaque HyTM implementation providing non-trivial
progress either has to pay a \emph{linear} number of metadata
accesses, or will have to allow slow-path transactions to abort
fast-path operations.  
The main motivation for our definition of metadata base objects (Definition~\ref{def:metadata}) is given by 
experiments suggesting that the cost of concurrency detection is a significant bottleneck 
for many HyTM implementations~\cite{MS13}.
To precisely characterize the costs incurred by hardware transactions, we made a distinction between the set of memory
locations that store the data values of the t-objects and the locations that store the metadata information.
To the best of our knowledge, all known 
HyTM proposals, such as  
\emph{HybridNOrec}~\cite{hybridnorec,riegel-thesis}, \emph{PhTM}~\cite{phasedtm} and others~\cite{damronhytm, kumarhytm}
avoid co-locating the data and metadata within a single base object. 

\vspace{1mm}\noindent\textbf{HyTM algorithms.}
Circa 2005, several papers introduced HyTM implementations~\cite{unboundedhtm1, damronhytm, kumarhytm}
that integrated HTMs with variants of \emph{DSTM}~\cite{HLM+03}.
These implementations provide nontrivial concurrency between hardware
and software transactions (progressiveness), by imposing instrumentation on  hardware
transactions: every t-read operation incurs at least one extra access to a
metadata base object.
Our Theorem~\ref{linearlowerbound} shows that this overhead is unavoidable.
Of note, write operations of these HyTMs are also instrumented, but
our Algorithm~\ref{alg:inswrite} shows that it is not necessary. 

Implementations like \emph{PhTM}~\cite{phasedtm} and \emph{HybridNOrec}~\cite{hybridnorec}
overcome the per-access instrumentation cost of \cite{damronhytm,kumarhytm} by realizing that if 
one is prepared to sacrifice progress, hardware transactions need
instrumentation only at the boundaries of transactions to detect pending software transactions.
Inspired by this observation, our HyTM implementation described in Algorithm~\ref{alg:inswrite2} overcomes the linear per-read
instrumentation cost by allowing hardware readers to abort due to a concurrent software writer, but
maintains progressiveness for software transactions, unlike \cite{phasedtm,hybridnorec,MS13}.

References~\cite{HLR10, riegel-thesis} provide detailed overviews on HyTM designs and  implementations. 
The software component of the HyTM algorithms presented in this paper is inspired by progressive STM
implementations~\cite{DSS06,norec,KR11} and is subject to the lower bounds for progressive STMs
established in~\cite{GK09-progressiveness,attiyaH13,tm-book,KR11}. 
%


\chapter{Optimism for boosting concurrency}
\label{ch:p2c1}
\epigraph{The wickedness and the foolishness of no man can avail against the fond optimism of mankind.}
{\textit{James Branch Cabell}-The Silver Stallion}
\section{Overview}
\label{sec:p1intro}
In previous chapters, we were concerned with the inherent complexities of implementing TM.
In this chapter, we are concerned with using TM to derive concurrent implementations and raise 
a fundamental question about the ability of the TM abstraction
to transform a sequential implementation to a concurrent one.
Specifically, does the optimistic nature of TM give it an inherent advantage in exploiting concurrency
that is lacking in pessimistic synchronization techniques like locking?
To exploit concurrency, conventional lock-based synchronization 
pessimistically protects accesses to the shared memory before
executing them.
Speculative synchronization, achieved using TMs,
optimistically executes memory operations with a risk of aborting them in the future.
A programmer typically uses these synchronization techniques 
as ``wrappers'' to allow every process (or thread) to \emph{locally} run its sequential code while ensuring 
that the resulting concurrent execution is \emph{globally} correct.

Unfortunately, it is difficult for programmers to tell in advance 
which of the synchronization techniques
will establish more concurrency in their resulting programs.
In this chapter, we analyze the ``amount of concurrency'' one can obtain by turning a 
sequential program into a concurrent one.
In particular, we compare the use of optimistic and pessimistic
synchronization techniques, whose prime examples are TMs and locks respectively.

To fairly compare concurrency provided by implementations
based on various techniques,   
one has (1)~to  define what it means for a concurrent program to be
correct regardless of the type of synchronization it uses and 
(2)~to define a metric of concurrency. 

\vspace{1mm}\noindent\textbf{Correctness.}
We begin by defining a consistency criterion,
namely \emph{locally-serializable linearizability}.
We say that a concurrent implementation of a given sequential data type is
\emph{locally serializable} if it
ensures that the local execution of each 
operation 
is equivalent to \emph{some} execution of its sequential implementation.
This condition is weaker than serializability
since it does not require that there exists a \emph{single} sequential 
execution  that is consistent with all local executions.
It is however sufficient to guarantee that optimistic
executions do not observe an inconsistent transient state that could 
lead, for example, to a fatal error like division-by-zero.

Furthermore, the implementation should ``make sense'' globally, 
given the \emph{sequential type} of the data structure we implement.
The high-level history of every execution 
of a concurrent implementation must be 
\emph{linearizable}~\cite{HW90,AW04} with respect to 
this sequential type.
The combination of local serializability and linearizability gives
a correctness criterion that we call \emph{\LS-linearizability},
where {\LS} stands for ``locally serializable''.
We show that LS-linearizability is, as the original  linearizability,
compositional~\cite{HW90,HS08-book}: a composition of LS-linearizable 
implementations is also LS-linearizable. 

We apply the criterion of LS-linearizability to  
two broad classes of \emph{pessimistic} and \emph{optimistic}
synchronization techniques. 
Pessimistic implementations capture what can be achieved 
using classic locks; in contrast, optimistic implementations proceed speculatively and
fail to return a response to the process in the case of conflicts, \emph{e.g.},  
relying on transactional memory.

\vspace{1mm}\noindent\textbf{Measuring concurrency.}
We characterize the amount of concurrency provided by an LS-linearizable implementation as the set of schedules it accepts.
To this end, we define a concurrency metric 
inspired by the analysis of parallelism in database concurrency control~\cite{Yan84,Her90}.
More specifically, we assume an external scheduler that defines which
processes execute which steps of the corresponding sequential program 
in a dynamic and unpredictable fashion. 
This allows us to define concurrency provided by an implementation as the set of \emph{schedules} 
(interleavings of steps of concurrent sequential operations) 
it \emph{accepts} (is able to effectively process).
Then, the more schedules the implementation would accept, the more concurrent it would be.

We provide a framework to compare the concurrency one can get
by choosing a particular synchronization technique for a specific data type.
For the first time, we analytically capture the
inherent concurrency provided by optimism-based and pessimism-based
implementations in exploiting concurrency.
We illustrate this using a popular sequential list-based set
implementation~\cite{HS08-book}, concurrent implementations of which
are our running examples.
More precisely, we show that there exist TM-based implementations that, for some workloads, 
allow for more concurrency than \emph{any} pessimistic implementation,
but we also show that there exist pessimistic implementations that, for other workloads, allow for more 
concurrency than \emph{any} TM-based implementation.

Intuitively, an implementation based on transactions 
may abort an operation based on the way
concurrent steps are scheduled, 
while a pessimistic implementation 
has to proceed eagerly without knowing about how future steps will be 
scheduled, sometimes over-conservatively rejecting 
a potentially acceptable schedule.  
By contrast, pessimistic implementations designed to exploit
the semantics of the data type can supersede the
``semantics-oblivious'' TM-based implementations.
More surprisingly, we demonstrate that combining the benefit of pessimistic implementations, 
namely their semantics awareness, and the benefit of TMs, namely their optimism, enables implementations that are strictly 
better-suited for exploiting concurrency than any of them individually.
We describe a generic optimistic implementation of 
the list-based set that is \emph{optimal} with respect to our
concurrency metric: we show that, essentially, it accepts \emph{all} 
correct concurrent schedules.

Our results suggest that ``relaxed'' TM models that are designed with the semantics of the high-level object in mind
might be central to exploiting concurrency.

\vspace{1mm}\noindent\textbf{Roadmap of Chapter~\ref{ch:p2c1}.}
In Section~\ref{sec:otpl}, we introduce the class of optimistic and pessimistic concurrent implementations we 
consider in this chapter. Section~\ref{sec:p2c2} introduces the definition of locally serializable linearizability
and Section~\ref{sec:p2c3} is devoted to the concurrency analysis of optimistic and pessimistic synchronization
techniques in the context of the list-based set. We wrap up with concluding remarks in Section~\ref{sec:p2c5}.
%
%
\section{Concurrent implementations}
\label{sec:otpl}
\vspace{1mm}\noindent\textbf{Objects and implementations.}
As with Chapter~\ref{ch:tm-model}, we assume an asynchronous shared-memory system in which a set of $n>1$ processes $p_1,\ldots , p_n$ communicate by
applying \emph{operations} on shared \emph{objects}.

An object is an instance of an \emph{abstract data type} which specifies a set of operations that provide the only means to
manipulate the object. Recall that an \emph{abstract data type} $\tau$ is a tuple
$(\Phi,\Gamma, Q, q_0, \delta)$ where $\Phi$ is a set of operations,
$\Gamma$ is a set of responses, $Q$ is a set of states, $q_0\in Q$ is an
initial state and $\delta \subseteq Q\times \Phi \times Q\times \Gamma$ 
is a transition relation that determines, for each state
and each operation, the set of possible
resulting states and produced responses. 
In this chapter, we consider only types that are \emph{total}, \emph{i.e.}, for every $q\in Q$,
$\pi \in \Phi$, there exist  $q' \in Q$ and
$r\in \Gamma$ such that $(q,\pi,q',r) \in \delta$.
We assume that every type $\tau=(\Phi,\Gamma, Q, q_0, \delta)$ is \emph{computable}, i.e., 
there exists a Turing machine that, 
for each input $(q,\pi)$, $q \in Q$, $\pi\in \Phi$, computes
a pair $(q',r)$ such that $(q,\pi,q',r) \in \delta$.

For any type $\tau$, each high-level object $O_{\tau}$ of this type has a \emph{sequential implementation}. 
For each operation $\pi \in \Phi$, 
$\ms{IS}$ specifies a deterministic procedure that 
performs \emph{reads} and \emph{writes} on a collection of objects
$X_1,\ldots , X_m$ that encode a state of $O_{\tau}$, and returns a response $r\in \Gamma$. 

\vspace{1mm}\noindent\textbf{Sequential list-based set.}
As a running example, we consider the sorted linked-list based implementation of the type \emph{set}, commonly referred to 
as the list-based set~\cite{HS08-book}.
Recall that the set type exports operations $\lit{insert}(v)$, $\lit{remove}(v)$ and
$\lit{contains}(v)$, with $v\in\mathbb{Z}$.
Formally, the \emph{set} type is defined by the tuple $(\Phi,\Gamma, Q, q_0, \delta)$ where:
\begin{enumerate}
\item[$\Phi$] $=\{\lit{insert}(v), \lit{remove}(v), \lit{contains}(v) \}$; $v \in \mathbb{Z}$ 
\item[$\Gamma$] $=\{\true,\false\}$ 
\item[$Q$] is the set of all finite subsets of $\mathbb{Z}$; $q_0=\emptyset$
\item[$\delta$] is defined as follows:
\begin{enumerate}
\item[$(1)$:]
$(q,\lit{contains}(v),q,(v\in q))$
\item[$(2)$:]
$(q,\lit{insert}(v),q \cup \{v\},(v \not\in q))$
\item[$(3)$:]
$(q,\lit{remove}(v),q \setminus \{v\},(v \in q))$
\end{enumerate}
\end{enumerate}
\begin{algorithm*}[t]
\caption{Sequential implementation {\LL} (\textit{sorted linked list}) of \emph{set} type}
\label{alg:lists}
  \begin{algorithmic}[1]
  	
	{\footnotesize
	\Part{Shared variables}{
		\State Initially $\ms{head}$, $\ms{tail}$,
		\State ~~~$\ms{head}.val=-\infty$, $\ms{tail}.val=+\infty$
		\State ~~~$\ms{head}.next=\ms{tail}$
	}\EndPart
	
	\Statex
	
	\	\Part{$\lit{insert}(v$)}{
		\State $\ms{prev} \gets \ms{head}$  			                \Comment{copy the address}
		\State $\ms{curr} \gets \lit{read}(\ms{prev.next})$ 		\Comment{fetch the next element}
		\While{$(\ms{tval} \gets \lit{read}(\ms{curr.val})) < v $}        
			\State $\ms{prev} \gets \ms{curr}$ 				
			\State $\ms{curr} \gets \lit{read}(\ms{curr.next})$ 	\Comment{fetch from memory}
		\EndWhile
		\If{$\ms{tval} \neq v$}							\Comment{$tval$ is stored locally}
			\State $X \gets \lit{new-node}(v,\ms{prev.next})$ 	\Comment{$v$ and address of \ms{curr}}\label{line:newnode}
                        \State $\lit{write}(\ms{prev.next}, X)$ 	\Comment{next points to the new element}    \label{line:seqinswrite}          
		\EndIf
		\Return $(\ms{tval}\neq v)$ 						
		\EndReturn
   	}\EndPart
	
	\newpage
	
	\Part{$\lit{remove}(v$)}{

		\State $\ms{prev} \gets \ms{head}$  			\Comment{copy the address}
		\State $\ms{curr} \gets \lit{read}(\ms{prev.next})$ 		\Comment{fetch next field}
		\While{$(\ms{tval} \gets \lit{read}(\ms{curr.val})) < v $} 	
					\Comment{$val$ local copy}
			\State $\ms{prev} \gets \ms{curr}$ 			
			\State $\ms{curr} \gets \lit{read}(\ms{curr.next})$ 	
		\EndWhile
		\If{$\ms{tval} = v$}							
                        \State $tnext \gets \lit{read}(\ms{curr.next})$ \Comment{fetch the node after \ms{curr}} 
			\State $\lit{write}(\ms{prev.next}, tnext)$  \Comment{delete the node} 
		\EndIf
		\Return $(\ms{tval}=v)$ 						
		\EndReturn	
   	}\EndPart
	
	
	\Part{$\lit{contains}(v$)}{
	  \State $\ms{curr} \gets \ms{head}$  			            
		\State $\ms{curr} \gets \lit{read}(\ms{prev.next})$ 		
		\While{$(\ms{tval} \gets \lit{read}(\ms{curr.val})) < v $} 	
			\State $\ms{curr} \gets \lit{read}(\ms{curr.next})$ 	
		\EndWhile
 	   	\Return $(\ms{tval}=v)$						
 	   	\EndReturn
   	 }\EndPart		

}

  \end{algorithmic}
\end{algorithm*}

We consider a sequential implementation $\LL$ (Algorithm~\ref{alg:lists})
of the set type using a sorted linked list where 
each element (or \emph{object}) stores an integer value, $\ms{val}$, and a pointer to its successor, $\ms{next}$, so that elements are 
sorted in the ascending order of their value.  

Every operation invoked with a parameter $v$ traverses the list starting from the
$\ms{head}$ up to the element storing value $v'\geq v$.
If $v'=v$, then $\lit{contains}(v)$ returns $\lit{true}$, $\lit{remove}(v)$ unlinks the 
corresponding element and returns $\lit{true}$, and $\lit{insert}(v)$ returns $\lit{false}$. Otherwise, 
$\lit{contains}(v)$ and $\lit{remove}(v)$ return
$\lit{false}$ while $\lit{insert}(v)$ adds a new element with value
$v$ to the list and returns $\lit{true}$. 
The list-based set 
is denoted by $(\LL,\ms{set})$.

\vspace{1mm}\noindent\textbf{Concurrent implementations.}
We tackle the problem of turning the sequential
implementation $\ms{IS}$ of type $\tau$ into a \emph{concurrent} one, shared by 
$n$ processes. The implementation provides the processes with algorithms 
for the reads and writes on objects.
We refer to the resulting implementation as a concurrent implementation of $(\id{IS},\tau)$.
As in Chapter~\ref{ch:tm-model}, we assume an asynchronous shared-memory system in which the processes communicate by
applying primitives on shared \emph{base objects}~\cite{Her91}.
We place no upper bounds on the number of versions an object may maintain or on the size of this object.

Throughout this chapter, the term \emph{operation} refers to some
high-level operation of the type, 
while read-write operations on objects are referred simply 
as \emph{reads} and \emph{writes}.

An implemented read or write may \emph{abort} by returning a special response
$\bot$. In this case, we say that the corresponding high-level
operation is \emph{aborted}. 
The $\bot$ event is treated both as the response event of the read or
write operation and as the response of the corresponding high-level operation.   

\vspace{1mm}\noindent\textbf{Executions and histories.}
An \emph{execution} of a concurrent implementation (of $(\id{IS},\tau)$) is a sequence
of invocations and responses of high-level operations of type $\tau$, 
invocations and responses of read and write
operations, and primitives applied on base-objects.
We assume that executions are \emph{well-formed}:
no process invokes a new read or write, or high-level operation before
the previous read or write, or a high-level operation, resp., 
returns, or takes steps outside its read or write operation's interval.

Let $\alpha|p_i$ denote the subsequence of an execution $\alpha$
restricted to the events of process $p_i$.
Executions $\alpha$ and $\alpha'$ are \emph{equivalent} if for every process
$p_i$, $\alpha|p_i=\alpha'|p_i$.
An operation $\pi$ \emph{precedes} another operation $\pi'$ in an execution
$\alpha$, 
denoted $\pi \rightarrow_{\alpha} \pi'$, 
if the response of $\pi$ occurs before the invocation of $\pi'$.
Two operations are \emph{concurrent} if neither precedes
the other. 
An execution is \emph{sequential} if it has no concurrent 
operations. 
A sequential execution $\alpha$ is \emph{legal} 
if for every object $X$, every read of $X$ in $\alpha$ 
returns the latest written value of $X$.
An operation is \emph{complete} in $\alpha$ if the invocation event is
followed by a \emph{matching} (non-$\bot$) response or aborted; otherwise, it is \emph{incomplete} in $\alpha$.
Execution $\alpha$ is \emph{complete} if every operation is complete in $\alpha$.

The \emph{history exported by an execution $\alpha$} is
the subsequence of $\alpha$ reduced to the invocations and responses
of operations,  reads and writes, except for the reads
and writes that return $\bot$. 

\vspace{1mm}\noindent\textbf{High-level histories and linearizability.}
A \emph{high-level history} $\tilde H$ of an execution $\alpha$ is the subsequence of $\alpha$ consisting of all
invocations and responses of 
(high-level) operations.
\begin{definition}[Linearizability]
A complete high-level history $\tilde H$ is \emph{linearizable} with 
respect to an object type $\tau$ if there exists
a sequential high-level history $S$ equivalent to $\tilde H$ such that
(1) $\rightarrow_{\tilde H}\subseteq \rightarrow_S$ and
(2) \emph{$S$ is consistent with the sequential specification of type $\tau$}.

Now a high-level history $\tilde H$ is linearizable if it can be
\emph{completed} (by adding matching responses to a subset of
incomplete operations in $\tilde H$ and removing the rest)
to a linearizable high-level history~\cite{HW90,AW04}.
\end{definition}
\vspace{1mm}\noindent\textbf{Obedient implementations.}
We only consider implementations that satisfy the following condition:
Let $\alpha$ be any complete sequential execution of a concurrent implementation $I$.
Then in every execution of $I$ of the form $\alpha\cdot\rho_1\cdots \rho_k$
where each $\rho_i$ ($i=1,\ldots,k$) is the complete execution of a
read, every read returns the value written by the last write that does
not belong to an aborted operation.

Intuitively, this assumption restricts our scope to
``obedient'' implementations of reads and writes, where no
read value may depend on some future write.   
In particular, we filter out implementations in which the
complete execution of a high-level operation is performed within the
first read or write of its sequential implementation.

\vspace{1mm}\noindent\textbf{Pessimistic implementations.}
Informally, a concurrent implementation is \emph{pessimistic} if the exported history contains
every read-write event that appears in the execution. 
More precisely, no execution of a pessimistic implementation includes
operations that returned $\bot$.  

For example, a class of pessimistic implementations are those based on locks.
A lock provides
shared or exclusive access to an object $X$ through 
synchronization primitives $\lit{lock}^S(X)$ (\emph{shared mode}),
$\lit{lock}(X)$ (\emph{exclusive mode}),  
and $\lit{unlock}(X)$.
When $\lit{lock}^S(X)$ (resp. $\lit{lock}(X)$) invoked
by a process $p_i$ returns, we say that $p_i$ \emph{holds
a lock on $X$ in shared (resp. exclusive) mode}.
A process \emph{releases} the object it holds by invoking
$\lit{unlock}(X)$.  
If no process holds a shared or exclusive
lock on $X$, then $\lit{lock}(X)$
eventually returns;
if no process holds an exclusive
lock on $X$, then $\lit{lock}^S(X)$
eventually returns; and
if no process holds a lock on $X$ forever, then every $\lit{lock}(X)$ or $\lit{lock}^S(X)$
eventually returns. Given a sequential implementation of a data type, 
a corresponding lock-based concurrent one 
is derived by inserting the synchronization primitives
to provide read-write access to an object. 

\vspace{1mm}\noindent\textbf{Optimistic implementations.}
In contrast with pessimistic ones, optimistic implementations may, under
certain conditions, abort an operation:
some read or write may return $\bot$,
in which case the corresponding operation also returns $\bot$.

Popular classes of optimistic implementations are those based on
``lazy synchronization''~\cite{HHL+05,HS08-book} (with the ability of
returning $\bot$ and re-invoking an operation) or transactional memory.
%
%
\section{Locally serializable linearizability}
\label{sec:p2c2}
We are now ready to define the correctness criterion that we impose on our
concurrent implementations.

Let $H$ be a history and let $\pi$ be a high-level operation in $H$. 
Then $H|\pi$ denotes the subsequence of $H$ consisting of the events
of $\pi$, except for the last aborted read or write, if any.
Let $\id{IS}$ be a sequential implementation of an object of type
$\tau$ and $\Sigma_{\id{IS}}$, the set of histories of $\id{IS}$. 
\begin{definition}[LS-linearizability]
\label{def:lin}
A history ${H}$ is \emph{locally serializable with respect to}
${\id{IS}}$ if for every high-level operation $\pi$ in $H$,
there exists $S \in \Sigma_{\id{IS}}$ such that $H|\pi=S|\pi$.
A history ${H}$ is \emph{\LS-linearizable with respect to
$(\id{IS},\tau)$}  (we also write $H$ is $(\id{IS},\tau)$-LSL)  if:
(1) ${H}$ is locally serializable with respect to
$\id{IS}$ and (2) the corresponding high-level history $\tilde H$ 
is linearizable with respect to $\tau$.
\end{definition}
Observe that local serializability stipulates that the execution is 
witnessed sequential by every operation.
Two different operations (even when invoked by the same process) are not
required to witness mutually consistent sequential executions.

A concurrent implementation $I$ is \emph{\LS-linearizable with respect to
$(\id{IS},\tau)$} (we also write $I$ is $(\id{IS},\tau)$-LSL)
if every history exported by $I$ is $(\id{IS},\tau)$-LSL.   
Throughout this paper, when we refer to a concurrent implementation of $(\id{IS},\tau)$, 
we assume that it is \LS-linearizable with respect to $(\id{IS},\tau)$.

\begin{figure*}[t]
 \includegraphics[scale=0.45]{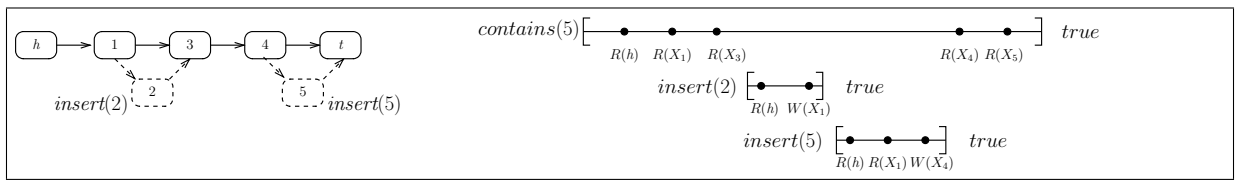}
 \caption{\small{A concurrency scenario for a list-based set, initially $\{1,3,4\}$, where value $i$ is stored at node $X_i$:
   $\lit{insert}(2)$ and  $\lit{insert}(5)$ can proceed
   concurrently with $\lit{contains}(5)$, the history is
   LS-linearizable but not serializable. (We only depict important read-write
   events here.)}}\label{fig:ex1}%
\end{figure*}
\vspace{1mm}\noindent\textbf{LS-linearizability is compositional.}
Just as linearizability, \LS-linearizability is
\emph{compositional}~\cite{HW90,HS08-book}: a composition of LSL 
implementations is also LSL.
We define the composition of two distinct object types $\tau_1$ and $\tau_2$ 
as a type $\tau_1\times\tau_2=(\Phi,\Gamma,Q,q_0,\delta)$ as follows: 
$\Phi=\Phi_1\cup \Phi_2$, $\Gamma=\Gamma_1\cup 
\Gamma_2$,\footnote{Here we treat each $\tau_i$ as a distinct type by adding
index $i$ to all elements of $\Phi_i$, $\Gamma_i$, and $Q_i$.}   
$Q=Q_1\times Q_2$,
$q_0=({q_0}_1,{q_0}_2)$, and  $\delta \subseteq Q\times \Phi \times Q\times
\Gamma$ is such that $((q_1,q_2),\pi,(q_1'q_2'),r)\in\delta$ if and only if for $i\in \{1,2\}$, if 
$\pi\in \Phi_i$ then $(q_i,\pi,q_i',r)\in\delta_i$ $\wedge$ $q_{3-i}=q^{\prime}_{3-i}$.

Every sequential implementation $\id{IS}$ of an object  $O_1\times O_2$ of a
composed type $\tau_1\times\tau_2$ naturally induces two sequential
implementations $I_{S1}$ and $I_{S2}$ of objects $O_1$ and $O_2$,
respectively. 
Now a correctness criterion 
$\Psi$
is \emph{compositional} if for every
history $H$ on an object composition $O_1\times O_2$, 
if 
$\Psi$
holds for $H|O_i$ with
respect to $I_{Si}$, for $i \in \{1,2\}$, then
$\Psi$
holds for $H$ with
respect to $\id{IS}=I_{S1}\times I_{S2}$.
Here, $H|O_i$ denotes the subsequence of $H$ consisting of events on $O_i$.
\begin{theorem}
\label{th:comp}
\LS-linearizability is compositional. 
\end{theorem}
\begin{proof}
Let $H$, a history on $O_1\times O_2$,  be \LS-linearizable
with respect to $\id{IS}$. 
Let each $H|O_i$, 
$i\in\{1,2\}$, 
be \LS-linearizable with respect to $I_{Si}$. 
Without loss of generality,  we assume that $H$ is complete (if $H$
is incomplete, we consider any completion of it containing
\LS-linearizable completions of  $H|O_1$ and $H|O_1$).

Let $\tilde H$ be a completion of the high-level history corresponding to $H$ such that
$\tilde H|O_1$ and $\tilde H|O_2$ are linearizable with respect to $\tau_1$
and $\tau_2$, respectively. Since linearizability is
compositional~\cite{HW90,HS08-book}, $\tilde H$ is linearizable with respect to $\tau_1\times\tau_2$.

Now let, for each operation $\pi$, $S_{\pi}^1$ and $S_{\pi}^2$ be any two sequential histories of
$I_{S1}$ and $I_{S2}$  such that $H|\pi|O_j=S_{\pi}^j|\pi$, for $j \in \{1,2\}$
(since 
$H|O_1$
and $H|O_2$ are \LS-linearizable such histories exist).
We construct a sequential history $S_{\pi}$ by interleaving events of
$S_{\pi}^1$ and $S_{\pi}^2$ so that $S_{\pi}|O_j=S_{\pi}^j$, 
$j\in\{1,2\}$.
Since each $S_{\pi}^j$ acts on a distinct component $O_j$ of $O_1\times
O_2$, every such $S_{\pi}$ is a sequential history of $\id{IS}$.
We pick one $S_{\pi}$ that respects the local history $H|\pi$,
which is possible, since $H|\pi$ is consistent with both    
$S_1|\pi$ and $S_2|\pi$. 

Thus, for each $\pi$, we obtain a history of $\id{IS}$ that agrees with
$H|\pi$. Moreover, the high-level history of $H$ is linearizable. Thus, $H$ is \LS-linearizable with respect to $\id{IS}$. 
\end{proof}
\vspace{1mm}\noindent\textbf{LS-linearizability versus other criteria.}
LS-linearizability is a two-level consistency criterion which makes it
suitable to compare concurrent implementations of a sequential data
structure, regardless of synchronization techniques they use.
It is quite distinct from related criteria designed for database and software
transactions, such as serializability~\cite{Pap79-serial,WV02-book} and
multilevel serializability~\cite{Wei86,WV02-book}.

For example, serializability~\cite{Pap79-serial} prevents sequences of reads and writes from conflicting in a cyclic way, 
establishing a  global order of transactions.
Reasoning only at the level of reads and writes may be overly conservative:
higher-level operations may commute even if their reads and writes conflict~\cite{Wei88}.
Consider an execution of a concurrent \emph{list-based set} depicted in 
Figure~\ref{fig:ex1}.
We assume here that the set initial state is $\{1,3,4\}$.
Operation $\lit{contains}(5)$ is concurrent, first with 
operation $\lit{insert}(2)$ and then with operation $\textsf{insert}(5)$. 
The history is not serializable:
$\lit{insert}(5)$ sees 
the effect of $\lit{insert}(2)$ because $R(X_1)$ by $\lit{insert}(5)$ returns the value
of $X_1$ that is updated by $\lit{insert}(2)$ and
thus should be serialized after it. But $\lit{contains}(5)$ misses
element $2$ in the linked list, but must see the
effect of $\lit{insert}(5)$ to perform the read of $X_5$, i.e., the element created by $\lit{insert}(5)$.  
However, this history is LSL since each of the three local histories is consistent with some
sequential history of $\LL$. 

Multilevel serializability~\cite{Wei86,WV02-book} was 
proposed to reason in terms of multiple semantic levels in the same execution.
\LS-linearizability, being defined for two levels only, does not require a global serialization of low-level operations as
$2$-level serializability does. 
LS-linearizability simply requires each process  to observe a local serialization, which can be different from one
process to another. Also, to make it more suitable for concurrency
analysis of a concrete data structure, instead of semantic-based commutativity~\cite{Wei88}, we use the sequential
specification of the high-level behavior of the object~\cite{HW90}.

Linearizability~\cite{HW90,AW04} only accounts for high-level
behavior of a data structure,  so it does not imply
LS-linearizability. For example, Herlihy's universal
construction~\cite{Her91} provides a linearizable implementation for
any given object type, but does not guarantee that each execution locally appears
sequential with respect to any sequential implementation of the type.    
Local serializability, by itself, does not require any synchronization
between processes and can be trivially implemented without
communication among the processes.
Therefore, the two parts of LS-linearizability indeed complement each other.  

%

%
\section{Pessimistic vs. optimistic synchronization}
\label{sec:p2c3}
In this section, we compare the relative abilities of optimistic and pessimistic synchronization techniques
to exploit concurrency in the context of the list-based set.

To characterize the ability of a concurrent implementation to process arbitrary interleavings of sequential code, we introduce 
the notion of a \emph{schedule}.
Intuitively, a schedule describes the order in which complete high-level
operations, and \emph{sequential} reads and writes are invoked by the user. 
More precisely, a schedule is 
an equivalence class of complete histories that agree on
the \emph{order} of invocation and response events of reads, writes and high-level operations, but 
not necessarily on read \emph{values} or high-level responses.
Thus, a schedule can be treated as a history, where responses of reads and operations
are not specified. 

We say that an implementation $I$ \emph{accepts} a schedule $\sigma$ if 
it exports a history $H$ such that $\ms{complete}(H)$ exhibits
the order of $\sigma$, where $\ms{complete}(H)$ is the subsequence of $H$
that consists of the events of the complete operations that returned a matching response. 
We then say that the execution (or history) \emph{exports} $\sigma$. 
A schedule $\sigma$ is 
$(\ms{IS},\tau)$-LSL if there
exists an $(\id{IS},\tau)$-LSL history that exports $\sigma$.

A \emph{synchronization technique} is a set of concurrent implementations.
We define a specific optimistic synchronization technique and then
a specific pessimistic one.

\vspace{1mm}\noindent\textbf{The class $\mathcal{SM}$.}
Formally, $\mathcal{SM}$ denotes the set of optimistic, safe-strict serializable LSL implementations.

Let $\alpha$ denote the execution of a concurrent implementation and
$\ms{ops}(\alpha)$, 
the set of operations each of which performs at least one event in $\alpha$.
Let ${\alpha}^k$ denote the prefix of $\alpha$ up to the last event of operation $\pi_k$.
Let $\ms{Cseq}(\alpha)$ denote the set of subsequences of ${\alpha}$  that
consist of all the events of operations that are complete in $\alpha$. 
We say that $\alpha$ is \emph{strictly serializable} if 
there exists a legal sequential execution $\alpha'$ equivalent to
a sequence in $\sigma\in\ms{Cseq}(\alpha)$
such that $\rightarrow_{\sigma} \subseteq \rightarrow_{\alpha'}$. 

We focus on optimistic implementations that are strictly
serializable and, in addition, guarantee that every operation (even aborted or incomplete) observes correct (serial)
behavior.  
More precisely, an execution $\alpha$ is  \emph{safe-strict serializable} if
(1) $\alpha$ is strictly serializable, and
(2) for each operation $\pi_k$, 
there exists a legal sequential execution
$\alpha'=\pi_0\cdots \pi_i\cdot \pi_k$ and
$\sigma\in\ms{Cseq}(\alpha^k)$ such that $\{\pi_0,\cdots, \pi_i\}
\subseteq \ms{ops}(\sigma)$ and $\forall \pi_m\in \ms{ops}(\alpha'):{\alpha'}|m={\alpha^k}|m$.

Similar to other relaxations of opacity~\cite{tm-book} like \emph{TMS1}~\cite{DGLM13} and \emph{VWC}~\cite{damien-vw-jv},  
safe-strict serializable implementations ($\mathcal{SM}$) require that every transaction (even aborted and incomplete) observes
``correct'' serial behavior. 
Safe-strict serializability captures nicely both local serializability
and linearizability. 
If we transform a sequential implementation
$\id{IS}$ of a type $\tau$ into a \emph{safe-strict serializable} concurrent one, 
we obtain an LSL implementation of $(\id{IS},\tau)$. 
Thus, the following lemma is immediate.
\begin{lemma}
Let $I$ be a safe-strict serializable implementation of $(\id{IS},\tau)$.
Then, $I$ is \LS-linearizable with respect to $(\id{IS},\tau)$.
\end{lemma}
Indeed, by running each operation of $\id{IS}$ within a transaction of
a safe-strict serializable TM, we make sure that completed operations witness the same execution of $\id{IS}$, and
every operation that returned $\bot$ is consistent with some
execution of $\id{IS}$ based on previously completed operations. 

\vspace{1mm}\noindent\textbf{The class $\mathcal{P}$.}
This denotes the set of \emph{deadlock-free} pessimistic
LSL implementations: assuming that every process takes
enough steps, at least one of the concurrent operations return a matching response~\cite{HS11-progress}.
Note that $\mathcal{P}$ includes implementations that are not necessarily safe-strict serializable.
\subsection{Concurrency analysis}
\label{sec:p2c3s1}
\begin{figure*}
 \includegraphics[scale=0.45]{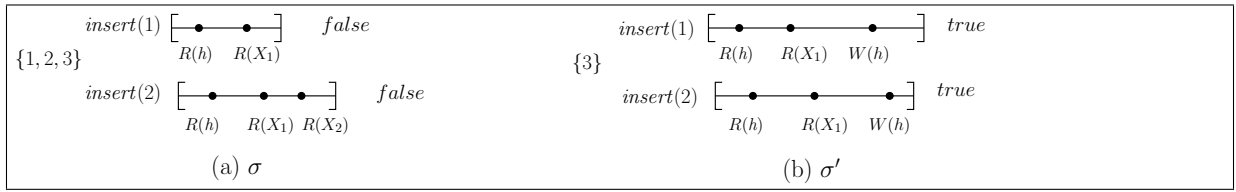}
 \caption{\small{%
(a) a history exporting schedule $\sigma$, with initial state
   $\{1,2,3\}$, accepted by $I^{LP}\in \mathcal{SM}$; 
(b) a history exporting a problematic schedule $\sigma'$, with initial state 
   $\{3\}$, which should be accepted by any $I\in\mathcal{P}$ if it accepts $\sigma$}}\label{fig:ex2}%
\vspace{-0.35mm}
\end{figure*}
We now provide a concurrency analysis of synchronization techniques $\mathcal{SM}$ and $\mathcal{P}$
in the context of the list-based set.

\vspace{1mm}\noindent\textbf{A pessimistic implementation $I^H\in \mathcal{P}$ of $(\LL,\ms{set})$.}
We describe a pessimistic implementation of $(\LL,\ms{set})$, $I^H \in \mathcal{P}$,
that accepts non-serializable schedules: each read operation performed by $\lit{contains}$ 
acquires the \emph{shared lock} on the object, 
reads the $\ms{next}$ field of the element before releasing the shared lock on the predecessor element 
in a \emph{hand-over-hand} manner~\cite{BS88}.
Update operations ($\lit{insert}$ and
$\lit{remove}$) acquire the \emph{exclusive lock} on the $\ms{head}$ during $\lit{read}(\ms{head})$ 
and release it at the end. Every other read operation performed 
by update operations simply reads the element $\ms{next}$ field to traverse the list. The write operation
performed by an $\lit{insert}$ or a $\lit{remove}$ acquires the exclusive lock, writes the value
to the element and releases the lock.
There is no real concurrency between any two update operations since the process holds the 
exclusive lock on the $\ms{head}$ throughout the operation execution.
Thus:
\begin{lemma}
$I^H$ is deadlock-free and LSL implementation of $(\LL,\ms{set})$.
\end{lemma}
On the one hand, the schedule of $(\LL,\ms{set})$ depicted in Figure~\ref{fig:ex1}, which we denote by $\sigma_0$, is not 
serializable and must be rejected by any implementation in $\mathcal{SM}$.
However, there exists an execution of $I^H$ that exports $\sigma_0$
since there is no read-write conflict on any two consecutive elements accessed.

On the other hand, consider the schedule $\sigma$ of $(\LL,\ms{set})$ in Figure~\ref{fig:ex2}(a).
Clearly, $\sigma$ is serializable and is accepted by implementations based on most progressive TMs since there
is no read-write conflict.
For example, let $I^{LP}$ denote an implementation of $(\id{IS},\tau)$ based on the progressive opaque
TM implementation $LP$ in Algorithm~\ref{alg:ic} (Chapter~\ref{ch:p3c2}). Then, $I^{LP}\in \mathcal{SM}$
and the schedule $\sigma$ is accepted by $I^{LP}$.
However, we prove that $\sigma$ is not accepted by any implementation in $\mathcal{P}$.
Our proof technique is interesting in its own right: we show that
if there exists any implementation in $\mathcal{P}$ that accepts $\sigma$, it must also
accept the schedule $\sigma'$ depicted in Figure~\ref{fig:ex2}(b). In $\sigma'$, 
$\lit{insert}(2)$ overwrites the write on \emph{head} performed by $\lit{insert}(1)$
resulting in a lost update. By deadlock-freedom, there exists an extension of $\sigma'$ in which
a $\lit{contains}(1)$ returns $\false$; but this is not a linearizable schedule. 
\begin{theorem}
\label{th:plm}
There exists a schedule $\sigma_0$ of $(\LL,\ms{set})$ that is accepted by an implementation in
$\mathcal{PL}$, but not accepted by \emph{any} implementation $I\in \mathcal{SM}$.
\end{theorem}
\begin{proof}
Let $\sigma_0$ be the schedule of $(\LL ,\ms{set})$ depicted in Figure~\ref{fig:ex1}.
Suppose by contradiction that $\sigma_0 \in\S(I)$, where $I$ is an implementation of $(\LL, \ms{set})$ based on any safe-strict serializable TM.
Thus, there exists an execution $\alpha$ of $I$ that exports $\sigma_0$.
Now consider two cases:
(1) Suppose that the read of $X_4$ by $\lit{contains}(5)$ 
returns the value of $X_4$ that is updated by $\lit{insert}(5)$.
Since $\lit{insert}(2) \rightarrow_{\alpha} \lit{insert}(5)$, 
$\lit{insert}(2)$ must precede $\lit{insert}(5)$ in any sequential execution $\alpha'$ equivalent to $\alpha$. 
Also, since $\lit{contains}(5)$ reads $X_1$ prior to its update by $\lit{insert}(2)$, 
$\lit{contains}(5)$ must precede $\lit{insert}(2)$ in $\alpha'$. 
But then the read of $X_4$ is not legal in $\alpha'$---a contradiction since $\alpha$ must be serializable.
(2) Suppose that $\lit{contains}(5)$
reads the initial value of $X_4$, \emph{i.e.}, its value prior to the write to $X_4$ by $\lit{insert}(5)$, 
where $X_4.\ms{next}$ points to the
\emph{tail} of the list (according to our sequential implementation $\LL$).
But then, according to $\LL$, $\lit{contains}(5)$ cannot access
$X_5$ in $\sigma_0$---a contradiction.  

Consider the pessimistic implementation $I^H \in \mathcal{P}$:
since the $\lit{contains}$ operation traverses the list using shared hand-over-hand locking, 
the process $p_i$ executing $\lit{contains}(5)$
can release the lock on element $X_1$ prior to the acquisition of the exclusive lock on $X_1$ by $\lit{insert}(2)$.
Similarly, $p_i$ can acquire the shared lock on $X_4$ immediately after the release of the 
exclusive lock on $X_4$ by the process executing $\lit{insert}(5)$ while still holding 
the shared lock on element $X_3$. Thus, there exists an execution of $I^H$ that exports $\sigma_0$.
\end{proof}%
\begin{theorem}
\label{th:mpl}
There exists a schedule $\sigma$ of $(\LL,\ms{set})$ 
that is accepted by an implementation in $ \mathcal{SM}$,
but not accepted by \emph{any} implementation in $\mathcal{P}$.
\end{theorem}
\begin{proof}
We show first that the schedule $\sigma$ of $(\LL,\ms{set})$ depicted
in Figure~\ref{fig:ex2}(a) is not accepted by any implementation in $\mathcal{P}$.
Suppose the contrary and let $\sigma$ be exported by an execution $\alpha$. 
Here $\alpha$ starts with three sequential $\lit{insert}$ operations with
parameters $1$, $2$, and $3$. The resulting ``state'' of the set is
$\{1,2,3\}$, where value $i\in \{1,2,3\}$ is stored in object $X_i$.   

Suppose, by contradiction, that some $I\in \mathcal{P}$ accepts $\sigma$. 
We show that $I$ then accepts the schedule $\sigma'$ depicted in Figure~\ref{fig:ex2}(b), which starts with a sequential
execution of $\lit{insert}(3)$ storing value $3$ in object $X_1$. 

Let $\alpha'$ be any history of $I$ that exports $\sigma'$.
Recall that we only consider obedient implementations:
in $\alpha'$: the read of \emph{head} by $\lit{insert}(2)$ in $\sigma'$ refers to $X_1$ (the next element to be read by $\lit{insert}(2)$). 
In $\alpha$, element $X_1$ stores value $1$,
\emph{i.e.}, $\lit{insert}(1)$ can safely return $\lit{false}$, while in
$\sigma'$, $X_1$ stores value $3$, i.e., the next step of
$\lit{insert}(1)$ must be a write to \emph{head}.       
Thus, no process
can distinguish $\alpha$ and $\alpha'$ 
before the
read operations on $X_1$ return. 
Let $\alpha''$ be the prefix of $\alpha'$ ending with $R(X_1)$
executed by $\lit{insert}(2)$.
Since $I$ is deadlock-free, we have an extension of $\alpha''$ in
which both $\lit{insert}(1)$ and $\lit{insert}(2)$ terminate; we show that this extension violates linearizability. 
Since $I$ is locally-serializable, to respect our sequential implementation
of $(LL,\ms{set})$, both operations should complete the write to \emph{head}
before returning. 
Let $\pi_1=\lit{insert}(1)$ be the first operation to write to
\emph{head} in this extended execution. 
Let $\pi_2=\lit{insert}(2)$ be the other insert operation.
It is clear that $\pi_1$ returns $\true$ even though $\pi_2$ overwrites the update of $\pi_1$ on $\ms{head}$
and also returns $\true$. 
Recall that implementations in $\mathcal{P}$ are deadlock-free. 
Thus,   
we can further extend the execution with a complete $\lit{contains}(1)$
that will return $\false$ (the element inserted to the list by $\pi_1$
is lost)---a contradiction since $I$ is linearizable with respect to \emph{set}. 
Thus, $\sigma\notin\S(I)$ for any $I\in \mathcal{P}$.

On the other hand, the schedule $\sigma$ is accepted by $I^{LP}\in \mathcal{SM}$, since
there is no conflict between the two concurrent update operations.
\end{proof}
\subsection{Concurrency optimality}
\label{sec:p2c3s2}
We now combine the benefits of semantics awareness of implementations in $\mathcal{P}$ and 
the optimism of $\mathcal{SM}$ to derive a generic optimistic implementation 
of the list-based set
that supersedes every implementation in classes $\mathcal{P}$ and
$\mathcal{SM}$ in terms of concurrency.
Our implementation, denoted $I^{RM}$ provides processes with algorithms for 
implementing read and write operations on the elements of the list for each operation of the 
list-based set (Algorithm~\ref{alg:elastic}).
\begin{algorithm*}[t]
\caption{Code for process $p_k$ implementing reads and writes in implementation $I^{RM}$}
\label{alg:elastic}
  \begin{algorithmic}[1]
	{\footnotesize
	
	\Part{Shared variables}{
	  \State for each object $X_{\ell}$:
	  \State ~~~~~~$\ms{t-var}[\ell]$, initially 0
	  \State ~~~~~~$r[\ell]$, initially $\false$
	  \State ~~~~~~$L[\ell]\in \N \times \{\lit{true}, \lit{false}\}$ supports $\lit{read}$
	  \State ~~~~~~~~~~~$\lit{write}, \lit{cas}$ operations, initially $\tup{0, \lit{false}}$
	}\EndPart	
	
   	\Statex
	  
	\Part{Local variables of process $p_k$}{
	  \State $\ms{rbuf}_k[i] \subset X \times \N$; $i=\{1,2\}$ cyclic buffer of size $2$,  
	  \State ~~~initially $\emptyset$
	}\EndPart	
	
	\Statex
%
	
   	\Part{$\lit{read}_k(X_{\ell})$ executed by $\lit{insert}$, $\lit{remove}$, $\lit{contains}$}{	  
	  \State $\tup{\ms{ver_1}, *} \gets L[\ell].read()$ \label{line:rver} \Comment{get versioned lock}
	  \State $val \gets \ms{t-var}[\ell].read()$ \Comment{get value} \label{line:linr}
	  \State $r \gets r[\ell].read()$ \label{line:reread}
	  \State $\tup{\ms{ver_2}, *} \gets L[\ell].read()$ \Comment{reget versioned lock}
	  \If{$(\ms{ver}_1 \neq \ms{ver}_2) \vee r$} \label{line:flagcheck}
	  	\Return $\bot$ \label{line:elastic:abort1} \EndReturn
	  \EndIf
	  \State $\ms{rbuf}_k.\lit{add}(\tup{X_{\ell},\ms{ver}_{1}})$
          \Comment{override penultimate entry}
	  \Return $\ms{val}$ \EndReturn \label{line:tx-read}
	}\EndPart
	  
	
	\Statex

	\Part{$\lit{write}_k(X_{\ell},v)$ executed by $\lit{remove}$}{
	  \State {\bf let} $\ms{oldver}_\ell$ be such that $\tup{X_\ell, \ms{oldver}_\ell} \in \ms{rbuf}_k$
	  \State $\ms{ver} \gets \ms{oldver}_\ell$
	  \If{$ \neg L[\ell].cas(\tup{\ms{ver}, \lit{false}}, \tup{\ms{ver}, \lit{true}})$} \label{line:linw}
	  		\Return $\bot$  \label{line:elastic:abort2}\Comment{grab lock or abort}	  	\EndReturn
	  \EndIf
	  \State {\bf let} $X_{\ell'}\neq X_{\ell}$ be such that $\{ X_{{\ell}^\prime},ver_{\ell'} \}\in \ms{rbuf}_k$
	  \If{$ \neg L[\ell'].cas(\tup{\ms{ver}_{\ell'}, \lit{false}}, \tup{\ms{ver}_{\ell'}, \lit{true}})$} \label{line:linw2}
	  		\Return $\bot$  \label{line:elastic:abort2}\Comment{grab lock or abort}	  	\EndReturn
	  \EndIf
	  \State $r[\ell^\prime].write(\true)$ \label{line:elastic:gc}  \Comment{mark element for deletion}
	      
	  \State $\ms{t-var}[\ell].write(v)$ \Comment{update memory} \label{line:commit-update}
	  \State $L[\ell].\lit{write}(\tup{\ms{ver}+1, \lit{false}})$\Comment{release locks} \label{line:release}
	  \State $L[\ell'].\lit{write}(\tup{\ms{ver}_{\ell'}+1, \lit{false}})$ \label{line:release2}
	  \Return $\lit{ok}$ \EndReturn
	}\EndPart

	\Statex
	\Part{$\lit{write}_k(X_{\ell},v)$ executed by $\lit{insert}$}{
	  \State {\bf let} $\ms{oldver}_\ell$ be such that $\tup{X_\ell, \ms{oldver}_\ell} \in \ms{rbuf}_k$
	  \State $\ms{ver} \gets \ms{oldver}_\ell$
	  \If{$ \neg L[\ell].cas(\tup{\ms{ver}, \lit{false}}, \tup{\ms{ver}, \lit{true}})$} \label{line:inswrite}
	  		\Return $\bot$  \label{line:elastic:abort3}\Comment{grab lock or abort}	  	\EndReturn
	  \EndIf
      \State $\ms{t-var}[\ell].write(v)$ \Comment{update memory} \label{line:inscommit}
      \State $L[\ell].\lit{write}(\tup{\ms{ver}+1, \lit{false}})$\Comment{release locks} \label{line:insrel}
	  \Return $\lit{ok}$ \EndReturn
	}\EndPart

  }
  \end{algorithmic}
\end{algorithm*}

%

%
%

Every object (or element) $X_{\ell}$ is
specified 
by the following shared variables: $\ms{t-var}[\ell]$
stores the \emph{value} $v\in V$ of $X_{\ell}$, $r[\ell]$ stores
a boolean indicating if $X_{\ell}$ is \emph{marked for deletion}, 
$L[\ell]$ stores a tuple of the \emph{version number} of $X_{\ell}$ and a \emph{locked} flag; 
the latter indicates whether a concurrent process is performing a write to $X_{\ell}$.

Any operation with input parameter $v$ traverses the list starting from the
$\ms{head}$ element up to the element storing value $v'\geq v$ without writing to shared memory.
If a read operation on an element conflicts with a write operation to the same element or
if the element is marked for deletion, the operation terminates by returning $\bot$.
While traversing the list, the process maintains the last two read elements and their version numbers 
in the local rotating buffer $\ms{rbuf}$. If none of the read operations performed by $\lit{contains}(v)$ return $\bot$
and if $v'=v$, then $\lit{contains}(v)$ returns $\lit{true}$; otherwise it returns $\lit{false}$.
Thus, the $\lit{contains}$ does not write to shared memory.

To perform write operation to an element as part of an update operation ($\lit{insert}$ and $\lit{remove}$), the process 
first retrieves the version of the object that belongs to its rotating buffer.
It returns $\bot$ if the version has been changed since the previous read of the element 
or if a concurrent process is executing a write to the same element.
Note that, technically, $\bot$ is returned only if $\ms{prev.next}\not\rightarrow \ms{curr}$.
If $\ms{prev.next}\rightarrow \ms{curr}$, then we attempt to lock the element with the current version
and return $\bot$ if there is a concurrent process executing a write to the same element.
But we avoid expanding
on this step in our algorithm pseudocode.
The write operation performed by the $\lit{remove}$ operation, additionally checks if the element to be removed
from the list is locked by another process; if not, it sets a flag on the element to mark it for deletion.
If none of the read or write operations performed during the $\lit{insert}(v)$ or $\lit{remove}(v)$ returned $\bot$,
appropriate matching responses are returned as prescribed by the sequential implementation $\LL$.
Any update operation of $I^{RM}$ uses at most two expensive synchronization
patterns~\cite{AGK11-popl}.

\vspace{1mm}\noindent\textbf{Proof of LS-linearizability.}
Let $\alpha$ be an execution of $I^{RM}$ and $<_\alpha$
denote the total-order on events in $\alpha$.
For simplicity, we assume that $\alpha$ starts with an artificial
sequential execution of an insert operation $\pi_0$ that inserts $\ms{tail}$ and sets $\ms{head}.\ms{next}=\ms{tail}$. 
Let $H$ be the history exported by $\alpha$, where 
all reads and writes are sequential. 
We construct $H$ by associating a linearization point $\ell_{op}$ with each
non-aborted read or write operation $op$ performed in $\alpha$  
as follows:
\begin{itemize}
\item  if $op$ is a read, then performed by process $p_k$,
  $\ell_{op}$ is the base-object $\lit{read}$ in line~\ref{line:linr};
\item  if $op$ is a write within an $\lit{insert}$ operation,
  $\ell_{op}$ is the base-object $\lit{cas}$ in line~\ref{line:linw};
\item  if $op$ is a write  within a $\lit{remove}$ operation,
  $\ell_{op}$ is the base-object  $\lit{cas}$ in line~\ref{line:inswrite}.
\end{itemize}
We say that a $\lit{read}$ of an element $X$ within an operation $\pi$
is \emph{valid} in $H$ 
(we also say that $X$ is \emph{valid}) if 
there does not exist any $\lit{remove}$ operation $\pi_1$ that \emph{deallocates} $X$ (removes $X$ from the list)
such that $\ell_{\pi_1.\lit{write}(X)} <_{\alpha}
\ell_{\pi.\lit{read}(X)}$. 
\begin{lemma}
\label{lem:invar1}
Let $\pi$ be any operation performing $\lit{read}(X)$ followed by $\lit{read}(Y)$ in $H$. 
Then (1) there exists an $\lit{insert}$ operation that sets $X.\ms{next}=Y$ prior to $\pi.\lit{read}(X)$, and
(2) $\pi.\lit{read}(X)$ and $\pi.\lit{read}(Y)$ are \emph{valid} in $H$.
\end{lemma}
\begin{proof} 
Let $\pi$ be any operation in $I^{RM}$ that performs $\lit{read}(X)$ followed by $\lit{read}(Y)$.
If $X$ and $Y$ are \emph{head} and \emph{tail} respectively, $\ms{head}.\ms{next}=\ms{tail}$ (by assumption). Since no $\lit{remove}$ operation
deallocates the \emph{head} or \emph{tail}, the $\lit{read}$ of $X$ and $Y$ are \emph{valid} in $H$.

Now, let $X$ be the \emph{head} element and suppose that $\pi$ performs $\lit{read}(X)$ followed by $\lit{read}(Y)$; $Y\neq \ms{tail}$ in $H$.
Clearly, if $\pi$ performs a $\lit{read}(Y)$, there exists an
operation $\pi'=\lit{insert}$ that has previously set $\emph{head}.\ms{next}=Y$.
More specifically, $\pi.\lit{read}(X)$ performs the action in line~\ref{line:linr} after the
write to shared memory by $\pi'$ in line~\ref{line:inscommit}. 
By the assignment of linearization points to tx-operations, $\ell_{\pi'} <_{\alpha} \ell_{\pi.\lit{read}(X)}$.
Thus, there exists an $\lit{insert}$ operation that sets $X.\ms{next}=Y$ prior to $\pi.\lit{read}(X)$ in $H$.

For the second claim, we need to prove that the $\lit{read}(Y)$ by $\pi$ is \emph{valid} in $H$.
Suppose by contradiction that $Y$ has been deallocated by some $\pi''=\lit{remove}$ operation prior to $\lit{read}(Y)$ by $\pi$.
By the rules for linearization of read and write operations, the action in line~\ref{line:commit-update} precedes the action in line~\ref{line:linr}.
However, $\pi$ proceeds to perform the check in line~\ref{line:flagcheck} and returns $\bot$ since the flag corresponding to the element $Y$ is previously set by $\pi''$. Thus, $H$ does not contain $\pi.\lit{read}(Y)$---contradiction.

Inductively, by the above arguments, every non-\emph{head} $\lit{read}$ by $\pi$ 
is performed on an element previously created by an $\lit{insert}$ operation
and is valid in $H$.
\end{proof} 
\begin{lemma}
\label{lem:rls}
$H$ is locally serializable with respect to $\LL$.
\end{lemma}
\begin{proof}
By Lemma~\ref{lem:invar1}, every element $X$ read 
within an operation $\pi$ 
is previously created by an $\lit{insert}$ operation and is valid in $H$.
Moreover, if the read operation on $X$ returns $v'$, then 
$X.\textit{next}$ stores a pointer to another valid element that
stores an integer value $v''>v'$.
Note that the series of reads performed by $\pi$ terminates as soon as 
an element storing value $v$ or higher is found. Thus, $\pi$ performs at most
$O(|v-v_0|)$ reads, where $v_0$ is the value of the second element read by $\pi$.  
Now we construct $S^{\pi}$ as a sequence of $\lit{insert}$ operations,
that insert values read by $\pi$, one by one, followed by $\pi$. 
By construction, $S^{\pi}\in\Sigma_{\ms{LL}}$.
\end{proof}
It is sufficient for us to prove that every finite high-level history $H$
of $I^{RM}$ is linearizable.
%
First, we obtain a completion $\tilde H$ of $H$ as follows.
The invocation of an incomplete \lit{contains} operation is discarded.
The invocation of an incomplete $\pi=\lit{insert} \vee \lit{remove}$
operation that has not returned successfully from the $\lit{write}$
operation is discarded; otherwise, it is completed with response $\true$.

We obtain a sequential high-level history $\tilde S$ equivalent to $\tilde H$ by associating a linearization point $\ell_{\pi}$ 
with each operation $\pi$ as follows.
For each $\pi=\lit{insert}\vee\lit{remove}$ that returns $\true$ in $\tilde H$, $\ell_{\pi}$ is associated with 
the first $\lit{write}$ performed by $\pi$ in $H$; otherwise
$\ell_{\pi}$ is associated with the last $\lit{read}$ performed by $\pi$ in $H$. For $\pi=\lit{contains}$ that returns $\true$, $\ell_{\pi}$
is associated with the last $\lit{read}$ performed in $I^{RM}$; otherwise $\ell_{\pi}$ is associated with the $\lit{read}$ of \emph{head}.
Since linearization points are chosen within the intervals of 
operations of $I^{RM}$, for any two operations
$\pi_i$ and $\pi_j$ in ${\tilde H}$, if $\pi_i \rightarrow_{\tilde H}
\pi_j$, then $\pi_i \rightarrow_{\tilde S} \pi_j$.
\begin{lemma}
\label{lem:rlegal}
$\tilde S$ is consistent with the sequential specification of type \textit{set}.
\end{lemma}
\begin{proof}
Let ${\tilde S}^k$ be the prefix of $\tilde S$ consisting of
the first $k$ complete operations. 
We associate each ${\tilde S}^k$ with a set $q^k$ of objects that were
successfully inserted and not subsequently successfully removed in ${\tilde S}^k$.
We show by induction on $k$ that the sequence of state transitions in
${\tilde S}^k$ is consistent with operations' responses in ${\tilde
  S}^k$ with respect to the \textit{set} type. 

The base case $k=1$ is trivial: the \textit{tail} element containing
$+\infty$ is successfully inserted.
Suppose that ${\tilde S}^k$ is consistent with the \emph{set} type and
let $\pi_1$ with argument $v\in \mathbb{Z}$ and response $r_{\pi_{1}}$
be the last operation of ${\tilde S}^{k+1}$.  
We want to show that $(q^k,\pi_1,q^{k+1},r_{\pi_{1}})$ is consistent with the \textit{set} type. 
%
%
\begin{enumerate}
\item[(1)]
If $\pi_1=\lit{insert}(v)$ returns $\true$ in ${\tilde S}^{k+1}$, there does not exist any other $\pi_2=\lit{insert}(v)$ that returns $\true$ in ${\tilde S}^{k+1}$ such that there does not exist any $\lit{remove}(v)$ that returns $\true$; $\pi_2 \rightarrow_{{\tilde S}^{k+1}} \lit{remove}(v) \rightarrow_{{\tilde S}^{k+1}} \pi_1$.
Suppose by contradiction that such a $\pi_1$ and $\pi_2$ exist.
Every successful $\lit{insert}(v)$ operation performs its penultimate $\lit{read}$ on an element $X$ that stores a value $v'<v$ and the last read is performed on an element that stores a value $v''>v$. Clearly, $\pi_1$ also performs a $\lit{write}$ on $X$.
By construction of $\tilde S$, $\pi_1$ is linearized at the release of the \emph{cas} lock on element $X$.
Observe that $\pi_2$ must also perform a $\lit{write}$ to the element $X$ (otherwise one of $\pi_1$ or $\pi_2$ would return $\false$).
By assumption, the write to $X$ in shared-memory by $\pi_2$
(line~\ref{line:inscommit}) precedes the corresponding write to $X$ in
shared-memory by $\pi_2$. If $\ell_{\pi_2} <_{\alpha}
\ell_{\pi_{1}.\lit{read}(X)}$, then $\pi_1$ cannot return $\true$---a contradiction.
Otherwise, if $\ell_{\pi_{1}.\lit{read}(X)}  <_{\alpha} \ell_{\pi_2}$,
then $\pi_1$ reaches line~\ref{line:linw} and return
$\bot$. This is because either $\pi_1$ attempts to acquire the
\emph{cas} lock on $X$ while it is still held by $\pi_2$ or the value
of $X$ contained in the \emph{rbuf} of the process executing $\pi_1$
has changed---a contradiction. 

If $\pi_1=\lit{insert}(v)$ returns $\false$ in ${\tilde S}^{k+1}$, there exists a $\pi_2=\lit{insert}(v)$ that returns $\true$ in ${\tilde S}^{k+1}$ such that there does not exist any $\pi_3=\lit{remove}(v)$ that returns $\true$; $\pi_2 \rightarrow_{{\tilde S}^{k+1}} \pi_3 \rightarrow_{{\tilde S}^{k+1}} \pi_1$. 
Suppose that such a $\pi_2$ does not exist. Thus, $\pi_1$ must perform
its last $\lit{read}$ on an element that stores value $v''>v$, perform
the action in Line~\ref{line:inscommit} and return $\true$---a contradiction.

It is easy to verify that the conjunction of the above two claims prove that $\forall q\in Q$; $\forall v\in \mathbb{Z}$, ${\tilde S}^{k+1}$ satisfies $(q,\lit{insert}(v),q \cup \{v\},(v \not\in q))$.
\item[(2)]
If $\pi_1=\lit{remove}(v)$, similar arguments as applied to $\lit{insert}(v)$ prove that $\forall q\in Q$; $\forall v\in \mathbb{Z}$, ${\tilde S}^{k+1}$ satisfies $(q,\lit{remove}(v),q \setminus \{v\},(v \in q))$.

\item[(3)]
If $\pi_1=\lit{contains}(v)$ returns $\true$ in ${\tilde S}^{k+1}$, there exists $\pi_2=\lit{insert}(v)$ that returns \emph{true} in ${\tilde S}^{k+1}$ such that there does not exist any $\lit{remove}(v)$ that returns \emph{true} in ${\tilde S}^{k+1}$ such that $\pi_2 \rightarrow_{{\tilde S}^{k+1}} \lit{remove}(v) \rightarrow_{{\tilde S}^{k+1}} \pi_1$.
The proof of this claim immediately follows from Lemma~\ref{lem:invar1}.

Now, if $\pi_1=\lit{contains}(v)$ returns $\false$ in ${\tilde S}^{k+1}$, 
there does not exist an $\pi_2=\lit{insert}(v)$ that returns \emph{true} such that 
there does not exist any $\lit{remove}(v)$ that returns \emph{true}; 
$\pi_2 \rightarrow_{{\tilde S}^{k+1}} \lit{remove}(v) \rightarrow_{{\tilde S}^{k+1}} \lit{contains}(v)$.
Suppose by contradiction that such a $\pi_1$ and $\pi_2$ exist. 
Thus, the action in line~\ref{line:inscommit} by the $\lit{insert}(v)$ operation that updates some element, 
say $X$ precedes the action in line~\ref{line:linr} by $\lit{contains}(v)$ that is associated with its 
first $\lit{read}$ (the \emph{head}).
We claim that $\lit{contains}(v)$ must read the element $X'$ newly created by
$\lit{insert}(v)$ and return $\true$---a contradiction to the initial assumption that it returns $\false$.
The only case when this can happen is if there exists a $\lit{remove}$
operation that forces $X'$ to be unreachable from \emph{head} i.e. concurrent to the $\lit{write}$
to $X$ by $\lit{insert}$, there exists a $\lit{remove}$ that sets $X''.\textit{next}$ to $X.\textit{next}$
after the action in line~\ref{line:inswrite} by $\lit{insert}$.
But this is not possible since the \emph{cas} on $X$ performed by the $\lit{remove}$
would return $\false$.
\end{enumerate}
Thus, inductively, the sequence of state transitions in ${\tilde S}$
satisfies the sequential specification of the \textit{set} type. 
\end{proof}
Lemmas~\ref{lem:rls} and~\ref{lem:rlegal} imply:
\begin{theorem}
\label{th:lr}
$I^{RM}$ is \LS-linearizable with respect to $(\LL,\ms{set})$.
\end{theorem}
\vspace{1mm}\noindent\textbf{Proof of concurrency optimality.}
Now we show that $I^{RM}$ supersedes, in terms of concurrency, \emph{any} implementation in classes
$\mathcal{P}$ or $\mathcal{SM}$.
The proof is based on a more general optimality result, interesting in its own right: 
any finite schedule rejected by $I^{RM}$ is not \emph{observably
LS-linearizable} (or simply \emph{observable}). 
We show that any finite schedule rejected by our algorithm  is not \emph{observably correct}. 

A correct schedule $\sigma$ is observably correct
if by completing update operations in $\sigma$ and 
extending, for any $v\in \mathbb{Z}$, the resulting schedule with a complete sequential execution 
$\lit{contains}(v)$, applied to the resulting contents of the list,
we obtain a correct schedule.
Here the contents of the list after a given correct schedule 
is determined based on the order of its write operations. For each
element, we define the resulting state of its
\textit{next} field based on the last write in the schedule. 
Since in a correct schedule, each new element is first created and then
linked to the list, we can reconstruct the \emph{state of the list} by
iteratively traversing it, starting from $\ms{head}$.    
 
Intuitively, a schedule is observably correct if it incurs no ``lost updates''.
Consider, for example a schedule (cf. Figure~\ref{fig:ex2}(b)) in which two operations,
$\lit{insert}(1)$ and $\lit{insert}(2)$ are applied to the list with state $\{3\}$. 
The resulting schedule is trivially correct (both operations
return $\ms{true}$ so the schedule can some from a complete linearizable history).  
However, in the schedule, one of the operations, say $\lit{insert}(1)$,
overwrites the effect of the other one.  
Thus, if we extend the schedule with a complete execution of
$\lit{contains}(2)$, the only possible response it may give is
$\ms{false}$ which obviously does not produce a linearizable high-level history.  
\begin{theorem}[Optimality]
\label{th:lrelaxed}
$I^{RM}$ accepts all schedules that are observable with respect to $(\ms{LL},\ms{set})$.
\end{theorem}
\begin{proof}
We prove that any schedule rejected by $I^{RM}$ is  
not observable.
We go through the cases when a read or write returns $\bot$ (implying the operation fails to return a matching response) and
thus the current schedule is rejected:
(1) $\lit{read}(X_{\ell})$ returns $\bot$ in line~\ref{line:elastic:abort1}
when $r[\ell]=\true$ or when $\ms{ver}_1\neq \ms{ver}_2$, (2)
$\lit{write}(X_{\ell})$ performed by $\lit{remove}(v)$ either returns $\bot$
in line~\ref{line:linw} when the $\lit{cas}$ operation on
$L[\ell]$ returns $\false$ or returns $\bot$
in line~\ref{line:linw2} when the $\lit{cas}$ operation on the element that
stores $v$ returns $\false$, and (3) $\lit{write}(X_{\ell})$ performed by $\lit{insert}$ returns $\bot$
in line~\ref{line:inswrite} when the $\lit{cas}$ operation on
$L[\ell]$ returns $\false$.

Consider the subcase (1a), $r[\ell]$ is set $\true$ by a
preceding or concurrent $\lit{write}(X_{\ell})$ (line~\ref{line:elastic:gc}).
The high-level operation performing this $\lit{write}$ is a $\lit{remove}$ 
that marks the corresponding list element as removed. 
Since no removed element can be read in a sequential execution
of $LL$, the corresponding history is not locally serializable.
Alternatively, in subcase (1b), the version of $X_{\ell}$ read previously in line~\ref{line:rver}
has changed. Thus, an update operation has concurrently performed a write to $X_{\ell}$.
However, there exist executions that export such schedules.

In case $(2)$, the $\lit{write}$ performed by a $\lit{remove}$ operation returns $\bot$.
In subcase (2a), $X_{\ell}$ is currently
locked. 
Thus, a concurrent high-level operation has previously locked $X_{\ell}$ (by successfully
performing $L[\ell].\lit{cas}()$ in line~\ref{line:linw}) and has not yet released
the lock (by writing $\tup{\ms{ver}', \lit{false}}$ to
$L[\ell]$ in line~\ref{line:release}). 
In subcase (2b), the current version of $X_{\ell}$ (stored in $L[\ell]$) differs 
from the version of $X_{\ell}$ witnessed by a preceding $\lit{read}$. 
Thus, a concurrent high-level operation completed a write to $X_{\ell}$
\emph{after} the current high-level operation $\pi$ performed a $\lit{read}$ of $X_{\ell}$.
In both (2a) and (2b), a concurrent high-level updating operation $\pi'$
($\lit{remove}$ or $\lit{insert}$) has written or is about to perform a $\lit{write}$ to $X_{\ell}$.  
In subcase (2c), the $\lit{cas}$ on the element $X_{\ell'}$ (element that stores the value $v$) executed by $\lit{remove}(v)$ 
returns $\false$ (line~\ref{line:linw2}).
Recall that by the sequential implementation $\LL$, 
$\lit{remove}(v)$ performs a $\lit{read}$ of
$X_{\ell'}$ prior to the $\lit{write}(X_{\ell})$, where $X_{\ell}.\textit{next}$
refers to $X_{\ell'}$.
If the \emph{cas} on $X_{\ell'}$ fails, there exists a process 
that concurrently performed a $\lit{write}$ to $X_{\ell'}$, but
after the $\lit{read}$ of $X_{\ell'}$ by $\lit{remove}(v)$.
In all cases, we observe that if we did not abort the write to $X_{\ell}$, then
the schedule extended by a complete execution of $\lit{contains}$ is not LSL.

In case $(3)$, the $\lit{write}$ performed by an $\lit{insert}$ operation returns $\bot$. Similar arguments to
case $(2)$ prove that any schedule rejected is not observable LSL.
\end{proof}
Theorem~\ref{th:lrelaxed} implies that the schedules exported by the
histories in Figures~\ref{fig:ex1} and \ref{fig:ex2}(a) and that are not
accepted by any $I'\in \mathcal{SM}$ and any $I\in \mathcal{P}$, respectively,
are indeed accepted by $I^{RM}$.
But it is easy to see that implementations in $\mathcal{SM}$ and $\mathcal{P}$ can only
accept observable schedules.  
As a result, $I^{RM}$ can be shown to strictly supersede any
pessimistic or TM-based implementation of the list-based set.  
\begin{corollary}
\label{cr:mrp}
$I^{RM}$ accepts every schedule accepted by any implementation in $\mathcal{P}$ and $\mathcal{SM}$.
Moreover, $I^{RM}$ accepts schedules $\sigma$ and $\sigma'$ that are rejected by any
implementation in $\mathcal{P}$ and $\mathcal{SM}$, respectively. 
\end{corollary}
%
%
%
\section{Related work and Discussion}
\label{sec:p2c5}
\vspace{1mm}\noindent\textbf{Measuring concurrency.}
Sets of accepted schedules are commonly used as a
metric of concurrency provided by a shared memory
implementation.
Gramoli et al.~\cite{GHF10} defined a concurrency metric, the \emph{input
acceptance}, as the ratio of committed transactions over aborted
transactions when TM executes the given schedule.   
Unlike our metric, input acceptance does not apply to
lock-based programs. 

For static database transactions, 
Kung and Papadimitriou~\cite{KP79} use the metric to 
capture the parallelism of a locking scheme,
While acknowledging that the metric is theoretical, they 
insist that it may
have ``practical significance as
well, if the schedulers in question have relatively small
scheduling times as compared with waiting and execution
times.'' 
Herlihy~\cite{Her90} employed the metric to compare various
optimistic and pessimistic synchronization techniques using
commutativity
of operations constituting high-level transactions.   
A synchronization technique is implicitly considered in~\cite{Her90} as highly
concurrent, namely ``optimal'',
if no other technique accepts more schedules. 
By contrast, we focus here on a \emph{dynamic} model where the scheduler cannot 
use the prior knowledge of all the shared addresses to be accessed. 
Also, unlike~\cite{KP79,Her90}, 
the results in this chapter require \emph{all} operations, including aborted ones, to observe (locally) consistent states.

\vspace{1mm}\noindent\textbf{Concurrency optimality.}
This chapter shows that ``semantics-oblivious'' optimistic TM and ``semantics-aware'' pessimistic locking are 
incomparable with respect
to exploiting concurrency of the list-based set. Yet, we have shown how to use the benefits
of optimism to derive a concurrency optimal implementation that is fine-tuned to the semantics of the list-based set.
Intuitively, the ability of an implementation to successfully process
interleaving steps of concurrent threads is an appealing property that
should be met by performance gains.
We believe this to be so.

In work that is not part of the thesis~\cite{optimistic-list15}, we confirm experimentally that
the concurrency optimal optimistic implementation of the list-based set based on $I^{RM}$
outperforms the state-of-the-art implementations of the list-based set, namely, the \emph{Lazy linked list}~\cite{HHL+05}
and the \emph{Harris-Michael linked list}~\cite{harris-set,michael-set}.
Does the claim also hold for other data structures? We suspect so. For
example, similar but more general data structures, such as skip-lists
or tree-based dictionaries, may allow for optimizations similar to 
proposed in this paper.      
Our results provides some preliminary hints in the quest for 
the ``right'' synchronization technique to develop highly concurrent and efficient implementations of data 
types.
%

%
\chapter{Concluding remarks}
\label{ch:conc}
\epigraph{Everything has to come to an end, sometime.}
{\textit{Lyman Frank Baum}-The Marvelous Land of Oz}
The inclusion of hardware support for transactions in mainstream CPU's~\cite{Rei12, asf, bluegene}
suggests that TM is an important concurrency abstraction.
However, hardware transactions are not going to be sufficient to support efficient concurrent programming 
since they may be aborted spuriously;
the fast but potentially unreliable hardware transactions must be complemented 
with slower, but more reliable software transactions.
Thus, understanding the inherent cost of both hardware and software transactions is of both theoretical and practical interest.

Below, we briefly recall the outcomes of the thesis and overview the future research directions.

\vspace{1mm}\noindent\textbf{Safety for TMs.}
We formalized the semantics of a safe TM: 
every transaction, including aborted and incomplete ones, must observe a view 
that is consistent with some sequential execution.
We introduced the notion of deferred-update semantics which explicitly precludes reading
from a transaction that has not yet invoked tryCommit.
We believe that our definition is useful to TM
practitioners, since it streamlines possible implementations of t-read
and tryCommit operations.

\vspace{1mm}\noindent\textbf{Complexity of TMs.}
The cost of the TM abstraction is parametrized by several properties: 
safety for transactions, conditions under which transactions must terminate, 
conditions under which transactions must commit/abort, bound on the number of versions that can be maintained
and a multitude of other implementation strategies like disjoint-access parallelism and invisible reads.

At a high-level, the complexity bounds presented in the thesis suggest that providing high degrees of concurrency 
in software transactional memory (STM) implementations incurs a considerable synchronization cost.
As we show, permissive STMs, while providing the best possible concurrency in theory,
require a strong synchronization primitive (AWAR) or a memory fence (RAW) per
read operation, which may result in excessively slow execution times.
Progressive STMs provide only basic concurrency by adapting to data conflicts, but perform considerably better in this respect: 
we present progressive implementations that incur constant RAW/AWAR complexity.

Since Transactional memory was originally proposed as an alternative to locking, 
early STMs implementations~\cite{HLM+03, astm, ST95,nztm,fraser} 
adopted optimistic concurrency control and guaranteed
that a prematurely halted transaction cannot not prevent
other transactions from committing. 
However, popular state-of-the-art STM implementations like \emph{TL2}~\cite{DSS06} and \emph{NOrec}~\cite{norec}
are progressive, providing no non-blocking progress guarantees for transactions, but perform empirically better than
obstruction-free TMs. Complexity lower and upper bounds presented in the thesis explain this performance gap.

Do our results mean that maximizing the ability of processing multiple transactions
in parallel or providing non-blocking progress should not be an important factor in STM design?
It would seem so.
Should we rather even focus on speculative ``single-lock'' solutions \'a la \emph{flat combining}~\cite{HendlerIST10}
or ``pessimistic'' STMs in which transactions never abort~\cite{pessimistic-stm}?    
Difficult to say affirmatively, but probably not, since our results suggest progressive STMs
incur low complexity overheads as also evidenced by their good empirical performance on most 
realistic TM workloads~\cite{DSS06,norec}.

Several questions yet remain open on the complexity of STMs. For instance,
the bounds in the thesis were derived for the TM-correctness property of strict serializability and
its restrictions. But there has been study of relaxations of strict serializability 
like \emph{snapshot isolation}~\cite{AHM09,BDFG14}. Verifying if the lower bounds presented in the thesis
hold under such weak TM-correctness properties and extending the proofs if indeed, presents interesting open
questions.
The discussion section of Chapters~\ref{ch:p3c2}, \ref{ch:p3c3} and \ref{ch:p3c4} additionally list some unresolved questions
closely related to the results in the thesis.

One problem of practical need that is not considered in the thesis concerns the interaction
of transactional code with \emph{non-transactional} code, \emph{i.e.}, the same data item is accessed both
transactionally and non-transactionally.
It is expected that code executed within a transaction behave as lock-based code within a 
single ``global lock''~\cite{scott-book, menon-sgl} to avoid memory races.
Techniques to ensure the safety of non-transactional accesses have been formulated through the notion of
\emph{privatization}~\cite{attiyaH13,spear-priv}. Devising techniques to ensure privatization for TMs and
understanding the cost of enforcing it is an important research direction.

In the thesis, we assumed that a rmw event is an access to a single base object.
However, there have been proposals to provide implementations with the ability to invoke 
$k$-rmw; $k\in \mathbb{N}$ primitives~\cite{AHCAS,dcas} that allow accessing up to $k$ base objects in a single atomic event.
For example, the $k$-cas instruction allows to perform $k$ cas instructions atomically
on a vector $\langle b_1,\ldots , b_k \rangle$ of base objects: it accepts as input a vector 
$\langle \ms{old}_1,\ldots , \ms{old}_k, \ms{new}_1,\ldots , \ms{new}_k \rangle$
and atomically updates the value of $\langle b_1,\ldots , b_k \rangle$ to $\langle \ms{new}_1,\ldots , \ms{new}_k \rangle$ 
and returns $\true$ \emph{iff} for all $i\in \{1,\ldots , k\}$, $\ms{old}_i=\ms{new}_i$; otherwise it returns $\false$. 
However, the ability to access such $k$-rmw primitives does not necessarily simplify the design
and improve the performance of non-blocking implementations nor overcome the compositionality issue~\cite{HS08-book,AHCAS,dcas}. 
Nonetheless, verifying if the lower bounds presented in the thesis hold in this shared memory model is an interesting problem.

\vspace{1mm}\noindent\textbf{HyTMs.}
We have introduced an analytical model for hybrid transactional memory
that captures the notion of cached accesses as performed by hardware transactions.
We then derived lower and upper bounds in this model to capture the inherent tradeoff between the degree of concurrency 
allowed among hardware and software transactions
and the instrumentation overhead introduced on the hardware.
In a nutshell, our results say that it is impossible to completely forgo instrumentation in a sequential HyTM, and 
that any opaque HyTM implementation providing non-trivial progress either has to pay a \emph{linear} number of 
metadata accesses, or will have to allow slow-path transactions to \emph{abort} fast-path operations.  

Our model of HTMs assumed that the hardware resources were \emph{bounded}, in the sense that,
a hardware transaction may only access a bounded number of data items, exceeding which, it incurs a capacity abort.
To overcome the inherent limitations of bounded HTMs, there have been proposals for ``unbounded HTMs''
that allow transactions to commit even if they exceed the hardware resources~\cite{hammondhytm,unboundedhtm1}.
The HyTM model from Chapter~\ref{ch:p4c4} can be easily extended to accommodate unbounded HTM designs
by disregarding capacity aborts.

Some papers have investigated alternatives to providing HTMs with an STM fallback, 
such as \emph{sandboxing}~\cite{ALM14,CTGM14}, or employing \emph{hardware-accelerated} STM~\cite{HASTM1,HASTM2}, 
and the use of both direct \emph{and} cached accesses 
within the same hardware transaction to reduce instrumentation 
overhead~\cite{riegel-thesis,hynorecriegel, kumarhytm}.
Another approach proposed \emph{reduced hardware transactions}~\cite{MS13}, 
where a part of the slow-path is executed using a short fast-path transaction, which allows to
partially eliminate instrumentation from the hardware fast-path.
Modelling and deriving complexity bounds for HyTM proposals outside the HyTM model described in the thesis
is an interesting future direction.

\vspace{1mm}\noindent\textbf{Relaxed transactional memory.}
The concurrency lower bounds derived in Chapter~\ref{ch:p2c1} illustrated that a strictly serializable TM, 
when used as a black-box to
transform a sequential implementation of the list-based set to a concurrent one, is not concurrency-optimal.
This is due to the fact that TM detects conflicts at the level of transactional reads and writes
resulting in \emph{false conflicts}, in the sense that, the read-write conflict may not affect the correctness
of the implemented high-level set type.
As we have shown, we can derive a \emph{concurrency optimal} optimistic (non-strictly serializable) implementation
that can process every correct schedule of the list-based set.
Indeed, several papers have studied ``relaxed'' TMs that are fined-tuned to the semantics of the high-level data type~\cite{HK08,FGG09, HK-commutativity}.
Exploring the complexity of such relaxed TM models represents a very important future research direction.
%

\listoffigures
\listoftables
\bibliography{references}

\def\noopsort#1{} \def\No{\kern-.25em\lower.2ex\hbox{\char'27}}
  \def\no#1{\relax} \def\http#1{{\\{\small\tt
  http://www-litp.ibp.fr:80/{$\sim$}#1}}}
\begin{thebibliography}{100}

\bibitem{asf}
{Advanced Synchronization Facility Proposed Architectural Specification}, March
  2009.
\newblock
  \url{http://developer.amd.com/wordpress/media/2013/09/45432-ASF_Spec_2.1.pdf}.

\bibitem{gcc-tm}
\href{https://gcc.gnu.org/wiki/TransactionalMemory} {Transactional Memory in
  GCC}.
\newblock 2012.

\bibitem{AdveG96}
S.~V. Adve and K.~Gharachorloo.
\newblock Shared memory consistency models: A tutorial.
\newblock {\em IEEE Computer}, 29(12):66--76, 1996.

\bibitem{ALM14}
Y.~Afek, A.~Levy, and A.~Morrison.
\newblock Software-improved hardware lock elision.
\newblock In {\em PODC}. ACM, 2014.

\bibitem{pessimistic-stm}
Y.~Afek, A.~Matveev, and N.~Shavit.
\newblock Pessimistic software lock-elision.
\newblock In {\em Proceedings of the 26th International Conference on
  Distributed Computing}, DISC'12, pages 297--311, Berlin, Heidelberg, 2012.
  Springer-Verlag.

\bibitem{AFHHT07}
M.~K. Aguilera, S.~Fr{\o}lund, V.~Hadzilacos, S.~L. Horn, and S.~Toueg.
\newblock Abortable and query-abortable objects and their efficient
  implementation.
\newblock In {\em PODC}, pages 23--32, 2007.

\bibitem{AlistarhEMMS14}
D.~Alistarh, P.~Eugster, M.~Herlihy, A.~Matveev, and N.~Shavit.
\newblock Stacktrack: An automated transactional approach to concurrent memory
  reclamation.
\newblock In {\em Proceedings of the Ninth European Conference on Computer
  Systems}, EuroSys '14, pages 25:1--25:14, New York, NY, USA, 2014. ACM.

\bibitem{hytm14}
D.~Alistarh, J.~Kopinsky, P.~Kuznetsov, S.~Ravi, and N.~Shavit.
\newblock Inherent limitations of hybrid transactional memory.
\newblock {\em CoRR}, abs/1405.5689, 2014.

\bibitem{hytm14disc}
D.~Alistarh, J.~Kopinsky, P.~Kuznetsov, S.~Ravi, and N.~Shavit.
\newblock Inherent limitations of hybrid transactional memory.
\newblock {\em CoRR}, abs/1405.5689, 2014.
\newblock To appear in 29th International Symposium on Distributed Computing
  (DISC'15), Japan.

\bibitem{WTTM3}
D.~Alistarh, J.~Kopinsky, P.~Kuznetsov, S.~Ravi, and N.~Shavit.
\newblock Inherent limitations of hybrid transactional memory.
\newblock {\em 6th Workshop on the Theory of Transactional Memory, Paris,
  France}, 2014.

\bibitem{AS85}
B.~Alpern and F.~B. Schneider.
\newblock Defining liveness.
\newblock {\em Inf. Process. Lett.}, 21(4):181--185, Oct. 1985.

\bibitem{unboundedhtm1}
C.~S. Ananian, K.~Asanovic, B.~C. Kuszmaul, C.~E. Leiserson, and S.~Lie.
\newblock Unbounded transactional memory.
\newblock In {\em Proceedings of the 11th International Symposium on
  High-Performance Computer Architecture}, HPCA '05, pages 316--327,
  Washington, DC, USA, 2005. IEEE Computer Society.

\bibitem{Anderson99-multi}
J.~H. Anderson and M.~Moir.
\newblock Universal constructions for multi-object operations.
\newblock In {\em Proceedings of the Fourteenth Annual ACM Symposium on
  Principles of Distributed Computing}, PODC '95, pages 184--193, New York, NY,
  USA, 1995. ACM.

\bibitem{anderson-90-tpds}
T.~E. Anderson.
\newblock The performance of spin lock alternatives for shared-memory
  multiprocessors.
\newblock {\em IEEE Trans. Parallel Distrib. Syst.}, 1(1):6--16, 1990.

\bibitem{AttiyaGHR2014}
H.~Attiya, A.~Gotsman, S.~Hans, and N.~Rinetzky.
\newblock Safety of live transactions in transactional memory: {TMS} is
  necessary and sufficient.
\newblock In {\em DISC}, pages 376--390, 2014.

\bibitem{AGHK09}
H.~Attiya, R.~Guerraoui, D.~Hendler, and P.~Kuznetsov.
\newblock The complexity of obstruction-free implementations.
\newblock {\em J. ACM}, 56(4), 2009.

\bibitem{AGK11-popl}
H.~Attiya, R.~Guerraoui, D.~Hendler, P.~Kuznetsov, M.~Michael, and M.~Vechev.
\newblock Laws of order: Expensive synchronization in concurrent algorithms
  cannot be eliminated.
\newblock In {\em POPL}, pages 487--498, 2011.

\bibitem{icdcs-opacity}
H.~Attiya, S.~Hans, P.~Kuznetsov, and S.~Ravi.
\newblock Safety of deferred update in transactional memory.
\newblock {\em 2013 IEEE 33rd International Conference on Distributed Computing
  Systems}, 0:601--610, 2013.

\bibitem{icdcs-TR}
H.~Attiya, S.~Hans, P.~Kuznetsov, and S.~Ravi.
\newblock Safety of deferred update in transactional memory.
\newblock {\em CoRR}, abs/1301.6297, 2013.

\bibitem{safety-tm14}
H.~Attiya, S.~Hans, P.~Kuznetsov, and S.~Ravi.
\newblock Safety and deferred update in transactional memory.
\newblock In R.~Guerraoui and P.~Romano, editors, {\em Transactional Memory.
  Foundations, Algorithms, Tools, and Applications}, volume 8913 of {\em
  Lecture Notes in Computer Science}, pages 50--71. Springer International
  Publishing, 2015.

\bibitem{AHCAS}
H.~Attiya and D.~Hendler.
\newblock Time and space lower bounds for implementations using k-cas.
\newblock {\em Parallel and Distributed Systems, IEEE Transactions on},
  21(2):162 --173, feb. 2010.

\bibitem{rmr-mutex}
H.~Attiya, D.~Hendler, and P.~Woelfel.
\newblock Tight rmr lower bounds for mutual exclusion and other problems.
\newblock In {\em Proceedings of the Twenty-seventh ACM Symposium on Principles
  of Distributed Computing}, PODC '08, pages 447--447, New York, NY, USA, 2008.
  ACM.

\bibitem{attiyaH13}
H.~Attiya and E.~Hillel.
\newblock The cost of privatization in software transactional memory.
\newblock {\em IEEE Trans. Computers}, 62(12):2531--2543, 2013.

\bibitem{AHM09}
H.~Attiya, E.~Hillel, and A.~Milani.
\newblock Inherent limitations on disjoint-access parallel implementations of
  transactional memory.
\newblock {\em Theory of Computing Systems}, 49(4):698--719, 2011.

\bibitem{Attiya09-tmread}
H.~Attiya and A.~Milani.
\newblock Transactional scheduling for read-dominated workloads.
\newblock In {\em Proceedings of the 13th International Conference on
  Principles of Distributed Systems}, OPODIS '09, pages 3--17, Berlin,
  Heidelberg, 2009. Springer-Verlag.

\bibitem{ARR10}
H.~Attiya, G.~Ramalingam, and N.~Rinetzky.
\newblock Sequential verification of serializability.
\newblock In {\em Proceedings of the 37th annual ACM SIGPLAN-SIGACT symposium
  on Principles of programming languages}, pages 31--42, 2010.

\bibitem{AW04}
H.~Attiya and J.~Welch.
\newblock {\em Distributed Computing. Fundamentals, Simulations, and Advanced
  Topics.}
\newblock John Wiley \& Sons, 2004.

\bibitem{Barnes93}
G.~Barnes.
\newblock A method for implementing lock-free shared-data structures.
\newblock In {\em Proceedings of the Fifth Annual ACM Symposium on Parallel
  Algorithms and Architectures}, SPAA '93, pages 261--270, New York, NY, USA,
  1993. ACM.

\bibitem{BS88}
R.~Bayer and M.~Schkolnick.
\newblock Concurrency of operations on {B}-trees.
\newblock In {\em Readings in database systems}, pages 129--139. Morgan
  Kaufmann Publishers Inc., 1988.

\bibitem{dobbs-pinv}
E.~Bruno.
\newblock
  \href{http://www.drdobbs.com/jvm/what-is-priority-inversion-and-how-do-yo/230600008}
  {What Is Priority Inversion (And How Do You Control It)?}
\newblock 2011.

\bibitem{BDFG14}
V.~Bushkov, D.~Dziuma, P.~Fatourou, and R.~Guerraoui.
\newblock The pcl theorem: Transactions cannot be parallel, consistent and
  live.
\newblock In {\em SPAA}, pages 178--187, 2014.

\bibitem{bushkov2012}
V.~Bushkov, R.~Guerraoui, and M.~Kapalka.
\newblock On the liveness of transactional memory.
\newblock In {\em Proceedings of the 2012 ACM Symposium on Principles of
  Distributed Computing}, PODC '12, pages 9--18, New York, NY, USA, 2012. ACM.

\bibitem{CTGM14}
I.~Calciu, T.~Shpeisman, G.~Pokam, and M.~Herlihy.
\newblock Improved single global lock fallback for best-effort hardware
  transactional memory.
\newblock In {\em Transact 2014 Workshop}. ACM, 2014.

\bibitem{michel-permissive}
T.~{C}rain, D.~{I}mbs, and M.~{R}aynal.
\newblock {R}ead invisibility, virtual world consistency and permissiveness are
  compatible.
\newblock {R}esearch {R}eport, {ASAP} - {INRIA} - {IRISA} - {CNRS} : {UMR}6074
  - {INRIA} - {I}nstitut {N}ational des {S}ciences {A}ppliqu{\'e}es de {R}ennes
  - {U}niversit{\'e} de {R}ennes {I}, 11 2010.

\bibitem{hybridnorec}
L.~Dalessandro, F.~Carouge, S.~White, Y.~Lev, M.~Moir, M.~L. Scott, and M.~F.
  Spear.
\newblock {Hybrid NOrec: a case study in the effectiveness of best effort
  hardware transactional memory}.
\newblock In R.~Gupta and T.~C. Mowry, editors, {\em ASPLOS}, pages 39--52.
  ACM, 2011.

\bibitem{norec}
L.~Dalessandro, M.~F. Spear, and M.~L. Scott.
\newblock Norec: Streamlining stm by abolishing ownership records.
\newblock {\em SIGPLAN Not.}, 45(5):67--78, Jan. 2010.

\bibitem{damronhytm}
P.~Damron, A.~Fedorova, Y.~Lev, V.~Luchangco, M.~Moir, and D.~Nussbaum.
\newblock Hybrid transactional memory.
\newblock {\em SIGPLAN Not.}, 41(11):336--346, Oct. 2006.

\bibitem{DiceLMN09}
D.~Dice, Y.~Lev, M.~Moir, and D.~Nussbaum.
\newblock Early experience with a commercial hardware transactional memory
  implementation.
\newblock In {\em Proceedings of the 14th International Conference on
  Architectural Support for Programming Languages and Operating Systems},
  ASPLOS XIV, pages 157--168, New York, NY, USA, 2009. ACM.

\bibitem{DSS06}
D.~Dice, O.~Shalev, and N.~Shavit.
\newblock Transactional locking ii.
\newblock In {\em Proceedings of the 20th International Conference on
  Distributed Computing}, DISC'06, pages 194--208, Berlin, Heidelberg, 2006.
  Springer-Verlag.

\bibitem{DStransaction06}
D.~Dice and N.~Shavit.
\newblock What really makes transactions fast?
\newblock In {\em Transact}, 2006.

\bibitem{Dijkstra1}
E.~W. Dijkstra.
\newblock Solution of a problem in concurrent programming control.
\newblock {\em Commun. ACM}, 8(9):569--, Sept. 1965.

\bibitem{dcas}
S.~Doherty, D.~L. Detlefs, L.~Groves, C.~H. Flood, V.~Luchangco, P.~A. Martin,
  M.~Moir, N.~Shavit, and G.~L. Steele, Jr.
\newblock Dcas is not a silver bullet for nonblocking algorithm design.
\newblock In {\em Proceedings of the Sixteenth Annual ACM Symposium on
  Parallelism in Algorithms and Architectures}, SPAA '04, pages 216--224, New
  York, NY, USA, 2004. ACM.

\bibitem{DGLM13}
S.~Doherty, L.~Groves, V.~Luchangco, and M.~Moir.
\newblock Towards formally specifying and verifying transactional memory.
\newblock {\em Formal Asp. Comput.}, 25(5):769--799, 2013.

\bibitem{DragojevicMLM11}
A.~Dragojevi\'{c}, M.~Herlihy, Y.~Lev, and M.~Moir.
\newblock On the power of hardware transactional memory to simplify memory
  management.
\newblock In {\em Proceedings of the 30th Annual ACM SIGACT-SIGOPS Symposium on
  Principles of Distributed Computing}, PODC '11, pages 99--108, New York, NY,
  USA, 2011. ACM.

\bibitem{EFKMT12}
F.~Ellen, P.~Fatourou, E.~Kosmas, A.~Milani, and C.~Travers.
\newblock Universal constructions that ensure disjoint-access parallelism and
  wait-freedom.
\newblock In {\em PODC}, pages 115--124, 2012.

\bibitem{G05}
F.~Ellen, D.~Hendler, and N.~Shavit.
\newblock On the inherent sequentiality of concurrent objects.
\newblock {\em SIAM J. Comput.}, 41(3):519--536, 2012.

\bibitem{Ennals-code}
R.~Ennals.
\newblock The lightweight transaction library.
\newblock http://sourceforge.net/projects/libltx/files/.

\bibitem{Ennals05}
R.~Ennals.
\newblock Software transactional memory should not be obstruction-free.
\newblock 2005.

\bibitem{Fatourou11}
P.~Fatourou and N.~D. Kallimanis.
\newblock A highly-efficient wait-free universal construction.
\newblock In {\em Proceedings of the Twenty-third Annual ACM Symposium on
  Parallelism in Algorithms and Architectures}, SPAA '11, pages 325--334, New
  York, NY, USA, 2011. ACM.

\bibitem{FGG09}
P.~Felber, V.~Gramoli, and R.~Guerraoui.
\newblock Elastic transactions.
\newblock In {\em DISC}, pages 93--107, 2009.

\bibitem{cond-04}
F.~Fich, D.~Hendler, and N.~Shavit.
\newblock On the inherent weakness of conditional synchronization primitives.
\newblock In {\em Proceedings of the Twenty-third Annual ACM Symposium on
  Principles of Distributed Computing}, PODC '04, pages 80--87, New York, NY,
  USA, 2004. ACM.

\bibitem{fraser}
K.~Fraser.
\newblock Practical lock-freedom.
\newblock Technical report, Cambridge University Computer Laborotory, 2003.

\bibitem{GHF10}
V.~Gramoli, D.~Harmanci, and P.~Felber.
\newblock On the input acceptance of transactional memory.
\newblock {\em Parallel Processing Letters}, 20(1):31--50, 2010.

\bibitem{PODCKRV12}
V.~Gramoli, P.~Kuznetsov, and S.~Ravi.
\newblock From sequential to concurrent: correctness and relative efficiency
  (ba).
\newblock In {\em Principles of Distributed Computing (PODC)}, pages 241--242,
  2012.

\bibitem{opti14}
V.~Gramoli, P.~Kuznetsov, and S.~Ravi.
\newblock Optimism for boosting concurrency.
\newblock {\em CoRR}, abs/1203.4751, 2012.

\bibitem{optimistic-list15}
V.~Gramoli, P.~Kuznetsov, S.~Ravi, and D.~Shang.
\newblock A concurrency-optimal list-based set.
\newblock {\em CoRR}, abs/1502.01633, 2015.

\bibitem{optimistic-list15-disc}
V.~Gramoli, P.~Kuznetsov, S.~Ravi, and D.~Shang.
\newblock A concurrency-optimal list-based set (ba).
\newblock {\em CoRR}, abs/1502.01633, 2015.
\newblock To appear in 29th International Symposium on Distributed Computing
  (DISC'15).

\bibitem{gray1992}
J.~Gray and A.~Reuter.
\newblock {\em Transaction Processing: Concepts and Techniques}.
\newblock Morgan Kaufmann Publishers Inc., San Francisco, CA, USA, 1st edition,
  1992.

\bibitem{GHS08-permissiveness}
R.~Guerraoui, T.~A. Henzinger, and V.~Singh.
\newblock Permissiveness in transactional memories.
\newblock In {\em DISC}, pages 305--319, 2008.

\bibitem{OFTM}
R.~Guerraoui and M.~Kapalka.
\newblock On obstruction-free transactions.
\newblock In {\em Proceedings of the twentieth annual symposium on Parallelism
  in algorithms and architectures}, SPAA '08, pages 304--313, New York, NY,
  USA, 2008. ACM.

\bibitem{GK08-opacity}
R.~Guerraoui and M.~Kapalka.
\newblock On the correctness of transactional memory.
\newblock In {\em Proceedings of the 13th ACM SIGPLAN Symposium on Principles
  and Practice of Parallel Programming}, PPoPP '08, pages 175--184, New York,
  NY, USA, 2008. ACM.

\bibitem{GK09-progressiveness}
R.~Guerraoui and M.~Kapalka.
\newblock The semantics of progress in lock-based transactional memory.
\newblock {\em SIGPLAN Not.}, 44(1):404--415, Jan. 2009.

\bibitem{tm-theory}
R.~Guerraoui and M.~Kapalka.
\newblock Transactional memory: Glimmer of a theory.
\newblock In {\em Proceedings of the 21st International Conference on Computer
  Aided Verification}, CAV '09, pages 1--15, Berlin, Heidelberg, 2009.
  Springer-Verlag.

\bibitem{tm-book}
R.~Guerraoui and M.~Kapalka.
\newblock {\em Principles of Transactional Memory, Synthesis Lectures on
  Distributed Computing Theory}.
\newblock Morgan and Claypool, 2010.

\bibitem{stmbench7}
R.~Guerraoui, M.~Kapalka, and J.~Vitek.
\newblock Stmbench7: A benchmark for software transactional memory.
\newblock {\em SIGOPS Oper. Syst. Rev.}, 41(3):315--324, Mar. 2007.

\bibitem{GR14-lin}
R.~Guerraoui and E.~Ruppert.
\newblock Linearizability is not always a safety property.
\newblock In {\em NETYS}, pages 57--69, 2014.

\bibitem{WTTM1}
P.~K. Hagit~Attiya, Sandeep~Hans and S.~Ravi.
\newblock What is safe in transactional memory.
\newblock {\em 4th Workshop on the Theory of Transactional Memory, Madeira,
  Portugal}, 2012.

\bibitem{hammondhytm}
L.~Hammond, V.~Wong, M.~Chen, B.~D. Carlstrom, J.~D. Davis, B.~Hertzberg, M.~K.
  Prabhu, H.~Wijaya, C.~Kozyrakis, and K.~Olukotun.
\newblock Transactional memory coherence and consistency.
\newblock {\em SIGARCH Comput. Archit. News}, 32(2):102--, Mar. 2004.

\bibitem{HLR10}
T.~Harris, J.~R. Larus, and R.~Rajwar.
\newblock {\em Transactional Memory, 2nd edition}.
\newblock Synthesis Lectures on Computer Architecture. Morgan {\&} Claypool
  Publishers, 2010.

\bibitem{harris-set}
T.~L. Harris.
\newblock A pragmatic implementation of non-blocking linked-lists.
\newblock In {\em DISC}, pages 300--314, 2001.

\bibitem{HHL+05}
S.~Heller, M.~Herlihy, V.~Luchangco, M.~Moir, W.~N. Scherer, and N.~Shavit.
\newblock A lazy concurrent list-based set algorithm.
\newblock In {\em OPODIS}, pages 3--16, 2006.

\bibitem{HendlerIST10}
D.~Hendler, I.~Incze, N.~Shavit, and M.~Tzafrir.
\newblock Flat combining and the synchronization-parallelism tradeoff.
\newblock In {\em SPAA}, pages 355--364, 2010.

\bibitem{hennessy-patterson}
J.~L. Hennessy and D.~A. Patterson.
\newblock {\em Computer Architecture: A Quantitative Approach}.
\newblock Morgan Kaufmann Publishers Inc., San Francisco, CA, USA, 3 edition,
  2003.

\bibitem{Her90}
M.~Herlihy.
\newblock Apologizing versus asking permission: optimistic concurrency control
  for abstract data types.
\newblock {\em ACM Trans. Database Syst.}, 15(1):96--124, 1990.

\bibitem{Her91}
M.~Herlihy.
\newblock Wait-free synchronization.
\newblock {\em ACM Trans. Prog. Lang. Syst.}, 13(1):123--149, 1991.

\bibitem{HK08}
M.~Herlihy and E.~Koskinen.
\newblock Transactional boosting: A methodology for highly-concurrent
  transactional objects.
\newblock In {\em PPoPP}, New York, NY, USA, 2008. ACM.

\bibitem{HK-commutativity}
M.~Herlihy and E.~Koskinen.
\newblock Composable transactional objects: A position paper.
\newblock In Z.~Shao, editor, {\em Programming Languages and Systems}, volume
  8410 of {\em Lecture Notes in Computer Science}, pages 1--7. Springer Berlin
  Heidelberg, 2014.

\bibitem{HLM03}
M.~Herlihy, V.~Luchangco, and M.~Moir.
\newblock Obstruction-free synchronization: Double-ended queues as an example.
\newblock In {\em ICDCS}, pages 522--529, 2003.

\bibitem{HLM+03}
M.~Herlihy, V.~Luchangco, M.~Moir, and W.~N. Scherer, III.
\newblock Software transactional memory for dynamic-sized data structures.
\newblock In {\em Proceedings of the Twenty-second Annual Symposium on
  Principles of Distributed Computing}, PODC '03, pages 92--101, New York, NY,
  USA, 2003. ACM.

\bibitem{HM93}
M.~Herlihy and J.~E.~B. Moss.
\newblock Transactional memory: architectural support for lock-free data
  structures.
\newblock In {\em ISCA}, pages 289--300, 1993.

\bibitem{HS08-book}
M.~Herlihy and N.~Shavit.
\newblock {\em The art of multiprocessor programming}.
\newblock Morgan Kaufmann, 2008.

\bibitem{HS11-progress}
M.~Herlihy and N.~Shavit.
\newblock On the nature of progress.
\newblock In {\em OPODIS}, pages 313--328, 2011.

\bibitem{HW90}
M.~Herlihy and J.~M. Wing.
\newblock Linearizability: A correctness condition for concurrent objects.
\newblock {\em ACM Trans. Program. Lang. Syst.}, 12(3):463--492, 1990.

\bibitem{lee-thesis}
L.~Hyonho.
\newblock \href{http://www.cs.toronto.edu/pub/hlee/thesis.ps} {Local-spin
  mutual exclusion algorithms on the DSM model using fetch-and-store objects}.
\newblock 2003.

\bibitem{damien-vw-jv}
D.~Imbs and M.~Raynal.
\newblock Virtual world consistency: A condition for {STM} systems (with a
  versatile protocol with invisible read operations).
\newblock {\em Theor. Comput. Sci.}, 444, July 2012.

\bibitem{israeli-disjoint}
A.~Israeli and L.~Rappoport.
\newblock Disjoint-access-parallel implementations of strong shared memory
  primitives.
\newblock In {\em PODC}, pages 151--160, 1994.

\bibitem{konig}
D.~K\"{o}nig.
\newblock {\em Theorie der Endlichen und Unendlichen Graphen: Kombinatorische
  Topologie der Streckenkomplexe}.
\newblock Akad. Verlag. 1936.

\bibitem{kumarhytm}
S.~Kumar, M.~Chu, C.~J. Hughes, P.~Kundu, and A.~Nguyen.
\newblock Hybrid transactional memory.
\newblock In {\em Proceedings of the Eleventh ACM SIGPLAN Symposium on
  Principles and Practice of Parallel Programming}, PPoPP '06, pages 209--220,
  New York, NY, USA, 2006. ACM.

\bibitem{KP79}
H.~T. Kung and C.~H. Papadimitriou.
\newblock An optimality theory of concurrency control for databases.
\newblock In {\em SIGMOD}, pages 116--126, 1979.

\bibitem{KR11}
P.~Kuznetsov and S.~Ravi.
\newblock On the cost of concurrency in transactional memory.
\newblock In {\em OPODIS}, pages 112--127, 2011.

\bibitem{KR11-TR}
P.~Kuznetsov and S.~Ravi.
\newblock On the cost of concurrency in transactional memory.
\newblock {\em CoRR}, abs/1103.1302, 2011.

\bibitem{TM-WF14}
P.~Kuznetsov and S.~Ravi.
\newblock On partial wait-freedom in transactional memory.
\newblock {\em CoRR}, abs/1407.6876, 2014.

\bibitem{WF14-icdcn}
P.~Kuznetsov and S.~Ravi.
\newblock On partial wait-freedom in transactional memory.
\newblock In {\em Proceedings of the 2015 International Conference on
  Distributed Computing and Networking, {ICDCN} 2015, Goa, India, January 4-7,
  2015}, page~10, 2015.

\bibitem{prog15}
P.~Kuznetsov and S.~Ravi.
\newblock Progressive transactional memory in time and space.
\newblock {\em CoRR}, abs/1502.04908, 2015.

\bibitem{prog15-pact}
P.~Kuznetsov and S.~Ravi.
\newblock Progressive transactional memory in time and space.
\newblock {\em CoRR}, abs/1502.04908, 2015.
\newblock To appear in 13th International Conference on Parallel Computing
  Technologies, Russia.

\bibitem{oftm15}
P.~Kuznetsov and S.~Ravi.
\newblock Why transactional memory should not be obstruction-free.
\newblock {\em CoRR}, abs/1502.02725, 2015.

\bibitem{oftm15disc}
P.~Kuznetsov and S.~Ravi.
\newblock Why transactional memory should not be obstruction-free.
\newblock {\em CoRR}, abs/1502.02725, 2015.
\newblock To appear in 29th International Symposium on Distributed Computing
  (DISC'15), Japan.

\bibitem{TMS-WTTM}
M.~Lesani, V.~Luchangco, and M.~Moir.
\newblock Putting opacity in its place.
\newblock In {\em WTTM}, 2012.

\bibitem{phasedtm}
Y.~Lev, M.~Moir, and D.~Nussbaum.
\newblock Phtm: Phased transactional memory.
\newblock In {\em In Workshop on Transactional Computing (Transact), 2007.
  research.sun.com/scalable/pubs/ TRANSACT2007PhTM.pdf}.

\bibitem{Lyn96}
N.~A. Lynch.
\newblock {\em Distributed Algorithms}.
\newblock Morgan Kaufmann, 1996.

\bibitem{astm}
V.~J. Marathe, W.~N.~S. Iii, and M.~L. Scott.
\newblock Adaptive software transactional memory.
\newblock In {\em In Proc. of the 19th Intl. Symp. on Distributed Computing},
  pages 354--368, 2005.

\bibitem{MS13}
A.~Matveev and N.~Shavit.
\newblock Reduced hardware transactions: a new approach to hybrid transactional
  memory.
\newblock In {\em Proceedings of the 25th ACM symposium on Parallelism in
  algorithms and architectures}, pages 11--22. ACM, 2013.

\bibitem{McKenney10}
P.~E. McKenney.
\newblock Memory barriers: a hardware view for software hackers.
\newblock Linux Technology Center, IBM Beaverton, June 2010.

\bibitem{menon-sgl}
V.~Menon, S.~Balensiefer, T.~Shpeisman, A.-R. Adl-Tabatabai, R.~L. Hudson,
  B.~Saha, and A.~Welc.
\newblock Single global lock semantics in a weakly atomic stm.
\newblock {\em SIGPLAN Not.}, 43(5):15--26, May 2008.

\bibitem{michael-set}
M.~M. Michael.
\newblock High performance dynamic lock-free hash tables and list-based sets.
\newblock In {\em SPAA}, pages 73--82, 2002.

\bibitem{MS96}
M.~M. Michael and M.~L. Scott.
\newblock Simple, fast, and practical non-blocking and blocking concurrent
  queue algorithms.
\newblock In {\em PODC}, pages 267--275, 1996.

\bibitem{bluegene}
M.~Ohmacht.
\newblock {Memory Speculation of the Blue Gene/Q Compute Chip}, 2011.
\newblock \url{http://wands.cse.lehigh.edu/IBM_BQC_PACT2011.ppt}.

\bibitem{OL82}
S.~S. Owicki and L.~Lamport.
\newblock Proving liveness properties of concurrent programs.
\newblock {\em ACM Trans. Program. Lang. Syst.}, 4(3):455--495, 1982.

\bibitem{Pap79-serial}
C.~H. Papadimitriou.
\newblock The serializability of concurrent database updates.
\newblock {\em J. ACM}, 26:631--653, 1979.

\bibitem{PFK10}
D.~Perelman, R.~Fan, and I.~Keidar.
\newblock On maintaining multiple versions in {STM}.
\newblock In {\em PODC}, pages 16--25, 2010.

\bibitem{Rei12}
J.~Reinders.
\newblock {Transactional Synchronization in Haswell}, 2012.
\newblock
  \url{http://software.intel.com/en-us/blogs/2012/02/07/transactional-synchronization-in-haswell/}.

\bibitem{riegel-thesis}
T.~Riegel.
\newblock
  \href{http://www.qucosa.de/fileadmin/data/qucosa/documents/11559/Riegel_Diss_final.pdf}{Software
  Transactional Memory Building Blocks}.
\newblock 2013.

\bibitem{hynorecriegel}
T.~Riegel, P.~Marlier, M.~Nowack, P.~Felber, and C.~Fetzer.
\newblock Optimizing hybrid transactional memory: The importance of
  nonspeculative operations.
\newblock In {\em Proceedings of the 23rd ACM Symposium on Parallelism in
  Algorithms and Architectures}, pages 53--64. ACM, 2011.

\bibitem{HASTM2}
B.~Saha, A.-R. Adl-Tabatabai, and Q.~Jacobson.
\newblock Architectural support for software transactional memory.
\newblock In {\em Proceedings of the 39th Annual IEEE/ACM International
  Symposium on Microarchitecture}, MICRO 39, pages 185--196, Washington, DC,
  USA, 2006. IEEE Computer Society.

\bibitem{dstm-contention}
W.~N. Scherer, III and M.~L. Scott.
\newblock Advanced contention management for dynamic software transactional
  memory.
\newblock In {\em Proceedings of the Twenty-fourth Annual ACM Symposium on
  Principles of Distributed Computing}, PODC '05, pages 240--248, New York, NY,
  USA, 2005. ACM.

\bibitem{scott-book}
M.~L. Scott.
\newblock {\em Shared-memory Synchronization, Synthesis Lectures on Distributed
  Computing Theory}.
\newblock Morgan and Claypool, 2013.

\bibitem{ST95}
N.~Shavit and D.~Touitou.
\newblock Software transactional memory.
\newblock In {\em PODC}, pages 204--213, 1995.

\bibitem{spear-priv}
M.~F. Spear, V.~J. Marathe, L.~Dalessandro, and M.~L. Scott.
\newblock Privatization techniques for software transactional memory.
\newblock In {\em Proceedings of the Twenty-sixth Annual ACM Symposium on
  Principles of Distributed Computing}, PODC '07, pages 338--339, New York, NY,
  USA, 2007. ACM.

\bibitem{HASTM1}
M.~F. Spear, A.~Shriraman, L.~Dalessandro, S.~Dwarkadas, and M.~L. Scott.
\newblock Nonblocking transactions without indirection using alert-on-update.
\newblock In {\em Proceedings of the Nineteenth Annual ACM Symposium on
  Parallel Algorithms and Architectures}, SPAA '07, pages 210--220, New York,
  NY, USA, 2007. ACM.

\bibitem{nztm}
F.~Tabba, M.~Moir, J.~R. Goodman, A.~W. Hay, and C.~Wang.
\newblock Nztm: Nonblocking zero-indirection transactional memory.
\newblock In {\em Proceedings of the Twenty-first Annual Symposium on
  Parallelism in Algorithms and Architectures}, SPAA '09, pages 204--213, New
  York, NY, USA, 2009. ACM.

\bibitem{Gadi-Bakery}
G.~Taubenfeld.
\newblock The black-white bakery algorithm and related bounded-space, adaptive,
  local-spinning and fifo algorithms.
\newblock In {\em DISC '04: Proceedings of the 23rd International Symposum on
  Distributed Computing}, 2004.

\bibitem{WTTM2}
P.~K. Vincent~Gramoli and S.~Ravi.
\newblock Sharing a sequential data structure: correctness definition and
  concurrency analysis.
\newblock {\em 4th Workshop on the Theory of Transactional Memory, Madeira,
  Portugal}, 2012.

\bibitem{Wei88}
W.~E. Weihl.
\newblock Commutativity-based concurrency control for abstract data types.
\newblock {\em IEEE Trans. Comput.}, 37(12):1488--1505, 1988.

\bibitem{Wei86}
G.~Weikum.
\newblock A theoretical foundation of multi-level concurrency control.
\newblock In {\em PODS}, pages 31--43, 1986.

\bibitem{WV02-book}
G.~Weikum and G.~Vossen.
\newblock {\em Transactional Information Systems: Theory, Algorithms, and the
  Practice of Concurrency Control and Recovery}.
\newblock Morgan Kaufmann, 2002.

\bibitem{Yan84}
M.~Yannakakis.
\newblock Serializability by locking.
\newblock {\em J. ACM}, 31(2):227--244, 1984.

\end{thebibliography}

\chapter*{Papers}
\nobibliography*

The content of the thesis is based on the following tech reports and publications.

\paragraph{Tech reports}

\bibverse{prog15}
\bibverse{oftm15}
\bibverse{hytm14}
\bibverse{opti14}
\bibverse{TM-WF14}
\bibverse{icdcs-TR}
\bibverse{KR11-TR}

\paragraph{Publications}
\bibverse{prog15-pact}
\bibverse{oftm15disc}
\bibverse{hytm14disc}
\bibverse{safety-tm14}
\bibverse{WF14-icdcn}
\bibverse{icdcs-opacity}
\bibverse{PODCKRV12}
\bibverse{KR11}

\paragraph{Workshop papers}

\bibverse{WTTM3}
\bibverse{WTTM2}
\bibverse{WTTM1}

Concurrently, I was also involved in the following paper whose contents are not included in the thesis.

\bibverse{optimistic-list15}
\bibverse{optimistic-list15-disc}

%
%

\end{document}